\documentclass[12pt, a4paper,twoside]{upf_uab}

\usepackage[english,catalan]{babel}

\usepackage[cam,a4,center,frame]{crop}

\usepackage{graphicx}

\usepackage{times}

\usepackage{makeidx}
\makeindex

\addto\captionscatalan
  {\renewcommand{\contentsname}{\Large \sffamily Sumari}}

\usepackage{hyperref}
\hypersetup{
 linktocpage,
}
\usepackage{tikz}
\usetikzlibrary{shapes}
\usetikzlibrary{positioning}
\usepackage{bbm}
\usepackage{bm}
\usepackage{pgfplots}
\usepackage{amssymb}   
\usepackage{stfloats}
\usepackage{mathtools}
\usepackage{accents}
\usepackage{booktabs} 
\usepackage{amsmath}
\usepackage{pgfplots}
\usepackage{colortbl}
\usepackage{tikz}
\usetikzlibrary{patterns}
\usepackage{amsthm}
\usepackage{lipsum, color}
\usepackage{amsfonts}
 \usepackage{cite}
\usepackage{blindtext}
\usepackage{cleveref}
\usepackage{cuted}
\usepackage{accents}
 \usepackage{booktabs}
\usepackage[T1]{fontenc}
\usepackage{lmodern}
\usepackage{amsfonts}
\usepackage[toc,page]{appendix}

\usepackage{bbold}
\usepackage{amsmath}
\usepackage{amssymb}
\usepackage{amsthm}
\usepackage{mathrsfs}
\usepackage{MnSymbol}
\usepackage{braket}
\usepackage{physics}
\usepackage{enumerate}
\usepackage{mathtools}
\usepackage{comment}

\renewcommand{\eqref}[1]{\ref{#1}} 
\renewcommand{\Tr}{{\operatorname{Tr}\,}}
\renewcommand{\Pr}{{\operatorname{Pr}}}
\renewcommand{\dim}{{\operatorname{dim}}}
\newcommand{\id}{{\operatorname{id}}}
\newcommand{\1}{\mathbb{1}}

\newcommand{\proj}[1]{|#1\rangle\!\langle #1|}
\newcommand{\cC}{{\mathcal{C}}}
\newcommand{\cD}{{\mathcal{D}}}
\newcommand{\cE}{{\mathcal{E}}}
\newcommand{\cT}{{\mathcal{T}}}
\newcommand{\cX}{{\mathcal{X}}}
\newcommand{\cS}{{\mathcal{S}}}
\newcommand{\cR}{{\mathcal{R}}}
\newcommand{\KI}{{\text{KI}}}
\newcommand{\aw}[1]{{#1}}             

\newcommand{\und}[1]{\underline{#1}}

\theoremstyle{plain}

\newtheorem*{remark*}{Remark}   
\renewcommand\qedsymbol{$\blacksquare$}
\newenvironment{proof-of}[1][{\hspace{-\blank}}]{{\medskip\noindent\textit{Proof~{#1}.\ }}}{\hfill\qedsymbol}

\newtheorem{theorem}{Theorem}
\numberwithin{theorem}{chapter}
\newtheorem{remark}{Remark}
\numberwithin{remark}{chapter}
\newtheorem{corollary}{Corollary}[theorem]
\numberwithin{corollary}{chapter}
\newtheorem{lemma}{Lemma}
\numberwithin{lemma}{chapter}

\numberwithin{proposition}{chapter}
\newtheorem{definition}{Definition}
\numberwithin{definition}{chapter}
\renewcommand{\Pr}{{\operatorname{Pr}}}
\renewcommand{\dim}{{\operatorname{dim}}}

\DeclarePairedDelimiter\floor{\lfloor}{\rfloor}

\newcommand{\nc}{\newcommand}
\nc{\rnc}{\renewcommand}
\nc{\avg}[1]{\langle#1\rangle}
\nc{\Rank}{\operatorname{Rank}}
\nc{\smfrac}[2]{\mbox{$\frac{#1}{#2}$}}
\nc{\ox}{\otimes}
\nc{\dg}{\dagger}
\nc{\dn}{\downarrow}
\nc{\cA}{{\cal A}}
\nc{\cB}{{\cal B}}
\nc{\cF}{{\cal F}}
\nc{\cG}{{\cal G}}
\nc{\cH}{{\cal H}}
\nc{\cI}{{\cal I}}
\nc{\cJ}{{\cal J}}
\nc{\cK}{{\cal K}}
\nc{\cL}{{\cal L}}
\nc{\cM}{{\cal M}}
\nc{\cN}{{\cal N}}
\nc{\cO}{{\cal O}}
\nc{\cP}{{\cal P}}
\nc{\cQ}{{\cal Q}}
\nc{\cY}{{\cal Y}}
\nc{\cZ}{{\cal Z}}
\nc{\csupp}{{\operatorname{csupp}}}
\nc{\qsupp}{{\operatorname{qsupp}}}
\nc{\rar}{\rightarrow}
\nc{\lrar}{\longrightarrow}
\nc{\polylog}{{\operatorname{polylog}}}
\nc{\wt}{{\operatorname{wt}}}

\nc{\RR}{{{\mathbb R}}}
\nc{\CC}{{{\mathbb C}}}
\nc{\FF}{{{\mathbb F}}}
\nc{\NN}{{{\mathbb N}}}
\nc{\ZZ}{{{\mathbb Z}}}
\nc{\PP}{{{\mathbb P}}}
\nc{\QQ}{{{\mathbb Q}}}
\nc{\UU}{{{\mathbb U}}}
\nc{\EE}{{{\mathbb E}}}

\nc{\Hom}[2]{\mbox{Hom}(\CC^{#1},\CC^{#2})}
\nc{\rU}{\mbox{U}}

\nc{\ob}[1]{#1}

\nc{\SEP}{{\text{SEP}}}
\nc{\NS}{{\text{NS}}}
\nc{\LOCC}{{\text{LOCC}}}
\nc{\PPT}{{\text{PPT}}}
\nc{\EXT}{{\text{EXT}}}
\nc{\Sym}{{\operatorname{Sym}}}

\nc{\ERLO}{{E_{\text{r,LO}}}}
\nc{\ERLOCC}{{E_{\text{r,LOCC}}}}
\nc{\ERPPT}{{E_{\text{r,PPT}}}}
\nc{\ERLOCCinfty}{{E^{\infty}_{\text{r,LOCC}}}}
\nc{\Aram}{{\operatorname{\sf A}}}

\begin{document}
\selectlanguage{english}

    




\frontmatter
\maketitle      
\cleardoublepage

\newpage
\selectlanguage{english}
\section*{\Large \sffamily Abstract}

This thesis addresses problems in the field of quantum information theory, specifically, quantum Shannon theory. The first part of the thesis is opened with concrete definitions of general quantum source models and their compression, and each subsequent chapter addresses the compression of a specific source model as a special case of the initially defined general models. First, we find the optimal compression rate of a general mixed state source which includes as special cases all the previously studied models such as Schumacher’s pure and ensemble sources and other mixed state ensemble models. For an interpolation between the visible and blind Schumacher’s ensemble model, we find the optimal compression rate region for the entanglement and quantum rates. Later, we comprehensively study the classical-quantum variation of the celebrated Slepian-Wolf problem and find the optimal rates considering per-copy fidelity; with block fidelity we find single letter achievable and converse bounds which match up to continuity of a function appearing in the bounds. The first part of the thesis is closed with a chapter on the ensemble model of quantum state redistribution for which we find the optimal compression rate considering per-copy fidelity and single-letter achievable and converse bounds matching up to continuity of a function which appears in the bounds.

The second part of the thesis revolves around information theoretical perspective of quantum thermodynamics. We start with a resource theory point of view of a quantum system with multiple non-commuting charges where the objects and allowed operations are thermodynamically meaningful; using tools from quantum Shannon theory we classify the objects and find explicit quantum operations which map the objects of the same class to one another. Subsequently, we apply this resource theory framework to study a traditional thermodynamics setup with multiple non-commuting conserved quantities consisting of a main system, a thermal bath and batteries to store various conserved quantities of the system. We state the laws of the thermodynamics for this system, and show that a purely quantum effect happens in some transformations of the system, that is, some transformations are feasible only if there are quantum correlations between the final state of the system and the thermal bath.

\newpage
\selectlanguage{catalan}
\section*{\Large \sffamily  Resum}

Aquesta tesi aborda problemes en el camp de la teoria de la informaci\'{o} qu\`{a}ntica, espec\'{i}ficament, la teoria qu\`{a}ntica de Shannon. La primera part de la tesi comen\c{c}a amb definicions concretes de models de fonts qu\`{a}ntiques generals i la seva compressi\'{o}, i cada cap\'{i}tol seg\"{u}ent aborda la compressi\'{o} d'un model de font espec\'{i}fic com a casos especials dels models generals definits inicialment. Primer, trobem la taxa de compressi\'{o} \`{o}ptima d'una font d'estats barreja general que inclou com a casos especials tots els models pr\`{e}viament estudiats, com les fonts pures i de co"lectivitats de Schumacher, i altres models de co"lectiuvitats d'estats barreja. Per a una interpolaci\'{o} entre els models de co"lectivitats visible i cec de Schumacher, trobem la regi\'{o} de compressi\'{o} \`{o}ptima per les taxes d'entrella\c{c}ament i les taxes qu\`{a}ntiques. A continuaci\'{o}, estudiem exhaustivament la variaci\'{o} cl\`{a}ssic-qu\`{a}ntica del fam\'{o}s problema de Slepian-Wolf i trobem les taxes \`{o}ptimes considerant la fidelitat per c\`{o}pia; per la fidelitat de bloc trobem expressions tancades per les fites assolibles i inverses que coincideixen, sota la condici\'{o} de que una funci\'{o} que apareix a les dues fites sigui continua. La primera part de la tesi tanca amb un cap\'{i}tol sobre el model de co"lectivitats per la redistribuci\'{o} d'estats qu\`{a}ntics per al qual trobem la taxa de compressi\'{o} \`{o}ptima considerant la fidelitat per c\`{o}pia i les fites assolibles i inverses, que de nou que coincideixen sota la condici\'{o} de continu\"{i}tat d'una certa funci\'{o}.

La segona part de la tesis gira al voltant de la term\'{o}dinamica qu\`{a}ntica sota de la perspectiva de la teoria de la informaci\'{o}. Comencem amb un punt de vista de la teoria de recursos d'un sistema qu\`{a}ntic amb m\'{u}ltiples c\`{a}rregues que no commuten i amb objectes i operacions permeses que son termodin\`{a}micament significatives; utilitzant eines de la teoria qu\`{a}ntica de Shannon classifiquem els objectes i trobem operacions qu\`{a}ntiques expl\'{i}cites que relacionen els objectes de la mateixa classe entre s\'{i}. Posteriorment, apliquem aquest marc de la teoria de recursos per estudiar una configuraci\'{o} termodin\`{a}mica tradicional amb m\'{u}ltiples quantitats conservades que no commuten que consta d'un sistema principal, un reservori cal\`{o}ric i bateries per emmagatzemar diverses quantitats conservades del sistema. Enunciem les lleis de la termodin\`{a}mica per a aquest sistema, i mostrem que un efecte purament qu\`{a}ntic t\'{e} lloc en algunes transformacions del sistema, \'{e}s a dir, algunes transformacions nom\'{e}s s\'{o}n factibles si hi ha correlacions qu\`{a}ntiques entre l’estat final del sistema i del reservori cal\`{o}ric.

\newpage
\selectlanguage{catalan}
\section*{\Large \sffamily  Resumen}

Esta tesis aborda problemas en el campo de la teor\'{i}a de la informaci\'{o}n cu\'{a}ntica, espec\'{i}ficamente, la teor\'{i}a cu\'{a}ntica de Shannon. La primera parte de la tesis comienza con definiciones concretas de modelos de fuentes cu\'{a}nticas generales y su compresi\'{o}n, y cada cap\'{i}tulo subsiguiente aborda la compresi\'{o}n de un modelo de fuente espec\'{i}fico como casos especiales de los modelos generales definidos inicialmente. Primero, encontramos la tasa de compresi\'{o}n \'{o}ptima de una fuente de estado mixto general que incluye como casos especiales todos los modelos previamente estudiados, como las fuentes pura y colectiva de Schumacher, y otros modelos colectivos de estado mixto. Para una interpolaci\'{o}n entre el modelo colectivo visible y ciego de Schumacher, encontramos la regi\'{o}n de tasa de compresi\'{o}n \'{o}ptima para el entrelazamiento y las tasas cu\'{a}nticas. A continuaci\'{o}n, estudiamos exhaustivamente la variaci\'{o}n cl\'{a}sico-cu\'{a}ntica del c\'{e}lebre problema de Slepian-Wolf y encontramos las tasas \'{o}ptimas considerando la fidelidad por copia; con la fidelidad de bloque encontramos l\'{i}mites alcanzables e inversos que coinciden con la continuidad de una funci\'{o}n que aparece en los l\'{i}mites. La primera parte de la tesis cierra con un cap\'{i}tulo sobre el modelo colectivo de redistribuci\'{o}n de estado cu\'{a}ntico para el cual encontramos la tasa de compresi\'{o}n \'{o}ptima considerando la fidelidad por copia y los l\'{i}mites alcanzables e inversos que coinciden con la continuidad de una funci\'{o}n que aparece en los l\'{i}mites.

La segunda parte de la tesis gira en torno a la perspectiva te\'{o}rica de la informaci\'{o}n de la termodin\'{a}mica cu\'{a}ntica. Comenzamos con un punto de vista de la teor\'{i}a de recursos de un sistema cu\'{a}ntico con m\'{u}ltiples cargas no conmutables con objetos y operaciones permitidas que son termodin\'{a}micamente significativas; usando herramientas de la teor\'{i}a cu\'{a}ntica de Shannon clasificamos los objetos y encontramos operaciones cu\'{a}nticas expl\'{i}citas que mapean los objetos de la misma clase entre s\'{i}. Posteriormente, aplicamos este marco de la teor\'{i}a de recursos para estudiar una configuraci\'{o}n termodin\'{a}mica tradicional con m\'{u}ltiples cantidades no conmutables compuesta por un sistema principal, un reservorio cal\'{o}rico y bater\'{i}as para almacenar varias cantidades conservadas del sistema. Enunciamos las leyes de la termodin\'{a}mica para este sistema, y mostramos que ocurre un efecto puramente cu\'{a}ntico en algunas transformaciones del sistema, es decir, algunas transformaciones solo son factibles si existen correlaciones cu\'{a}nticas entre el estado final del sistema y del reservorio cal\'{o}rico.

\selectlanguage{english}


\cleardoublepage
\noindent {\Large   \textbf{Acknowledgements}}

\vspace{1cm}

\noindent I express my sincere gratitude to my supervisors Andreas Winter and Maciej Lewenstein for their continuous support and care in any aspect that I could possibly ask for. I started with studying various problems in quantum Shannon theory, and I enjoyed and learned from immense knowledge of Andreas Winter who gave me unlimited freedom and time to submerge myself in problems and wrap my head around them; I have been fascinated to see his perspective, scientific discipline and how he approaches science in general, and I feel privileged to have him as my mentor and  role model both in academic and personal life. He later introduced me to Maciej Lewenstein in ICFO where I have learned from his profound knowledge in physics and how to make sense of complicated mathematical notions through physical interpretations without obsessing about equations. I cannot thank Maciej enough for his kindness, support and also the freedom and time that he gave me.

I am honored and delighted to defend my thesis in front of
the experts of the field 
John Calsamiglia, Patrick Hayden and  Micha{\l} Horodecki, whose scientific works have
been a source of inspiration and guidance to me.

I have learned a lot and enjoyed discussing problems during my academic visits that I have had, specifically, I would like to thank Paul Skrzypczyk and Tony Short in the University of Bristol, Nilanjana Datta in Cambridge and Masahito Hayashi in Peng Cheng Lab.

Apart from scientific perspective, I got to enjoy my time in ICFO as a Phd student which is a great institute for anyone seeking professional academic training thanks to the organization, generosity and supportive environment of the institute. I am indebted to both academic and administrative staff. Moreover, thanks to the social environment there, I have made great friendships and enjoyed the fun specifically during various annual events. 

I had great time and experience in UAB in our own quantum information group (Giq), where it is my academic home, thanks to the supportive and encouraging atmosphere that has been fostered here. I am indebted to Anna, Emili, John, Ramon and other Giq members for all their support and care; I have enjoyed the seminars, the time we have spent together during lunch and our annual Cal\c{c}otadas.

I could never accomplish what I have accomplished so far without my background and the training that I have had in great schools in Iran, I specifically thank my teachers in Sharif university of technology where I was exposed to information theory and I got fascinated for the first time about quantum information. I am specifically grateful to Salman Beigi and Amin Gohari for their teaching, advice, scientific manner, support and recommendations.       

Despite the challenges I have faced, I got to enjoy living in Barcelona thanks to the beauty of this city and life-long friendships that I have made here.
I am thankful to my friends Arezou, Hara, Lisa, Marzieh, Susanna and other friends in Giq and ICFO, specifically Roger for translating the summary of my thesis to Catalan and Spanish. 

During this period, I have gone back to my family  all the time. In particular, 
I am grateful and indebted to my parents for their love, support and encouragement and planting the initial seeds of love and passion for science.  
I have had their 
continuous support throughout my life and especially in my academic endeavours.

Finally, I cannot express enough my happiness and gratitude to have Farzin as my husband and friend. We started the Phd at the same time, and despite ups and downs of his own path, he never failed to support and encourage me; he kept my attitude positive and optimistic to overcome challenges of this path even when he was facing his own hurdles.  I am thankful to his kindness and care and all the discussions we had regarding quantum information and the Phd life.

\cleardoublepage


\noindent This thesis has been supported by the Spanish MINECO (projects FIS2016-86681-P, 
FISICATEAMO FIS2016-79508-P, SEVERO OCHOA No. SEV-2015-0522, FPI 
and PID2019-107609GB-I00/AEI/10.13039/501100011033), the FEDER funds,  the Generalitat de Catalunya(project 2017-SGR-1127, 
2017-SGR-1341 and CERCA/Program), ERC AdG OSYRIS, EU FETPRO QUIC, 
and the National Science Centre, Poland-Symfonia grant no.2016/20/
W/ST4/00314.

\tableofcontents

\mainmatter    
\chapter{Introduction} 

\section{Background and motivation}

Information theory studies the transmission, processing, extraction, and utilization of information. 
The notion of classical information was first introduced
by Shannon \cite{Shannon1948}, 
who defined it operationally, as the minimum number of bits needed to communicate the message
produced by a statistical source. This gave meaning to the Shannon entropy $H(X)$ of a source producing a random variable $X$. 
The amount of information that two random variables
$X$ and $Y$ have in common was given a meaning through
the mutual information $I(X:Y)$. Operationally it is the rate of communication possible
through a noisy channel taking
$X$ to $Y$.

Quantum Shannon theory is a more general field which studies information on physical systems governed by the rules of 
quantum mechanics, therefore encompasses classical information as sub-field, and
was mathematically founded   
by Holevo in 1973 \cite{Holevo1973}
to study the transmission of information 
over quantum channels following the earliest understanding of the connection between quantum physics and information theory \cite{Gordon64,Levitin69,Forney63,Stratonovich66}.

Surprisingly, von Neumann entropy, which is a generalization of Shannon entropy, was formulated  before Shannon entropy  in the context of thermodynamics and statistical mechanics, and it was not contemplated  to convey  informational interpretation. Despite this fact and Holevo's study of classical information on quantum systems \cite{Holevo1973,Holevo1979}, the concept of quantum information 
was obscure till 1995, when Schumacher
showed that the von Neumann entropy
has the operational interpretation of the number of \emph{qubits}
needed to transmit quantum states emitted by a statistical  source \cite{Schumacher1995}.

After Schumacher's quantitative notion of quantum information, i.e. qubit, and understanding its complementary nature to classical information, quantum Shannon theory has been further established in the last three decades by fundamental discoveries from source and channel coding to quantum cryptography, quantum error-correcting, measures of entanglement and so on 
\cite{decoherence-Shor-1995,err-correction-Schumacher-1996,Schumi1996,err-codes1997,Barnum1998,privacy-coherence-1998,E-assisted-C-capacity1999,Lo-Popescu-1999,fidelities-capacities-2000,On-C-capacity-Holevo-2002,Schumacher2002,Bennett2002,Keydistill-DW-2005,Devetak-capacity-2005}.

In particular, the notion of a quantum source 
as a quantum state together with correlations with a reference system and its compression led to the discovery of operational meaning for quantum quantities such as quantum conditional entropy, which as opposed to its classical counterpart can obtain negative values. In this source compression task with side information, which is called state merging, the negative values of 
conditional entropy imply that the entanglement is generated after the compression is accomplished, and it can be used as a resource for future communications \cite{Horodecki2007,SM_nature}.   

Other quantum source compression problems 
such as quantum state redistribution and visible compression of mixed states gave operational meaning to quantum conditional mutual information and regularized entanglement of purification \cite{Devetak2008_2,Yard2009,visible_Hayashi}, respectively, and they have been used successfully as sub-protocols to accomplish tasks other than data compression \cite{QRST2014}. 
Various source models and their compression have been considered throughout these years 
and each source appeared to be a distinct case with a unique compression behavior \cite{Devetak2003,Winter1999,KI2001,visible_Hayashi,Horodecki2007,Devetak2008_2,Yard2009}, and the compression of many other source models has been left open \cite{Winter1999,Ahn2006}.
These open questions and the lack of a source model, which can unify all these seemingly distinct models, is the underlying motivation for the first part of this thesis which focuses on the compression of quantum sources. 
We specifically solve the Schumacher's compression problem when the overall state together with the reference is a general mixed state. When there are side information systems, a general reference system appears to be  hard to tackle, therefore we attack compression problems with classical references or so called ensemble sources. 

Understanding compression and capacity problems apart from finding fundamental limits on the amount of communication and storage rates, has developed tools and quantitative notions, e.g typical subspaces and entropic quantities, which has been successfully used to deal with and interpret other 
quantum effects such as quantum thermodynamics and quantum coherence \cite{Brandao2013,Weilenmann2016,Winter-Dong-2016}.
In particular, the innate relationship between information theory and thermodynamics has proved that integrated ideas from both fields are fruitful \cite{Jaynes_2_1957,Jaynes1957,Brillouin62,Jaynes82}.
This has been the motivation for the second part of this thesis which focuses on quantum thermodynamics, where we consider a general framework with multiple conserved quantities and apply information theoretic tools to construct charge conserving operations. These explicit operations are extremely helpful to study traditional thermodynamics settings and laws.

Perhaps the most up-to-date and comprehensive review of the fast-growing field of quantum thermodynamics is contained in the collection of essays
in \cite{Thermo_Quantum_Regime}.
Still, some very fundamental questions concerning quantum thermodynamics have been answered in this thesis.
%
For a non-specialist, 
these questions can be formulated as follows. Normally in both classical and quantum thermodynamics one deals with large systems interacting with an even larger bath. In addition to energy, the system maybe characterized by many macroscopic conserved (on average) quantities, called here charges like total electric charge, total dipole moment, angular momentum, magnetization, total spin components etc. In the quantum case, these quantities may correspond to non-commuting operators. How come that with repeated measurements on the system prepared in the same state, we obtain well defined average values of this charges? The repeated measurements of equally prepared systems can be mathematically treated by considering tensor product states of many copies of the systems. This mathematical construction is used in the thesis to define thermodynamically allowed transformation, which have to conserve all average values of the charges. To any quantum state, we associate a vector with entries of the expected charge values and entropy of that state. The set of all these vectors forms the phase diagram of the system, and show that it characterizes the equivalence classes of states under thermodynamically allowed transformations, which are proven rigorously to correspond to asymptotic unitary transformations that approximately conserve the charges. 

Our theory provides a general theoretical framework, but leads also to predictions of very concrete effects. In particular, we estimate how large an asymptotically large 
bath is necessary  to attain the second law of thermodynamics, and permit a specified 
work transformation of a given system. In some situations, the necessary bath extension is relatively small, and then quantum setting requires an extended phase diagram exhibiting negative entropies. This corresponds to the purely quantum effect that at the end of the process, system and bath are entangled. Obviously, such processes are impossible classically! For large thermal bath, thermodynamically allowed transformation leave the system and the bath uncorrelated. In such case, the heat capacity of the bath becomes a function of how tightly the second law is attained.

\section{The structure of the thesis}

The reminder of this chapter is dedicated to introducing some notation and preliminary material, which are prerequisite for the subsequent chapters. In summary as mentioned above, the thesis is based on two main themes: part I and part II revolving around quantum source compression and quantum    thermodynamics, respectively.

As for the source compression part, we start with chapter~\ref{chap:source-coding} where we first expand on the notion of a quantum source and continue with a rigorous definition
of asymptotic source compression task which 
encompasses as special cases all the reviewed
compression problems in the context of asymptotic quantum source compression and unifies them under a common base. Later in the chapter, we define side information and distributed settings for compressing the information.  As for the resources available for communication, we consider noiseless qubit channel and shared entanglement between the parties. 

In chapter \ref{chap:mixed state}, we consider the most general source where the overall state with the reference system is a general mixed state. This model covers all the previously studied models such as Schumacher's ensemble and pure sources \cite{Schumacher1995} and the ensemble of mixed state source \cite{KI2001}. We find the optimal trade-off between the entanglement and quantum communication rates. The optimal rates are in terms of a decomposition of the source introduced in \cite{Hayden2004} which is a generalization of the well-known decomposition discovered by Koashi and Imoto in \cite{KI2002}.
When there are side information systems or the compression task is distributed, the general models defined in chapter~\ref{chap:source-coding} appears to be very complicated, and even much simpler models
have been left open since the early exploration of the compression problems \cite{Winter1999, Ahn2006,Devetak2003}; therefore in the subsequent chapters we consider special cases where the source states are from an ensemble, that is the reference system is partly classical.

In chapter~\ref{chap: E assisted Schumacher}, we consider an interpolation between visible and blind Schumacher compression, that is the encoder has access to a side information system which can reduce to a classical system with the information about the identity of the states and a trivial system in the visible and blind scenarios, respectively. We find the optimal trade-off between the entanglement and quantum rates which depending on whether the ensemble is reducible or not, the entanglement consumption reduces the quantum rate or does not help it at all. 

Chapter~\ref{chap:cqSW} is about the distributed  compression  of a hybrid classical-quantum source which is an extension of the celebrated Slepian-Wolf problem \cite{Slepian1973}. Two important sub-problems of this distributed compression problem are classical data compression with quantum side information (at the decoder), which is addressed in \cite{Winter1999,Devetak2003},  and quantum data compression with classical side information (at the decoder), which is the main focus of this chapter. For a class of generic sources we show that the compression rate can be strictly larger than the conditional entropy contrary to the fully classical problem of Slepian-Wolf where the rate of the side information case is always governed by the conditional entropy. However, in general the quantum compression rate reduces by a factor of half of the mutual information between the classical variable and the environment system of the encoder.

Chapter~\ref{chap:QSR ensemble} closes
the first part of the thesis where we consider the most general ensemble model of pure states with side information available both at the encoder and decoder side. When the overall state of the parties and the reference system is pure, the problem is known as quantum state redistribution \cite{Devetak2008_2,Yard2009,Oppenheim2008}. 
We find  the optimal quantum compression rate and confirm that preserving correlations with a 
hybrid classical-quantum reference, which is less stringent than preserving the correlations with the purified reference, can lead to strictly smaller quantum rates.  Indeed, this model 
includes as special cases the sources considered in  chapter~\ref{chap: E assisted Schumacher} and chapter~\ref{chap:cqSW}, however, in the former chapter the figure of merit is block fidelity whereas in the last two chapters the optimal rates are obtained by considering per-copy fidelity; considering block fidelity in the last two chapters, we find  upper and lower bounds 
which would match if the corresponding function defining the bounds is continuous.


The second part of the thesis consists of two chapters. In chapter~\ref{chap:resource theory}, we develop a general resource theory with allowed operations which are thermodynamically meaningful. The objects of this resource theory are quantum states and the allowed operations are those asymptotically commuting with a general set of charges associated with the quantum system.  In order to explicitly construct these operations we use tools and notions such as quantum typicality and approximate  microcanonical subspace. 
Later in chapter~\ref{chap:thermo}, we use the developed operations to study a traditional thermodynamics setting with multiple conserved quantities consisting of a work system, a thermal bath and many batteries to store each charge. We extend the notion of charge-entropy diagram to a diagram with conditional entropy to find out
which transformations are feasible and show that some transformations are feasible only if the final states of the work system and the thermal bath are entangled, i.e. a
purely quantum effect enlarges the set of feasible transformations for the work system.

\bigskip

Finally, the last six chapters are essentially based on the following publications and preprints:

\begin{itemize}

\item 
\textbf{Chapter 3:} 

\noindent { \cite{ZK_mixed_state_2019} Z. B. Khanian and A. Winter, “General mixed state quantum data compression with and without entanglement assistance,”  \emph{pre-print (2019)}, arXiv: 1912.08506.}

\cite{ZK_mixed_state_ISIT_2020} Z. B. Khanian and A. Winter, “General mixed state quantum data compression with and without entanglement assistance,” in: \emph{Proc. IEEE Int. Symp. Inf. Theory (ISIT)}, Los Angeles, CA, USA, pp. 1852-1857, June 2020.

\item 
\textbf{Chapter 4:}

 \cite{Schumacher_Assisted_arXiv_Z_2019} Z. B. Khanian and A. Winter, “Entanglement-assisted quantum data
compression,” \emph{preprint (2019)}, arXiv: 1901.06346.  

\cite{ZK_Eassisted_ISIT_2019} Z. B. Khanian and A. Winter, “Entanglement-assisted quantum data compression,” in: \emph{Proc. IEEE Int. Symp. Inf. Theory (ISIT)}, Paris, France, pp. 1147–1151, July 2019.

\item \textbf{Chapter 5:}

\cite{ZK_cqSW_2018} Z. B. Khanian and A. Winter, “Distributed compression of correlated classical-quantum sources or: the price of ignorance,” \emph{IEEE Trans. Inf. Theory}, vol. 66, no. 9,  pp. 5620-5633, Sep 2020. arXiv: 1811.09177.

\cite{ZK_cqSW_ISIT_2019} Z. B. Khanian and A. Winter, “Distributed compression of correlated classical-quantum sources,” in: \emph{Proc. IEEE Int. Symp. Inf. Theory (ISIT)}, Paris, France, pp. 1152-1156, July 2019.

\item \textbf{Chapter 6: }


\cite{QSR_ensemble_full} Z. B. Khanian and A. Winter, “Rate distortion perspective of quantum state redistribution,” \textit{in preparation}.

\cite{ZK_QSR_ensemble_ISIT_2020} Z. B. Khanian and A. Winter, “Quantum state redistribution for ensemble sources,” in: \emph{Proc. IEEE Int. Symp. Inf. Theory (ISIT)}, Los Angeles, CA, USA, pp. 1858-1863, June 2020.

\item \textbf{Chapter 7 and Chapter 8:}

 
\cite{thermo_ZBK_2020} Z. B. Khanian, M.~N. Bera, A. Riera, M. Lewenstein and A. Winter, “Resource theory of heat and work with non-commuting charges: yet another new foundation of thermodynamics,” \emph{preprint (2020)}, arXiv: 2011.08020. 
\end{itemize}

During my Phd, I have also worked on the following thermodynamics project which is not included in this thesis:
 
\begin{itemize}
\item \cite{brl19} M. N. Bera, A. Riera, M. Lewenstein, Z. B. Khanian, and
A. Winter, “Thermodynamics as a Consequence of Information Conservation,”
\emph{Quantum}, vol. 3, 2018. arXiv[quant-ph]:1707.01750.
\end{itemize}

\section{Notation and preliminaries}
In this section, we introduce some conventions, notation and facts that we use throughout this thesis.

Quantum systems are associated with (finite dimensional) Hilbert spaces $A$, $B$, etc.,
whose dimensions are denoted $|A|$, $|B|$, respectively.
The state of such
quantum system is entirely characterized by a density operator, say $\rho$, acting on the associated Hilbert space which is a positive semidefinite operator with trace 1.
Also, we use the notation $\phi= \ketbra{\phi}{\phi}$ to denote  the density 
operator of the pure state vector $\ket{\phi}$. 
Moreover, a system is called  classical  if all the states of the system are diagonal in a fixed orthonormal basis. 

The evolution of a quantum system is characterized by a quantum channel or a so-called completely positive and
trace preserving (CPTP) map which is a linear map taking operators on a Hilbert space 
to operators on the same or a different Hilbert space \cite{Stinespring1955}, however, since there is no risk of confusion, we  denote a CPTP map by the input and output Hilbert spaces, for example, the operator $\mathcal{N}:A \longrightarrow B$  takes the input state $\rho$  on $A$ to the output state $\mathcal{N}(\rho)$ on $B$.

Furthermore, according to Stinespring's factorization theorem \cite{Stinespring1955}, if $\mathcal{N}:A \longrightarrow B$ is a CPTP 
map, then it can be  dilated  to the isometry 
$U_{\mathcal{N}}: A \hookrightarrow B W$ with $W$ as the environment system such that $\mathcal{N}(\rho)=\Tr_W(U_{\mathcal{N}} \rho U_{\mathcal{N}}^{\dagger})$ where $\Tr_W(\cdot)$ denotes the partial trace on system $W$.

\medskip

The fidelity, which is a measure of closeness,  between two states $\rho$ and $\sigma$ is defined as \cite{Jozsa1994_2}
\begin{align}\label{eq:fidelity_def}
F(\rho, \sigma) := \|\sqrt{\rho}\sqrt{\sigma}\|_1 
                 = \Tr \sqrt{\rho^{\frac{1}{2}} \sigma \rho^{\frac{1}{2}}},
\end{align}
where the trace norm is defined as 
\begin{align}\label{eq:norm_def}
\|X\|_1 := \Tr|X| = \Tr\sqrt{X^\dagger X}
\end{align}
It relates to the trace distance in the following well-known way \cite{Fuchs1999}:
\begin{align}\label{eq:Fidelity_norm_relation}
1-F(\rho,\sigma) \leq \frac12\|\rho-\sigma\|_1 \leq \sqrt{1-F(\rho,\sigma)^2}.
\end{align}

The \emph{von Neumann entropy} of a quantum state $\rho$ on system $A$ is defined as 
\begin{align}\label{eq:von Neumann}
S(\rho)_A := - \Tr\rho\log\rho,
\end{align}
where throughout this thesis, $\log$ denotes by default the binary logarithm,
and its inverse function $\exp$, unless otherwise stated, is also to basis $2$. $S(\rho)_A$ is also denoted by $S(A)_{\rho}$.
For the diagonalization of $\rho$, i.e $\rho=\sum_x p_x \proj{v_x}$ with orthonormal basis $\{\ket{v_x}\}$, the von Neumann entropy reduces to the \emph{Shannon entropy} $H(X)$ of 
a random variable $X$ with probability  distribution $p_x$:
\begin{align}\label{eq:Shannon entropy}
 H(X):=-\sum_x p_x \log p_x =S(\rho)_A.
\end{align}
The von Neumann entropy is always bounded as the following:
\begin{align}
  0\leq S(\rho) \leq \log |A|,
\end{align}
where $|A|$ is the dimension of the underlying Hilbert space of $\rho$, i.e. the support of $\rho$. Moreover, $S(\rho)=0$ if and only if $\rho$ is a pure state, and
$S(\rho) = \log |A|$ if and only if  it is  a maximally mixed state, i.e. $\rho = \frac{\1}{\log |A|}$.

\medskip

The mutual information for  a state $\rho^{AB}$ on a bipartite Hilbert space $A \otimes B$ is defined as:
\begin{align}
I(A:B):=S(A)_{\rho}+S(B)_{\rho}-S(AB)_{\rho},
\end{align}
which is always non-negative due to sub-additivity of the von Neumann entropy \cite{Araki_Lieb_1970} and  is equal to 0 if and only if $\rho^{AB}=\rho^{A}\otimes \rho^{B}$, that is an uncorrelated state.

\medskip
Quantum conditional entropy and quantum conditional mutual information, $S(A|B)_{\rho}$ and $I(A:B|C)_{\rho}$,
respectively, are defined in the same way as their classical counterparts: 
\begin{align}\label{eq:conditional_def}
    S(A|B)_{\rho}   &:= S(AB)_\rho-S(B)_{\rho}, \text{ and} \nonumber\\ 
    I(A:B|C)_{\rho} &:= S(A|C)_\rho-S(A|BC)_{\rho} \nonumber \\
                     &= \!S(AC)_\rho\!+\! S(BC)_\rho\!\!-\!S(ABC)_\rho\!\!-\!S(C)_\rho.
\end{align}
Quantum conditional entropy can acquire negative values, however, it is always positive if at least one of the systems $A$ or $B$ is classical. 
Araki-Lieb inequality holds for the conditional entropy as the following \cite{Araki_Lieb_1970}:
\begin{align}\label{eq:Araki_Leib}
  -S(A)_{\rho} \leq S(A|B)_{\rho} \leq  S(A)_{\rho}, 
\end{align}
where the inequality on the right hand side is known as sub-additivity of the entropy.
Quantum conditional mutual information is always positive due to strong sub-additivity of the entropy as the following \cite{SSA_1973}:
\begin{align}\label{eq:SSA}
  S(A|BC)_{\rho} \leq  S(A|C)_{\rho}.
\end{align}

The quantum  \emph{relative entropy} between two quantum states $\rho$ and $\sigma$ is defined as:
\begin{align}
D(\rho || \sigma):=
\begin{cases} 
\Tr (\rho (\log \rho - \log \sigma)) & \text{supp}(\rho) \subseteq \text{supp}(\sigma) \\ 
\infty & \text{otherwise}, 
\end{cases}
\end{align}
which is always non-negative. Pinsker's inequality \cite{Schumacher2002} relates the quantum relative entropy and the trace norm  by 
\begin{align}
  \|\rho-\sigma\|_1 \leq \sqrt{2 \ln 2 D(\rho\|\sigma)}. 
\end{align}










\part{Quantum Source Compression}

\chapter{Formulation of quantum source compression problems} 
\label{chap:source-coding}
In this chapter, we first expand on the concept of quantum sources and the literature on that  and mathematically define an asymptotic compression task as a general model which
include all previously studied asymptotic models.
Then, we introduce quantum compression problems with side information and review
the literature, and later we proceed with defining the most general asymptotic compression task
with side information.
Finally at the end of the chapter, we summarize the results that we have accomplished on quantum source compression.













\section{What is a quantum source?}

A statistical quantum source is a quantum system together with correlations with a \emph{reference system}.
A criterion of how well a source is reproduced in a communication 
task is to measure how well the correlations are preserved with 
the reference system. Without correlation, the information does not make sense 
because a known quantum state without correlations can be reproduced at the 
destination without any communication.

A special case is a classical statistical source, which is modeled by a random variable. Since classical information can be copied, a copy of a random variable can be always stored as a reference, and the final processed information is compared with the copy as a reference 
to analyse the performance of the communication task. 
%
%
However, in the classical information theory literature, the reference is not usually considered explicitly in the description of classical 
information theory tasks, but arguably it is conceptually necessary in quantum 
information. This is because it allows us to present the figure of merit quantifying 
the decoding error as operationally accessible, for example via the probability of
passing a test in the form of a measurement on the combined source and reference systems. This point 
is made eloquently in the early work of Schumacher on quantum information transmission
\cite{Schumi1996,Barnum1998}.

To elaborate more on the reference system, consider the source that Schumacher defined in his 1995 paper \cite{Schumacher1995,Jozsa1994_1} as an ensemble of pure states $\{p(x),\ket{\psi_x}^{A} \}$,
where the source generates the state $\ket{\psi_x}$ with probability $p(x)$. The figure of merit for the encoding-decoding process is to keep the decoded quantum states \emph{on average} very close to the original states with respect to the fidelity, where the average is taken over the probability distribution $p(x)$. 
By basic algebra one can show that this is equivalent to preserving the classical-quantum state 
$\rho^{AR}=\sum_x p(x) \proj{\psi_x}^A \otimes \proj{x}^R$, where system $A$ is the quantum system to be compressed and $R$ is the reference system; namely the following fidelity relation holds:
\begin{align*}
    \sum_{x} p(x) F(\proj{\psi_x}^A,\xi_x^{\hat{A}})=F(\rho^{AR}, \xi^{\hat{A}R}),
\end{align*}
where $\xi_x^{\hat{A}}$ is the decoded state for the realization $x$ and $\xi^{\hat{A}X}=\sum_x  p(x) \xi_x^{\hat{A}} \otimes \proj{x}^R$.
Another source model that Schumacher considered was the purification of the source ensemble, 
that is the state $\ket{\psi}^{AR}=\sum_x \sqrt{p(x)}\ket{\psi_x}^{A} \ket{x}^R$, where the 
figure of merit for the encoding-decoding process was to preserve the pure state correlations with the 
reference system $R$ by maintaining a high fidelity between the decoded state and $\psi$. 
He showed that both definitions lead to the same compression rate, namely, 
the von Neumann entropy of the source $S(A)_{\rho} = S(\rho^A)$, where 
$\rho^A = \Tr_R \rho^{AR}$. Incidentally, the full proof of optimality in the first 
model, without any additional restrictions on the encoder, had to wait until \cite{Barnum1996}
(see also \cite{Horodecki1998});
the strong converse, i.e. the optimality of the entropy rate even for constant error 
bounded away from $1$, was eventually given in \cite{Winter1999}.

Another example of a quantum source is the mixed state source considered by Horodecki \cite{Horodecki1998} and Barnum \emph{et al.} \cite{Barnum2001}, and finally solved by Koashi and Imoto \cite{KI2001},
where the source is defined as an ensemble of mixed states $\{p(x),\rho_x^{A} \}$. Preserving 
these mixed quantum states, on average, in the process of encoding-decoding, the task is 
equivalent to preserving the state $\rho^{AR}=\sum_x p(x) \rho_x^A \otimes \proj{x}^R$, that is 
the quantum system $A$ together with its correlation with the classical reference system $R$.    

In this thesis, we consider the most general finite-dimensional source in the realm of 
quantum mechanics, namely a quantum system $A$ that is correlated with a reference system 
$R$ in an arbitrary way, described by the overall state $\rho^{AR}$. In particular, the 
reference does not necessarily purify the source, nor is it assumed to be classical. 
%
%
The ensemble source and the pure source defined by Schumacher are special cases of this model, 
where the reference is a classical system in the former and a purifying system in the latter. 
So is the source considered by Koashi and Imoto in \cite{KI2001}, where the reference system 
is classical, too. 

Understanding the compression of the source $\rho^{AR}$ has paramount importance in the 
field of quantum information theory and unifies all the models that have been considered 
in the literature. Schumacher's pure source model in a sense is the most stringent model 
because it requires preserving the correlations with a purifying reference system which 
implies that the correlations with any other reference system is preserved which follows 
from the fact that the fidelity is non-decreasing under quantum channels. 
However, the converse is not necessarily true: if in a compression task the parties are 
required to preserve the correlations with a given reference system which does not purify 
the source state, they might be able to compress more efficiently compared to the scenario 
where the reference system purifies the source. This is exactly what we show in Chapter~\ref{chap:mixed state}: 
we characterise the gap precisely depending on the reference system. 

\section{Mathematical definition of quantum noiseless compression} 

A source compression task consists of an encoder which maps the source to compressed information which is stored or sent to another party. When it is needed, a decoder maps the compressed information to decoded information, and the aim is to preserve the correlations with the reference system and reconstruct a source which is very close to the original source in some distance measure. In the quantum realm the most general encoding and decoding maps which can be performed on the information is a quantum operation or a CPTP map.
The communication means or quantum storage device is assumed to be an ideal channel acting as an identity on the encoded information which can be simulated through various resources such as a qubit channel, sharing entanglement and sending classical information and etc. 
The resource is the dimension of the Hilbert space of the encoding operation.

\begin{figure}[!t]
\centering
\includegraphics[width=1\textwidth]{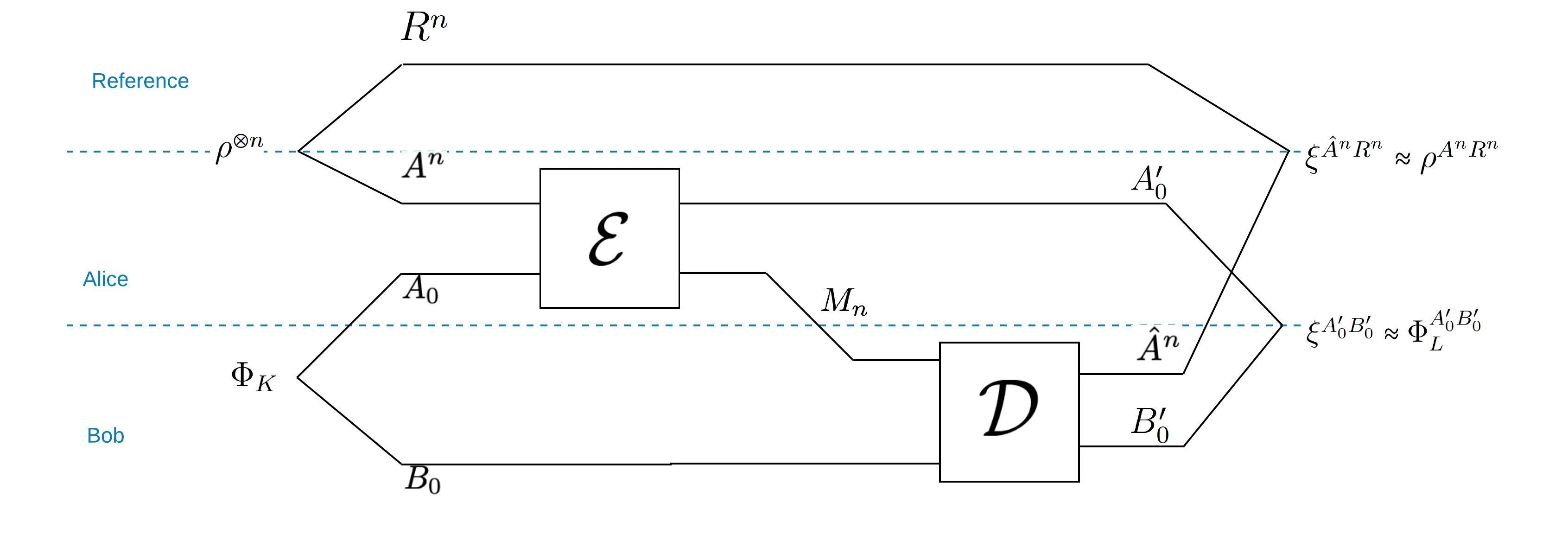} 
\caption{Circuit diagram of the compression task: the source is composed of $n$ copies of the state $\rho^{AR}$ where
$A^n$ is the system to be compressed and $R^n$ is an inaccessible reference system.
Dotted lines are used to demarcate domains controlled by the different participants here the reference, the encoder, Alice and the decoder, Bob. The solid lines represent quantum information registers.
The encoder sends the compressed information, i.e. system $M_n$, to the decoder through a noiseless quantum channel;
moreover, they share initial entanglement in the registers $A_0$ and $B_0$, respectively. 
The aim of the compression task is to reconstruct the source at the decoder side, that is the final state $\xi^{\hat{A}^n R^n}$ has the fidelity converging to 1 with the source state
$\rho^{A^nR^n}$; this ensures that the correlations between the reconstructed system $\hat{A}^n$ and the 
reference system $R^n$ are preserved. Furthermore, the encoder and decoder distill entanglement in
their registers $A'_0$ and $B'_0$, respectively. 
}
\label{fig: rho AR compression task}
\end{figure}

Throughout the thesis, we consider the information theoretic asymptotic limit of
$n$ copies of a finite dimensional source with state $\rho^{AR}$, i.e.~$\rho^{A^n R^n} = \left(\rho^{AR}\right)^{\otimes n}$ where system $A$ is the system to be compressed
and system $R$ is an inaccessible reference system.
We assume that the encoder, Alice, and the decoder, Bob, share initially a maximally 
entangled state $\Phi_K^{A_0B_0}$ on registers $A_0$ and $B_0$ (both of dimension $K$).
The encoder, Alice, performs the encoding compression operation 
$\mathcal{E}:A^n A_0 \longrightarrow M_n A_0'$ on the system $A^n$ and her part $A_0$ of the entanglement, which is CPTP map. 
%
%
Alice's encoding operation produces the state $\sigma^{M_n  R^n A_0' B_0}$ with $M$, $A_0'$ and $B_0$ as the compressed system of Alice, Alice's new entanglement system and Bob's part of the entanglement, respectively.
The dimension of the compressed system is without loss of 
generality not larger than  the dimension of the
original source, i.e. $|M_n| \leq  \abs{A}^n$. 
The system $M_n$ is then sent to Bob via a noiseless quantum channel, who performs
a decoding operation $\mathcal{D}:M_n B_0 \longrightarrow \hat{A}^n B_0'$ on the system 
$M_n$ and his part of the entanglement $B_0$ where $\hat{A}^n$ and $B'_0$ are the reconstructed 
source and Bob's new entanglement system.
Ideally the encoder and decoder want to distill entanglement in the form of maximally entangled
state $\Phi_L^{A'_0 B'_0}$ of dimension $L$ in their corresponding registers $A'_0$ and $B'_0$.

We call $\frac1n \log (K-L)$ and $\frac1n \log|M_n|$ the \emph{entanglement
rate} and \emph{quantum
rate} of the compression protocol, respectively.
We say the encoding-decoding scheme has \emph{fidelity} $1-\epsilon$, or \emph{error} $\epsilon$, if 
\begin{align}
  \label{eq:fidelity criterion no side info task}
  F\left( \rho^{A^n R^n } \otimes \Phi_L^{A_0' B_0'},\xi^{\hat{A}^n  R^n A_0' B_0'} \right)  
          \geq 1-\epsilon,  
\end{align}
%
where $\xi^{\hat{A}^n  R^n A_0' B_0'}=\left((\mathcal{D}\circ\mathcal{E})\otimes \id_{R^n}\right) \rho^{A^n R^n } \otimes \Phi_K^{A_0 B_0}$.
%
%
Moreover, we say that $(E,Q)$ is an (asymptotically) achievable rate pair if for all $n$
there exist codes (encoders and decoders) such that the fidelity converges to $1$, and
the entanglement and quantum rates converge to $E$ and $Q$, respectively.
The compression schemes where the error converges to zero are called noiseless compression schemes which we consider throughout the thesis. 
The rate region is the set of all achievable rate pairs, as a subset of 
$\mathbb{R}\times\mathbb{R}_{\geq 0}$.

\medskip

A schematic description of the quantum source and its compression is illustrated in Fig.~\ref{fig: rho AR compression task}
where the system to be compressed and the reference are denoted by $A^n$ and $R^n$, respectively.
This compression problems is addressed in Chapter~\ref{chap:mixed state} where we find the
optimal trade-off rate region for the entanglement and quantum rates, that is the pairs $(E,Q)$.


\section{Quantum noiseless compression with side information}\label{sec:side info intro}

Side information in information theory is referred to as extra information, which
is correlated with an information source and  is available to encoder, decoder or both of them, and they can use this extra information to use less resources, for example reduce the dimension of the compressed information. 
Slepian and Wolf for the first time studied the compression of a classical source, i.e a random variable, where a decoder has access to another random variable, which is correlated with the source, and showed that the compression rate is equal to the conditional Shannon entropy
\cite{Slepian1973}.

The \textit{visible} paradigm of source compression problems are basically
compression problems where an \emph{encoder} has access to side information, i.e. the identity of states from an ensemble generated by a source \cite{Jozsa1994_1,Barnum1996,Horodecki2000,Barnum2001_2,Bennett2005,Hayden2002,visible_Hayashi}.
For example, the source in the visible Schumacher compression \cite{Jozsa1994_1,Barnum1996}
is modeled by a classical-quantum state $\rho^{ACR}=\sum_x p(x) \proj{\psi_x}^A \otimes \proj{x}^C \otimes \proj{x}^{R}$, where system $A$ is the system to be compressed, and
systems $C$ and $R$ are the side information system of the encoder and the reference system,  respectively. It is shown that both visible model and  \textit{blind} model, where the encoder
does not have access to system $C$, lead to the same compression rate, i.e. $S(A)$ \cite{ Schumacher1995,Jozsa1994_1,Barnum1996} whereas this is not the case when system $A$ is composed of mixed states, that is visible and blind models for mixed states lead to different compression rates. 
In the visible mixed state compression problem, the source is modeled by many copies of the state $\rho^{ACR}=\sum_x p(x) \rho_x^A \otimes \proj{x}^C \otimes \proj{x}^{R}$ where system $A$ with mixed states is the system to be compressed, and
systems $C$ and $R$ are the side information system of the encoder and the reference system,  respectively. Hayashi showed that the optimal compression rate is equal to the regularized entanglement of  purification of the source \cite{visible_Hayashi} which is different from the
blind compression ($C$ is not available to the encoder) rate obtained by Koashi and Imoto \cite{KI2001,KI2002}.
The visible compression of  this source when the encoder and decoder share unlimited entanglement is a special case of the remote state preparation considered in \cite{Bennett2005}, and the optimal quantum compression rate is equal to $\frac{1}{2}S(A)$.

Winter in his Phd thesis \cite{Winter1999} generalized the notion of correlated sources and side information at the \emph{decoder} to a quantum setting by modeling it as a multipartite quantum source which generates multipartite quantum states where different parties have access
to some parts of a source. The first example studied in this context was a hybrid classical-quantum source $\rho^{ABR_1R_2}=\sum_x p(x) \proj{x}^A \otimes \proj{\psi_x}^{BR_1} \otimes  \proj{x}^{R_2}$ where an encoder 
compresses the classical system $A$, and a decoder aims to reconstruct this system
while having access to quantum side information system $B$ such that the correlations with the
reference systems $R=R_1R_2$ are preserved \cite{Winter1999,Devetak2003}.
This example is one of the earliest attempts to find operational meaning to quantum conditional entropy in analogy to the classical conditional Shannon entropy which characterizes the optimal compression rate of a classical source with classical side information at the decoder side,
a.k.a. fully classical Slepian-Wolf problem \cite{Slepian1973}.

The compression of a purified source with side information at the decoder is known as
state merging or fully quantum Slepian-Wolf (FQSW) and its discovery was an important milestone in the quantum information field which gave an operational meaning to the quantum conditional entropy \cite{Horodecki2007,Abeyesinghe2009}; in this task, a source generates many copies of the state $\ket{\psi}^{ABR}$ where an encoder compresses system $A$ and sends it to a decoder who has access to system $B$ and aims to reconstruct system $A$ while preserving the correlations with the reference system $R$.
Depending on the communication means which has been considered shared entanglement with free classical communication or quantum communication, the compression rate is equal to  $S(A|B)$ ebits or $\frac{1}{2}I(A:R)$ qubits, respectively \cite{Horodecki2007,Abeyesinghe2009}.
An ensemble version of FQSW is considered in \cite{Ahn2006} with the source $\rho^{ABR}=
\sum_x p(x) \proj{\psi_x}^{AB}\otimes\proj{x}^R$ and $A$, $B$ and $R$ as the system to be compressed, the side information at the decoder and the reference system, respectively; the optimal quantum compression rate is 
found for some special cases, but the problem has been left open in general.

A generalization of state merging, which is known as quantum state redistribution (QSR), is proposed in \cite{Devetak2008_2,Yard2009}, where 
both encoder and decoder have access to side information systems. Namely, a source generates
many copies of the state $\ket{\psi}^{ACBR}$, where an encoder compresses system $A$ while having access to side information system $C$ and sends the compressed information to a decoder who has access to system $B$ and aims to reconstruct system $A$ while preserving the correlations with the reference system $R$; in this compression task systems $C$ and $B$ remain at the disposal 
of the encoder and decoder, respectively. This gave an operational meaning to the quantum
conditional mutual information since the optimal quantum compression rate was obtained to be
$\frac{1}{2}I(A:R|C)$.

In the remainder of this section, we define mathematically the most general model for
the compression of quantum sources with side information which includes as special cases all the aforementioned side information problems of this section (considering block fidelity defined in Eq.~\ref{eq:block fidelity criterion side info problem}).

\bigskip

We consider a source generates asymptotic limit of
$n$ copies of a finite dimensional state $\rho^{ACBR}$, i.e.~$\rho^{A^n C^n B^n R^n} = \left(\rho^{ACBR}\right)^{\otimes n}$, and distributes the copies of the systems $AC$, $B$ and $R$ between an encoder, a decoder and an inaccessible reference system, respectively.
We assume that the encoder, Alice, and the decoder, Bob, share initially a maximally 
entangled state $\Phi_K^{A_0B_0}$ on registers $A_0$ and $B_0$ (both of dimension $K$).
The encoder, Alice, performs the encoding compression operation 
$\mathcal{E}:A^n C^n A_0 \longrightarrow M_n \hat{C}^n A_0'$ on the system $A^n C^n$ and her part $A_0$ of the entanglement, which is CPTP map. 
Alice's encoding operation produces the state $\sigma^{M_n \hat{C}^n B^n R^n A_0' B_0}$ with $M_n$, $\hat{C}^n$, $A_0'$ and $B_0$ as the compressed system of Alice, a reconstruction of system $C^n$, Alice's new entanglement system and Bob's part of the entanglement, respectively.
The dimension of the compressed system is without loss of 
generality not larger than  the dimension of the
original source, i.e. $|M_n| \leq  \abs{A}^n$. 
The system $M_n$ is then sent to Bob via a noiseless quantum channel, who performs
a decoding operation $\mathcal{D}:M_n B^n B_0 \longrightarrow \hat{A}^n \hat{B}^n B_0'$ on the compressed information $M_n$, system $B^n$ and his part of the entanglement $B_0$ where $\hat{A}^n$, $\hat{B}^n$ and $B'_0$ are the reconstruction of systems  
$A^n$, $B^n$ and Bob's new entanglement system, respectively. 
In this task, the side information systems remain at the disposal of their corresponding 
parties, that is the encoder and decoder respectively reconstruct systems $C^n$ and $B^n$ after 
using them as side information.
Ideally the encoder and decoder want to distill entanglement in the form of maximally entangled
state $\Phi_L^{A'_0 B'_0}$ of dimension $L$ in their corresponding registers $A'_0$ and $B'_0$.

\begin{figure}[!t]
\centering
\includegraphics[width=1\textwidth]{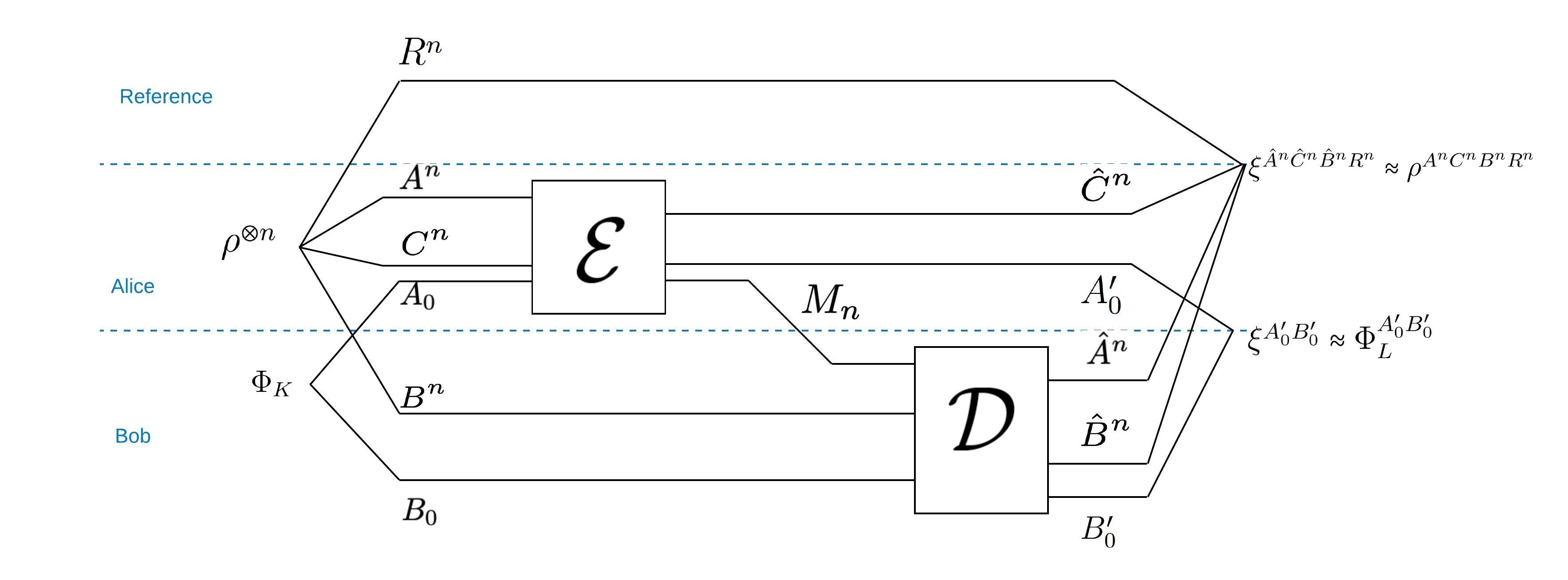} 
\caption{Circuit diagram of the compression task with side information: the source is composed of $n$ copies of the state $\rho^{ACBR}$ where
$A^n$ is the system to be compressed and $R^n$ is an inaccessible reference system; systems $C^n$ and $B^n$ are the side information available for the encoder and the decoder, respectively.
Dotted lines are used to demarcate domains controlled by the different participants here the reference, the encoder, Alice and the decoder, Bob. The solid lines represent quantum information registers.
The encoder sends the compressed information, i.e. system $M_n$, to the decoder through a noiseless quantum channel;
moreover, they share initial entanglement in the registers $A_0$ and $B_0$, respectively. 
The aim of the compression task is to reconstruct system $A^n$ at the decoder side while each party reconstructs its own corresponding side information as well, that is the final state $\xi^{\hat{A}^n \hat{C}^n \hat{B}^n R^n}$ has the fidelity converging to 1 with the source state
$\rho^{A^n C^n B^n R^n}$; this ensures that the correlations between the reconstructed systems $\hat{A}^n \hat{C}^n \hat{B}^n$ and the 
reference system $R^n$ are preserved. Furthermore, the encoder and decoder distill entanglement in
their registers $A'_0$ and $B'_0$, respectively.  
}
\label{fig: rho ACBR compression task}
\end{figure}

We call $\frac1n \log (K-L)$ and $\frac1n \log|M_n|$ the \emph{entanglement
rate} and \emph{quantum
rate} of the compression protocol, respectively.
We say the encoding-decoding scheme has \emph{block fidelity} $1-\epsilon$, or \emph{block error} $\epsilon$, if 
\begin{align}
  \label{eq:block fidelity criterion side info problem}
  F\left( \rho^{A^n C^n B^n R^n } \otimes \Phi_L^{A_0' B_0'},\xi^{\hat{A}^n \hat{C}^n \hat{B}^n  R^n A_0' B_0'} \right)  
          \geq 1-\epsilon,  
\end{align}
where $\xi^{\hat{A}^n \hat{C}^n \hat{B}^n  R^n A_0' B_0'}=\left((\mathcal{D}\circ\mathcal{E})\otimes \id_{R^n}\right) \rho^{A^n C^n B^n R^n } \otimes \Phi_K^{A_0 B_0}$.
Moreover, we say that $(E_b,Q_b)$ is an (asymptotically) achievable block-error rate pair if for all $n$
there exist codes (encoders and decoders) such that the block fidelity converges to $1$, and
the entanglement and quantum rates converge to $E_b$ and $Q_b$, respectively.
The rate region is the set of all achievable rate pairs, as a subset of 
$\mathbb{R}\times\mathbb{R}_{\geq 0}$.
A schematic description of the source compression task with side information  is illustrated in Fig.~\ref{fig: rho ACBR compression task}.

We also consider another figure of merit which turns out to be an easier criterion to evaluate side information problems;
we say a code has \emph{per-copy fidelity} $1-\epsilon$, 
or \emph{per-copy error} $\epsilon$, if 
\begin{align}
  \label{eq:per-copy-error side information problems}
  \frac{1}{n}\sum_{j=1}^n F(\rho^{ACBR},\xi^{\hat{A}_j\hat{C}_j\hat{B}_jR_j}) \geq 1-\epsilon, 
\end{align}
where 
$\xi^{\hat{A}_j\hat{C}_j\hat{B}_jR_j}=\Tr_{[n]\setminus j}\,\xi^{\hat{A}^n\hat{C}^n\hat{B}^nR^n}$, and `$\Tr_{[n]\setminus j}$' denotes the partial trace over all systems 
with indices in $[n]\setminus j$. 
Similarly, we say that $(E_c,Q_c)$ is an (asymptotically) achievable per-copy-error rate pair if for all $n$
there exist codes (encoders and decoders) such that the per-copy fidelity converges to $1$, and
the entanglement and quantum rates converge to $E_c$ and $Q_c$, respectively.
The rate region is the set of all achievable rate pairs, as a subset of 
$\mathbb{R}\times\mathbb{R}_{\geq 0}$.

\begin{table}[!t]
\renewcommand{\arraystretch}{1.5}
\scriptsize
\begin{center}
    \begin{tabular}{ | p{7.5cm} || p{3.5cm} | p{1.5cm}  |} 
    \hline
    \small{Source}           & \small{$(0,Q_b^*)$}            & \small{$(\infty,Q_b^*)$}   
\\ \hline\hline             
\parbox[t]{8cm}{    \cite{Schumacher1995,Jozsa1994_1} 
$\rho^{AR}=\sum_x p(x) \proj{\psi_x}^A \otimes \proj{x}^R$}
& \parbox[t]{3cm}%
{$S(A)_{\rho}$} 
& $-$ \parbox[t]{2cm}{}    
   \\ \hline
\parbox[t]{8cm}{    \cite{Schumacher1995,Jozsa1994_1} 
    $\ket{\psi}^{AR}=\sum_x \sqrt{p(x)}\ket{\psi_x}^{A} \ket{x}^R$}
& \parbox[t]{3cm}%
{$S(A)_{\rho}$} 
&$-$ \parbox[t]{2cm}{}    
   \\ \hline
 \parbox[t]{8cm}{    \cite{Schumacher1995,Jozsa1994_1,Barnum1996} 
    $\rho^{ACR}=\sum_x p(x) \proj{\psi_x}^A \otimes \proj{x}^C \otimes \proj{x}^{R}$}
& \parbox[t]{3cm}%
{$S(A)_{\rho}$} 
& $-$ \parbox[t]{2cm}{}    
   \\ \hline
\parbox[t]{8cm}{    \cite{KI2001} 
    $\rho^{AR}=\sum_x p(x) \rho_x^A \otimes \proj{x}^R$}
& \parbox[t]{3cm}%
{$S(CQ)_{\omega}$} 
& $-$ \parbox[t]{2cm}{}    
   \\ \hline
\parbox[t]{8cm}{    \cite{visible_Hayashi} 
    $\rho^{ACR}=\sum_x p(x) \rho_x^A \otimes \proj{x}^C \otimes \proj{x}^{R}$}
& \parbox[t]{3cm}%
{$\lim_{n \to \infty } \frac{E_p((\rho^{AC})^{\otimes n})}{n}$} 
& $-$ \parbox[t]{2cm}{}    
   \\ \hline
\parbox[t]{8cm}{    \cite{Bennett2005}
    $\rho^{ACR}=\sum_x p(x) \rho_x^A \otimes \proj{x}^C \otimes \proj{x}^{R}$}
& \parbox[t]{3cm}%
{$-$} 
& \parbox[t]{2cm}{$\frac{1}{2}S(A)$}    
   \\ \hline
\parbox[t]{8cm}{    \cite{Winter1999,Devetak2003}
    $\rho^{ABR_1R_2}=\sum_x p(x) \proj{x}^A \otimes \proj{\psi_x}^{BR_1} \otimes \proj{x}^{R_2}$}
& \parbox[t]{3cm}%
{$S(A|B)_{\rho}$} 
& $-$ \parbox[t]{2cm}{}    
  \\ \hline
\parbox[t]{8cm}{    \cite{Horodecki2007,Abeyesinghe2009}
    $\ket{\psi}^{ABR}$}
& \parbox[t]{3.5cm}%
{$\max \{S(A|B)_{\rho}, \frac{1}{2}I(A:R)\}$ } 
& \parbox[t]{2cm}{$\frac{1}{2}I(A:R)$}    
   \\ \hline
\parbox[t]{8cm}{    \cite{Ahn2006}
    $\rho^{ABR}=\sum_x p(x) \proj{\psi_x}^{AB}\otimes\proj{x}^R$ }
& \parbox[t]{3.7cm}%
{solved for specific examples} 
& $-$ \parbox[t]{2cm}{}    
 \\ \hline
\parbox[t]{8cm}{    \cite{Devetak2008_2,Yard2009}
    $\ket{\psi}^{ACBR}$}
& \parbox[t]{3.7cm}%
{$\max \{ S(A|B)_{\rho}$,$\frac{1}{2}I(A:R|C)\}$} 
& \parbox[t]{2cm}{$\frac{1}{2}I(A:R|C)$}       \\ \hline

\end{tabular}
\end{center}
\caption{A summary of the asymptotic source compression problems, that have been studied in the literature so far, is presented in this table. The rate pairs $(0,Q_b^*)$ and $(\infty,Q_b^*)$ denote the unassisted and entanglement-assisted qubit rates, respectively. Here $E_p(\cdot)$ denotes the entanglement of purification, moreover, $S(CQ)_{\omega}$ is the von Neumann entropy with respect to Koashi-Imoto decomposition of the source; for more information see chapter~\ref{chap:mixed state}.}
\label{table-of-sources}
\end{table}


\bigskip

The special cases of this general problem that have been addressed so far is summarized in table~\ref{table-of-sources}. This general compression problem has a complex nature; for example,
consider the special case of the visible mixed state source by Hayashi \cite{visible_Hayashi} with classical reference $R$ and classical side information at the encoder $C$, with no side information at the decoder $B= \emptyset$, i.e. $\rho^{ACR}=\sum_x p(x) \rho_x^A \otimes \proj{x}^C \otimes \proj{x}^{R}$;
with no entanglement consumption, the optimal block-error quantum rate, i.e. the pair $(0,Q_b^*)$ is equal to the regularized entanglement of purification whereas with free entanglement the optimal block-error quantum rate, i.e. the pair $(\infty,Q_b^*)$ is equal to $(\infty,\frac{1}{2}S(A))$ \cite{Bennett2005}.
Therefore, it is insightful to first study the pairs $(0,Q_b^*)$ and $(\infty,Q_b^*)$
for some other special cases of the source $\rho^{ACBR}$. 

Moreover as we will show in the subsequent chapters, unlike the classical 
scenario where conditional entropy characterizes the classical compression rate,
for non-pure sources, the quantum conditional entropy or mutual information
does not necessary play a
role and more complicated functions of the source determine the compression rate. 
In section~\ref{sec:Summary of our results in quantum source compression}, we briefly
go through the special cases of the general source $\rho^{ACBR}$ with side information 
which we address in this thesis and discuss the challenges of each particular case
in its corresponding chapter.

\section{Distributed noiseless quantum source compression}
This thesis mainly focuses on the side information compression problems, however, in chapter \ref{chap:cqSW}, aside from a side information problem we study the distributed compression of correlated classical-quantum sources. This motivates us to define a general distributed compression problem, which the side information problem of section~\ref{sec:side info intro} can be considered a sub-problem of this distributed scenario since the decoder can use successive decoding,  that is it can first decode the information of one of the encoders and treat it as its own side information, and later decode the information of the other encoder.

\bigskip

Here we define the problem for two encoders, however, the definition can be easily extended to three or more encoders. 
We consider a source generates asymptotic limit of
$n$ copies of a finite dimensional  state $\rho^{A_1C_1 A_2 C_2BR}$, i.e.~$\rho^{A_1^n C_1^n A_2^n C_2^n B^n R^n} = \left(\rho^{A_1C_1A_2C_2BR}\right)^{\otimes n}$, and distributes the copies of the systems $A_1C_1$, $A_2C_2$, $B$ and $R$ between  encoder 1, encoder 2, a decoder and an inaccessible reference system, respectively.
We assume that  both encoder 1, Alice and encoder 2, Ava,  share initially maximally 
entangled states $\Phi_{K_1}^{A_{01}B_{01}}$ and $\Phi_{K_2}^{A_{02}B_{02}}$ with the decoder, Bob, respectively (of dimension $K_1$ and $K_2$ respectively).
The encoder $i$ ($i=1,2$)  performs the encoding compression operation, i.e. the CPTP map 
$\mathcal{E}_i:A_i^n C_i^n A_{0i} \longrightarrow M_{i_n} \hat{C}_i^n A_{0i}'$ on the systems $A_i^n C_i^n$ and the entanglement part $A_{0i}$.
The encoding operations are distributed in the sense that each encoder applies her own operation locally without having access to the information of the other encoder.  
%
%
The dimension of the compressed systems are without loss of 
generality not larger than  the dimension of the
original sources, i.e. $|M_{i_n}| \leq  \abs{A_i}^n$. 
The systems $M_{i_n}$ ($i=1,2$) are then sent to Bob via a noiseless quantum channel, who performs
the decoding operation $\mathcal{D}:M_{1_n} M_{2_n} B^n B_{01} B_{02} \longrightarrow \hat{A_1}^n \hat{A_2}^n \hat{B}^n B_{01}' B_{02}'$ on the compressed information systems $M_{1_n} M_{2_n}$, system $B^n$ and his parts of the entanglement $B_{01} B_{02}$ where $\hat{A_1}^n \hat{A_2}^n$, $\hat{B}^n$ and $B'_{01} B'_{02}$ are the reconstruction of systems  
$A_1^n A_2^n$, $B^n$ and Bob's new entanglement systems, respectively. 
In this task, the systems $C_1^n$, $C_2^n$ and $B^n$ remain at the disposal of their corresponding 
parties, that is the encoders and the decoder respectively reconstruct systems $C_1^n$, $C_2^n$ and $B^n$ after 
using them as side information.
Ideally the encoder $i$ ($i=1,2$) and the decoder aim to distill entanglement in the form of maximally entangled
state $\Phi_{L_i}^{A'_{0i} B'_{0i}}$ of dimension $L_i$ in their corresponding registers $A'_{0i}$ and $B'_{0i}$, respectively.


\begin{figure}[!t]
\centering
\includegraphics[width=1\textwidth]{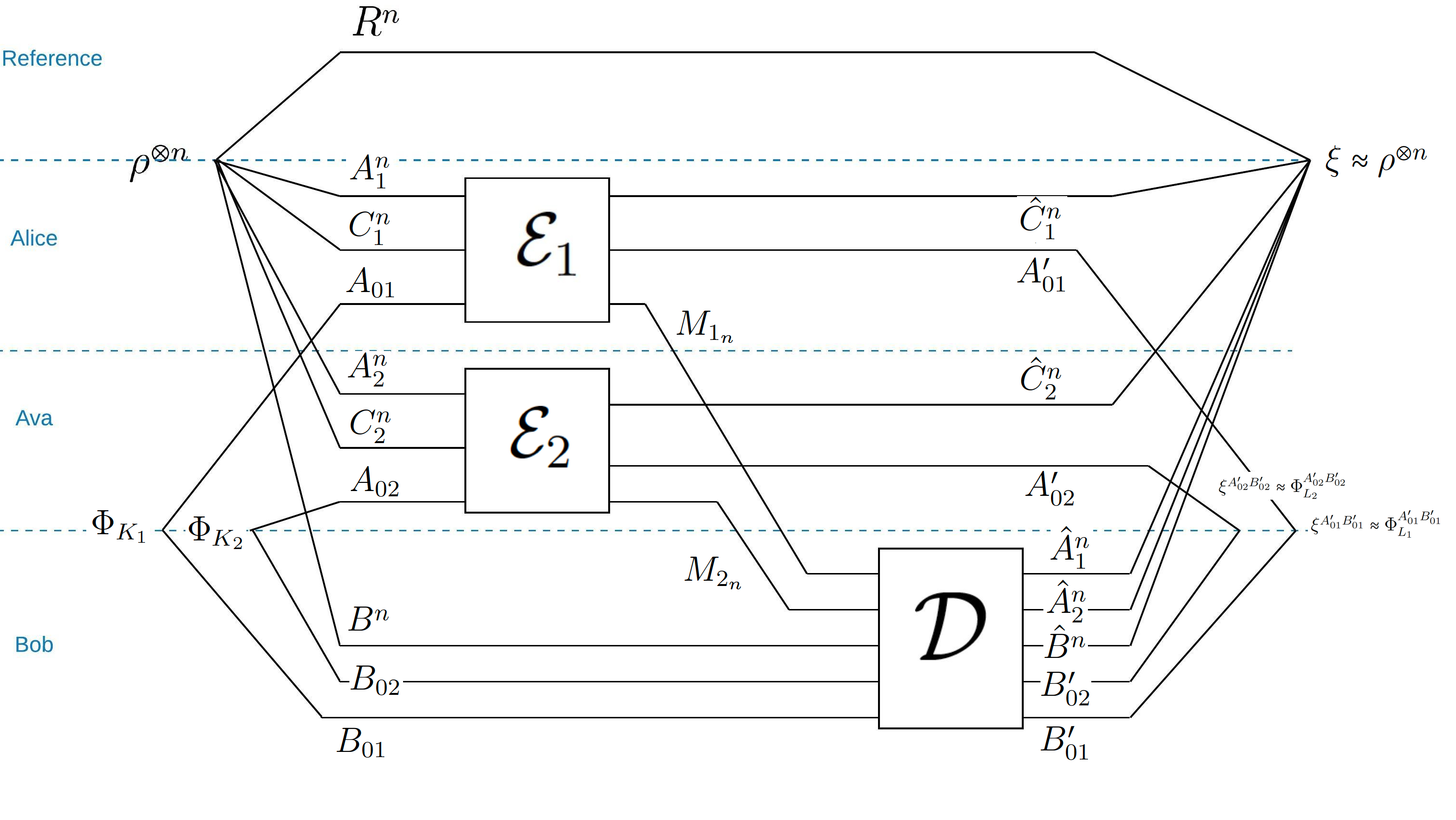} 
\caption{Circuit diagram of the distributed compression task: the source is composed of $n$ copies of the state $\rho^{A_1C_1 A_2C_2 BR}$ where
$A_i^n$ ($i=1,2$) is the system to be compressed and $R^n$ is an inaccessible reference system; systems $C_i^n$ and $B^n$ are the side information available for the encoder $i$ and the decoder, respectively.
Dotted lines are used to demarcate domains controlled by the different participants here the reference, the encoders, Alice and Ava, and the decoder, Bob. The solid lines represent quantum information registers.
The encoder $i$ sends the compressed information, i.e. system $M_{i_n}$, to the decoder through a noiseless quantum channel;
moreover, they share initial entanglement in the registers $A_{0i}$ and $B_{0i}$, respectively. 
The aim of the compression task is to reconstruct systems $A_1^n$ and $A_2^n$ at the decoder side while each party reconstructs its own corresponding side information as well, that is the final state $\xi^{\hat{A_1}^n \hat{C_1}^n \hat{A_2}^n \hat{C_2}^n \hat{B}^n R^n}$ has the fidelity converging to 1 with the source state
$\rho^{A_1^n C_1^n A_2^n C_2^n B^n R^n}$; this ensures that the correlations between the reconstructed systems $\hat{A_1}^n \hat{C_1}^n \hat{A_2}^n \hat{C_2}^n \hat{B}^n$ and the 
reference system $R^n$ are preserved. Furthermore, the encoder $i$ and the decoder distill entanglement in
their registers $A'_{0i}$ and $B'_{0i}$, respectively.  
}
\label{fig: distributed rho ACBR compression task}
\end{figure}


We call $\frac1n \log (K_i-L_i)$ and $\frac1n \log|M_{i_n}|$ the \emph{entanglement
rate} and \emph{quantum
rate} of the compression protocol, respectively (for $i=1,2$).
Moreover, we say the encoding-decoding scheme has \emph{block fidelity} $1-\epsilon$, or \emph{block error} $\epsilon$, if 
\begin{align}
  \label{eq:block fidelity criterion distributed problem}
  F\left( \rho^{A_1^n C_1^n A_2^n C_2^n B^n R^n } \otimes \Phi_{L_1}^{A'_{01} B'_{01}} \otimes \Phi_{L_2}^{A'_{02} B'_{02}},\xi^{\hat{A_1}^n \hat{C_1}^n \hat{A_2}^n \hat{C_2}^n \hat{B}^n  R^n A_{01}' B_{01}' A_{02}' B_{02}'} \right)  
          \geq 1-\epsilon,  
\end{align}
where 
\begin{align*}
&\xi^{\hat{A_1}^n \hat{C_1}^n \hat{A_2}^n \hat{C_2}^n \hat{B}^n  R^n A_{01}' B_{01}' A_{02}' B_{02}'}=\\
& \quad \quad \quad \quad \left((\mathcal{D}\circ(\mathcal{E}_1 \otimes \mathcal{E}_2))\otimes \id_{R^n}\right) \rho^{A_1^n C_1^n A_2^n C_2^n B^n R^n } \otimes \Phi_{K_1}^{A_{01} B_{01}} \otimes \Phi_{K_2}^{A_{02} B_{02}}.
\end{align*}
Moreover, we say that $(E_{b_1},E_{b_2},Q_{b_1},Q_{b_2})$ is an (asymptotically) achievable block-error rate tuple if for all $n$
there exist codes (encoders and decoders) such that the block fidelity converges to $1$, and
the $i$th entanglement and quantum rates converge to $E_{b_i}$ and $Q_{b_i}$ for encoder $i$, respectively.
The rate region is the set of all achievable rate pairs, as a subset of 
$\mathbb{R}\times \mathbb{R}\times\mathbb{R}_{\geq 0} \times\mathbb{R}_{\geq 0}$.
A schematic description of the source compression task with side information  is illustrated in Fig.~\ref{fig: distributed rho ACBR compression task}.

In chapter~\ref{chap:cqSW}, we consider block fidelity, however, the results follow for the per-copy fidelity as well which is defined as follows: 
we say a code has \emph{per-copy fidelity} $1-\epsilon$, 
or \emph{per-copy error} $\epsilon$, if 
\begin{align}
  \label{eq:per-copy  fidelity criterion distributed problem}
\frac{1}{n} \sum_{j=1}^n F\left( \rho^{A_{1}  C_1  A_2  C_2  B R } ,\xi^{\hat{A_{1j}} \hat{C_{1j}} \hat{A_{2j}} \hat{C_{2j}} \hat{B}  R } \right)  
          \geq 1-\epsilon,  
\end{align}
where 
$\xi^{\hat{A_{1j}} \hat{C_{1j}} \hat{A_{2j}} \hat{C_{2j}} \hat{B}  R }=\Tr_{[n]\setminus j}\,\xi^{\hat{A_1}^n \hat{C_1}^n\hat{A_2}^n\hat{C_2}^n\hat{B}^nR^n}$, and `$\Tr_{[n]\setminus j}$' denotes the partial trace over all systems 
with indices in $[n]\setminus j$. 
Similarly, we say that $(E_{c_1},E_{c_2},Q_{c_1},Q_{c_2})$ is an (asymptotically) achievable per-copy-error rate tuple if for all $n$
there exist codes  such that the per-copy fidelity converges to $1$, and
the $i$th entanglement and quantum rates converge to $E_{c_i}$ and $Q_{c_i}$ for encoder $i$, respectively.
The rate region is the set of all achievable rate pairs, as a subset of 
$\mathbb{R}\times \mathbb{R}\times\mathbb{R}_{\geq 0} \times\mathbb{R}_{\geq 0}$.

\bigskip
In \cite{Savov_Mscthesis2007,Savov_distributed_2008},  compression of a pure source $\ket{\psi}^{A_1A_2R}$ with side information at the encoders is considered ($C_1,C_2,B=\emptyset$). The achievable rate region is a convex hull of various points where each point corresponding to an encoder is achieved by applying fully quantum Slepian-Wolf (FQSW) compression and treating the rest of the systems as a reference.  
The converse bounds are in terms of 
the multipartite squashed entanglement,
which is a measure of multipartite entanglement.

\section{Summary of our results in quantum source compression and discussion}
\label{sec:Summary of our results in quantum source compression} 

In this section, we briefly explain the special cases of problems, defined in the previous sections, that we address in this thesis.
Notice that in the subsequent chapters we do not necessarily respect the notation $A$, $C$, $B$ and $R$ for the system to be compressed, the side information at the encoder, the side information at the decoder and the reference system, however, we clearly define the task and specify the notation for the corresponding registers. Moreover, we specify whether the error criterion is block fidelity or per-copy fidelity. 

\bigskip 

In chapter \ref{chap:mixed state}, we consider the compression of a general mixed state source
$\rho^{AR}$ (no side information) and find the optimal trade-off between the entanglement and quantum rates, i.e. the pair $(E,Q)$.

\medskip

In chapter \ref{chap: E assisted Schumacher}, we unify the visible and blind Schumacher compression by considering an interpolation between them as side information, that is the source
$\rho^{ACR}=\sum_x p(x) \proj{\psi_x}^A \otimes \proj{c_x}^C \otimes \proj{x}^R$ with $A$, $C$ and $R$ as the system to be compressed, the side information at the encoder and the classical reference system. For this source, we find optimal trade-off between the block-error entanglement and quantum rate pairs $(E_b,Q_b)$.

\medskip

In chapter \ref{chap:cqSW}, we consider quantum source compression with classical side information with the source $\rho^{AR_1BR_2}=\sum_x p(x) \proj{\psi_x}^{AR_1} \otimes \proj{x}^B \otimes \proj{x}^{R_2}$ and $A$, $B$ and $R=R_1 R_2$ as the system to be compressed, the side information at the decoder and the hybrid classical-quantum reference systems, respectively.
We study the entanglement assisted case $(\infty,Q_b)$, the unassisted case
$(0,Q_b)$ then distributed scenario considering block fidelity. We find achievable and converse bounds for each scenario and  show that the two bounds match for the entanglement assisted quantum block-error rate $Q_b$ up to continuity of a function which appears in the bounds. Finally, considering per-copy fidelity we find the optimal entanglement assisted quantum per-copy-error rate, i.e. the pair $(\infty,Q_c^*)$.

\medskip

In chapter \ref{chap:QSR ensemble}, we consider an ensemble generalization of the quantum state redistribution (QSR), i.e. the source $\rho^{ACBR}=\sum_x p(x) \proj{\psi_x}^{ACBR_1} \otimes \proj{x}^{R_2}$ with $A$, $C$, $B$ and $R=R_1R_2$ as the system to be compressed, the side information at the encoder, the side information at the decoder and the hybrid classical-quantum reference systems, respectively. We consider free entanglement scenario and find the optimal quantum per-copy-error rate, i.e. the pair $(\infty,Q_c^*)$.
With block fidelity, we find achievable and converse bounds which match up to continuity of a function appearing in the bounds.

\medskip

In summary, for a general mixed state we solve the problem when there is no side information, and the rate region is in terms of an extension of the decomposition of the source state which is discovered by Koashi and Imoto in \cite{KI2002}, and later this decomposition extended to a general mixed state in \cite{Hayden2004}. However, for multipartite states this decomposition does not necessarily preserve the tensor structure over various systems; this turns out to be the main hurdle in dealing with general mixed state problems with side information.
This is not an issue for pure or ensemble sources mainly because the structure of maps
which preserve these states are well-understood. For these sources the environment systems of the encoding and decoding operations are decoupled from the reconstructed source given the identity of the state from the ensemble. This property is one of the guiding intuitions behind the converse proofs for the side information problems.

\chapter{Compression of a general mixed state source}
\label{chap:mixed state}

In this chapter, we consider the most general (finite-dimensional) quantum 
mechanical information source, which is given by a quantum system 
$A$ that is correlated with a reference system $R$. The task is to 
compress $A$ in such a way as to reproduce the joint source state 
$\rho^{AR}$ at the decoder with asymptotically high fidelity. This 
includes Schumacher's original quantum source coding problem of a 
pure state ensemble and that of a single pure entangled state, as 
well as general mixed state ensembles. 
Here, we determine the optimal compression rate (in qubits per 
source system) in terms of the Koashi-Imoto decomposition of the 
source into a classical, a quantum, and a redundant part. The same 
decomposition yields the optimal rate in the presence of unlimited 
entanglement between compressor and decoder, and indeed the full 
region of feasible qubit-ebit rate pairs. 
This chapter is based on the papers in \cite{ZK_mixed_state_ISIT_2020,ZK_mixed_state_2019}.
%
%
%


\medskip


\section{The source model and the compression task}
\label{sec:The Compression task}
We consider a general mixed state source
$\rho^{AR}$ with $A$ and $R$ as the system to be compressed and the reference system, respectively, where the source generates 
the information theoretic limit of
many copies of the state $\rho^{AR}$, i.e.~$\rho^{A^n R^n} = \left(\rho^{AR}\right)^{\otimes n}$.
We assume that the encoder, Alice, and the decoder, Bob, have initially a maximally 
entangled state $\Phi_K^{A_0B_0}$ on registers $A_0$ and $B_0$ (both of dimension $K$).
The encoder, Alice, performs the encoding compression operation 
$\mathcal{C}:A^n A_0 \longrightarrow M $ on the system $A^n$ and her part $A_0$ of the entanglement, which is a quantum channel,
i.e.~a completely positive and trace preserving (CPTP) map. 
Notice that as functions CPTP maps act on the operators (density matrices) over 
the respective input and output Hilbert spaces, but as there is no risk of confusion,
we will simply write the Hilbert spaces when denoting a CPTP map.
Alice's encoding operation produces the state $\sigma^{M B_0 R^n}$ with $M$ and $B_0$ as the compressed system of Alice and Bob's part of the entanglement, respectively.
The dimension of the compressed system is without loss of 
 generality not larger than  the dimension of the
original source, i.e. $|M| \leq  \abs{A}^n$. 
We call $\frac1n \log K$ and $\frac1n \log|M|$ the \emph{entanglement
rate} and \emph{quantum
rate} of the compression protocol, respectively.
The system $M$ is then sent to Bob via a noiseless quantum channel, who performs
a decoding operation $\mathcal{D}:M B_0 \longrightarrow \hat{A}^n$ on the system 
$M$ and his part of the entanglement $B_0$.
We say the encoding-decoding scheme has \emph{fidelity} $1-\epsilon$, or \emph{error} $\epsilon$, if 
\begin{align}
  \label{eq:fidelity criterion}
  F\left( \rho^{A^n R^n },\xi^{\hat{A}^n R^n} \right)  
          \geq 1-\epsilon,  
\end{align}
where $\xi^{\hat{A}^n R^n}=\left((\mathcal{D}\circ\mathcal{C})\otimes \id_{R^n}\right) \rho^{A^n R^n }$.
Moreover, we say that $(E,Q)$ is an (asymptotically) achievable rate pair if for all $n$
there exist codes such that the fidelity converges to $1$, and
the entanglement and quantum rates converge to $E$ and $Q$, respectively.
The rate region is the set of all achievable rate pairs, as a subset of 
$\mathbb{R}_{\geq 0}\times\mathbb{R}_{\geq 0}$. 

According to Stinespring's theorem \cite{Stinespring1955}, a CPTP map 
$\cT: A \longrightarrow \hat{A}$ can be dilated to an isometry $U: A \hookrightarrow \hat{A} E$ 
with $E$ as an environment system, called an isometric extension of a CPTP map, such that 
$\cT(\rho^A)=\Tr_E (U \rho^A U^{\dagger})$. 
Therefore, the encoding and decoding operations are can in general be viewed as 
isometries $U_{\cE} : A^n A_0 \hookrightarrow M W$ and
$U_{\cD} : M B_0 \hookrightarrow \hat{A}^n V$, respectively, 
with the systems $W$ and $V$ as the environment systems
of Alice and Bob, respectively. 

\medskip

We say a source $\omega^{BR}$ is equivalent to a source $\rho^{AR}$ if there are
CPTP maps $\cT:A \longrightarrow B$ and $\cR:B \longrightarrow A$ in both directions
taking one to the other: 
\begin{align} \label{def: equivalent sources}
    \omega^{BR}=(\cT \otimes \id_R) \rho^{AR} \text{ and } 
    \rho^{AR}=(\cR \otimes \id_R) \omega^{BR}.
\end{align}
The rate regions of equivalent sources are the same, because any achievable rate pair for 
one source is achievable for the other source as well. This follows from the fact that for
any code $(\cC,\cD)$ of block length $n$ and error $\epsilon$ for $\rho^{AR}$, 
concatenating the encoding and decoding operations with $\cT$ and $\cR$, i.e. letting
$\cC'=\cC\circ\cR^{\otimes n}$ and $\cD'=\cT^{\otimes n}\circ\cD$, we get a code 
of the same error $\epsilon$ for $\omega^{BR}$. Analogously we can turn a code for 
$\omega^{BR}$ into one for $\rho^{AR}$.

\section{The qubit-ebit rate region}
\label{sec:The optimal rate region}
The idea behind the compression of the source $\rho^{AR}$ is based on a decomposition 
of this state introduced in \cite{Hayden2004}, which is a generalization of the decomposition
introduced by Koashi and Imoto in \cite{KI2002}. Namely, for any set of quantum states $\{ \rho_x\}$, 
there is a unique decomposition of the Hilbert space describing
the structure of CPTP maps which preserve the set $\{ \rho_x^A\}$. This idea was generalized 
in \cite{Hayden2004} for a general mixed state $\rho^{AR}$ describing the structure of 
CPTP maps acting on system $A$ which preserve the overall state $\rho^{AR}$. 
This was achieved by showing that any such map preserves the set of all possible
states on system $A$ which can be obtained by measuring system $R$, and 
conversely any map preserving the set of all possible
states on system $A$ obtained by measuring system $R$, preserves the state $\rho^{AR}$,
thus reducing the general case to the case of classical-quantum states 
\[
  \rho^{AY} = \sum_y q(y) \rho_y^A \otimes \proj{y}^Y
            = \sum_y \Tr_R \rho^{AR}(\1_A\otimes M_y^R) \otimes \proj{y}^Y, 
\]
which is the ensemble case considered by Koashi and Imoto. As a matter of fact, 
looking at the algorithm presented in \cite{KI2002} to compute the decomposition,
it is enough to consider an informationally complete POVM $(M_y)$ on $R$, with 
no more than $|R|^2$ many outcomes.
The properties of this decomposition are stated in the following theorem.

\begin{theorem}[\cite{KI2002,Hayden2004}]
\label{thm: KI decomposition}
Associated to the state $\rho^{AR}$, there are Hilbert spaces $C$, $N$ and $Q$
and an isometry $U_{\KI}:A \hookrightarrow C N Q$ such that:
\begin{enumerate}
  \item The state $\rho^{AR}$ is transformed by $U_{\KI}$ as
    \begin{equation}
      \label{eq:KI state}
      (U_{\KI}\otimes \1_R)\rho^{AR} (U_{\KI}^{\dagger}\otimes \1_R)
        = \sum_j p_j \proj{j}^{C} \otimes \omega_j^{N} \otimes \rho_j^{Q R}
        =:\omega^{C N Q R},
    \end{equation}
    where the set of vectors $\{ \ket{j}^{C}\}$ form an orthonormal basis for Hilbert space 
    $C$, and $p_j$ is a probability distribution over $j$. The states $\omega_j^{N}$ and 
    $\rho_j^{Q R}$ act  on the Hilbert spaces $N$ and $Q \otimes R$, respectively.

  \item For any CPTP map $\Lambda$ acting on system $A$ which leaves the state $\rho^{AR}$ 
    invariant, that is $(\Lambda \otimes \id_R )\rho^{AR}=\rho^{AR}$, every associated 
    isometric extension $U: A\hookrightarrow A E$ of $\Lambda$ with the environment system 
    $E$ is of the following form
    \begin{equation}
      U = (U_{\KI}\otimes \1_E)^{\dagger}
            \left( \sum_j \proj{j}^{C} \otimes U_j^{N} \otimes \1_j^{Q} \right) U_{\KI},
    \end{equation}
    where the isometries $U_j:N \hookrightarrow N E$ satisfy 
    $\Tr_E [U_j \omega_j U_j^{\dagger}]=\omega_j$ for all $j$.
    The isometry $U_{KI}$ is unique (up to trivial change of basis of the Hilbert spaces 
    $C$, $N$ and $Q$). Henceforth, we call the isometry $U_{\KI}$ and the state 
    $\omega^{C N Q R}=\sum_j p_j \proj{j}^{C} \otimes \omega_j^{N} \otimes \rho_j^{Q R}$ 
    the Koashi-Imoto (KI) isometry and KI-decomposition of the state $\rho^{AR}$, respectively. 

  \item In the particular case of a tripartite system $CNQ$ and a state $\omega^{CNQR}$ already 
    in Koashi-Imoto form (\ref{eq:KI state}), property 2 says the following:
    For any CPTP map $\Lambda$ acting on systems $CNQ$ with 
    $(\Lambda \otimes \id_R )\omega^{CNQR}=\omega^{CNQR}$, every associated 
    isometric extension $U: CNQ\hookrightarrow CNQ E$ of $\Lambda$ with the environment system 
    $E$ is of the form
    \begin{equation}
      U = \sum_j \proj{j}^{C} \otimes U_j^{N} \otimes \1_j^{Q},
    \end{equation}
    where the isometries $U_j:N \hookrightarrow N E$ satisfy 
    $\Tr_E [U_j \omega_j U_j^{\dagger}]=\omega_j$ for all $j$.
\end{enumerate} 
\end{theorem}

According to the discussion at the end of Sec. \ref{sec:The Compression task}, the 
sources $\rho^{AR}$ and $\omega^{C N Q R}$ are equivalent because there are the isometry 
$U_{\KI}$ and the reversal CPTP map $\cR: C N Q \longrightarrow A$, which reverses the 
action of the KI isometry, such that:
\begin{align}
    \omega^{C N Q R}&= (U_{\KI}\otimes \1_R)\rho^{AR} (U_{\KI}^{\dagger}\otimes \1_R), \nonumber \\
    \rho^{AR}&=(\cR \otimes \id_R)\omega^{C N Q R}\nonumber \\
    &=(U_{\KI}^{\dagger }\otimes \1_R) \omega^{C N Q R} (U_{\KI}\otimes \1_R)+\Tr [(\1_{C N Q }-\Pi_{C N Q})\omega^{C N Q}]\sigma,
\end{align}
where $\Pi_{C N Q}=U_{\KI}U_{\KI}^{\dagger}$ is the projection onto the subspace 
$U_{\KI}A \subset C \otimes N \otimes Q$, and $\sigma$ is an arbitrary state acting on $A\otimes R$.
Henceforth we assume that the source is $\omega^{C N Q R}$, which is convenient because
our main result is expressed in terms of the systems $C$ and $Q$. Notice that
the source $\omega^{C N Q R}$ is in turn equivalent to $\omega^{C Q R}$,
a fact we will exploit in the proof.

%
Moreover, since the information in $C$ is classical, we can reduce the 
compression rate even more if the sender and receiver share  
entanglement, by using dense coding of $j$. In the following
theorem we show the optimal qubit-ebit rate tradeoff for the compression of the source $\rho^{AR}$. 

\begin{theorem}
  \label{theorem:complete rate region mixed state}
  For the compression of the source $\rho^{AR}$, all asymptotically achievable entanglement and
  quantum rate pairs $(E,Q)$ satisfy
  \begin{align*}
    Q   &\geq S(CQ)_{\omega}-\frac{1}{2}S(C)_{\omega},\\
    Q+E &\geq  S(CQ)_{\omega}, 
  \end{align*}
  where the entropies are with respect the KI decomposition of the state $\rho^{AR}$, i.e. 
 the state $\omega^{C N Q R}$.
  Conversely, all the rate pairs satisfying the above inequalities are asymptotically achievable.
\end{theorem}

\begin{remark}
This theorem implies that the optimal asymptotic quantum rates for the compression of 
the source $\rho^{AR}$ with and without entanglement assistance are 
$S(CQ)_{\omega}-\frac{1}{2}S(C)_{\omega}$ and $S(CQ)_{\omega}$ qubits, respectively,
and $\frac{1}{2}S(C)_{\omega}$ ebits of entanglement 
are sufficient and necessary in the entanglement assisted case. 
\end{remark}

\begin{remark}
If in the compression task the parties were required to preserve the correlations with a purifying reference system, then due to Schumacher compression the optimal qubit rate would be $S(A)_{\rho}=S(CNQ)_{\omega}$. However, Theorem~\ref{theorem:complete rate region mixed state} shows that 
the parties can compress more if they are only required to preserve the correlations with a mixed state reference. This gap can be strictly positive if the redundant system $N$ is mixed given the classical information $j$ in system $C$, that is $S(CNQ)_{\omega}-S(CQ)_{\omega}=S(N|CQ)_{\omega}>0$. 
\end{remark}

\begin{figure}[ht] 
  \includegraphics[width=0.8\textwidth]{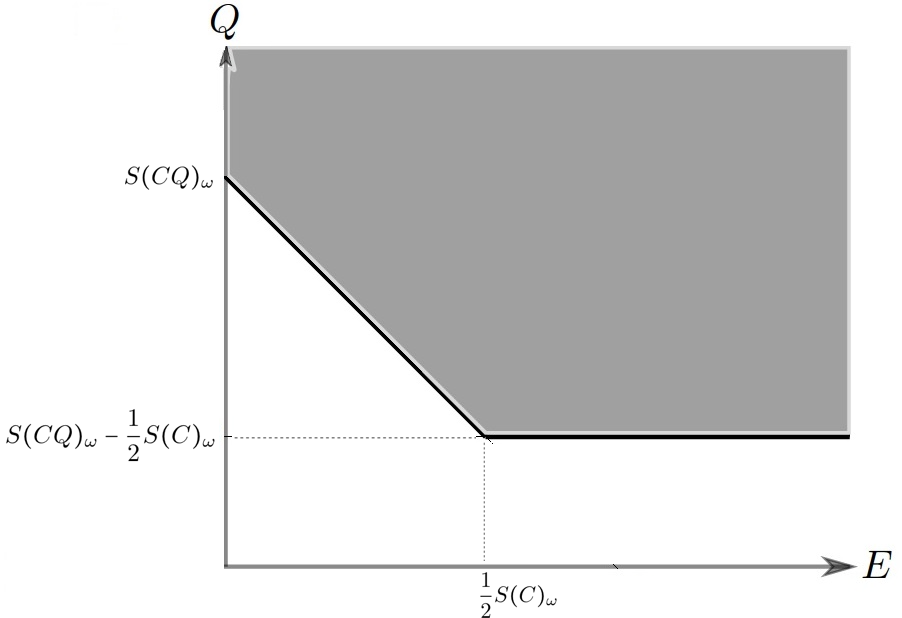} 
  \caption{The achievable rate region of the entanglement and quantum rates.}
  \label{fig:rate region}
\end{figure}

\begin{proof}
We start with the achievability of these rates. The converse proofs need more tools, so we will 
leave them to the subsequent sections. 
Looking at Fig. \ref{fig:rate region}, it will be enough to 
prove the achievability of the corresponding corner points $(E,Q)=(0,S(CQ)_{\omega})$ 
and $(E,Q)=(\frac{1}{2}S(C)_{\omega},S(CQ)_{\omega}-\frac{1}{2}S(C)_{\omega})$ 
for the unassisted and entanglement assisted cases, respectively. This is because
by definition (and the time-sharing principle) the rate region is convex and 
upper-right closed.
Indeed, all the points on the line $Q+E = S(CQ)_{\omega}$ for 
$Q \geq S(CQ)_{\omega}-\frac{1}{2}S(C)_{\omega}$ are achievable because one 
ebit can be distributed by sending a qubit. 
All other rate pairs are achievable by resource wasting. 
The rate region is depicted in Fig.~\ref{fig:rate region}.

As we discussed, we can assume that the source is 
$(\omega^{C N Q R})^{\otimes n}=\omega^{C^n N^n Q^n R^n}$. To achieve the point $(0,S(CQ)_{\omega})$, 
Alice traces out the redundant part $N^n$ of the source, to get the state $\omega^{C^n Q^n R^n}$ and applies Schumacher compression to send the systems $C^n Q^n$ to Bob. Since the 
Schumacher compression preserves the purification of the systems $C^n Q^n$, it preserves the state $\omega^{C^n Q^n R^n}$ as well. To be more specific, let $\Lambda_S$ denote the composition of the encoding and decoding operations for the Schumacher compression of the state
$\ket{\omega}^{C^n Q^n R^n {R'}^n}$ where the system ${R'}^n$ is a purifying reference system which of course the parties do not have access to. The Schumacher compression preserves the following fidelity on the left member of the equation, therefore it preserves the fidelity on the right member:
\begin{align}
   1-\epsilon 
     &\leq F \left({\omega}^{C^n Q^n R^n {R'}^n} ,(\Lambda_S \otimes \id_{R^n {R'}^n}) {\omega}^{C^n Q^n R^n {R'}^n}\right) \nonumber\\
     &\leq F \left({\omega}^{C^n Q^n R^n } ,(\Lambda_S \otimes \id_{R^n } ){\omega}^{C^n Q^n R^n }\right), \nonumber
\end{align}
where the inequality is due to monotonicity of the fidelity under partial trace. 
The rate achieved by this scheme is $S(CQ)_{\omega}$. After applying this scheme, 
Bob has access to the systems $\hat{C}^n \hat{Q}^n$, which is correlated with the reference 
system $R^n$:
\begin{align*}
  \zeta^{\hat{C}^n \hat{Q}^n R^n}=(\Lambda_S \otimes \id_{R^n } ){\omega}^{C^n Q^n R^n }.
\end{align*}
Then, to reconstruct the system $N^n$, Bob applies the CPTP map 
$\mathcal{N}:C Q  \longrightarrow C  N Q$ to each copy, which acts as follows:
\begin{align}
  \mathcal{N}(\rho^{CQ})=\sum_j (\proj{j}^{C} \otimes \1_{Q}) \rho^{CQ} (\proj{j}^{C} \otimes\1_{Q}) \otimes \omega_j^{N}.\nonumber   
\end{align}
%
%
This map satisfies the fidelity criterion of Eq.~(\ref{eq:fidelity omega})  because of monotonicity of the fidelity under CPTP maps:
\begin{align}\label{eq:fidelity omega}
    1-\epsilon &\leq F \left({\omega}^{C^n Q^n R^n } ,\zeta^{\hat{C}^n \hat{Q}^n R^n}\right) \nonumber\\
    & \leq F \left((\mathcal{N}^{\otimes n} \otimes \id_{R^n } ){\omega}^{C^n Q^n R^n } ,(\mathcal{N}^{\otimes n} \otimes \id_{R^n } ) \zeta^{\hat{C}^n \hat{Q}^n R^n}\right) \nonumber \\
    &= F \left({\omega}^{C^n N^n Q^n R^n } ,\tau^{\hat{C}^n \hat{N}^n \hat{Q}^n R^n}\right).
\end{align}

To achieve the point $(\frac{1}{2}S(C)_{\omega},S(CQ)_{\omega}-\frac{1}{2}S(C)_{\omega})$, 
Alice applies dense coding to send the classical system $C^n$ to Bob which requires 
$\frac{n}{2}S(C)_{\omega}$ ebits of initial entanglement and $\frac{n}{2}S(C)_{\omega}$ 
qubits \cite{Bennett1992}. When both Alice and Bob have access to system $C^n$, 
Alice can send the quantum system $Q^n$ to Bob by applying Schumacher compression, which 
requires sending $nS(Q|C)$ qubits to Bob.
Therefore, the overall qubit rate is 
$\frac{1}{2}S(C)_{\omega}+S(Q|C)=S(CQ)_{\omega}-\frac{1}{2}S(C)_{\omega}$.
\end{proof}

\section{Converse}
\label{sec:converse}
In this section, we will provide the converse bounds for the qubit rate $Q$ and the sum 
rate $Q+E$ of Theorem~\ref{theorem:complete rate region mixed state}.
We obtain these bounds based on the structure of the CPTP maps which preserve the source state
$\omega^{CNQR}$. Namely, according to Theorem~\ref{thm: KI decomposition} the CPTP maps 
acting on systems $CNQ$, which preserve the state $\omega^{CNQR}$, act only on the 
redundant system $N$. This implies that the environment systems of such CPTP maps are 
decoupled from systems $Q R$ given the classical information $j$ in the classical system $C$.
This gives us an insight into the structure of the  encoding-decoding maps, which preserve 
the overall state \emph{asymptotically} intact. 

To proceed with the proof, we first define two functions that emerge in the converse bounds.
Then, we state some important properties of these functions in 
Lemma~\ref{lemma:J_epsilon Z_epsilon properties} which we will use to
compute the tight asymptotic converse bounds.

\begin{definition}
  \label{def:J_epsilon Z_epsilon}
  For the KI decomposition  
  $\omega^{C N Q R}=\sum_{j} p_j \proj{j}^{C}\otimes \omega_j^{N} \otimes \rho_{j}^{Q R}$
  of the state $\rho^{AR}$ and $\epsilon \geq 0$, define
  \begin{align*}
    J_\epsilon(\omega) &:=  
        \max I(\hat{N} E:\hat{C}\hat{Q}|C')_\tau 
                  \text{ s.t. } \\
                       & \quad \quad U:C N Q \rightarrow \hat{C} \hat{N} \hat{Q} E
                  \text{ is an isometry with } 
                  F( \omega^{C N Q R},\tau^{\hat{C} \hat{N} \hat{Q}R})  \geq 1- \epsilon,  \\
    Z_\epsilon(\omega) &:=  
        \max S(\hat{N} E|C')_\tau 
          \text{ s.t. } \\
          & \quad \quad U:C N Q \rightarrow \hat{C} \hat{N} \hat{Q} E 
          \text{ is an isometry with } 
          F( \omega^{C N Q R},\tau^{\hat{C} \hat{N} \hat{Q}R})  \geq 1- \epsilon,  
\end{align*}
where
\begin{align*}
  \omega^{C N Q  R C'} 
     &=\sum_{j} p_j \proj{j}^{C}\otimes \omega_j^{N} \otimes \rho_{j}^{Q R} \otimes \proj{j}^{C'}, \\
  \tau^{\hat{C} \hat{N} \hat{Q} ER C'}
     &= (U \otimes \1_{RC'}) \omega^{C N Q R C'}    (U^{\dagger} \otimes \1_{RC'}), \\
  \tau^{\hat{C} \hat{N} \hat{Q} R}
     &=\Tr_{E C'} [ \tau^{\hat{C} \hat{N} \hat{Q}  ER C'}].
\end{align*}
\end{definition}
In this definition, the dimension of the environment is w.l.o.g. bounded as $|E| \leq (|C||N||Q|)^2$
because the input and output dimensions of the channel are fixed as $|C||N||Q|$; 
hence, the optimisation is of a continuous function over a compact domain, so we have a 
maximum rather than a supremum.

\begin{lemma}
  \label{lemma:J_epsilon Z_epsilon properties}
  The functions $Z_\epsilon(\omega)$ and $J_\epsilon(\omega)$ have the following properties:
  \begin{enumerate}
    \item They are non-decreasing functions of $\epsilon$. 
    \item They are concave in $\epsilon$.
    \item They are continuous for $\epsilon \geq 0$. 
    \item For any two states $\omega_1^{C_1 N_1 Q_1 R_1}$ and $\omega_2^{C_2 N_2 Q_2 R_2}$ and for $\epsilon \geq 0$,
    \begin{align*}
        &J_{\epsilon}(\omega_1 \otimes \omega_2) \leq  J_{\epsilon}(\omega_1) +J_{\epsilon}(\omega_2),\\  
        &Z_{\epsilon}(\omega_1 \otimes \omega_2) \leq  Z_{\epsilon}(\omega_1) +Z_{\epsilon}(\omega_2).  
    \end{align*}
    \item At $\epsilon=0$, $Z_0(\omega) =S(N|C)_\omega$ and $J_0(\omega) =0$.
  \end{enumerate}
\end{lemma}

The proof of this lemma follows in the next section. Now we show how it is
used to prove the converse (optimality) of Theorem \ref{theorem:complete rate region mixed state}.
As a guide to reading the subsequent proof, we remark that in Eqs.~(\ref{eq:converse_Q_1}) 
and (\ref{eq:converse_Q+E_2}), the environment systems $VW$ of the encoding-decoding 
operations appear in the terms $I(\hat{N}^n VW : \hat{C}^n  \hat{Q}^n| {C'}^n)$ and 
$S(\hat{N}^n VW | {C'}^n)$, which are bounded by  
the functions $J_{\epsilon}(\omega^{\otimes n})$ and $Z_{\epsilon}(\omega^{\otimes n})$, 
respectively.
As stated in point 4 of Lemma~\ref{lemma:J_epsilon Z_epsilon properties}, these functions 
are sub-additive, so basically we can single-letterize the terms appearing in the converse.
Moreover, from point 3 of Lemma~\ref{lemma:J_epsilon Z_epsilon properties}, we know that
these functions are continuous for $\epsilon \geq 0$; therefore, the limit points of these 
functions are equal to the values of these functions at $\epsilon=0$.
When the fidelity is equal to 1 ($\epsilon=0$), the structure of the CPTP maps 
preserving the state $\omega^{C N Q R}$ in Theorem~\ref{thm: KI decomposition} 
implies that $J_{0}(\omega)=0$ and  $Z_{0}(\omega)=S(N|C)_{\omega}$, as stated 
in point 5 of Lemma~\ref{lemma:J_epsilon Z_epsilon properties}. Thereby, we conclude 
the converse bounds in Eqs.~(\ref{eq:converse_Q_asymptotics}) and (\ref{eq:converse_Q+E_asymptotics}).

\begin{proof-of}[of Theorem \ref{theorem:complete rate region mixed state} (converse)]
We first get the following chain of inequalities considering the process of the 
decoding of the information: 
\begin{align}
  nQ+S(B_0)&\geq S(M)+S(B_0)      \label{eq:converse_decoding_1_1}  \\
           &\geq S(M B_0)         \label{eq:converse_decoding_1_2} \\
           &= S(\hat{C}^n \hat{N}^n \hat{Q}^n V)     \label{eq:converse_decoding_1_4} \\
           &= S(\hat{C}^n\hat{Q}^n) + S(\hat{N}^n V| \hat{C}^n  \hat{Q}^n)     \label{eq:converse_decoding_1_5} \\
           &\geq nS(C  Q) + S(\hat{N}^n V| \hat{C}^n  \hat{Q}^n)    -n \delta(n,\epsilon)    \label{eq:converse_decoding_1_6}\\
           &\geq nS(C  Q) + S(\hat{N}^n V| \hat{C}^n  \hat{Q}^n {C'}^n)    -n \delta(n,\epsilon)    \label{eq:converse_decoding_1_7}\\
           &= nS(C  Q) + S(\hat{N}^n V| \hat{C}^n  \hat{Q}^n {C'}^n) - S(\hat{N}^n V| {C'}^n)\nonumber \\
           &\quad+ S(\hat{N}^n V| {C'}^n)-n \delta(n,\epsilon)  \nonumber\\
           &=  nS(C  Q) - I(\hat{N}^n V : \hat{C}^n  \hat{Q}^n| {C'}^n)
            + S(\hat{N}^n V| {C'}^n)  -n \delta(n,\epsilon)\nonumber\\
           &\geq nS(C Q) - I(\hat{N}^n VW : \hat{C}^n  \hat{Q}^n| {C'}^n)+ S(\hat{N}^n V| {C'}^n)-n \delta(n,\epsilon) \label{eq:converse_decoding_1}
\end{align}
where Eq.~(\ref{eq:converse_decoding_1_1}) follows because the entropy of a system 
is bounded by the logarithm of the dimension of that system;
Eq.~(\ref{eq:converse_decoding_1_2}) is due to sub-additivity of the entropy;
Eq.~(\ref{eq:converse_decoding_1_4}) follows because the decoding isometry 
$U_{\cD}:M B_0 \hookrightarrow \hat{C}^n \hat{N}^n \hat{Q}^n V$ does not change the entropy;
Eq.~(\ref{eq:converse_decoding_1_5}) is due to the chain rule;
Eq.~(\ref{eq:converse_decoding_1_6}) follows from the decodability: the 
output state on systems $\hat{C}^n \hat{Q}^n$ is $2\sqrt{2\epsilon}$-close 
to the original state $C^n  Q^n$ in trace norm; then the inequality follows 
by applying the Fannes-Audenaert inequality 
\cite{Fannes1973,Audenaert2007}, where 
$\delta(n,\epsilon)=\sqrt{2\epsilon} \log(|C||Q|) + \frac1n h(\sqrt{2\epsilon})$;
Eq.~(\ref{eq:converse_decoding_1_7}) is due to strong sub-additivity of the entropy,
and system $C'$  is a copy of classical system $C$;
Eq.~(\ref{eq:converse_decoding_1}) 
follows from data processing inequality where $W$ is the environment system 
of the encoding isometry $U_{\cE}:C^n N^n Q^n A_0 \hookrightarrow M W$.

Moreover, considering the process of encoding the information, $Q$ is bounded as follows:
\begin{align}
  nQ &\geq S(M)  \nonumber \\                   
     &\geq S(M|W  {C'}^n)           \label{eq:converse_encoding_1_1} \\
     &=    S(M W  {C'}^n) -S(W  {C'}^n) \label{eq:converse_encoding_1_2} \\
     &=    S(C^n N^n Q^n A_0  {C'}^n)-S(W {C'}^n) \label{eq:converse_encoding_1_4} \\
     &=    S(C^n N^n Q^n {C'}^n)+S(A_0)-S(W {C'}^n) \label{eq:converse_encoding_1_5} \\
     &=    S(C^n N^n Q^n  {C'}^n)+S(A_0)-S({C'}^n)-S(W  |{C'}^n) \label{eq:converse_encoding_1_6} \\
     &=    S(C^n N^n Q^n)+S(A_0)-S({C'}^n)-S(W |{C'}^n) \label{eq:converse_encoding_1_7} \\
     &=    nS(C Q)+nS(N|C Q)+S(A_0)-nS(C')-S(W  |{C'}^n) \label{eq:converse_encoding_1_8} \\
     &=    nS(C   Q )+nS(N |C)+S(A_0)-nS(C')-S(W |{C'}^n),  \label{eq:converse_encoding_1}
\end{align}
where Eq.~(\ref{eq:converse_encoding_1_1}) is due to sub-additivity of the entropy;
%
Eq.~(\ref{eq:converse_encoding_1_2}) is due to the chain rule;
Eq.~(\ref{eq:converse_encoding_1_4}) follows because the encoding isometry 
$U_{\cE}:C^n N^n Q^n A_0 \hookrightarrow M W$ does not the change the entropy; 
Eq.~(\ref{eq:converse_encoding_1_5}) follows because the initial entanglement $A_0$ 
is independent from the source;
Eq.~(\ref{eq:converse_encoding_1_6}) is due to the chain rule;
%
%
Eq.~(\ref{eq:converse_encoding_1_7}) follows because $C'$ is a copy of the system $C$, 
so $S(C'|C  N Q)=0$;
Eq.~(\ref{eq:converse_encoding_1_8}) is due to the chain rule and the fact that the entropy is additive
for product states;
Eq.~(\ref{eq:converse_encoding_1}) follows because conditional on system $C$ 
the system $N$ is independent from system $Q$. 

Now, we add Eqs.~(\ref{eq:converse_decoding_1}) and 
(\ref{eq:converse_encoding_1}); the entanglement terms $S(A_0)$ and $S(B_0)$ cancel out,
and by dividing by $2n$ we obtain
\begin{align}
  Q \!&\geq S(C Q)-\frac{1}{2}S(C)\! +\frac{1}{2}S(N |C)\!-\frac{1}{2n} I(\hat{N}^n VW  : \hat{C}^n  \hat{Q}^n| {C'}^n) \nonumber \\
  & \quad + \frac{1}{2n}S(\hat{N}^n V| {C'}^n)\!-\frac{1}{2n}S(W |{C'}^n)\!-\frac{1}{2} \delta(n,\epsilon)  \nonumber\\
  &\geq S(C Q)-\frac{1}{2}S(C) +\frac{1}{2}S(N |C)-\frac{1}{2n} I(\hat{N}^n VW : \hat{C}^n  \hat{Q}^n| {C'}^n) \nonumber \\
  &\quad - \frac{1}{2n}S(\hat{N}^n VW | {C'}^n) -\frac{1}{2} \delta(n,\epsilon)  \label{eq:converse_Q_1}\\
  &\geq S(C Q)-\frac{1}{2}S(C) +\frac{1}{2}S(N |C)-\frac{1}{2n}J_{\epsilon}(\omega^{\otimes n})-\frac{1}{2n} Z_{\epsilon}(\omega^{\otimes n})- \frac{1}{2} \delta(n,\epsilon) \label{eq:converse_Q_2} \\ 
  &\geq S(C Q)-\frac{1}{2}S(C) +\frac{1}{2}S(N |C)-\frac{1}{2}J_{\epsilon}(\omega)-\frac{1}{2} Z_{\epsilon}(\omega)- \frac{1}{2} \delta(n,\epsilon), \label{eq:converse_Q}   
\end{align} 
where Eq.~(\ref{eq:converse_Q_1}) follows from strong sub-additivity of the entropy, 
$S(\hat{N}^n V| {C'}^n)+S(\hat{N}^n V| W  {C'}^n)\geq 0$; 
Eq.~(\ref{eq:converse_Q_2}) follows from Definition~\ref{def:J_epsilon Z_epsilon};
Eq.~(\ref{eq:converse_Q}) is due to point 4 of Lemma~\ref{lemma:J_epsilon Z_epsilon properties}.

In the limit of $\epsilon \to 0$ and 
$n \to \infty $, the qubit rate is thus bounded by
\begin{align}\label{eq:converse_Q_asymptotics}
  Q &\geq S(C Q)-\frac{1}{2}S(C) +\frac{1}{2}S(N |C)-\frac{1}{2}J_{0}(\omega)-\frac{1}{2} Z_{0}(\omega) \nonumber\\
    &= S(C Q)-\frac{1}{2}S(C),
\end{align}
where the equality follows from point 5 of Lemma~\ref{lemma:J_epsilon Z_epsilon properties}. 

Moreover, from Eq.~(\ref{eq:converse_decoding_1}) we have:
\begin{align}
   nQ+S(B_0)&= nQ+nE     \nonumber  \\
            &\geq nS(C Q) - I(\hat{N}^n VW : \hat{C}^n  \hat{Q}^n| {C'}^n)+ S(\hat{N}^n V| {C'}^n)-n \delta(n,\epsilon) \nonumber \\
            &\geq nS(C Q) - I(\hat{N}^n VW : \hat{C}^n  \hat{Q}^n| {C'}^n)-n \delta(n,\epsilon) \label{eq:converse_Q+E_2} \\
            &\geq nS(C Q) - J_{\epsilon}(\omega^{\otimes n})-n \delta(n,\epsilon) \label{eq:converse_Q+E_3} \\
            &\geq nS(C Q) - nJ_{\epsilon}(\omega)-n \delta(n,\epsilon), \label{eq:converse_Q+E}  
\end{align}
where Eq.~(\ref{eq:converse_Q+E_2}) follows because the entropy conditional on a 
classical system is positive, $S(\hat{N}^n V| {C'}^n) \geq 0$;
Eq.~(\ref{eq:converse_Q+E_3}) follows from Definition~\ref{def:J_epsilon Z_epsilon};
Eq.~(\ref{eq:converse_Q+E})  is due to point 4 of Lemma~\ref{lemma:J_epsilon Z_epsilon properties}. 

In the limit of $\epsilon \to 0$ and 
$n \to \infty $, we thus obtain the following bound on the rate sum:
\begin{align}
  Q+E \geq S(CQ) - J_{0}(\omega)
      =    S(CQ) \label{eq:converse_Q+E_asymptotics},  
\end{align}
where the equality follows from point 5 of Lemma~\ref{lemma:J_epsilon Z_epsilon properties}. 
\end{proof-of}

\begin{remark}
Our lower bound on $Q+E$ in Eq. (\ref{eq:converse_Q+E_asymptotics}) reproduces the 
result of Koashi and Imoto \cite{KI2001} for the case of a classical-quantum source
$\rho^{AX} = \sum_x p(x) \rho_x^A \otimes \proj{x}^X$. This is because a code with
qubit-ebit rate pair $(Q,E)$ gives rise to a compression code in the sense of 
Koashi and Imoto using a rate of qubits $Q+E$ and no prior entanglement, simply by
first distributing $E$ ebits and then using the entanglement assisted code. 

It is worth noting that conversely, Eq. (\ref{eq:converse_Q+E_asymptotics}) can 
be obtained from the Koashi-Imoto result, as follows. Any good code for $\rho^{AR}$
is automatically a good code for the classical-quantum source of mixed states
\[
  \rho^{AY} = \sum_y q(y) \rho_y^A \otimes \proj{y}^Y
            = \sum_y \Tr_R \rho^{AR}(\1_A\otimes M_y^R) \otimes \proj{y}^Y, 
\]
for any POVM $(M_y)$ on $R$, simply by the monotonicity of the fidelity 
under CPTP maps. As discussed before, by choosing an informationally complete
measurement, the KI-decomposition of the ensemble $\{q(y),\rho_y^A\}$ is 
identical to that of $\rho^{AR}$ in Theorem \ref{thm: KI decomposition}.
Thus the unassisted qubit compression rate of $\rho^{AY}$ and of $\rho^{AR}$
are lower bounded by the same quantity, the right hand side of Eq. (\ref{eq:converse_Q+E_asymptotics}).
\end{remark}

\section{Proof of Lemma~\ref{lemma:J_epsilon Z_epsilon properties}}
\label{sec: Proof of Lemma}

\begin{enumerate}
\item The definitions of the functions $J_{\epsilon}(\omega)$ and $Z_{\epsilon}(\omega)$ 
  directly imply that they are non-decreasing functions of $\epsilon$.

\item We first prove the concavity of $Z_{\epsilon}(\omega)$. 
  Let $U_1:C N Q \hookrightarrow \hat{C} \hat{N} \hat{Q} E$ and 
  $U_2:C N Q \hookrightarrow \hat{C} \hat{N} \hat{Q} E$ be the isometries attaining the 
  maximum for $\epsilon_1$ and $\epsilon_2$, respectively, which act as 
  follows on the purification $\ket{\omega}^{C N Q R C' R'}$ of the previously
  introduced state $\omega^{C N Q R C'}$:
\begin{align*}
    &\ket{\tau_1}^{\hat{C} \hat{N} \hat{Q} E R C' R'}
        =(U_1 \otimes \1_{R C' R'}) \ket{\omega}^{C N Q R C' R'}
    \quad \text{ and } \\
    &\ket{\tau_2}^{\hat{C} \hat{N} \hat{Q} E R C' R'}
        =(U_2 \otimes \1_{R C' R'}) \ket{\omega}^{C N Q R C' R'},   
\end{align*}
  where $\Tr_{R'}[\proj{\omega}^{C N Q R C' R'}]=\omega^{C N Q R C'}$. 
  For $0\leq \lambda \leq 1$, define the isometry 
  $U_0:C N Q \hookrightarrow \hat{C} \hat{N} \hat{Q} E F F'$ which acts as 
  \begin{equation}
    \label{eq: isometry U in convexity}
    U_0 := \sqrt{\lambda} U_1 \otimes \ket{11}^{FF'} + \sqrt{1-\lambda} U_2 \otimes \ket{22}^{FF'},
  \end{equation}
  where systems $F$ and $F'$ are qubits, and
  which leads to the state
  \begin{align*}
    (U_0 \otimes \1_{R C' R'})& \ket{\omega}^{C N Q R C' R'}\\
       &= \sqrt{\lambda}\ket{\tau_1}^{\hat{C} \hat{N} \hat{Q} E R C' R'} \ket{11}^{FF'}
        + \sqrt{1-\lambda}\ket{\tau_2}^{\hat{C} \hat{N} \hat{Q} E R C' R'} \ket{22}^{FF'}.   
  \end{align*}
  Then, $U_0$ defines its state $\tau$. for which the reduced state on the systems 
  $\hat{C} \hat{N} \hat{Q}  R C'$ is 
  \begin{align} \label{eq: tau in convexity proof}
    \tau^{\hat{C} \hat{N} \hat{Q} R C'} 
      =\lambda \tau_1^{\hat{C} \hat{N} \hat{Q} R C'}+ (1-\lambda) \tau_2^{\hat{C} \hat{N} \hat{Q} R C'}. 
  \end{align}  
  Therefore, the fidelity for the state $\tau$ is bounded as follows:
  \begin{align}\label{eq:fidelity in convexity}
    F(\omega^{C N Q R} &,\tau^{\hat{C} \hat{N} \hat{Q} R} ) \nonumber \\
      &= F(\omega^{C N Q R} ,\lambda \tau_1^{\hat{C} \hat{N} \hat{Q} R}
        + (1-\lambda) \tau_2^{\hat{C} \hat{N} \hat{Q} R}) \nonumber \\
      &= F(\lambda \omega^{C N Q R}+(1-\lambda)\omega^{C N Q R},
           \lambda \tau_1^{\hat{C} \hat{N} \hat{Q} R}
            + (1-\lambda) \tau_2^{\hat{C} \hat{N} \hat{Q} R}) \nonumber\\
      &\geq \lambda F( \omega^{C N Q R},\tau_1^{\hat{C} \hat{N} \hat{Q} R})
            +(1-\lambda)F( \omega^{C N Q R},\tau_2^{\hat{C} \hat{N} \hat{Q} R}) \nonumber\\
     &\geq 1-\left( \lambda\epsilon_1 +(1-\lambda)\epsilon_2 \right).
  \end{align}
  The first inequality is due to simultaneous concavity of the fidelity in both
  arguments;
  the last line follows by the definition of the isometries $U_1$ and $U_2$.
  Thus, the isometry $U_0$ yields a fidelity of at least 
  $1-\left( \lambda\epsilon_1 +(1-\lambda)\epsilon_2 \right) =: 1-\epsilon$.
  Now let $E'=E FF'$ denote the environment of the isometry $U_0$ defined above. 
  According to Definition \ref{def:J_epsilon Z_epsilon}, we obtain
  \begin{align}
    Z_\epsilon(\omega) &\geq S(\hat{N} E'|C')_{\tau} \nonumber\\ 
                     &= S(\hat{N} EFF'|C')_{\tau} \nonumber\\ 
                     &= S(F|C')_{\tau}+S(\hat{N} E|F C')_{\tau}+S(F'|\hat{N} EF C')_{\tau} \label{eq:Z_concavity_1}\\
                     &\geq S(\hat{N}E|FC')_{\tau} \label{eq:Z_concavity_2}\\
                     &= \lambda S(\hat{N} E|C')_{\tau_1}+(1-\lambda) S(\hat{N} E|C')_{\tau_2}\label{eq:Z_concavity_3}\\
                     &= \lambda Z_{\epsilon_1}(\omega)+(1-\lambda)Z_{\epsilon_2}(\omega) \label{eq:Z_concavity_4},
  \end{align}
  where the state $\tau$ in the entropies is given in Eq.~(\ref{eq: tau in convexity proof});
  Eq.~(\ref{eq:Z_concavity_1}) is due to the chain rule; 
  Eq.~(\ref{eq:Z_concavity_2}) follow because 
  for the state on systems $\hat{N} EFF' C' $ we have $S(F'|C')+S(F'|\hat{N} E F C')\geq 0$ 
  which follows from strong sub-additivity of the entropy; 
  Eq.~(\ref{eq:Z_concavity_3}) follows by expanding the conditional entropy on the classical system $F$; 
  Eq.~(\ref{eq:Z_concavity_4}) follows from the definitions of the isometries $U_1$ and $U_2$.

  Moreover, let $U_1:C N Q \hookrightarrow \hat{C} \hat{N} \hat{Q} E$ and 
  $U_2:C N Q \hookrightarrow \hat{C} \hat{N} \hat{Q} E$ be the isometries attaining the 
  maximum for $\epsilon_1$ and $\epsilon_2$ in the definition of $J_{\epsilon}(\omega)$, respectively.
  Again, define the isometry $U_0$ as in Eq.~(\ref{eq: isometry U in convexity}),
  which leads to the bound on the fidelity as in Eq.~(\ref{eq:fidelity in convexity}),
  letting $E'=EFF'$ be the environment of the isometry $U_0$. 
  According to Definition \ref{def:J_epsilon Z_epsilon}, we obtain
  \begin{align}
    J_\epsilon(\omega) &\geq I(\hat{N} E FF':\hat{C} \hat{Q}|C')_{\tau} \nonumber \\
                     &\geq I(\hat{N} E F:\hat{C} \hat{Q}|C')_{\tau} \label{eq:concavity_J_1} \\
                     &= I(F:\hat{C} \hat{Q}|C')_{\tau}+I(\hat{N} E :\hat{C} \hat{Q}|F C')_{\tau} \label{eq:concavity_J_2} \\
                     &\geq I(\hat{N} E :\hat{C} \hat{Q}|F C')_{\tau} \label{eq:concavity_J_3} \\
                     &=    \lambda I(\hat{N} E :\hat{C} \hat{Q}| C')_{\tau_1}+(1-\lambda) I(\hat{N} E :\hat{C} \hat{Q}| C')_{\tau_2} \label{eq:concavity_J_4}\\
                     &=    \lambda J_{\epsilon_1}(\omega)+(1-\lambda)J_{\epsilon_2}(\omega)\label{eq:concavity_J_5},
  \end{align}
  where Eq.~(\ref{eq:concavity_J_1}) follows from data processing; 
  Eq.~(\ref{eq:concavity_J_2}) is due to the chain rule for mutual information;
  Eq.~(\ref{eq:concavity_J_3}) follows from strong sub-additivity of the 
  entropy, $I(F:\hat{C} \hat{Q}|C')_{\tau} \geq 0$;
  Eq.~(\ref{eq:concavity_J_4}) is obtained by expanding the conditional mutual 
  information  on the classical system $F$; 
  finally, Eq.~(\ref{eq:concavity_J_5}) follows from the definitions of the isometries $U_1$ and $U_2$.

\item The functions are non-decreasing and concave for $\epsilon \geq 0 $, so they are continuous 
  for $\epsilon > 0$. 
  The concavity implies furthermore that $J_{\epsilon}$ and $Z_{\epsilon}$ are lower semi-continuous at 
  $\epsilon=0$. On the other hand, since the fidelity, the conditional entropy and the conditional 
  mutual information are all continuous functions of CPTP maps, and the domain of both optimizations 
  is a compact set, we conclude that $J_\epsilon(\omega)$ and $Z_{\epsilon}$ are also upper 
  semi-continuous at $\epsilon=0$, so they are continuous at $\epsilon=0$ 
  \cite[Thms.~10.1 and 10.2]{Rockafeller}. 

\item We first prove 
  $Z_{\epsilon}(\omega_1 \otimes \omega_2) \leq  Z_{\epsilon}(\omega_1) +Z_{\epsilon}(\omega_2)$.
  In the definition of $Z_{\epsilon}(\omega_1 \otimes \omega_2)$, let the isometry 
  $U_0:C_1 N_1 Q_1 C_2 N_2 Q_2 \hookrightarrow \hat{C}_1 \hat{N}_1 \hat{Q}_1 \hat{C}_2 \hat{N}_2 \hat{Q}_2 E$
  be the one attaining the maximum, which acts on the following purified source states with purifying 
  systems $R'_1$ and $R'_2$: 
  \begin{align}
    &\ket{\tau}^{\hat{C}_1 \hat{N}_1 \hat{Q}_1 \hat{C}_2 \hat{N}_2 \hat{Q}_2 E R_1 C'_1 R'_1 R_2 C'_2 R'_2}\\
    \quad \quad \quad \quad &=(U_0 \otimes \1_{R_1 C'_1 R'_1 R_2 C'_2 R'_2})\ket{\omega_1}^{C_1 N_1 Q_1 R_1 C'_1 R'_1}
                                                   \otimes    \ket{\omega_2}^{C_2 N_2 Q_2 R_2 C'_2 R'_2}. \label{eq:U0-action}
  \end{align}
  By definition, the fidelity is bounded by
  \begin{align*}
    F(\omega_1^{C_1 N_1 Q_1 R_1} \otimes \omega_2^{C_2 N_2 Q_2 R_2},
      \tau^{\hat{C}_1 \hat{N}_1 \hat{Q}_1 \hat{C}_2 \hat{N}_2 \hat{Q}_2 R_1 R_2}) \geq 1- \epsilon.   
  \end{align*}
  Now, we can define an isometry 
  $U_1:C_1 N_1 Q_1 \hookrightarrow \hat{C}_1 \hat{N}_1 \hat{Q}_1 E_1$ 
  acting only on systems $C_1 N_1 Q_1$, by letting
  $U_1 = (U_0 \otimes \1_{R_2 C_2' R_2'})(\1_{C_1 N_1 Q_1} \otimes \ket{\omega_2}^{C_2 N_2 Q_2 R_2 C_2' R_2})$
  and with the environment $E_1 := \hat{C}_2 \hat{N}_2 \hat{Q}_2 E R_2  C'_2 R'_2$.
  It has the property that 
  $\ket{\tau}^{\hat{C}_1 \hat{N}_1 \hat{Q}_1 R_1 C_1' R_1' E} 
   = (U_1 \otimes \1_{R_1 C_1' R_1'})\ket{\omega_1}^{C_1 N_1 Q_1 R_1 C_1' R_1'}$ 
  has the same reduced state on $\hat{C}_1 \hat{N}_1 \hat{Q}_1 R_1$ as $\tau$ from
  Eq. (\ref{eq:U0-action}).
  This isometry preserves the fidelity for $\omega_1$, which follows from monotonicity 
  of the fidelity under partial trace:
  \begin{align*}
     F(\omega_1^{C_1 N_1 Q_1 R_1},\tau_1^{\hat{C}_1 \hat{N}_1 \hat{Q}_1 R_1}) 
       &= F(\omega_1^{C_1 N_1 Q_1 R_1},\tau^{\hat{C}_1 \hat{N}_1 \hat{Q}_1 R_1}) \\
       &\geq F(\omega_1^{C_1 N_1 Q_1 R_1} \otimes \omega_2^{C_2 N_2 Q_2 R_2},
               \tau^{\hat{C}_1 \hat{N}_1 \hat{Q}_1 \hat{C}_2 \hat{N}_2 \hat{Q}_2 R_1 R_2}) \\
       &\geq 1- \epsilon.   
  \end{align*}
  By the same argument, there is the following isometry 
  \begin{align*}
      U_2:C_2 N_2 Q_2\hookrightarrow \hat{C}_1 \hat{N}_1 \hat{Q}_1 \hat{C}_2 \hat{N}_2 \hat{Q}_2 E R_1 C'_1 R'_1,
  \end{align*}
  with output system $\hat{C}_2 \hat{N}_2 \hat{Q}_2$ and 
  environment $E_2:=\hat{C}_1 \hat{N}_1 \hat{Q}_1 E R_1 C'_1 R'_1$, such that
  \begin{align*}
    F(\omega_2^{C_2 N_2 Q_2 R_2},\tau_2^{\hat{C}_2 \hat{N}_2 \hat{Q}_2 R_2}) 
      &=    F(\omega_2^{C_2 N_2 Q_2 R_2},\tau^{\hat{C}_2 \hat{N}_2 \hat{Q}_2 R_2}) \\
      &\geq F(\omega_1^{C_1 N_1 Q_1 R_1} \otimes \omega_2^{C_2 N_2 Q_2 R_2},
              \tau^{\hat{C}_1 \hat{N}_1 \hat{Q}_1 \hat{C}_2 \hat{N}_2 \hat{Q}_2 R_1 R_2}) \\
      &\geq 1- \epsilon.   
  \end{align*}
  Therefore, we obtain:
  \begin{align}
     Z_{\epsilon}(\omega_1) &+Z_{\epsilon}(\omega_2)-Z_{\epsilon}(\omega_1 \otimes \omega_2) \nonumber\\
     &\geq
     S(\hat{N}_1 E_1|C'_1)_{\tau}+S(\hat{N}_2 E_2|C'_2)_{\tau}-S(\hat{N}_1
     \hat{N}_2 E |C'_1 C'_2)_{\tau} \label{eq:Z_additivity_1}\\
     &=S(\hat{N}_1 E_1 C'_1)_{\tau}+S(\hat{N}_2 E_2C'_2)_{\tau}-S(\hat{N}_1
     \hat{N}_2 E C'_1 C'_2)_{\tau}\nonumber \\
      &\quad \quad\quad\quad\quad\quad \quad\quad\quad -S(C'_1)-S(C'_2)+S(C'_1 C'_2) \label{eq:Z_additivity_2}\\
     &=S(\hat{N}_1 E_1 C'_1)_{\tau}+S(\hat{N}_2 E_2C'_2)_{\tau}-S(\hat{N}_1
     \hat{N}_2 E C'_1 C'_2)_{\tau} \label{eq:Z_additivity_3}\\
     &=S(\hat{C}_1\hat{Q}_1 R_1 R'_1)+S(\hat{C}_2\hat{Q}_2 R_2 R'_2)-S(\hat{C}_1\hat{Q}_1 \hat{C}_2\hat{Q}_2 R_1 R'_1 R_2 R'_2) \label{eq:Z_additivity_4}\\
     &=I(\hat{C}_1\hat{Q}_1 R_1 R'_1:\hat{C}_2\hat{Q}_2 R_2 R'_2) \nonumber\\
     &\geq 0 \label{eq:Z_additivity_5},
  \end{align}
  where Eq.~(\ref{eq:Z_additivity_1}) is due to Definition~\ref{def:J_epsilon Z_epsilon};
  Eq.~(\ref{eq:Z_additivity_2}) is due to the chain rule;
  Eq.~(\ref{eq:Z_additivity_3}) because the systems $C'_1$ and $C'_2$ are independent from each other;
  Eq.~(\ref{eq:Z_additivity_4}) follows because the overall state on systems 
  $\hat{C}_1 \hat{N}_1 \hat{Q}_1 \hat{C}_2 \hat{N}_2 \hat{Q}_2 E R_1 C'_1 R'_1 R_2 C'_2 R'_2$ 
  is pure;
  Eq.~(\ref{eq:Z_additivity_5}) is due to sub-additivity of the entropy. 

  To prove prove 
  $J_{\epsilon}(\omega_1 \otimes \omega_2) \leq J_{\epsilon}(\omega_1) +J_{\epsilon}(\omega_2)$,
  let the isometry 
  $U_0:C_1 N_1 Q_1 C_2 N_2 Q_2 \hookrightarrow \hat{C}_1 \hat{N}_1 \hat{Q}_1 \hat{C}_2 \hat{N}_2 \hat{Q}_2 E$ 
  be the one attaining the maximum in definition of $J_{\epsilon}(\omega_1 \otimes \omega_2)$,
  which acts on the following purified source states with purifying 
  systems $R'_1$ and $R'_2$, as in Eq. (\ref{eq:U0-action}).
  By definition, the fidelity is bounded as
  \begin{align*}
    F(\omega_1^{C_1 N_1 Q_1 R_1} \otimes \omega_2^{C_2 N_2 Q_2 R_2},
      \tau^{\hat{C}_1 \hat{N}_1 \hat{Q}_1 \hat{C}_2 \hat{N}_2 \hat{Q}_2 R_1 R_2}) 
               \geq 1- \epsilon.   
  \end{align*}
  Now define 
  $U_1:C_1 N_1 Q_1\hookrightarrow \hat{C}_1 \hat{N}_1 \hat{Q}_1 \hat{C}_2 \hat{N}_2 \hat{Q}_2 E R_2  C'_2 R'_2$ 
  and $U_2:C_2 N_2 Q_2\hookrightarrow \hat{C}_1 \hat{N}_1 \hat{Q}_1 \hat{C}_2 \hat{N}_2 \hat{Q}_2 E R_1 C'_1 R'_1$ 
  as in the above discussion, with the environments 
  $E_1:=\hat{C}_2 \hat{N}_2 \hat{Q}_2 E R_2 C'_2 R'_2$ and 
  $E_2:=\hat{C}_1 \hat{N}_1 \hat{Q}_1 E R_1 C'_1 R'_1$, respectively. 
  Recall that the fidelity for the states $\omega_1$ and $\omega_2$ is at least 
  $1-\epsilon$, because of the monotonicity of the fidelity under partial trace. 
  Thus we obtain
  \begin{align}
  J_{\epsilon}(\omega_1) 
     &+J_{\epsilon}(\omega_2)-J_{\epsilon}(\omega_1 \otimes \omega_2) \nonumber\\
     &\geq I(\hat{N}_1 E_1:\hat{C}_1\hat{Q}_1|C'_1)_\tau+
       I(\hat{N}_2 E_2:\hat{C}_2\hat{Q}_2|C'_2)_\tau \nonumber \\
        &\quad-I(\hat{N}_1\hat{N}_2 E:\hat{C}_1\hat{Q}_1\hat{C}_2\hat{Q}_2|C'_1 C'_2)_\tau
       \label{eq:J_additivity_1}\\
     &=S(\hat{N}_1 E_1 C'_1)+S(\hat{C}_1\hat{Q}_1C'_1)-S( \hat{C}_1 \hat{N}_1 \hat{Q}_1 E_1 C'_1)-S(C'_1) \nonumber \\
     &\quad +S(\hat{N}_2 E_2 C'_2)+S(\hat{C}_2\hat{Q}_2C'_2)-S(     \hat{C}_2 \hat{N}_2 \hat{Q}_2 E_2 C'_2)-S(C'_2) \nonumber\\
     &\quad \!-\!S(\hat{N}_1\hat{N}_2 E C'_1 C'_2) \!- \!S(\hat{C}_1\hat{Q}_1\hat{C}_2\hat{Q}_2 C'_1 C'_2)\! \nonumber \\
     &\quad +\!S( \hat{C}_1\hat{N}_1\hat{Q}_1\hat{C}_2 \hat{N}_2\hat{Q}_2E C'_1 C'_2)\!+\! S(C'_1 C'_2) \label{eq:J_additivity_2} \\
     &=S(\hat{C}_1 \hat{Q}_1 R_1 R'_1)+S(\hat{C}_1\hat{Q}_1C'_1)-S(R_1 R'_1)-S(C'_1) \nonumber \\
     &\quad +S(\hat{C}_2 \hat{Q}_2 R_2 R'_2)+S(\hat{C}_2\hat{Q}_2C'_2)-S(     R_2 R'_2)-S(C'_2) \nonumber\\
     &\quad \!-\!S(\hat{C}_1 \hat{Q}_1\hat{C}_2 \hat{Q}_2 R_1 R'_1 R_2 R'_2) \!- \!S(\hat{C}_1\hat{Q}_1\hat{C}_2\hat{Q}_2 C'_1 C'_2)\! \nonumber \\
     &\quad +\!S(R_1 R'_1 R_2 R'_2)\!+\! S(C'_1 C'_2) \label{eq:J_additivity_3}\\
     &=I(\hat{C}_1 \hat{Q}_1 R_1 R'_1:\hat{C}_2 \hat{Q}_2 R_2 R'_2)
     -I(R_1 R'_1:R_2 R'_2)\nonumber \\
     &\quad +I(\hat{C}_1\hat{Q}_1C'_1:\hat{C}_2\hat{Q}_2C'_2)
     -I(C'_1:C'_2) \nonumber \\
     &\geq I( R_1 R'_1:R_2 R'_2)
     -I(R_1 R'_1:R_2 R'_2)
     +I(C'_1:C'_2)
     -I(C'_1:C'_2) \label{eq:J_additivity_4}\\
     &=0, \nonumber
  \end{align}
  where Eq.~(\ref{eq:J_additivity_1}) is due to Definition~\ref{def:J_epsilon Z_epsilon}; 
  In Eq.~(\ref{eq:J_additivity_2}) we expand the mutual informations in terms of entropies;
  Eq.~(\ref{eq:J_additivity_3}) follows because the overall state on systems
  $\hat{C}_1 \hat{N}_1 \hat{Q}_1 \hat{C}_2 \hat{N}_2 \hat{Q}_2 E R_1 C'_1 R'_1 R_2 C'_2 R'_2$
  is pure; 
  Eq.~(\ref{eq:J_additivity_4}) is due to data processing. 

\item According to Theorem~\ref{thm: KI decomposition} \cite{KI2002,Hayden2004}, 
  any isometry $U:C N Q \rightarrow \hat{C} \hat{N} \hat{Q} E$ acting on the state 
  $\omega^{C N Q R C'}$ which preserves the reduced state on systems $C N Q R C'$ 
  ($C'$ here is considered as a part of the reference system), acts as the following:
  \begin{align*}
    (U \otimes \1_{RC'}) \omega^{C N Q R C'}(U^{\dagger} \otimes \1_{RC'})
      =\sum_{j} p_j \proj{j}^{C}\otimes U_j \omega_j^{N} U_j^{\dagger} \otimes \rho_{j}^{Q R} \otimes \proj{j}^{C'},
  \end{align*}
  where the isometry $U_j: N \rightarrow \hat{N} E$ satisfies 
  $\Tr_E [U_j \omega_j^{N} U_j^{\dagger}]=\omega_j$.
  Therefore,  in Definition~\ref{def:J_epsilon Z_epsilon} for $\epsilon=0$, the final state is
  \begin{align*}
    \tau^{\hat{C} \hat{N} \hat{Q}  E R C'}
      = \sum_{j} p_j \proj{j}^{C}\otimes U_j \omega_j^{N} U_j^{\dagger} \otimes \rho_{j}^{Q R} \otimes \proj{j}^{C'}.
  \end{align*}
  Thus we can directly evaluate
\begin{align*}
      Z_0(\omega)=S(\hat{N} E|C')_\tau=S(N |C)_\omega \text{ and } 
      J_0(\omega)=I(\hat{N} E:\hat{C}\hat{Q}|C')_\tau=0,
\end{align*}
concluding the proof.
\hfill\qedsymbol
\end{enumerate}

\section{Discussion}
\label{sec:Discussion}
We have introduced a common framework for all single-source quantum compression 
problems, i.e. settings without side information at the encoder or the decoder, 
by defining the compression task as the reproduction of a given bipartite state 
between the system to be compressed and a reference. That state, which defines 
the task, can be completely general, and special instances recover Schumacher's
quantum source compression (in both variants of a pure state ensemble and of 
a pure entangled state) \cite{Schumacher1995} 
and compression of a mixed state ensemble source in the blind variant 
\cite{Horodecki1998,KI2001}.

Our general result gives the optimal quantum compression rate in terms of
qubits per source state, both in the settings without and with entanglement, and 
indeed the entire qubit-ebit rate region, reproducing the aforementioned 
special cases, along with other previously considered problems \cite{ZK_Eassisted_ISIT_2019}. 
Despite the technical difficulties in obtaining it, the end result has a 
simple and intuitive interpretation. Namely, the given source $\rho^{AR}$ 
is equivalent to a source in standard Koashi-Imoto form,
\[
  \omega^{CQR} = \sum_j p_j \proj{j}^C \otimes \rho_j^{QR},
\]
so that $j$ has to be compressed as classical information, at rate $S(C)$,
and $Q$ as quantum information, at rate $S(Q|C)$; in the presence of 
entanglement, the former rate is halved while the latter is maintained. 
Indeed, what our Theorem \ref{theorem:complete rate region mixed state}
shows is that the original source has the same qubit-ebit 
rate region as the clean classical-quantum mixed source
\[
  \Omega^{CQRR'C'} = \sum_j p_j \proj{j}^C \otimes \proj{\psi_j}^{QRR'} \otimes \proj{j}^{C'},
\]
where $\ket{\psi_j}^{QRR'}$ purifies $\rho_j^{QR}$, and $RR'C'$ is considered
the reference. In $\Omega$, $C$ is indeed a manifestly classical source, 
since it is duplicated in the reference system, and conditional on $C$,
$Q$ is a  genuinely quantum source since it is purely entangled with the
reference system. As $\Tr_{R'C'} \Omega^{CQRR'C'} = \omega^{CQR}$, any 
code and any achievable rates for $\Omega$ are good for $\omega$, and 
that is how the achievability of the rate region in Theorem \ref{theorem:complete rate region mixed state} 
can be described. The opposite, that a code good for $\omega$ should be 
good for $\Omega$, is far from obvious. Indeed, if that were true, it would 
not only yield a quick and simple proof of our converse bounds, but would 
imply that the rate region of Theorem \ref{theorem:complete rate region mixed state} satisfies a 
strong converse! However, as we do not know this reduction to the source $\Omega$,
our converse proceeds via a more complicated, indirect route, and yields only
a weak converse. Whether the strong converse holds, and what the detailed
relation between the sources $\omega^{CQR}$ and $\Omega^{CQRR'C'}$ is, 
remain open questions. 

\medskip

\chapter{Unification of the blind and visible  Schumacher compression}
\label{chap: E assisted Schumacher}

In this chapter, we ask how the quantum compression of ensembles of pure states is
affected by the availability of entanglement, and in settings where
the encoder has access to side information.
We find the optimal asymptotic quantum rate and the optimal tradeoff 
(rate region) of quantum and entanglement rates. 
It turns out that the amount by which the quantum rate beats
the Schumacher limit, the entropy of the source, is precisely half
the entropy of classical information that can be extracted from the
source and side information states without disturbing them at all
(``reversible extraction of classical information'').

In the special case that the encoder has no side information, or that 
she has access to the identity of the states, this problem reduces to the known
settings of \textit{blind} and \textit{visible} Schumacher compression, respectively,
albeit here additionally with entanglement assistance. 
We comment on connections to previously studied and further rate
tradeoffs when also classical information is considered.
This chapter is based on the papers in \cite{Schumacher_Assisted_arXiv_Z_2019,ZK_Eassisted_ISIT_2019}.

\section{The source model} 
The task of data compression of a quantum source, introduced by Schumacher
\cite{Schumacher1995}, marks one of the foundations of quantum
information theory: not only did it provide an information theoretic 
interpretation of the von Neumann entropy $S(\rho) = -\Tr\rho\log\rho$
as the minimum compression rate, it also motivated the very concept of the qubit!
In the Schumacher modelling, a source is given by an ensemble 
$\cE = \{ p(x), \proj{\psi_x} \}$ of pure states $\psi_x=\proj{\psi_x}\in\cS(A)$,
$\ket{\psi_x}\in A$, with a Hilbert space $A$ of
finite dimension $|A|<\infty$; $\cS(A)$ denotes the set of states (density operators).
Furthermore, $x\in\cX$ ranges over a 
discrete alphabet, so that we can can describe the source equivalently by
the classical-quantum (cq) state $\omega = \sum_x p(x) \proj{x}^X \otimes \proj{\psi_x}^A$.

While the achievability of the rate $S(A)_\omega = S(\omega^A)$ was
shown in \cite{Schumacher1995,Jozsa1994_1} (see also \cite[Thm.~1.18]{OhyaPetz:entropy}),
the full (weak) converse was established in \cite{Barnum1996}, a simplified
proof being given by M. Horodecki \cite{Horodecki1998}; the strong 
converse was proved in \cite{Winter1999}. 

\medskip
In this chapter, we consider a more comprehensive model, where on the one hand
the sender/encoder of the compressed data (Alice) has access to side
information, namely a pure state $\sigma_x^C$ in addition to the source state
$\psi_x^A$, and on the other hand, she and the receiver/decoder of the compressed
data (Bob) share pure state entanglement in the form of EPR pairs at a
certain rate.

Thus, the source is now an ensemble $\cE = \{ p(x), \proj{\psi_x}^A\otimes\proj{\sigma_x}^C \}$
of product states, which can be described equivalently by the cqq-state
\begin{align}
  \label{eq:source state omega}
  \omega^{XAC}=\sum_{x\in \mathcal{X}} p(x) \proj{x}^X\otimes \proj{\psi_x}^A\otimes \proj{\sigma_x}^C. 
\end{align}
Yet another equivalent description is via the
random variable $X \in \cX$, distributed according to $p$, i.e. $\Pr\{X=x\}=p_x$;
this also makes the pure states $\psi_X$ and $\sigma_X$ random variables.

We will consider the information theoretic limit of
many copies of $\omega$, i.e.~$\omega^{X^n A^n C^n} = \left(\omega^{XAC}\right)^{\otimes n}$:
\[
  \omega^{X^n A^n C^n}
    \!\!\! = \!\!\!\!
             \sum_{x^n \in \mathcal{X}^n}\!\!\!\! p(x^n) \proj{x^n}^{X^n} 
                       \!\otimes\! \proj{\psi_{x^n}}^{A^n}
                       \!\otimes\! \proj{\sigma_{x^n}}^{C^n}\!\!\!\!\!,
\]
using the notation
\begin{align*}
  x^n              &= x_1 x_2 \ldots x_n,\quad\;
  p(x^n) = p(x_1) p(x_2)  \cdots p(x_n), \\
  \ket{x^n}        &= \ket{x_1} \ket{x_2} \cdots \ket{x_n},\ 
  \ket{\psi_{x^n}} = \ket{\psi_{x_1}} \ket{\psi_{x_2}} \cdots \ket{\psi_{x_n}}.
\end{align*}

\section{Compression assisted by entanglement}
\label{sec:Compression assisted by entanglement}
We assume that the encoder, Alice, and the decoder, Bob, have initially a maximally 
entangled state $\Phi_K^{A_0B_0}$ on registers $A_0$ and $B_0$ (both of dimension $K$).
With probability $p(x^n)$, the source provides Alice  
with the state $\psi_{x^n}^{A^n}\otimes\sigma_{x^n}^{C^n}$.
Then, Alice performs her encoding operation 
$\mathcal{C}:A^nC^nA_0 \longrightarrow \hat{C}^nC_A$ on the systems $A^n$, 
$C^n$ and her part $A_0$ of the entanglement, which is a quantum channel,
i.e.~a completely positive and trace preserving (CPTP) map. 
(Note that our notation is a slight abuse, which we maintain as it is simpler 
while it cannot lead to confusions, since channels really are maps between
the trace class operators on the involved Hilbert spaces.)
The dimension of the compressed system obviously has to be smaller than the 
original source, i.e. $|C_A| \leq  \abs{A}^n$. 
We call $Q=\frac1n \log|C_A|$ and $E=\frac{1}{n}\log K$ the quantum and entanglement 
rates of the compression protocol, respectively.
The system $C_A$ is then sent to Bob via a noiseless quantum channel, who performs
a decoding operation $\mathcal{D}:C_A B_0 \longrightarrow \hat{A}^n$ on the system 
$C_A$ and his part of entanglement $B_0$. 

According to Stinespring's theorem \cite{Stinespring1955}, all these CPTP maps can be 
dilated to isometries 
$V_A : A^nC^nA_0 \hookrightarrow \hat{C}^n C_A W_A$ and  
$V_B : C_A B_0 \hookrightarrow {\hat{A}^n W_B}$, 
where the new systems $W_A$ and $W_B$ are the environment systems
of Alice and Bob, respectively. 

We say the encoding-decoding scheme has fidelity $1-\epsilon$, or error $\epsilon$, if 
\begin{align}
  \label{eq:Schumcaher assisted fidelity}
  \overline{F} &:= F\left( \omega^{X^n\hat{A}^n\hat{C}^n},\xi^{X^n\hat{A}^n\hat{C}^n} \right) \nonumber\\
               &=\sum_{x^n \in \mathcal{X}^n}\!\!\! p(x^n)F\!\left(\proj{\psi_{x^n}}^{A^n}\!
                          \otimes\!\proj{\sigma_{x^n}}^{C^n}\!,\xi_{x^n}^{\hat{A}^n\hat{C}^n}\right) \\
         &\geq 1-\epsilon, \nonumber
\end{align}
where $\xi^{X^n\hat{A}^n\hat{C}^n}=\sum_{x^n}p(x^n)\proj{x}^{X^n} \otimes \xi_{x^n}^{\hat{A}^n\hat{C}^n}$ 
and 
$\xi_{x^n}^{\hat{A}^n\hat{C}^n}=(\mathcal{D}\circ\mathcal{C})\!\proj{\psi_{x^n}\!}^{A^n} \!\otimes\! \proj{\sigma_{x^n}\!}^{C^n} \!\otimes\!\Phi_K^{A_0\!B_0}\!\!$.
We say that $(E,Q)$ is an (asymptotically) achievable rate pair if for all $n$
there exist codes such that the fidelity converges to $1$, and
the entanglement and quantum rates converge to $E$ and $Q$, respectively.
The rate region is the set of all achievable rate pairs, as a subset of 
$\mathbb{R}\times\mathbb{R}_{\geq 0}$. 

Note that this means that we demand not only that Bob can reconstruct the
source states $\psi_{x^n}$ with high fidelity on average, but that Alice
retains the side information states $\sigma_{x^n}$ as well with high fidelity.

There are two extreme cases of the side information that have been considered
in the literature: 
If $C$ is a trivial system, or more generally if the states
$\sigma_x^C$ are all identical, then the aforementioned task is the 
entanglement-assisted version of \textit{blind} Schumacher compression. If $C=X$, or more
precisely $\ket{\sigma_x}=\ket{x}$, then Alice has access to classical random variable 
$X$, and the task reduces to \textit{visible} Schumacher compression with entanglement assistance.  
The blind-visible terminology is originally from \cite{Barnum1996,Horodecki2000}.

\begin{remark}
\label{remark:E=0}
In the case of no entanglement being available, i.e. $E=0$ ($K=1$), the 
problem is fully understood: The asymptotic rate $Q=S(A)$ 
from \cite{Schumacher1995,Jozsa1994_1} is achievable without touching
the side information, and it is optimal, even in the visible case
(which includes all other side informations), by the weak and strong
converses of \cite{Barnum1996,Horodecki1998} and \cite{Winter1999}. 
\qed
\end{remark}


\section{Optimal quantum rate}
To formulate the minimum compression rate under unlimited entanglement
assistance, we need the following concept.

\begin{definition}
  \label{def:reducibility}
  An ensemble of pure states 
  $\cE=\{p(x),\proj{\psi_x}^A \otimes \proj{\sigma_x}^C \}_{x\in \mathcal{X}}$ 
  is called \emph{reducible} if its states belong to two or more orthogonal subspaces.
  Otherwise the ensemble $\cE$ is called \emph{irreducible}.
  We apply the same terminology to the source cqq-state $\omega^{X A C}$.
\end{definition}

Notice that a reducible ensemble can be written uniquely as a disjoint union 
of irreducible ensembles $\mathcal{E} = \bigcupdot_{y \in \mathcal{Y}} q(y) \mathcal{E}_y$,
with a partition $\mathcal{X} = \bigcupdot_{y\in \mathcal{Y}} \mathcal{X}_y$ and 
irreducible ensembles 
\begin{align*}
\mathcal{E}_y = \{ p(x|y), \proj{\psi_x}^A \otimes \proj{\sigma_x}^C \}_{x \in \mathcal{X}_y},
\end{align*}
where $q(y)p(x|y)=p(x)$ for $x \in \mathcal{X}_y$ and 
$q(y) =\sum_{x \in \mathcal{X}_y} p(x)$.
We define the subspace spanned by the vectors of each irreducible ensemble as 
$F_y := \text{span} \{\ket{\psi_x}\otimes\ket{\sigma_x} : x \in \mathcal{X}_y\}$.
The irreducible ensembles $\mathcal{E}_y$ are pairwise orthogonal, i.e.~$F_{y'} \perp F_y$ 
for all $y' \neq y$.
We may thus introduce the random variable $Y=Y(X)$ taking values in the set 
$\mathcal{Y}$ with probability distribution $q(y)$; namely, $Y$ is a deterministic
function of $X$ such that $\Pr\{X\in\mathcal{X}_Y\}=1$. 

We define the \textit{modified} source as
\begin{align*}
\omega^{XACY}=\sum_x p(x) \proj{x}^X\otimes \proj{\psi_x}^A \otimes \proj{\sigma_x}^C \otimes \proj{y(x)}^Y,
\end{align*} 
with side information systems $CY$.
Because there is an isometry $V:AC \rightarrow ACY$ which acts as
\begin{equation}
  \label{eq:iso}
  V\ket{\psi_x}^A \otimes \ket{\sigma_x}^C=\ket{\psi_x}^A \otimes \ket{\sigma_x}^C \otimes \ket{y(x)}^Y,
\end{equation}
the extended source $\omega^{XACY}$ is equivalent to the original
source and side information $\omega^{XAC}$ modulo a local operation of Alice.

We first present the optimal asymptotic compression rate in the following 
theorem and prove the achievability of it, but we leave 
the converse proof to the end of this section, as it requires introducing 
further machinery.

\begin{theorem}
  \label{theorem: main}
  For the given source $\omega^{XACY}$, the optimal asymptotic compression rate 
  assisted by unlimited entanglement is
  $Q=\frac12 (S(A)+S(A|CY))$.
  
  Furthermore, there is a protocol achieving this communication 
  rate with entanglement consumption at rate $E=\frac12 (S(A)-S(A|CY))$.
\end{theorem}
\begin{proof}
We first show that this rate is achievable.
Consider the following purification of $\omega^{XACY}$,
\begin{align*}
  \ket{\omega}^{XX'ACY}=\sum_x \sqrt{p(x)}\ket{x}^X\ket{x}^{X'}\ket{\psi_x}^A\ket{\sigma_x}^C\ket{y(x)}^Y,
\end{align*}
with side information systems $CY$. This is obtained from 
\begin{align*}
\ket{\omega}^{XX'AC}=\sum_x \sqrt{p(x)}\ket{x}^X\ket{x}^{X'}\ket{\psi_x}^A\ket{\sigma_x}^C,
\end{align*}
by Alice applying the isometry $V$ from Eq.~(\ref{eq:iso}).

We apply quantum state redistribution (QSR) \cite{Devetak2008_2,Oppenheim2008} 
as a subprotocol, where the objective is for Alice to send to Bob $A^n$,
using $C^nY^n$ as side information, while $(XX')^n$ serves as reference system;
the figure of merit is the fidelity with the original pure state $(\omega^{XX'ACY})^{\otimes n}$.
Denoting the overall encoding-decoding CPTP map 
$\Lambda:A^nC^nY^n \rightarrow \hat{A}^n\hat{C}^n\hat{Y}^n$, 
QSR gives us the first inequality of the following chain:
\begin{align*} 
  1-o(1) &\leq F\!\left( \omega^{X^nX'^nA^nC^nY^n}\!\!,
                         (\id_{X^nX'^n} \otimes \Lambda) \omega^{X^nX'^nA^nC^nY^n}\! \right) \\
         &\leq F\!\left( \omega^{X^nA^nC^nY^n}\!\!,
                         (\id_{X^n} \otimes \Lambda) \omega^{X^nA^nC^nY^n}\! \right),
\end{align*}
where the second inequality follows from monotonicity of the fidelity under partial trace.
Thus, the protocol satisfies our fidelity criterion (\ref{eq:Schumcaher assisted fidelity}).

The communication rate we obtain from QSR is
$Q = \frac{1}{2}I(A:XX') = \frac{1}{2} (S(A)+S(A|CY))$. 
Furthermore, QSR guarantees entanglement consumption at the rate 
$E = \frac12 I(A:CY) = \frac12 (S(A)-S(A|CY))$.
\end{proof}

To prove optimality (the converse), we first need a few preparations.
The following definition is inspired by the ``reversible extraction of classical
information'' in \cite{Barnum2001_2}.

\begin{definition}
  \label{def:I_epsilon}
  For a source $\omega^{XAC}$ and $\epsilon \geq 0$, define
  \begin{align*}
    I_\epsilon(\omega) \!\!:= \!\!
                  \max_{V:AC\rightarrow \hat{A}\hat{C} W \text{ isometry}} \!\! \!\!I(X\!:\!\hat{C}W)_\xi 
          \text{ s.t. } F(\!\omega^{\!X\!A\!C\!}\!,\xi^{\!X\hat{A}\!\hat{C}\!})\! \geq \!1\!\!-\!\epsilon,    
\end{align*}
where  
\[
  \xi^{X\!\hat{A}\hat{C}W} \!\!\!
      =\!(\!\1_X \otimes V\!) \omega^{XAC}\!(\!\1_X \otimes V^{\dagger}\!) \!
      =\!\sum_x p(x) \proj{x}^X \!\otimes \proj{\xi_x}^{\!\hat{A}\hat{C}W} \!\!\!.
\]
\end{definition}
In this definition, the dimension of the environment is w.l.o.g. bounded as $|W| \leq |A|^2|C|^2$; 
hence, the optimisation is of a continuous function over a compact domain, so we have a 
maximum rather than a supremum.

\begin{lemma}
  \label{lemma:I_epsilon properties}
  The function $I_\epsilon(\omega)$ has the following properties:
  \begin{enumerate}
    \item It is a non-decreasing function of $\epsilon$. 
    \item It is concave in $\epsilon$.
    \item It is continuous for $\epsilon \geq 0$. 
    \item For any two states $\omega_1^{X_1 A_1 C_1}$ and $\omega_2^{X_2 A_2 C_2}$ and for $\epsilon \geq 0$,
          \(
            I_{\epsilon}(\omega_1 \otimes \omega_2) \leq  I_{\epsilon}(\omega_1) +I_{\epsilon}(\omega_2).
          \)
    \item For any state $\omega^{XAC}$, $I_0(\omega) \leq S(CY)$.
  \end{enumerate}
\end{lemma}

\begin{proof}
1. The definition of $I_\epsilon(\omega)$ directly implies that it 
is a non-decreasing function of $\epsilon$.

2. To prove the concavity, let $V_1:AC \rightarrow \hat{A}\hat{C}W$ and 
$V_2:AC \rightarrow \hat{A}\hat{C}W$ be the isometries attaining the 
maximum for $\epsilon_1$ and $\epsilon_2$, respectively, which act as 
follows:
\[
    V_1 \ket{\psi_x}^A\ket{\sigma_x}^C=\ket{\xi_x}^{\hat{A}\hat{C}W} \text{ and } 
    V_2 \ket{\psi_x}^A\ket{\sigma_x}^C=\ket{\zeta_x}^{\hat{A}\hat{C}W}.   
\]
For $0\leq \lambda \leq 1$,
define the isometry $U:AC \rightarrow \hat{A}\hat{C}WRR'$ by letting, for all $x$,
\[
  U\!\ket{\psi_x}^A\!\ket{\sigma_x}^C 
    \!:=\!\sqrt{\!\lambda}\ket{\xi_x}^{\hat{A}\hat{C}W}\!\!\ket{00}^{RR'}\!\!\!+\!\!\sqrt{1\!-\!\lambda}\ket{\zeta_x}^{\hat{A}\hat{C}W}\!\!\ket{11}^{RR'},
\]
where systems $R$ and $R'$ are qubits. 
Then, the reduced state on the systems $X\hat{A}\hat{C}$ is 
$\tau^{X\hat{A}\hat{C}}=\sum_x p(x) \proj{x}^X\otimes \tau_x^{\hat{A}\hat{C}}$, 
where $\tau_x^{\hat{A}\hat{C}}=\lambda \xi_x^{\hat{A}\hat{C}}+(1-\lambda)\zeta_x^{\hat{A}\hat{C}}$; 
therefore, the fidelity is bounded as follows:
\begin{align*}
  F(\omega^{XA\hat{C}}\!,\tau^{X\hat{A}\hat{C}}) 
    &= \sum_x p(x) \sqrt{\bra{\psi_x}
                          \left(\lambda \xi_x^{\hat{A}\hat{C}}+(1-\lambda)\zeta_x^{\hat{A}\hat{C}}\right)
                         \ket{\psi_x}} \\
    &\!\!\!\!\!\!\!\!\!\!\!\!\!\!\!\!\!\!\!\!\!\!\!\!\!\!\!\!\!\!\!\!\!\!\!\!\!\!\!\!\!
     \geq \lambda \sum_x p(x) \sqrt{\!\bra{\psi_x} \xi_x^{\hat{A}\hat{C}} \ket{\psi_x}\!}
          + (1-\lambda)\sum_x p(x) \sqrt{\!\bra{\psi_x}\zeta_x^{\hat{A}\hat{C}} \ket{\psi_x}\!} \\
    &\!\!\!\!\!\!\!\!\!\!\!\!\!\!\!\!\!\!\!\!\!\!\!\!\!\!\!\!\!\!\!\!\!\!\!\!\!\!\!\!\!
     \geq 1-\left( \lambda\epsilon_1 +(1-\lambda)\epsilon_2 \right),
\end{align*}
where the second line follows from the concavity of the function $\sqrt{x}$, 
and the last line follows by the definition of the isometries $V_1$ and $V_2$.
Now, define $W':=WRR'$ and let $\epsilon=\lambda\epsilon_1 +(1-\lambda)\epsilon_2$. 
According to Definition \ref{def:I_epsilon}, we obtain
\begin{align*}
  I_\epsilon(\omega) &\geq I(X:\hat{C}W')_{\tau}\\ 
                     &=    I(X:R)_{\tau}+I(X:\hat{C}W|R)_{\tau}+I(X:R'|\hat{C}WR)_{\tau}\\
                     &\geq I(X:\hat{C}W|R)_{\tau}
                      =    \lambda I_{\epsilon_1}(\omega)+(1-\lambda)I_{\epsilon_2}(\omega),
\end{align*}
where the third line is due to strong subadditivity of the quantum mutual information.

3. The function is non-decreasing and concave for $\epsilon \geq 0 $, so it is continuous 
for $\epsilon > 0 $. 
The concavity implies furthermore that $I_{\epsilon}$ is lower semi-continuous at 
$\epsilon=0$. On the other hand, since the fidelity and mutual information are 
both continuous functions of CPTP maps, and the domain of the optimization is a 
compact set, we conclude that $I_\epsilon(\omega)$ is also upper semi-continuous at 
$\epsilon=0$, so it is continuous at $\epsilon=0$ \cite[Thms.~10.1,~10.2]{Rockafeller}. 

4. In the definition of $I_{\epsilon}(\omega_1 \otimes \omega_2)$, let the isometry 
$V_0:A_1 C_1 A_2C_2 \rightarrow \hat{A}_1\hat{C}_1\hat{A}_2\hat{C}_2 W$ be the 
one attaining the maximum which acts on the purified source state with purifying 
systems $X_1'$ and $X_2'$ as follows: 
\begin{align*}
  \ket{\xi}&^{X_1\!X_1'\!X_2\!X_2'\!\hat{A}_1\!\hat{C}_1\!\hat{A}_2\!\hat{C}_2\!W} \\
           &\phantom{====}
            =(\1_{X_1\!X_1'\!X_2\!X_2'}\otimes V_0)\ket{\omega_1}^{X_1\!X'_1\!A_1\!C_1}
                                                   \ket{\omega_1}^{X_2\!X'_2\!A_2\!C_2}\!\!.
\end{align*}
Now, define the isometry $V_1:A_1 C_1 \rightarrow \hat{A}_1\hat{C}_1\hat{A}_2 \hat{C}_2W X_2X'_2$ 
acting only on the systems $A_1 C_1$ with the output state $\hat{A}_1\hat{C}_1$ and the 
environment $W_1:=\hat{A}_2\hat{C}_2 W X_2X'_2$ as follows:
\[
  \ket{\xi}^{X_1\!X_1'\!X_2\!X_2'\!\hat{A}_1\!\hat{C}_1\!\hat{A}_2\!\hat{C}_2\!W}
                           = (\1_{X_1X_1'}\otimes V_1)\ket{\omega_1}^{X_1X_1'A_1C_1}. 
\]
Hence, we obtain
\begin{align*}
  F(\omega_1^{X_1\!A_1\!C_1}\!,\xi^{X_1\!\hat{A}_1\!\hat{C}_1})
                             &\geq F\!\left(\omega_1^{X_1\!A_1\!C_1}\!\otimes\!\omega_2^{X_2\!A_2\!C_2}\!,
                                       \xi^{X_1\!X_2\!\hat{A}_1\!\hat{C}_1\!\hat{A}_2\!\hat{C}_2}\!\right) \\
                             &\geq 1 - \epsilon, 
\end{align*}
where the first inequality is due to monotonicity of the fidelity under 
CPTP maps, and the second inequality follows by the definition of $V_0$. 
Consider the isometry 
$V_2:A_2 C_2 \rightarrow \hat{A}_1\hat{C}_1\hat{A}_2 \hat{C}_2W X_1X'_1$ 
defined in a similar way, with the output state $\hat{A}_2\hat{C}_2$ and 
the environment $W_2:=\hat{A}_1\hat{C}_1 W X_1X'_1$. 
Therefore, we obtain
\begin{align*}
    I_{\epsilon}(\omega_1) +I_{\epsilon}(\omega_2) &\geq I(X_1:\hat{C}_1W_1)+I(X_2:\hat{C}_2W_2)\\
    &\geq I(X_1:\hat{C}_1\hat{C}_2W)+I(X_2:\hat{C}_1\hat{C}_2WX_1)\\
    &= I(X_1X_2:\hat{C}_1\hat{C}_2W)
     = I_{\epsilon}(\omega_1 \otimes \omega_2),
\end{align*}
where the second line is due to data processing. 

5. In the definition of $I_0(\omega)$ let $V_0:AC \rightarrow \hat{A}\hat{C}W$ be the isometry attaining the maximum with $F(\omega^{XAC},\xi^{X\hat{A}\hat{C}})=1$. Hence, we obtain
\begin{align*}
  I_0(\omega)&= I(X:\hat{C}W)
       =    I(XY:\hat{C}W) \\
      &=    I(Y:\hat{C}W)+I(X:\hat{C}W|Y) \\
      &\leq S(Y)+I(X:\hat{C}W|Y) \\
      &=    S(Y)+I(X:W|Y)+I(X:\hat{C}|WY) \\
      &\leq S(Y)+I(X:W|Y)+S(C|WY)\\
      &\leq S(Y)+I(X:W|Y)+S(C|Y),
\end{align*}
where the first line follows because $Y$ is a function of $X$. The second and fourth
line are due to the chain rule. The third line follows because for the classical 
system $Y$ the conditional entropy $S(Y|\hat{C}W)$ is non-negative. The penultimate 
line follows because for any $x$ the state on the system $\hat{C}$ is pure. The 
last line is due to strong sub-additivity of the entropy. Furthermore, for every 
$y$, the ensemble $\cE_y$ is irreducible; hence, the conditional mutual information 
$I(X:W|Y)=0$ which follows from the detailed discussion on page 2028 of \cite{Barnum2001_2}. 
%
%
%
\end{proof}

\begin{proof-of}[{of the converse part of Theorem \ref{theorem: main}}]
We start by observing 
\[
  nQ+S(B_0) \geq S(C_A)+S(B_0)
            \geq S(C_A B_0)    
            =    S(\hat{A}^n W_B),
\]
where the second inequality is due to subadditivity of the entropy, 
and the equality follows because the decoding isometry $V_B$ does not change 
the entropy. Hence, we get 
\begin{align}
  \label{eq:converse Schumacher assisted 1}
  nQ+S(B_0) &\geq S(\hat{A}^n)+S(W_B|\hat{A}^n)            \nonumber\\
    &\geq S(\hat{A}^n)+S(W_B|\hat{A}^n X^n)                \nonumber\\
    &\geq S(A^n)+S(W_B|\hat{A}^n X^n)-n \delta(n,\epsilon) \nonumber\\
    &=    S(A^n) \!+\! S(\hat{A}^nW_B| X^n\!) \!-\! S(\hat{A}^n| X^n)\!-\!n \delta(n,\epsilon)\nonumber\\
    &=    S(A^n) \!+\! S(\hat{C}^nW_A| X^n) \!-\! S(\hat{A}^n| X^n) \!-\! n \delta(n,\epsilon)\nonumber\\
    &\geq S(A^n)+S(\hat{C}^nW_A| X^n)- 3n \delta(n,\epsilon),
\end{align}
where in the first and second line we use the chain rule and subadditivity of entropy.
The inequality in the third line follows from the decodability of the system $A^n$:
the fidelity criterion (\ref{eq:Schumcaher assisted fidelity}) implies that the 
output state on systems $\hat{A}^n$ is $2\sqrt{2\epsilon}$-close to the 
original state $A^n$ in trace norm; then apply the Fannes-Audenaert inequality 
\cite{Fannes1973,Audenaert2007} where 
$\delta(n,\epsilon)=\sqrt{2\epsilon} \log|A| + \frac1n h(\sqrt{2\epsilon})$.   
The equalities in the fourth and the fifth line are due to the chain rule 
and the fact that  for any $x^n$ the overall state of $\hat{A}^n\hat{C}^nW_AW_B$ is pure.
In the last line, we use the decodability of the systems $X^nA^n$, that is the 
output state on systems $X^n\hat{A}^n$ is $2\sqrt{2\epsilon}$-close to the 
original states $X^nA^n$ in trace norm, then we apply the Alicki-Fannes 
inequality \cite{Alicki2004,Winter2016}. 

Moreover, we bound $Q$ as follows:
\begin{align}\label{eq:converse Schumacher assisted 2}
  nQ &\geq S(C_A)                     
      \geq S(C_A|\hat{C}^nW_A)           \nonumber \\
     &=    S(A^nC^nA_0) -S(\hat{C}^nW_A) \nonumber \\
     &=    S(A^nC^nY^n)+S(A_0) -S(\hat{C}^nW_A),
\end{align}
where the first equality follows because the encoding isometry 
$V_A:A^nC^nA_0 \rightarrow C_A\hat{C}^nW_A$ does not the change the entropy. 
Adding Eqs. (\ref{eq:converse Schumacher assisted 1}) and 
(\ref{eq:converse Schumacher assisted 2}), we thus obtain
\begin{align}\label{eq:converse Schumacher}
  Q &\geq \frac{1}{2}(S(A)+S(ACY))-\frac{1}{2n}I(\hat{C}^nW_A:X^n)-\frac{3}{2}\delta(n,\epsilon)\nonumber \\
  &\geq \frac{1}{2}(S(A)+S(ACY))-\frac{1}{2n}I(\hat{C}^nW_AW_B:X^n)-\frac{3}{2}\delta(n,\epsilon)\nonumber \\
  &\geq \frac{1}{2}(S(A)+S(ACY))-\frac{1}{2n} I_\epsilon(\omega^{\otimes n})
                                -\frac{3}{2}\delta(n,\epsilon) \nonumber \\
  &\geq \frac{1}{2}(S(A)+S(ACY))-\frac{1}{2} I_\epsilon(\omega)-\frac{3}{2}\delta(n,\epsilon) \nonumber 
\end{align}
where the second line is due to data processing. The third line follows from 
Definition \ref{def:I_epsilon}. The last line follows from point 4 of Lemma 
\ref{lemma:I_epsilon properties}. In the limit of $\epsilon \to 0$ and 
$n \to \infty $, the rate is bounded by
\begin{align*}
  Q &\geq \frac{1}{2}(S(A)+S(ACY))-\frac{1}{2} I_0(\omega)\\
    &\geq \frac{1}{2}(S(A)+S(ACY))-\frac{1}{2}S(CY)\\
    &=    \frac{1}{2}(S(A)+S(A|CY)),
\end{align*}
where the first line follows from point 3 of Lemma \ref{lemma:I_epsilon properties}
stating that $I_\epsilon(\omega)$ is continuous at $\epsilon=0$. 
The second line is due to point 5 of Lemma \ref{lemma:I_epsilon properties}. 
\end{proof-of}

\section{Complete rate region}
In this section, we find the complete rate region of achievable rate pairs $(E,Q)$.

\begin{theorem}
  \label{theorem:complete rate region}
  For the source $\omega^{XACY}$, all asymptotically achievable entanglement and 
  quantum rate pairs $(E,Q)$ satisfy
  \begin{align*}
    Q   &\geq \frac{1}{2}(S(A)+S(A|CY)),\\
    Q+E &\geq  S(A). 
  \end{align*}
  Conversely, all the rate pairs satisfying the above inequalities are achievable.
\end{theorem}
\begin{proof}
The first inequality comes from Theorem \ref{theorem: main}. 
For the second inequality, consider any code with quantum communication
rate $R$ and entanglement rate $E$. By using an additional communication
rate $E$, Alice and Bob can distribute the entanglement first, and then
apply the given code, converting it into one without preshared 
entanglement and communication rate $Q+E$, having exactly the same
fidelity. By Remark \ref{remark:E=0}, $Q+E \geq S(A)$.

As for the achievability, the corner point 
$(\frac{1}{2} I(A:CY),\frac{1}{2}(S(A)+S(A|CY)))$ is achievable, 
because QSR which is used as the achievability protocol in 
Theorem \ref{theorem: main} uses $\frac{1}{2} I(A:CY)$ ebits of 
entanglement between Alice and Bob. 
Furthermore, all the points on the line $Q+E = S(A)$ for 
$Q \geq \frac{1}{2}(S(A)+S(A|CY))$ are achievable because one 
ebit can be distributed by sending a qubit. 
All other rate pairs are achievable by resource wasting. The rate region is depicted in
Fig.~\ref{fig:E-Q}
\end{proof}

\begin{figure}[!t]   
\centering
  \includegraphics[width=0.7\textwidth]{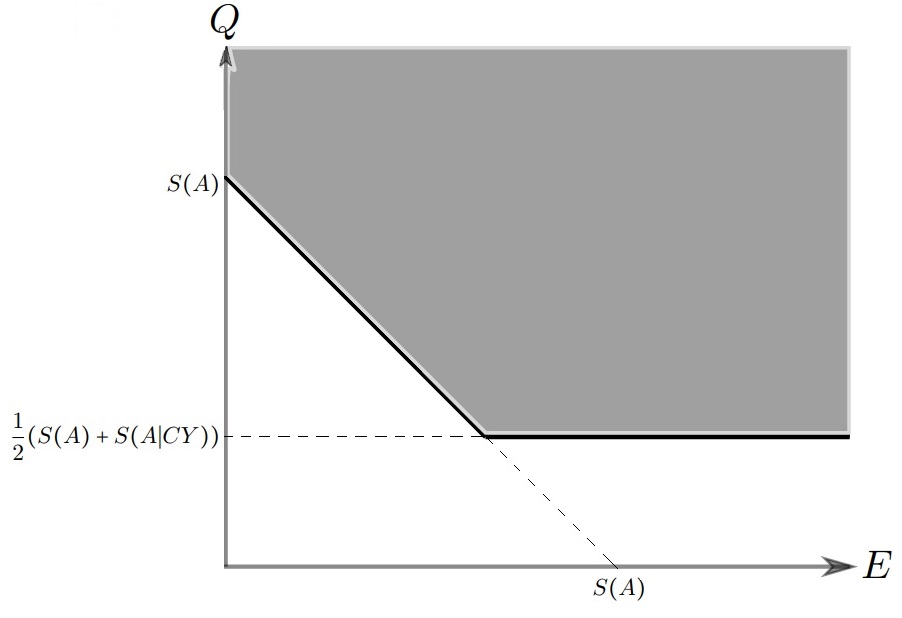}
  \caption{The optimal rate region of quantum and entanglement rates.}
  \label{fig:E-Q}
\end{figure}

\section{Discussion}
First of all, let us look what our result tell us in the cases of
blind and visible compression. 

\begin{corollary}
  \label{corollary:blind}
  In blind compression (i.e. if $C$ is trivial, or more generally the
  states $\sigma_x$ are all identical), the compression of the source 
  $\omega^{XACY}$ reduces to the entanglement-assisted Schumacher 
  compression for which Theorem \ref{theorem: main} gives the optimal asymptotic quantum rate
  \[
    Q = \frac{1}{2}(S(A)+S(A|Y))=S(A)-\frac{1}{2}S(Y).
  \]
  This implies that if the source is irreducible, then this rate is equal to the
  Schumacher limit $S(A)$. In other words, the entanglement does not help the 
  compression. Moreover, due to Theorem \ref{theorem:complete rate region}, a 
  rate $\frac{1}{2}S(Y)$ of entanglement is consumed in the compression,
  and $E+Q\geq S(A)$ in general.
\qed
\end{corollary}

The blind compression of a source $\omega^{XAY}$ is also considered 
in \cite{Barnum2001_2}, but there instead of entanglement, a noiseless classical
channel was assumed in addition to the quantum channel.
It was shown that the optimal quantum rate assisted with free classical communication 
is equal to $S(A)-S(Y)$, while a rate $S(Y)$ of classical communication suffices.
By sending the classical information using dense coding \cite{Bennett1992},
spending $\frac12$ ebit and $\frac12$ qubit per cbit, we can recover the 
quantum and entanglement rates of Corollary \ref{corollary:blind}.
This means that our converse implies the optimality of the quantum rate
from \cite{Barnum2001_2}.

Thus we are motivated to look at a modified compression model where the resources
used are classical communication and entanglement. Namely, we let Alice and Bob 
share entanglement at rate $E$ and use classical communication at rate $C$, but
otherwise the objective is the same as in Section \ref{sec:Compression assisted by entanglement};
define the rate region as the set of all asymptotic achievable classical 
communication and entanglement rate pairs $(C,E)$, such that the decoding fidelity
asymptotically converges to $1$.  

\begin{theorem}
For a source $\omega^{XAY}$, a rate pair $(C,E)$ is achievable if and only if
\begin{align*}
  C \geq 2S(A)-S(Y),\ 
  E \geq  S(A)-S(Y). 
\end{align*}
\end{theorem}

\begin{proof}
We start with the converse. 
The first inequality follows from Theorem \ref{theorem: main}, because with 
unlimited entanglement shared between Alice and Bob, 
$\frac{1}{2}(S(A)+S(A|Y))=S(A)-\frac{1}{2}S(Y)$ qubits of quantum 
communication is equivalent to $2S(A)-S(Y)$ bits of classical communication 
due to teleportation \cite{Bennett1993} and dense coding \cite{Bennett1992}. 
The second inequality follows from \cite{Barnum2001_2}, because with free 
classical communication, the quantum rate is lower bounded by $S(A)-S(Y)$ 
which, due to teleportation \cite{Bennett1993}, is equivalent to sharing 
$S(A)-S(Y)$ ebits when classical communication is for free. 

The achievability of the corner point $(2S(A)-S(Y),S(A)-S(Y))$ 
follows from \cite{Barnum2001_2} because the compression protocol 
uses $S(A)-S(Y)$ qubits and $S(Y)$ bits of classical communication 
which is equivalent to using $S(A)-S(Y)$ ebits of entanglement and 
$2S(A)-2S(Y)+S(Y)$ bits of classical communication, due to dense coding \cite{Bennett1992}.
Other rate pairs are achievable by resource wasting. 
The rate region is depicted in Fig.~\ref{fig:C-E}. 
\end{proof}

\begin{figure}[!t] 
\centering
  \includegraphics[width=0.7\textwidth]{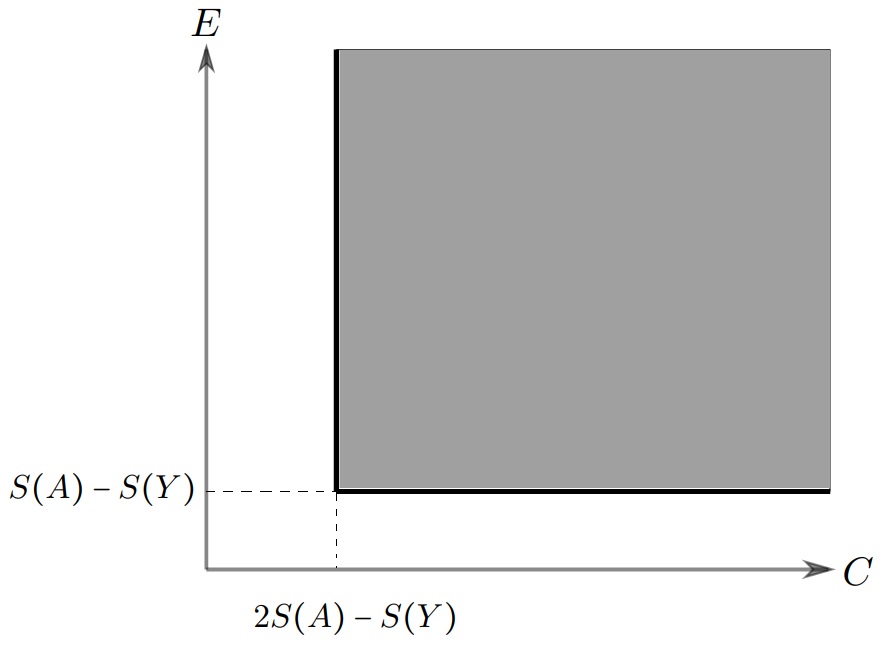} 
  \caption{The optimal rate region of classical and entanglement rates.}
  \label{fig:C-E}
\end{figure}

\begin{corollary}
\label{cor:visible}
In the visible case, our compression problem reduces to the visible version of 
Schumacher compression with entanglement assistance. In this case, 
according to Theorem \ref{theorem: main} the optimal asymptotic quantum 
rate is $Q=\frac{1}{2}S(A)$.
Moreover, a rate $E=\frac{1}{2}S(A)$ of entanglement
is consumed in the compression scheme, and $E+Q\geq S(A)$ in general.
\qed
\end{corollary}

We remark that the visible compression assisted by unlimited entanglement is 
also a special case of remote state preparation considered in \cite{Bennett2005},
from which we know that the rate $Q=\frac12 S(A)$ is achievable and optimal.

The visible analogue of \cite{Barnum2001_2}, of compression using
qubit and cbit resources, was treated in \cite{Hayden2002}, where the
achievable region was determined as the union of all all pairs $(C,Q)$ such
that $Q\geq S(A|Z)$ and $C\geq I(X:Z)$, for any random
variable $Z$ forming a Markov chain $Z$---$X$---$A$. Compare to the
complicated boundary of this region the much simpler one of Corollary \ref{cor:visible},
which consists of two straight lines.

We close by discussing several open questions for future work: 
First, the final discussion of different pairs of resources to compress suggests
that an interesting target would be the characterisation of the full triple resource
tradeoff region for $Q$, $C$ and $E$ together.

Secondly, we recall that our definition of successful decoding included 
preservation of the side information $\sigma_x^C$ with high fidelity. What is the
optimal compression rate $Q$ if the side information does not have to be preserved? 
For an example where this change has a dramatic effect on the optimal communication
rate, consider the ensemble $\cE$ consisting of the three two-qubit states
$\ket{0}^A\ket{0}^C$, $\ket{1}^A\ket{0}^C$ and $\ket{+}^A\ket{+}^C$
(where $\ket{+}=\frac{1}{\sqrt{2}}(\ket{0}+\ket{1})$), 
with probabilities $\frac12-t$, $\frac12-t$ and $2t$, respectively.
Note that $\cE$ is irreducible, hence for $t\approx 0$, we get an optimal
quantum rate of $Q\approx 1$, because $S(A) \approx S(A|C) \approx 1$.
However, by applying a CNOT unitary (with $A$ as control and $C$ as target),
the ensemble is transformed into $\cE'$ consisting of the states
$\ket{0}^A\ket{0}^{C'}$, $\ket{1}^A\ket{1}^{C'}$ and $\ket{+}^A\ket{+}^{C'}$.
The state of $A$ is not changed, only the side information, which is why we 
denote it $C'$. Hence we can apply Theorem \ref{theorem: main} to get a
quantum rate $Q\approx\frac12$, because $S(A) \approx 1$, $S(A|C) \approx 0$.

Thirdly, note that the lower bound $Q+E\geq S(A)$ in Theorem \ref{theorem:complete rate region}
holds with a strong converse (see the proof and \cite{Winter1999}).
But does $Q\geq \frac12 (S(A)+S(A|CY))$ hold as a strong converse rate
with unlimited entanglement? Likewise, in the setting of \cite{Barnum2001_2}
with unlimited classical communication, is $Q\geq S(A)-S(Y)$ a strong converse 
bound for the quantum rate?

\chapter{Distributed compression of  correlated  classical-quantum sources}
\label{chap:cqSW}

In this chapter, we resume the investigation of the problem of independent 
local compression of correlated quantum sources, the classical case 
of which is covered by the celebrated Slepian-Wolf theorem. 
We focus specifically on classical-quantum (cq) sources, for which one edge 
of the rate region, corresponding to the compression of the classical
part, using the quantum part as side information at the decoder, 
was previously determined by Devetak and Winter [Phys. Rev. A 68, 042301 (2003)].
Whereas the Devetak-Winter protocol attains a rate-sum equal to the von 
Neumann entropy of the joint source, here we show that the full rate 
region is much more complex, due to the partially quantum nature of
the source. In particular, in the opposite case of compressing the
quantum part of the source, using the classical part as side information
at the decoder, typically the rate sum is strictly larger
than the von Neumann entropy of the total source.

We determine the full rate region in the 
\textit{generic} case, showing that, apart from the Devetak-Winter
point, all other points in the achievable region 
have a rate sum strictly larger than the joint entropy. We can interpret 
the difference as the price paid for the quantum encoder being ignorant 
of the classical side information.
In the general case, we give an achievable rate region, via protocols 
that are built on the decoupling principle, and the protocols of quantum 
state merging and quantum state redistribution. 
Our achievable region is matched almost by a single-letter converse,
which however still involves asymptotic errors and an unbounded
auxiliary system.
This chapter is based on the papers in \cite{ZK_cqSW_ISIT_2019,ZK_cqSW_2018}.


\section{The source and the compression model}\label{sec:Source and compression model}
\label{introduction}

The Slepian-Wolf problem of two sources  
correlated in a known way, but subject to separate, local compression \cite{Slepian1973}
has proved to provide a unifying principle for much of Shannon
theory, giving rise to natural information theoretic interpretations
of entropy and conditional entropy, and exhibiting deep 
connections with error correction, channel capacities and 
mutual information (cf.~\cite{csiszar_korner_2011}).
The quantum case has been investigated for two decades, starting with the
second author's PhD thesis \cite{Winter1999} \aw{and subsequently in
\cite{Devetak2003}}, up to the systematic study \cite{Ahn2006}, 
and while we still do not have a complete understanding of the rate region,
it has become clear that the problem is of much higher
complexity than the classical case. The quantum Slepian-Wolf
problem, and specifically quantum data compression with side
information at the decoder, has resulted in many fundamental
advances in quantum information theory, including the protocols
of quantum state merging \cite{Horodecki2007,Abeyesinghe2009} and quantum state 
redistribution \cite{Devetak2008_2}, 
which have given operational meaning to the conditional von 
Neumann entropy, the mutual information and the conditional quantum mutual 
information, respectively. 

A variety of resource models and different tasks have been 
considered over the years: The source and its recovery was
either modelled as an ensemble of pure states (following
Schumacher \cite{Schumacher1995}), or as a pure state between the encoders and a
reference system; the communication resource required was
either counted in qubits communicated, in addition either
allowing or disallowing entanglement, or it was counted in
ebits shared between the agents, but with free classical 
communication. While this latter model has led to the most
complete picture of the general rate region, in the present
chapter we will go back to the original idea \cite{Schumacher1995,Winter1999} 
of quantifying the communication, counted in qubits, between the
encoders and the decoder.

\bigskip

\textbf{Source model.} 
The source model we shall consider is a hybrid classical-quantum one,
with two agents, Alice and Bob, whose task is is to compress the 
classical and quantum parts of the source, respectively. They then send their
shares to a decoder, \aw{Debbie}, who has to reconstruct the classical
information with high probability and the quantum information with
high (average) fidelity.

In detail, the source is characterised by a classical source, i.e.~a probability
distribution $p(x)$ on a discrete (in fact: finite) alphabet $\mathcal{X}$
which is observed by Alice, and a family of quantum states $\rho_x$
on a quantum system $B$, given by a Hilbert space of finite dimension $|B|$. 
To define the problem of independent local compression (and
decompression) of such a correlated \aw{classical-quantum} source, we 
shall consider purifications $\psi_x^{BR}$ of the $\rho_x$,
i.e.~$\rho_x^B = \Tr_R \psi_x^{RB}$. Thus the source can be described
compactly by the cq-state
\[
  \omega^{XBR} = \sum_{x \in \mathcal{X}} p(x) \ketbra{x}{x}^X \otimes \ketbra{\psi_x}{\psi_x}^{BR}.
\]
We will be interested in the information theoretic limit of
many copies of $\omega$, i.e.
\begin{align*}
\omega^{X^n B^n R^n}
    &= \left(\omega^{XBR}\right)^{\otimes n} \\
    &= \sum_{x^n \in \mathcal{X}^n} p(x^n) \ketbra{x^n}{x^n}^{X^n} 
                                   \otimes \ketbra{\psi_{x^n}}{\psi_{x^n}}^{B^nR^n} \! \!\!\!,
\end{align*}
where we use the notation
\begin{align*}
  x^n              &= x_1 x_2 \ldots x_n, \\
  \ket{x^n}        &= \ket{x_1} \ket{x_2} \cdots \ket{x_n}, \\
  p(x^n)           &= p(x_1) p(x_2)  \cdots p(x_n), \text{ and} \\
  \ket{\psi_{x^n}} &= \ket{\psi_{x_1}} \ket{\psi_{x_2}} \cdots \ket{\psi_{x_n}}.
\end{align*}

Alice and Bob, receiving their respective 
parts of the source, separately encode these using the most general allowed 
quantum operations; the compressed quantum information, living on 
a certain number of qubits, is passed to the decoder who has to
output, again acting with a quantum operation, an element of $\mathcal{X}^n$
and a state on $B^n$, in such a way as to attain a low error probability
for $x^n$ and a high-fidelity approximation of the conditional quantum
source state, $\psi_{x^n}^{B^nR^n}$.
We consider two models: unassisted and entanglement-assisted, which we
describe formally in the following 
(see Figs.~\ref{fig:una} and \ref{fig:ea}).

\medskip
\textbf{Unassisted model.}
With probability $p(x^n)$, the source provides Alice and Bob respectively 
with states $\ket{x^n}^{X^n}$ and $\ket{\psi_{x^n}}^{B^nR^n}$.
Alice and Bob then perform their respective
encoding operations $\mathcal{E}_X:X^n \longrightarrow C_X$ and 
$\mathcal{E}_B:B^n \longrightarrow C_B$, 
\aw{respectively,} which are quantum operations, i.e.~completely positive and trace preserving (CPTP)
maps. \aw{Of course, as functions they act on the operators (density matrices) over 
the respective input and output Hilbert spaces. But as there is no risk of confusion,
we will simply write the Hilbert spaces when
denoting a CPTP map. Note that since $X$ is a classical random variable, $\mathcal{E}_X$
is entirely described by a cq-channel.}
We call $R_X=\frac1n \log|C_X|$ and $R_B=\frac1n \log|C_B|$ 
\aw{the} quantum rates of the compression protocol.
Since Alice and Bob are required to act independently, the joint encoding operation 
is $\mathcal{E}_X \otimes \mathcal{E}_B$. 
The systems $C_X$ and $C_B$ are then sent to \aw{Debbie} who performs
a decoding operation \aw{$\mathcal{D}:C_X C_B \longrightarrow \hat{X}^n\hat{B}^n$}.
$\hat{X}^n$ and $ \hat{B}^n$ are output systems with Hilbert spaces $\hat{X}^n$ and $\hat{B}^n$ which are isomorphic to Hilbert spaces $X^{ n}$ and $B^{n}$, respectively. 
We define the \aw{extended source state}
\begin{align}\label{eq: extended source model}
&\omega^{X^n {X'}^n B^n R^n}  \nonumber \\
 \quad     &= \left( \omega^{X{X'}BR}\right)^{\otimes n} \nonumber \\
\quad      &=\!\!\!\! \! \!\sum_{x^n \in \mathcal{X}^n} \!\!\!\!p(x^n) \!\ketbra{x^n}{x^n}^{X^n} \!\!\!\otimes \ketbra{x^n}{x^n}^{{X'}^n} 
                    \! \!\!\! \otimes \ketbra{\psi_{x^n}}{\psi_{x^n}}^{B^nR^n} \!\!\!\!,
\end{align}
and say the encoding-decoding
scheme has average fidelity $1-\epsilon$ if 
\begin{align} 
  \label{F_QCSW_unassisted1}
  \overline{F} := F\left(\omega^{X^n {X'}^n B^n R^n },\xi^{ \hat{X}^n {X'}^n \hat{B}^n R^n }
      \right)   
  \geq 1-\epsilon,
\end{align} 
where 
\[\xi^{ \hat{X}^n {X'}^n \hat{B}^n R^n }\!\!\!\!=\!\left(\mathcal{D} \circ (\mathcal{E}_X \otimes \mathcal{E}_B) \otimes \id_{{X'}^n R^n}\right) \omega^{X^n {X'}^n B^n R^n}\!\!\!, \]
and $\id_{{X'}^n R^n}$ is the identity (ideal) channel acting on ${X'}^n R^n$.
By the above fidelity definition and the linearity of CPTP maps, 
the average fidelity defined in (\ref{F_QCSW_unassisted1}) \aw{can be expressed equivalently as}
\begin{align} 
  \label{F_QCSW_unassisted2}
  \overline{F} \!\!= \!\!\!\!\!\sum_{x^n \in \mathcal{X}^n } \!\!\!p(x^n) F\! \left( \ketbra{x^n}{x^n}^{X^n}\!\!\! \!\otimes\! \ketbra{\psi_{x^n}}{\psi_{x^n}}^{B^nR^n} \!\!\!\!, \xi_{x^n}^{\hat{X}^n  \hat{B}^n R^n}\right)\nonumber
\end{align}
where 
\begin{align*}
&\xi_{x^n}^{\hat{X}^n  \hat{B}^n R^n}\!\!\!\!= \\
&\quad \!(\mathcal{D} \circ (\mathcal{E}_X \otimes \mathcal{E}_B) \otimes \id_{R^n}) \ketbra{x^n}{x^n}^{X^n} \!\!\!\otimes \ketbra{\psi_{x^n}}{\psi_{x^n}}^{B^nR^n}\!\!\!\! .
\end{align*}

We say that $(R_X,R_B)$ is an (asymptotically) achievable rate pair if 
there exist codes $(\mathcal{E}_X,\mathcal{E}_B,\mathcal{D})$ as above 
for every $n$, with fidelity $\overline{F}$ converging to $1$,
and classical and quantum rates converging to $R_X$ and $R_B$, respectively.
\aw{The rate region is the set of all achievable rate pairs, as a subset of $\mathbb{R}_{\geq 0}^2$.}

\begin{figure}[!t]
\centering
  \includegraphics[scale=.4]{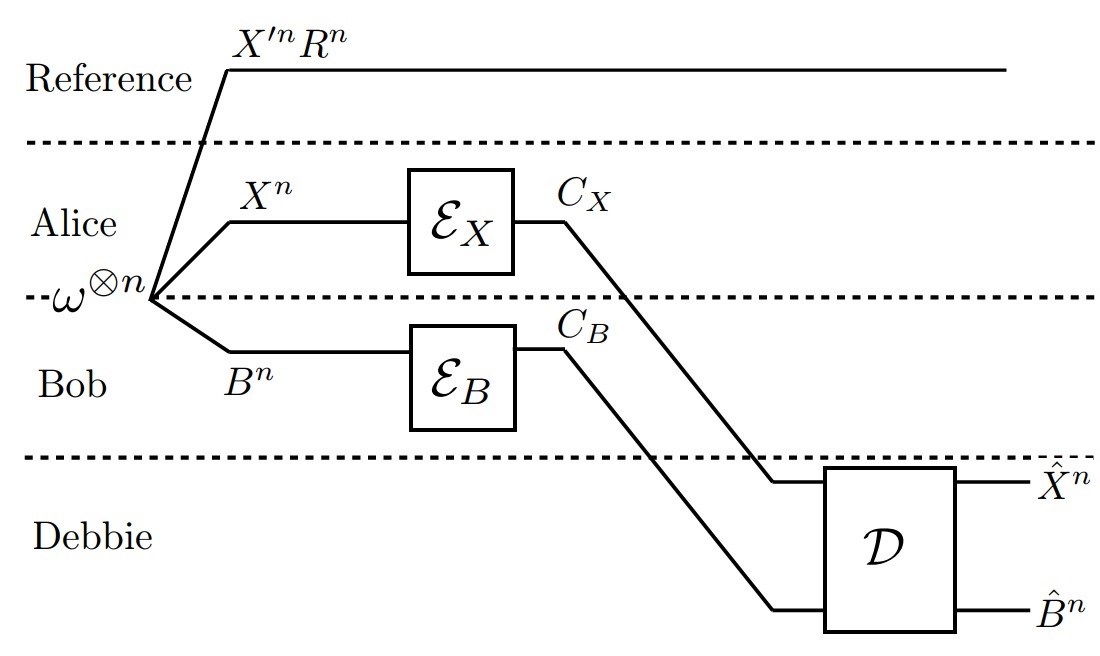}
  \caption{\aw{Circuit diagram of} the unassisted model. Dotted lines are used to 
           demarcate domains controlled by the different participants. 
           The solid lines represent quantum information \aw{registers}.}
  \label{fig:una}
\end{figure}

It is shown \aw{by Devetak and Winter} in \cite[Theorem 1]{Devetak2003}  and \cite[Corollary IV.13]{Winter1999} that the \aw{rate pair
\begin{equation}
  \label{eq:DW}
  (R_X,R_B) = (S(X|B),S(B))
\end{equation}
is achievable and optimal. The optimality is two-fold; first, the rate sum
achieved, $R_X+R_B=S(XB)$ is minimal, and secondly, even with unlimited $R_B$,
$R_X \geq S(X|B)$. This shows that the Devetak-Winter point is an extreme point
of the rate region. Interestingly,} Alice can achieve the rate $S(X|B)$ using only classical 
communication. However, we \aw{will} prove the converse theorems considering 
a quantum channel for Alice, which are obviously stronger statements.    
In Theorem \ref{theorem: generic full rate region}, we show that our system model is equivalent 
to the model considered in \cite{Devetak2003,Winter1999}, which implies the achievability and 
optimality of this rate pair in our system model. 
\aw{We remark that in \cite{Devetak2003}, the rate $R_B=S(B)$ was not explicitly
discussed, but it is clear that it can always be achieved by Schumacher's quantum
data compression \cite{Schumacher1995}, introducing an arbitrarily small additional error.}

\medskip
\textbf{Entanglement-assisted model.} 
This model \aw{generalizes the unassisted model, and it is basically the same,} 
except that we let Bob and \aw{Debbie} share entanglement \aw{and use it in encoding and decoding, respectively.
In addition, we take care of any possible entanglement that is produced in the process.
Consequently, while Alice's encoding  $\mathcal{E}_X:X^n \longrightarrow C_X$ remains the same,
the Bob's encoding and the decoding map now act as
$\mathcal{E}_B:B^n B_0 \longrightarrow C_B B_0'$ and
$\mathcal{D}:C_X C_B D_0 \longrightarrow \hat{X}^n\hat{B}^n D_0'$, respectively,
where $B_0$ and $D_0$ are $K$-dimensional quantum registers of Bob and \aw{Debbie}, 
respectively, designated to hold the initially shared entangled state, and $B_0'$ and $D_0'$
are $L$-dimensional registers for the entanglement produced by the protocol.
Ideally, both initial and final entanglement are given by maximally
entangled states $\Phi_K$ and $\Phi_L$, respectively.}
Correspondingly, we say \aw{that} the encoding-decoding scheme has average fidelity $1-\epsilon$ if 
\begin{align} 
  \label{F_QCSW_assisted}
  \overline{F} &:= F\left(\omega^{X^n {X'}^n B^n R^n }\otimes \Phi_L^{B_0'D_0'},
    \xi^{ \hat{X}^n {X'}^n \hat{B}^n R^n B_0' D_0'}                \right) \nonumber\\
   &\geq 1-\epsilon,
\end{align} 
where
\begin{align*}
\xi^{ \hat{X}^n {X'}^n \hat{B}^n R^n B_0' D_0'} 
 \!\!&=\!\left(\!\mathcal{D} \circ (\mathcal{E}_X \!\otimes \mathcal{E}_{BB_0} \!\otimes\! \id_{D_0}\!) \!\otimes\! \id_{{X'}^n R^n}\!\right)\\
 &  \quad \quad \quad \quad \quad \omega^{X^n {X'}^n B^n R^n }\otimes \Phi_L^{B_0'D_0'}.
\end{align*}
We call $E=\frac{1}{n}(\log K - \log L)$ the entanglement rate of the scheme.
The CPTP map $\mathcal{E}_{B}$ takes the input systems $B^nB_0$ to the compressed system 
$C_B$ \aw{plus Bob's share of the output entanglement, $B_0'$.}
\aw{Debbie} applies the decoding operation $\mathcal{D}$ on the received systems 
$C_XC_B$ and \aw{her part of the initial} entanglement $D_0$, 
to produce an output state on systems $\hat{X}^n \hat{B}^n$ \aw{plus her share of the output
entanglement, $D_0'$}. 
Similar to the unassisted model, $\hat{X}^n$ and $ \hat{B}^n$ are output systems with Hilbert spaces $\hat{X}^n$ and $\hat{B}^n$ which are isomorphic to Hilbert spaces $X^{n}$ and $B^{n}$, respectively.
We say $(R_X, R_B, E)$ is an (asymptotically) achievable rate triple if for all $n$
there exist entanglement-assisted codes as before, such that the
fidelity $\overline{F}$ converges to $1$, and
the classical, quantum and entanglement rates converge to
$R_X$, $R_B$ and $E$, respectively.
\aw{The rate region is the set of all achievable rate pairs, as a subset of 
$\mathbb{R}_{\geq 0}^2\times\mathbb{R}$. In the following we will be mostly
interested in the projection of this region onto the first two coordinates,
$R_X$ and $R_B$, corresponding to unlimited entanglement assistance.}

\medskip
\aw{It is a simple consequence of the time sharing principle that the rate regions,
both for the unassisted and the entanglement-assisted model, are closed convex regions.
Furthermore, since one can always waste rate, the rate regions are open to the ``upper right''.
This means that the task of characterizing the rate regions boils down to describing
the lower boundary, which can be achieved by convex inequalities. In the Slepian-Wolf
problem, they are in fact linear inequalities, and we will find analogues of these
in the present investigation.}

\medskip
Stinespring's dilation theorem \aw{\cite{Stinespring1955}} states that any CPTP map can be built 
from the basic operations of isometry and reduction to a subsystem by 
tracing out the environment system \cite{Stinespring1955}. 
Thus, the encoders and the decoder are without loss of generality isometries 
\begin{align*}
    U_X : {X^n} &\longrightarrow {C_X W_X},                          \\
    U_B : {B^n B_0} &\longrightarrow {C_B B_0' W_B},                 \\
    V   : {C_X C_B D_0} &\longrightarrow {\hat{X}^n \hat{B}^n D_0' W_D},
\end{align*}
\aw{where the new systems $W_X$, $W_B$ and $W_D$ are the environment systems
of Alice, Bob and \aw{Debbie}, respectively. They simply remain locally in 
possession of the respective party.}

\begin{figure}[!t]
\centering
  \includegraphics[scale=.4]{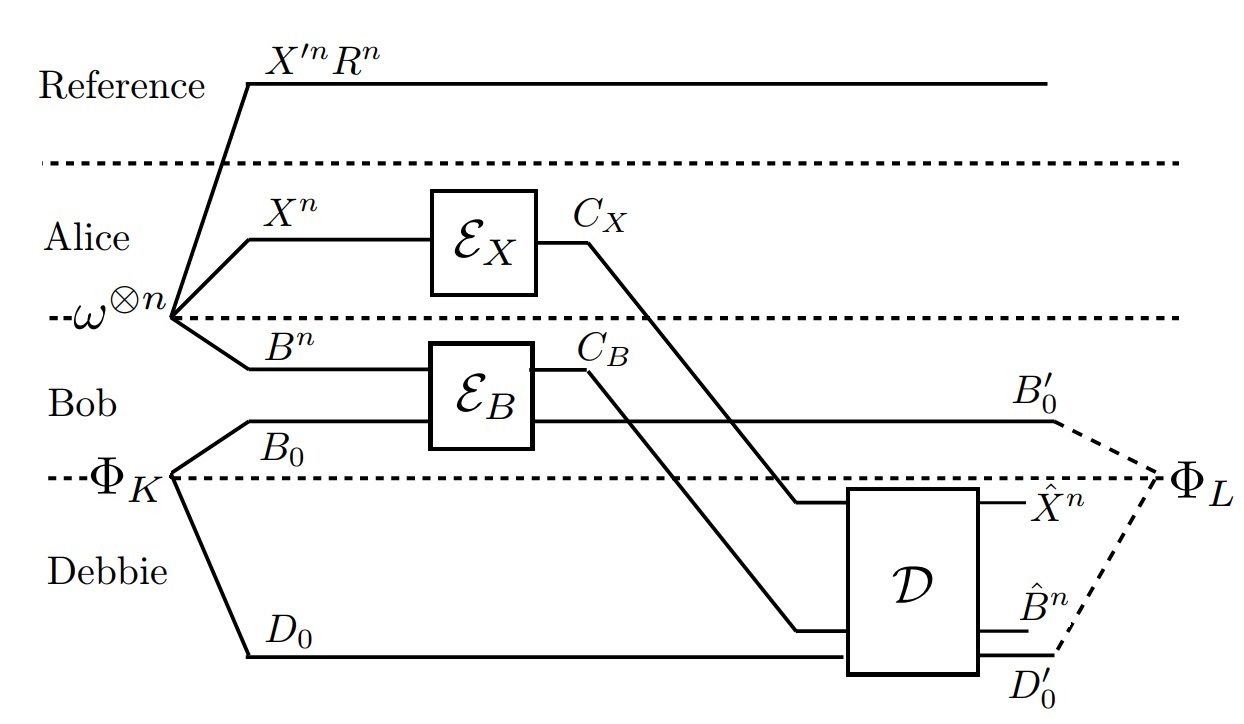}
  \caption{\aw{Circuit diagram of} the entanglement-assisted model. Dotted lines are used to 
           demarcate domains controlled by the different participants. 
           The solid lines represent quantum information \aw{registers}.}
  \label{fig:ea}
\end{figure}

The following lemma states that for a code of block length $n$ and error $\epsilon$, 
the environment parts of the encoding and decoding isometries, i.e.~$W_X$, $W_B$
and $W_D$, as well as the entanglement output registers $B_0'$ and $D_0'$, are decoupled from 
the reference $R^n$, conditioned on $X^n$. 
This lemma plays a crucial role in the proofs of converse theorems.

\begin{lemma}(Decoupling condition) 
\label{decoupling condition} 
For a code of block length $n$ and error $\epsilon$ in the entanglement-assisted model,  
let $W_X$, $W_B$ and $W_D$ be the environments of Alice's and Bob's encoding and of
Debbie's decoding isometries, respectively. Then,
\[
  I(W_XW_BW_D B_0'D_0':\hat{X}^n\hat{B}^nR^n|{X'}^n)_\xi \leq n \delta(n,\epsilon) ,
\] 
where $\delta(n,\epsilon) = 4\sqrt{6\epsilon}\log(|X| |B|) + \frac2n h(\sqrt{6\epsilon})$, 
\aw{with the binary entropy $h(\epsilon)=-\epsilon \log \epsilon -(1-\epsilon)\log (1-\epsilon)$;}
the conditional mutual information is with respect to the state 
\begin{align*}
 &\xi^{{X'}^n \hat{X}^n \hat{B}^n B_0' D_0' W_XW_BW_D R^{n}} \\
     &\quad \quad \quad =\left(V \circ (U_X \otimes U_{B} \otimes \1_{D_0}) \otimes \1_{{X'}^n R^n}\right)  \\
       &\quad \quad \quad \quad \quad \quad(\omega^{X^n {X'}^n B^n R^n } \otimes \Phi_K^{B_0D_0}) \\
        &\quad \quad \quad \quad\quad \quad \quad\left(V \circ (U_X \otimes U_{B} \otimes \1_{D_0}) \otimes \1_{{X'}^n R^n}\right)^{\dagger}.   
\end{align*}
\end{lemma}

\begin{proof}
We show that the fidelity criterion (\ref{F_QCSW_assisted}) 
implies that given $x^n$, the environments $W_X$, $W_B$ and $W_D$ of Alice's, Bob's and \aw{Debbie}'s isometries
are decoupled from the the rest of the output systems.

The parties share $n$ copies of the state $\omega^{X^{\prime} X B R}$, where 
Alice and Bob have access to systems $X^n$ and $B^n$, respectively, and ${X'}^n$ and $R^n$ are the reference systems.  
Alice and Bob apply the following isometries to encode their systems, respectively:
\begin{align*}
    U_{X}:{X^n}         &\longrightarrow {C_X W_X},  \\
    U_{B}:{B^n B_0}     &\longrightarrow {C_B B_0' W_B},                   
\end{align*}
where Alice and Bob send respectively their compressed information $C_X$ and $C_B$ to \aw{Debbie} 
and keep the environment parts $W_X$ and $W_B$ of their respective isometries for themselves. 
\aw{Debbie} applies \aw{the} decoding isometry \aw{$V:{C_X C_B D_0} \longrightarrow {\hat{X}^n \hat{B}^n D_0' W_D}$ 
to the systems $C_XC_B$ and her part of the entanglement $D_0$,
to generate the output systems $\hat{X}^n \hat{B}^n D_0'$, with $W_D$ the environment of her isometry.
This leads to the following final state after decoding:
\begin{align*}
  &\xi^{X'^n \hat{X}^n \hat{B}^n B_0'D_0' W_X W_B W_D R^n}\\
     & =\! \!\!\!\sum_{x^n} \! p(x^n)\! \ketbra{x^n}^{X'^n} \!\!\! \!\! \otimes \! \ketbra{\xi_{x^n}}^{\hat{X}^n \hat{B}^n B_0'D_0' W_X W_B W_D R^n} \!\!\!\!,     
\end{align*}
where 
\begin{align*}
 \ket{\xi_{x^n}}&^{\hat{X}^n \hat{B}^n B_0'D_0' W_X W_B W_D R^n} 
     \!\! \\
    &= V^{{C_XC_BD_0 \to \hat{X}^n\hat{B}^nD_0'W_D}}\\
      & \quad \quad  \big( U_X^{X^n \to C_XW_X}\!\!\ket{x^n}^{X^n} 
              \!\otimes U_B^{B^nB_0 \to C_BB_0'W_B}\\
              &\quad \quad\quad \quad\quad \quad\quad\quad \quad\quad (\ket{\psi_{x^n}}^{B^nR^n}\!\ket{\Phi_K}^{B_0D_0}) \bigr).
\end{align*}
}

The fidelity defined in \aw{Eq.}~(\ref{F_QCSW_assisted}) is \aw{now} bounded as follows:
\begin{align} 
  \label{eq-B1}
  \overline{F} 
    &= F\left(\omega^{X'^n X^n B^n R^n} \otimes \Phi_L^{B_0'D_0'},
               \xi^{X'^n \hat{X}^n \hat{B}^n B_0' D_0' R^n} \right)                                  \nonumber \\  
   &\leq F\left(\omega^{X'^n X^n B^n R^n}, \xi^{X'^n \hat{X}^n \hat{B}^n R^n} \right)                          \nonumber \\  
   &=   \!\!\! \!\sum_{x^n \in \mathcal{X}^n} \!\!\! p(x^n) \!F\!\left(\!\ketbra{x^n}{x^n}^{X^n}\!\!\!\! \otimes\! \ketbra{\psi_{x^n}}{\psi_{x^n}}^{B^nR^n}\!\!\!\!,
                                                     \xi_{x^n}^{\hat{X}^n \hat{B}^n R^n}\!\! \right)               \nonumber \\
   &=   \!\!\!\! \sum_{x^n} p(x^n)\! \sqrt{\!\bra{x^n}\bra{\psi_{x^n}}^{B^n\! R^n} \!
                                  \xi_{x^n}^{\hat{X}^n\hat{B}^nR^n} \!\ket{x^n}\!\ket{\psi_{x^n}}^{B^n\! R^n}\!}       \nonumber \\
   &\leq \sum_{x^n} p(x^n) \sqrt{\| \xi_{x^n}^{\hat{X}^n\hat{B}^n R^n} \|},
\end{align}
where in the first line $\xi^{X'^n \hat{X}^n \hat{B}^n B_0' D_0' R^n} =\left(\mathcal{D} \circ (\id_{X^nD_0} \otimes \mathcal{E}_{B}) \otimes \id_{X'^n R^n} \right) \omega^{X^n X'^n B^n R^n } \otimes \Phi_K^{B_0D_0}$.
The \aw{inequality in the second line} is due to the monotonicity of fidelity under partial trace, 
and \aw{$\|\xi_{x^n}^{\hat{X}^n \hat{B}^n R^n}\|$ denotes the operator norm, which in this case
of a positive semidefinite operator is the maximum eigenvalue of $\xi_{x^n}^{\hat{X}^n \hat{B}^n R^n}$. 
Now, consider the Schmidt decomposition of the state 
$\ket{\xi_{x^n}}^{\hat{X}^n \hat{B}^n B_0'D_0' W_X W_B W_DR^n}$ with respect to the partition
$\hat{X}^n \hat{B}^n R^n$ : $B_0' D_0'W_X W_B W_D$, i.e.
\begin{align*}
&\ket{\xi_{x^n}}^{\hat{X}^n \hat{B}^n B_0'D_0' W_X W_B W_DR^n}\\
   & \quad =\!\! \sum_{i} \!\!\sqrt{\lambda_{x^n}(i)}\ket{v_{x^n}(i)}^{\hat{X}^n \!\hat{B}^n\! R^n} \!\!\ket{w_{x^n}(i)}^{B_0'\!D_0' \!W_X\! W_B \!W_D}\!\!\!. 
\end{align*}}
High average fidelity $\overline{F} \geq 1-\epsilon$ implies that \emph{on average} 
the above states are approximately product states. In other words, the two subsystems are 
nearly decoupled on average:  
\begin{align}
  \label{eq:almost-pure}
  \sum_{x^n}& p(x^n) \!F\!\left(\! \ketbra{\xi_{x^n}}{\xi_{x^n}},
                             \xi_{x^n}^{\hat{X}^n \!\hat{B}^n\! R^n} \!\!\!\!\otimes \xi_{x^n}^{B_0'D_0' W_X W_B W_D} \right) \nonumber\\
    &=\! \sum_{x^n} \!p(x^n)\! \sqrt{\bra{\xi_{x^n}}
                                {\xi_{x^n}^{\hat{X}^n \!\hat{B}^n\! R^n} \!\!\otimes \xi_{x^n}^{B_0'\! D_0'\! W_X \!W_B\! W_D} 
                              \!\! \ket{\xi_{x^n}}}}                                                                 \nonumber\\
    &= \sum_{x^n} p(x^n) \sum_i \lambda_{x^n}(i)^{\frac32}                                                 \nonumber\\
    &\geq \sum_{x^n} p(x^n) \|\xi_{x^n}^{\hat{X}^n \hat{B}^n R^n}\|^{\frac32}                        \nonumber\\
    &\geq \left( \sum_{x^n} p(x^n) \sqrt{\|\xi_{x^n}^{\hat{X}^n \hat{B}^n R^n}\|} \right)^{3}  \nonumber\\
    &\geq (1-\epsilon)^3 
     \geq 1 - 3\epsilon,  
\end{align}
where in the first line $\ketbra{\xi_{x^n}}$ is a state on systems ${\hat{X}^n \!\hat{B}^n \!B_0'\!D_0'\! W_X\!W_B W_D\!R^n}$. The \aw{inequality in the fifth line} follows from the convexity of $x^3$ for $x \geq 0$, 
and \aw{in the sixth line we have used Eq.}~(\ref{eq-B1}). 
Based on the relation between fidelity and trace distance \aw{(Lemma \ref{lemma:FvdG}), 
we thus obtain for the product ensemble
\begin{align*}
& \zeta^{X'^n \hat{X}^n \hat{B}^n B_0'D_0' W_X W_B W_D R^n} \\
   & \quad:= \sum_{x^n}\! p(x^n)\!\ketbra{x^n}^{X'^n}\!\!\! \!\otimes \xi_{x^n}^{\hat{X}^n \hat{B}^n R^n} \!\!\!\otimes \xi_{x^n}^{B_0'D_0' W_X W_B W_D}\! \!\!,
\end{align*}
that
\begin{align*}
  &\| \xi - \zeta \|_1\\ 
    &=    \sum_{x^n} p(x^n)\\ 
    &\>\> \norm{\! \ketbra{\xi_{x^n}\! }{\xi_{x^n}\! }^{\! \hat{X}^n \! \hat{B}^n\!  B_0' \! D_0' \! W_X\! W_B \! W_D\! R^n}
                            \! \!  \! \! \!  \!-\! \xi_{x^n}^{\hat{X}^n\!  \hat{B}^n\!  R^n} \! \! \!\! \! \otimes\! \xi_{x^n}^{B_0'\! D_0' \! W_X\!  W_B \! W_D\! } \! }_{\! 1} \\
    & \leq 2\sqrt{6\epsilon}.   
\end{align*}}
By \aw{the Alicki-Fannes inequality (Lemma \ref{AFW lemma}), this implies
\begin{align}
  \label{decoupling_I}
 & I(\hat{X}^n\hat{B}^nR^n : B_0'D_0'W_XW_B W_D | {X'}^n)_\xi \nonumber \\
     &=    S(\hat{X}^n \hat{B}^n R^n | {X'}^n)_\xi \nonumber \\
     & \quad \quad \quad \quad - S(\hat{X}^n \hat{B}^n R^n | {X'}^n B_0'D_0' W_XW_B W_D)_\xi \nonumber \\ 
     &\leq 2\sqrt{6\epsilon} \log(|X|^n |B|^n |R|^n) + 2 h(\sqrt{6\epsilon}) \nonumber \\
     &\leq 2\sqrt{6\epsilon} \log(|X|^{2n} |B|^{2n}) + 2 h(\sqrt{6\epsilon}) \nonumber\\
         &  =: n \delta(n,\epsilon),  
\end{align}
where we note in the second line that 
$S(\hat{X}^n \hat{B}^n R^n|X'^n B_0'D_0' W_XW_B W_D)_\zeta 
 = S(\hat{X}^n \hat{B}^n R^n)_\zeta = S(\hat{X}^n \hat{B}^n R^n)_\xi$,
and in the forth line that we can without loss of generality assume $|R| \leq |X| |B|$, 
since that is the maximum possible dimension of the support of $\omega^R$.}
\end{proof}


\section{Quantum data compression with classical side information}
\label{sec:seiteninformation}
In this section, we assume that Alice sends her information to \aw{Debbie} 
at rate $R_X=\log \abs{\mathcal{X}}$ such that \aw{Debbie} can decode it 
perfectly, and we ask how much Bob can compress his system given that 
the decoder has access to classical side information $X^n$.  
\aw{This problem is a special case of the \emph{classical-quantum Slepian-Wolf problem}} (CQSW problem),
and we call it quantum data compression with classical side information at the decoder,
in analogy to the problem of classical data compression 
with quantum side information at the decoder which is addressed 
in \cite{Devetak2003,Winter1999}. Note we do not speak about the
compression and decompression of the classical part at all, and the 
decoder \aw{may} depend directly on $x^n$.
\aw{Of course, by Shannon's data compression theorem \cite{Shannon1948}, $X$ can always be 
compressed to a rate $R_X = H(X)$, introducing an arbitrarily small
error probability}.

We know from previous section that  Bob's encoder, in the entanglement-assisted model, is without loss of generality 
an isometry \aw{$U \equiv U_B:{B^n B_0} \longrightarrow {C W B_0'}$, 
taking $B^n$ and Bob's part of the entanglement $B_0$ to systems 
$C\otimes W \otimes B_0'$, where $C \equiv C_B$ is the compressed 
information of rate $R_B=\frac{1}{n}\log |C|$; $W \equiv W_B$ is the environment 
of Bob's encoding CPTP map, and $B_0'$ is the register carrying Bob's share of 
the output entanglement
(in this section, we drop subscript $B$ from $C_B$ and $W_B$). 
Having access to side information $X^n$, \aw{Debbie} applies the decoding isometry  
$V:X^n C D_0 \to \hat{X}^n \hat{B}^n W_D D_0'$ to generate 
the output systems $\hat{X}^n \hat{B}^n$ and entanglement share
$D_0'$, and where $W_D$ is the environment of the isometry.} 
We call this encoding-decoding scheme a side information code of 
block length $n$ and error $\epsilon$ for the entanglement-assisted model if the average fidelity 
(\ref{F_QCSW_assisted}) is at least $1-\epsilon$. 
Similarly, we define a side information code for the unassisted model by removing the corresponding systems of entanglement in the encoding and decoding isometries,
that is systems $B_0$, $B'_0$, $D_0$ and $D'_0$.

\medskip

To state our lower bound on the necessary compression rate, we introduce the
following quantity, which emerges naturally from the converse proof.

\begin{definition}
\label{I_delta}
For the state $\omega^{XBR} = \sum_x p(x) \ketbra{x}{x}^X \otimes \ketbra{\psi_x}{\psi_x}^{BR}$
and $\delta \geq 0$, define
\begin{align*}
&I_\delta(\omega) := \sup_{\cT} I(X:W)_\sigma 
                  \\
                  &\quad \quad \quad  \text{ s.t. } \cT:B\rightarrow W \text{ CPTP with }  
                  I(R:W|X)_\sigma \leq \delta,
\end{align*}
where the mutual informations are understood with respect to the
state $\sigma^{XWR} = (\id_{XR}\otimes \cT)\omega$ \aw{and $W$ ranges over arbitrary
finite dimensional quantum systems}.
Furthermore, let
\(
  \widetilde{I}_0 := \lim_{\delta \searrow 0} I_\delta = \inf_{\delta>0} I_\delta.
\)
\end{definition}

Note that the system $W$ is not restricted in any way, which is
the reason why in this definition we have a supremum and an infimum, rather 
than a maximum and a minimum. 
(It is a simple consequence of compactness of the domain of optimisation, 
together with the continuity of
the mutual information, that if we were to impose a bound on the dimension of $W$
in the above definition, the supremum in $I_\delta$ would be attained, and
for the infimum in $\widetilde{I}_0$, it would hold that $\widetilde{I}_0 = I_0$.)


\begin{lemma}
  \label{lemma:I-delta}
  The function $I_{\delta}(\omega)$ introduced in Definition~\ref{I_delta},
  has the following properties:
  \begin{enumerate}
      \item It is a non-decreasing function of $\delta$. 
      \item It is a concave function of $\delta$.
      \item It is continuous for $\delta > 0$.
      \item For any two states
  $\omega_1^{X_1B_1R_1}$ and $\omega_2^{X_2B_2R_2}$ and for $\delta,\delta_1,\delta_2 \geq 0$ 
  \[
    I_\delta(\omega_1\otimes\omega_2)
        = \max_{\delta_1+\delta_2= \delta} \left (I_{\delta_1}(\omega_1) + I_{\delta_2}(\omega_2)\right).
  \]
  \item $I_{n\delta}(\omega^{\otimes n}) = n I_\delta(\omega)$.
  \item $I_0$ and $\widetilde{I}_0$ are additive:
 \begin{align*}
     I_0(\omega_1\otimes\omega_2) &= I_0(\omega_1) + I_0(\omega_2)
    \quad \text{and} \\ \quad
    \widetilde{I}_0(\omega_1\otimes\omega_2) &= \widetilde{I}_0(\omega_1) + \widetilde{I}_0(\omega_2).
\end{align*}  
 \end{enumerate}
\end{lemma}

\begin{proof}
1) The non-decrease with $\delta$ is evident from the definition.

2) For this consider $\delta_1,\delta_2\geq 0$, $0<p<1$,
and let $\delta = p\delta_1+(1-p)\delta_2$. 
Let furthermore channels $\cT_i:B\rightarrow W_i$ be given ($i=1,2$) such that for the
states $\sigma_i^{XW_iR} = (\id_{XR}\otimes \cT_i)\omega$, $I(R:W_i|X)_{\sigma_i} \leq \delta_i$.
\par
Now define $W := W_1 \oplus W_2$, so that $W_1$ and $W_2$ can be considered
mutually orthogonal subspaces of $W$, and define the new channel
$\cT := p\cT_1 + (1-p)\cT_2:B\rightarrow W$. By the chain rule for the mutual
information, one can check that w.r.t.~$\sigma^{XWR} = (\id_{XR}\otimes \cT)\omega$,
\begin{align*}
  I(R:W|X)_\sigma &= p I(R:W_1|X)_{\sigma_1} + (1-p) I(R:W_2|X)_{\sigma_2} \\
  &\leq p\delta_1+(1-p)\delta_2 \\
  &= \delta,
\end{align*}
and likewise
\[
  I(X:W)_\sigma = p I(X:W_1)_{\sigma_1} + (1-p) I(X:W_2)_{\sigma_2}.
\] 
Hence, $I_\delta \geq p I(X:W_1)_{\sigma_1} + (1-p) I(X:W_2)_{\sigma_2}$; by maximizing over
the channels, the concavity follows.

3) Properties 1 and 2 imply that it is continuous for $\delta > 0$.

4) First, we prove that 
$I_\delta(\omega_1\otimes\omega_2) \leq \max_{\delta_1+\delta_2= \delta} \left( I_{\delta_1}(\omega_1) + I_{\delta_2}(\omega_2) \right)$; 
the other direction \aw{of the inequality is trivial from the definition}. 
Let $\cT:B_1 B_2 \to W$ be a CPTP map such that 
\begin{align}
  \label{eq1}
  \delta & \geq I(W:R_1R_2|X_1X_2) \nonumber \\
  &=  I(W:R_1|X_1X_2)+I(W:R_2|X_1R_1X_2)  \\ \nonumber
                                 &=  I(WX_2:R_1|X_1)+I(WX_1R_1:R_2|X_2), 
\end{align}
where the first line is to chain rule, and the second line is due to the independence of $\omega_1$ and $\omega_2$. 
We now define the new systems $W_1:=WX_2$ and $W_2:=WX_1R_1$. \aw{Then we have,}
\begin{align}
  \label{eq2}
  I(W:X_1 X_2) &=    I(W:X_2)+I(W:X_1|X_2)  \\ \nonumber
               &=    I(W:X_2)+I(WX_2:X_1)\\  \nonumber
               &\leq I(\underbrace{WX_1R_1}_{W_2}:X_2)+I(\underbrace{WX_2}_{W_1}:X_1),
\end{align}
where the second equality is due to the independence of $X_1$ and $X_2$. 
The inequality follows \aw{from data processing}.
From \aw{Eq.~(\ref{eq1})} we know that $I(W_1:R_1|X_1)\leq \delta_1$ and $I(W_2:R_2|X_2)\leq \delta_2$ 
for some $\delta_1+\delta_2= \delta$. Thereby, from \aw{Eq.~(\ref{eq2})} we obtain
\begin{align*}
  I_\delta(\omega_1 \otimes \omega_2) &\leq I_{\delta_1}(\omega_1 )+I_{\delta_2}(\omega_2 )\\
                                      &\leq \max_{\delta_1+\delta_2=\delta} I_{\delta_1}(\omega_1 )+I_{\delta_2}(\omega_2 ).
\end{align*}

5) \aw{Now, the multi-copy additivity follows easily from property 4:}
According to the first statement of the \aw{lemma}, we have
\[
  I_{n\delta}(\omega^{\otimes n}) = \max_{\delta_1+\ldots+\delta_n=n\delta} I_{\delta_1}(\omega)+\ldots+I_{\delta_n}(\omega). 
\] 
\aw{Here, the right hand side is clearly $\geq n I_\delta(\omega)$ since we can choose all 
$\delta_i = \delta$. By the concavity of $I_{\delta}(\omega)$ in $\delta$, on the other hand, 
we have for any $\delta_1+\ldots+\delta_n=n\delta$ that
\[
 \frac{1}{n}(I_{\delta_1}(\omega)+\ldots+I_{\delta_n}(\omega)) \leq I_{\delta}(\omega),  
\]
so the maximum is attained at $\delta_i=\delta$ for all $i=1,\ldots,n$}.

6) The property 4 of the \aw{lemma also} implies that $I_0$ and $\widetilde{I}_0$ are additive. 
\end{proof}

\begin{remark}
 There is a curious resemblance of our function
$I_\delta$ with the so-called \emph{information bottleneck function} introduced 
by Tishby \emph{et al.}~\cite{info-bottleneck}, whose generalization to
quantum information theory is recently being discussed \cite{Salek-QIB,Hirche-QIB}.
Indeed, the concavity and additivity properties of the two functions are proved 
by the same principles, although it is not evident to us, what --if any--, the 
information theoretic link between $I_\delta$ and the information bottleneck is.
\end{remark}

\subsection{Converse bound}\label{subsec: Converse bound}
In this subsection, we use the properties of the function $I_{\delta}(\omega)$ (Lemma~\ref{lemma:I-delta}) to prove
a lower bound on Bob's quantum communication rate.

\begin{theorem}
\label{converse_QCSW}
In the entanglement-assisted model, consider any side information code of block length $n$ and error $\epsilon$. 
Then, Bob's quantum communication rate is lower bounded as
\[
  R_B \geq \frac12 \left( S(B)+S(B|X) - I_{\delta(n,\epsilon)} - \delta(n,\epsilon) \right),
\]
where $\delta(n,\epsilon) = 4\sqrt{6\epsilon}\log(|X| |B|) +\frac2n h(\sqrt{6\epsilon})$.
Any asymptotically achievable rate $R_B$ is consequently lower bounded 
\[
  R_B \geq\frac{1}{2}\left( S(B)+S(B|X)-\widetilde{I}_0 \right).
\]
\end{theorem}

\begin{proof}
\aw{As already discussed in the introduction to this section,}
the encoder of Bob is without loss of generality 
an isometry $\aw{U:{B^n B_0} \longrightarrow {C W B_0'}}$.
The existence of a high-fidelity  
decoder using $X^n$ as side information 
implies that systems $W B_0'$ are decoupled from 
system $R^n$ conditional on \aw{$X^n$; indeed, by Lemma \ref{decoupling condition},} 
$I(R^n:WB_0'|X'^n) \leq n \delta(n,\epsilon)$. 
%
The first part of the converse reasoning is as follows:
\begin{align*}
  nR_B  =    \log |C| 
       &\geq S(C)  \\
       &\geq S(CW B_0')-S(W B_0')  \\
       &=    S(B^n)+S(B_0)-S(W B_0'), 
\end{align*} 
\aw{where the second inequality is a version of subadditivity, and
the equality in the last line holds because the encoding isometry \aw{$U$} 
does not change the entropy; furthermore, $B^n$ and $B_0$ are initially
independent.}
Moreover, the decoder can be dilated to an isometry $\aw{V}: X^n C D_0 \longrightarrow \hat{X}^n \hat{B}^n D_0' W_D$, 
where $W_D$ and $D_0'$ are the environment of \aw{Debbie}'s decoding operation and 
\aw{the output of \aw{Debbie}'s entanglement, respectively.}
Using the decoupling condition of Lemma \ref{decoupling condition} once more, we have
\begin{align*}
 nR_B+S(D_0)&= \log |C| + S(D_0) \\
     &\geq S(C)+S(D_0)  \\
     &\geq S(C D_0)  \\
     &\geq S(X^nCD_0|{X'}^n)  \\
     &= S(\hat{X}^n\hat{B}^n D_0' W_D|{X'}^n) \\
     &= S(W B_0' R^n|{X'}^n) \\
     &\geq S(R^n|{X'}^n)\!+S(WB_0'|{X'}^n) \!\! -\! n \delta(n,\epsilon) \\
     &= S(B^n|X^n)\!+S(WB_0'|{X'}^n)  \!\! -\! n \delta(n,\epsilon), 
\end{align*}
\aw{where the third and fourth line are by subadditivity of the entropy;
the fifth line follows because the decoding isometry $V$ does not change the entropy. 
The sixth line holds because for any given $x^n$  
the overall state of the systems $\hat{X}^n\hat{B}^n B_0'D_0' W W_DR^n$ is pure. 
The penultimate line is due to the decoupling condition (Lemma \ref{decoupling condition}), 
and the last line follows because for a given $x^n$ the overall state 
of the systems $B^nR^n$ is pure.}
Adding these two relations and dividing by $2n$, we obtain 
\begin{align*}
  R_B \!\geq \! \frac{1}{2} (S(B)\!+\!S(B|X)) \!-\! \frac{1}{2n} I(X'^n\!\!:\!WB_0') \!-\! \frac{1}{2}\delta(n,\epsilon),
\end{align*}
where the terms $S(B_0)$ and $S(D_0)$ cancel out each other because
$B_0$ and $D_0$ are $K$-dimensional quantum registers with maximally
entangled states $\Phi_K$. 

In the above \aw{inequality, the mutual information on the right hand side} 
is bounded as
\begin{align*}
 I(X'^n:WB_0') \leq I_{n\delta(n,\epsilon)}({\omega^{\otimes n}}) = nI_{\delta(n,\epsilon)}({\omega}), 
\end{align*}
\aw{To see this, define the CPTP map $\mathcal{T}:B^n \longrightarrow \widetilde{W}:= WB_0'$ as
$\mathcal{T}(\rho):= \Tr_{CD_0} (U\otimes\1)(\rho\otimes\Phi_K^{B_0D_0})(U\otimes\1)^\dagger$.
Then we have $I(R^n:\widetilde{W}|{X'}^n) \leq n\delta(n,\epsilon)$, and hence
the above inequality follows directly from Definition \ref{I_delta}.}

The second statement of the theorem follows because
$\delta(n,\epsilon)$ tends to zero as \aw{$n \rightarrow \infty$ and $\epsilon \rightarrow 0$.}
\end{proof}

\begin{remark}
\label{rem:example}
Notice that the term $\frac{1}{n} I({X'}^n:WB_0')$ is not necessarily small. 
For example, suppose \aw{that the source is of the form
$\ket{\psi_x}^{BR} = \ket{\psi_x}^{B'R} \otimes \ket{\psi_x}^{B''}$ for all $x$}; 
clearly it is possible to perform the coding task by
coding only $B'$ and trashing $B''$ (i.e.~putting it into $W$), because 
by having access to $x$ the decoder can reproduce $\psi_x^{B''}$ locally. In 
this setting, characteristically $\frac{1}{n} I({X'}^n:WB_0')$ does not go 
to zero because ${B''}^n$ ends up in $W$.  
\end{remark}

\subsection{Achievable rates}\label{subsec:Achievable rates}
In this subsection, we provide achievable rates \aw{both for the unassisted and entanglement-assisted} model.

\begin{theorem}
\label{State_merging_rate}
\aw{In the unassisted model, there exists a sequence of side information codes that
compress Bob's system $B^n$} at the asymptotic qubit rate
\begin{align*}
  R_B = \frac{1}{2}\left( S(B)+S(B|X) \right).
\end{align*}
\end{theorem}
\begin{proof}
We recall that in a side information code, Bob aims to send his system $B^n$ to Debbie while she has access to side information system $X^n$ as explained at the beginning of this section.
We can use the fully quantum Slepian-Wolf protocol (FQSW), also called coherent state merging protocol
(\cite{Abeyesinghe2009} section 7), as a subprotocol since it considers the entanglement 
fidelity as the decodability criterion, which is more stringent than
the average fidelity defined in (\ref{F_QCSW_unassisted1}). 
Namely, let 
\[
  \ket{\Omega}^{X X^{\prime} B R}
   =\sum_{x \in \mathcal{X}} \sqrt{p(x)} \ket {x}^{X} \ket {x}^{X'} \ket{\psi_{x}}^{B R}
\]
be the source in \aw{the} FQSW problem, where $B$ is the system to be compressed, $X$ is the side information at the decoder,  $R$ and $X'$ are the reference systems. Bob applies the corresponding encoding map of the FQSW protocol $\cE_B: B^n \longrightarrow C$ and sends system $C$ to Debbie who then applies the decoding map of the FQSW protocol $\cD: X^n C \longrightarrow X^n \hat{B}^n$ to her side information system $X^n$ and the compressed information $C$ to reconstruct system $\hat{B}^n$. These encoding and decoding operations preserve the entanglement fidelity $F_e$ which  is the decodability criterion  of the FQSW problem:
\begin{align*} 
  \label{F_QCSW}
   F_e &\!
   \!=\! \! F \! \! \left(\! \Omega^{X^{\! n}\!  {X'}^n\!  B^{\! n} \! R^n } \! \! \! ,\! \left(\mathcal{D} \! \circ\!  (\id_{\! X^n}\!  \! \otimes \! \mathcal{E}_{\! B}\! )\!  \otimes\!  \id_{{X'}^n\!   R^n}\! \right) \! \Omega^{\! X^{\! n} \!  {X'}^{ n}  \!  B^n\!  R^n } \! \! \right)  \nonumber  \\
       &\! \! \! \leq \!  \!  F \! \! \left(\! \omega^{X^n \! {X'}^n\!  B^n \! R^n \! } \! \! ,\! \left(\mathcal{D} \! \circ\!  (\id_{X^{\! n}}\!  \! \otimes \! \mathcal{E}_{\! B}\! ) \! \otimes\!  \id_{{X'}^n \! R^n}\! \right) \! \omega^{X^{\! n} \! {X'}^n\!  B^{\! n} \! R^{n} \! } \! \right) \\
        & \! \! = \overline{F},
\end{align*}
where the inequality is due to the monotonicity of fidelity 
under \aw{CPTP maps, namely the projective measurement on system $X'$ in the computational
basis $\{ \ketbra{x}{x}\}$)}. 
Therefore, if an encoding-decoding scheme attains an entanglement fidelity for 
the FQSW problem going to $1$, then it will have the average fidelity for 
the CQSW problem going to $1$ as well. Hence, the FQSW rate 
\begin{align*}
  R_B= \frac{1}{2}I(B:X^{\prime} R)_{\Omega}=\frac{1}{2}(S(B)_{\omega}+S(B|X)_{\omega}),
\end{align*}
\aw{is achievable.}
\end{proof}

\begin{remark}
Notice that for the source considered at the end of the previous subsection
\aw{in Remark \ref{rem:example}}, where 
$\ket{\psi_x}^{BR} = \ket{\psi_x}^{B'R} \otimes \ket {\psi_x}^{B''}$ 
for all $x$, we can achieve a rate strictly smaller than the rate 
stated in the above theorem. The reason is that $R$ is only 
entangled with $B'$, so clearly it is possible to perform the coding task by
coding only $B'$ and trashing $B''$ because by having access to $x$ 
the decoder can reproduce the state $\psi_x^{B''}$ locally. 
Thereby, the rate $\frac{1}{2} (S(B')+S(B'|X))$ is achievable by 
applying coherent state merging as above.
\end{remark}

\medskip
The previous observation shows that in general, the rate $\frac12(S(B)+S(B|X))$
from Theorem \ref{State_merging_rate} is not optimal. By looking for a systematic
way of obtaining better rates, we have the following result in the entanglement-assisted
model.

\begin{theorem}
\label{QSR_achievability}
\aw{In the entanglement-assisted model, there exists a sequence of side information codes} 
with the following asymptotic entanglement and qubit rates:
\begin{align*}
  E&=\frac{1}{2}\left(I(C:W)_{\sigma}-I(C:X)_{\sigma}\right) 
    \quad \text{and}\quad\\
  R_B&= \frac{1}{2}\left(S(B)_{\omega}+S(B|X)_{\omega}-I(X:W)_{\sigma} \right),   \nonumber
\end{align*}
where $C$ and $W$ are, respectively, the system and environment of an
isometry $V:{B\rightarrow CW}$ on $\omega^{XBR}$ 
producing the state $\sigma^{XCWR} = (\1_{XR}\otimes V)\omega^{XBR} (\1_{XR}\otimes V)^{\dagger}$, such that $I(W:R|X)_{\sigma}=0$.
\end{theorem}
\begin{proof}
Notice that there is always an isometry $V:{B\rightarrow CW}$  
with $I(W:R|X)_{\sigma}=0$, and the trivial example is the isometry $V:{B\rightarrow B W}$ where system $W$ is a trivial system with state $\ketbra{0}^W$.

First, Bob applies \aw{the} isometry $V$ to each copy \aw{of the $n$ systems $B_1,\ldots,B_n$}:  
\begin{align*}
   &\sigma^{X X^{\prime} CWR} \\
       &\quad \quad =\!\! (V^{B \to C W} \!\!\otimes\! \1_{X X^{\prime}R}) 
             \omega^{X X^{\prime}  BR}
          (V^{B \to C W} \!\!\otimes\! \1_{X X^{\prime}R})^\dagger   \\
       &\quad \quad = \sum_x p(x) \ketbra{x}^{X}\otimes \ketbra{x}^{X^{\prime}}\otimes \ketbra{\phi_{x}}^{CWR}. 
\end{align*}
Now consider the following source state from which the state $\sigma^{X X^{\prime} CWR}$ is obtained by applying projective measurement on system $X'$ in the computational
basis $\{ \ketbra{x}{x}\}$,
\[
  \ket{\Sigma}^{X X^{\prime} CW R}
   =\sum_{x \in \mathcal{X}} \sqrt{p(x)} \ket {x}^{X} \ket {x}^{X'} \ket{\phi_{x}}^{CW R}.
\]
For this source, consider Bob and \aw{Debbie} respectively hold 
the $CW$ and $X$ systems, and Bob wishes to send system $C$ to \aw{Debbie} while keeping 
$W$ for himself.
For many copies of the above state, \aw{the} parties can apply \aw{the} quantum state redistribution 
(QSR) protocol \cite{Yard2009,Oppenheim2008} for transmitting $C$, having access to system $W$ as 
side information at the encoder and to $X$ as side information at the decoder. 
According to this protocol, Bob needs exactly the rate of 
$R_B = \frac{1}{2}I(C:X^{\prime} R |X)_{\Sigma}
     = \frac{1}{2}(S(B)_{\omega}+S(B|X)_{\omega}-I(X:W)_{\sigma})$ 
qubits of communication.
\aw{The protocol requires the rate of $\frac{1}{2} I(C:W)_{\Sigma}=\frac{1}{2} I(C:W)_{\sigma}$
ebits of entanglement shared between the encoder and decoder, and at the end of the protocol
the rate of $\frac{1}{2} I(C:X)_{\Sigma}=\frac{1}{2} I(C:X)_{\sigma}$ ebits of entanglement is 
distilled between the encoder and the decoder (see equations (1) and (2) in \cite{Yard2009}).}
This protocol \aw{attains high} fidelity for \aw{the} state $\Sigma^{X^n {X'}^n C^n W^n R^n }$, 
and consequently for \aw{the} state $\sigma^{X^n {X'}^n  C^n W^n R^n }$ due to \aw{the} monotonicity 
of fidelity under \aw{CPTP maps}:
\begin{align} \label{F_QSR}
  1\!\!-\!\epsilon \! &\leq \!F \!\left(\!{\Sigma}^{X^n {X'}^n  C^n W^n R^n }\!\!\!\!\otimes\! \Phi_L^{B_0'D_0'}\!\!,{\hat{\Sigma}}^{\hat{X}^n {X'}^n \hat{ C}^n \! \hat{W}^n\!  R^n \! B'_0 \! D'_0 }
                          \right) \nonumber \\
             &\leq \!F \!\left(\!{\sigma}^{X^n {X'}^n  C^n W^n R^n }\!\!\!\!\otimes\! \Phi_L^{B_0'D_0'}\!\!,{\hat{\sigma}}^{\hat{X}^n {X'}^n \hat{ C}^n \! \hat{W}^n\!  R^n \! B'_0 \! D'_0 }
                          \right)\!\!,
\end{align} 
where
\begin{align*}
&{\hat{\Sigma}}^{\hat{X}^n \!{X'}^n \!\hat{ C}^n \! \hat{W}^n\!  R^n \! B'_0 \! D'_0 } \\
&\!\! =\!\!\left(\!\mathcal{D} \!\circ\! (\!\id_{X^n \!D_0} \!\otimes \!\mathcal{E}_{C\!W \!B_0}\!)\! \otimes\! \id_{{X'}^n\! R^n}\!\right)  
 \!\!\Sigma^{\!X^n\! {X'}^n\! C^n\! W^n \!R^n } \!\!\!\otimes\! \Phi_K^{\!B_0\!D_0}\!\!,
\end{align*}
and 
\begin{align*}
&{\hat{\sigma}}^{\hat{X}^n \!{X'}^n \!\hat{ C}^n \! \hat{W}^n\!  R^n \! B'_0 \! D'_0 } \\
&\!\! =\!\!\left(\!\mathcal{D} \!\circ\! (\!\id_{X^n \!D_0} \!\otimes \!\mathcal{E}_{C\!W \!B_0}\!)\! \otimes\! \id_{{X'}^n\! R^n}\!\right)  
 \!\sigma^{\!X^n\! {X'}^n\! C^n\! W^n \!R^n } \!\!\!\otimes\! \Phi_K^{\!B_0\!D_0}\!\!\!,
\end{align*}
and $\mathcal{E}_{CW B_0}$ and $\mathcal{D}$ are respectively the 
encoding and decoding operations of the QSR protocol. 
The condition $I(W:R|X)_{\sigma}=0$ implies that for every $x$ the systems $W$ and $R$ are decoupled:  
\begin{align*}\label{Decoupling_0} 
\phi_{x}^{WR}=\phi_{x}^W  \otimes \phi_{x}^R.   
\end{align*} 
By Uhlmann's theorem \cite{UHLMANN1976,Jozsa1994_2}, 
there exist isometries $V_{x}:{C\rightarrow VB}$ for all $x \in \mathcal{X}$, such that 
\[ 
  (\1\otimes V_{x}^{C\rightarrow VB}) \ket{\phi_{x}}^{CWR}
             =\ket{\nu_{x}}^{VW}  \otimes \ket{\psi_{x}}^{BR}.   
\]
After applying the decoding operation $\mathcal{D}$ of QSR, 
\aw{Debbie} applies the isometry $V_{x}:{C\rightarrow VB}$ for each $x$, 
which does not change the fidelity (\ref{F_QSR}).
By tracing out the unwanted systems $V^n W^n$, 
due to the monotonicity of \aw{the} fidelity under partial trace, 
the fidelity defined in (\ref{F_QCSW_assisted}) \aw{will go to $1$} in this encoding-decoding scheme. 
\end{proof}



\begin{remark}
In Theorem \ref{QSR_achievability}, the smallest achievable rate, 
when unlimited entanglement is available,
is equal to $\frac{1}{2}(S(B)+S(B|X)-I_0)$.
This rate resembles the converse bound 
$R_B \geq \frac{1}{2}(S(B)+S(B|X)-\widetilde{I}_0)$,
except \aw{that $\widetilde{I}_0 \geq I_0$}. 
In the definition of $\widetilde{I}_0$, it seems unlikely that we can take the 
limit  of $\delta$ going to 0 directly because there is no dimension bound on 
the systems $C$ and $W$, so compactness cannot be used directly to prove that 
$\widetilde{I}_0$ and $I_0$ are equal. 
\end{remark}

\aw{\begin{remark}
Looking again at the entanglement rate in Theorem \ref{QSR_achievability},
$E=\frac{1}{2}\left(I(C:W)_{\sigma}-I(C:X)_{\sigma}\right)$, we reflect that
there may easily be situations where $E\leq 0$, meaning that no entanglement is
consumed, and in fact no initial entanglement is necessary. In this case,
the theorem improves the rate of Theorem \ref{State_merging_rate} by the
amount $\frac12 I(X:W)$. 
This motivates the definition of the following variant of $I_0$,
\begin{align*}
I_{0-}(\omega) :=& \sup I(X:W) \text{ s.t. } I(R:W|X)=0,\\
& I(C:W)-I(C:X) \leq 0,
\end{align*}
where the supremum is over all isometries $V:B\rightarrow CW$. 
\par
As a corollary to these considerations, in the unassisted model
the rate $\frac{1}{2}\left(S(B)+S(B|X)-I_{0-} \right)$ is achievable.
\end{remark}}

\subsection{Optimal compression rate for generic sources}
\label{sec:generic side info}

In this subsection, we find the optimal compression rate for  
\emph{generic} sources, by which we mean any source except for a
submanifold of lower dimension within the set of all sources.
Concretely, we will consider sources where there is at least one $x$ 
for which the reduced state $\psi_x^B= \Tr_R \ketbra{\psi_x}{\psi_x}^{BR}$ 
has full support on $B$. 
In this setting, coherent state merging as a subprotocol gives the optimal 
compression rate, so not only does the protocol not use any initial 
entanglement, but some entanglement is distilled at the end of the protocol.

\begin{theorem}
\label{theorem:generic optimal rate}
In both unassisted and entanglement-assisted models, for any side information code of a generic source, the asymptotic compression rate $R_B$ of Bob is lower bounded
\begin{align*}
  R_B \geq \frac{1}{2}\left(S(B)+S(B|X)\right),
\end{align*}
so the protocol of Theorem \ref{State_merging_rate} has optimal rate for a generic source.
Moreover, in that protocol no prior entanglement is needed and 
a rate $\frac{1}{2}I(X:B)$ ebits of entanglement 
is distilled between the encoder and decoder.
\end{theorem}

\begin{proof}
The converse bound of Theorem \ref{converse_QCSW} states that 
\aw{the} asymptotic quantum communication rate of Bob is lower bounded as
\begin{align}
  R_B \geq \frac{1}{2}\left(S(B)+S(B|X)-\widetilde{I}_0 \right),   \nonumber
\end{align}
where $\widetilde{I}_0$ \aw{comes from} Definition \ref{I_delta}. We will show that for generic sources, $\widetilde{I}_0 = I_{0}= 0$.
Moreover, Theorem \ref{State_merging_rate} states that using coherent state 
merging, the asymptotic qubit rate of $\frac{1}{2}(S(B)+S(B|X))$ is achievable, 
that no prior entanglement is required and a rate of 
$\frac{1}{2}I(X:B)$ ebits of entanglement is distilled between the encoder and the decoder.

We show that for any CPTP map $\cT:B\to W$, which acts on a generic $\omega^{XBR}$ and produces  
state $\sigma^{XWR} = (\id_{XR}\otimes \cT)\omega^{XBR}$ such that $I(R:W|X)_{\sigma}\leq \delta$ 
for  $\delta \geq 0$, the quantum mutual information 
$I(X:W)_{\sigma} \leq \delta' \log |X| +2h (\frac12\delta')$ where $\delta'$ is defined in 
\aw{Eq.~(\ref{phix_phi0_distance}) below}. 
Thus, we obtain
\[ 
  \widetilde{I}_0 = \lim_{\delta \searrow 0} I_\delta = 0.
\]
\aw{To show this claim, we proceed as follows.} From $I(R:W|X)_{\sigma}\leq \delta$ we have
\begin{align*}
  I(R:W|X=x)_\sigma \leq \frac{\delta}{p(x)} \quad \quad \forall x \in \mathcal{X},
\end{align*}
\aw{so} by Pinsker's inequality \cite{Schumacher2002} we obtain
\begin{align*}
  \left\| \phi_x^{W R}- \phi_x^{W} \otimes \phi_x^{R} \right\|_1  \leq \sqrt{\frac{2 \delta \ln 2}{p(x)}} \quad \quad \forall x \in \mathcal{X}.
\end{align*}
By Uhlmann's theorem (Lemma~\ref{lemma:Uhlmann1} and Lemma~\ref{lemma:Uhlmann2}), there exists an isometry $V_x:{C \to BV}$ such that
\begin{align}\label{eq3}
 &\left\| (V_x \otimes \1_{WR}) \phi_x^{CW R} (V_x \otimes \1_{WR})^{\dagger}
                                        - \theta_x^{WV} \otimes \psi_x^{BR} \right\|_1  \nonumber \\
 & \quad \quad \quad \quad \quad \quad  \leq \sqrt{\sqrt{ \frac{\delta\ln 2}{2p(x)}} \left(2-\sqrt{ \frac{\delta\ln 2}{2p(x)}}\right)}, 
\end{align}
where $\theta_x^{WV}$ is a purification of $\phi_x^{W}$.
Since the source is generic by definition there is an $x$, say $x=0$, 
for which $\psi_0^B$ has full support on
$\mathcal{L}(\mathcal{H}_B)$, i.e.~$\lambda_0:=\lambda_{\min}(\psi_0^B)>0$. 
By Lemma \ref{full_support_lemma} in Appendix \ref{Miscellaneous_Facts}, 
for any $\ket{\psi_x}^{BR}$ there is an operator $T_x$ acting on the reference system such that
\begin{align*}
  \ket{\psi_x}^{BR} = (\1_B \otimes T_x) \ket{\psi_0}^{BR}. \nonumber
\end{align*}
Using this fact, we show that the decoding isometry $V_0$ in \aw{Eq.~(\ref{eq3})} works for all states:
\begin{align*}
 & \bigl\| (V_0 \otimes \1_{WR}) \phi_x^{CWR} (V_0^{\dagger} \otimes \1_{W R})
- \theta_0^{WV}  \otimes  \psi_x^{BR} \bigr\|_1   \\
   &= \bigl\|\!(\!V_0 \!\otimes \!\1_{W\!R}) (\!\1_{C\!W} \!\otimes \!T_x \!) \phi_0^{C\!W\!R} (\1_{C\!W} \!\otimes \!T_x)^{\dagger} (\!V_0^{\dagger}\otimes \1_{W \!R}\!) \\
   &\quad \quad\quad\quad - \theta_0^{W\!V} \!\otimes \!(\1_B \otimes T_x)\psi_0^{B\!R}(\1_B \!\otimes\! T_x)^{\dagger} \bigr\|_1   \\
   &= \!\!\bigl\| \!(\!\1_{BVW}\!\otimes \! T_x\!)(V_0\!\otimes\!\1_{W\!R}) \phi_0^{C\!W\!R} \!(\!V_0^{\dagger}\!\otimes\!\1_{W \!R}\!) (\!\1_{B\!V\!W}\!\otimes \!T_x^{\dagger}\!) 
         \\
          &\quad \quad\quad\quad   - (\1_{BVW}\! \otimes \! T_x)\theta_0^{WV} \otimes\psi_0^{B\!R}(\1_{B\!V\!W} \!\otimes\! T_x^{\dagger}) \bigr\|  \\
   &\leq  \norm{\1_{BVW} \otimes T_x}_{\infty}^2  \\
       &  \quad \quad\quad \norm{ \!(\!V_0 \!\otimes\! \1_{WR})\! \phi_0^{C\!W\!R}  (V_0^{\dagger}  \!\otimes\! \1_{W\!R}) -  \theta_0^{W\!V} \!\otimes\!\psi_0^{B\!R}\!}_1   \\
   &\leq   \frac{1}{\lambda_0}  \sqrt{\sqrt{ \frac{\delta\ln 2}{2p(0)}} \left(2-\sqrt{ \frac{\delta\ln 2}{2p(0)}}\right)}, 
\end{align*}
where the \aw{last two} inequalities follow from Lemma \ref{T_norm1_inequality} 
and Lemma \ref{full_support_lemma}, respectively. 
By tracing out the systems $VBR$ in the above chain of inequalities, we get 
\begin{align}
  \label{phix_phi0_distance}
  \norm{\phi_x^{W}- \phi_0^{W}}_1 
      \leq \frac{1}{\lambda_0}\sqrt{\sqrt{ \frac{\delta\ln 2}{2p(0)}} \left(2-\sqrt{ \frac{\delta\ln 2}{2p(0)}}\right)} \aw{=: \delta'}.   
\end{align}
Thus, by triangle inequality we obtain
\begin{align}
  \label{phi_phi0_distance} 
&  \norm{\underbrace{\sum_x p(x) \ketbra{x}{x}^X \otimes \phi_x^{W}}_{\sigma^{XW}}
         - \underbrace{\sum_x p(x) \ketbra{x}{x}^X \otimes \phi_0^{W}}_{\aw{=:}\sigma_0^{XW}}}_1 \nonumber \\
    &\quad \quad \quad \quad\quad \quad \leq \sum_x p(x) \norm{  \phi_x^{W}- \phi_0^{W}}_1 \nonumber \\ 
    &\quad \quad\quad \quad\quad \quad\leq \frac{1}{\lambda_0}\sqrt{\sqrt{ \frac{\delta\ln 2}{2p(0)}} 
          \left(2-\sqrt{ \frac{\delta\ln 2}{2p(0)}}\right)} = \delta'.    
\end{align}
By applying \aw{the Alicki-Fannes inequality in the form of Lemma \ref{AFW lemma},} to 
\aw{Eq.}~(\ref{phi_phi0_distance}), we have
\begin{align*}
  I(X\!:\!W)_\sigma &\!\!=\!S(\!X\!)_{\sigma}\!\!-\!S(X|W)_{\sigma} \!+\!S(X|W)_{\sigma_0}\!\!-\!\!S(X|W)_{\sigma_0} \\
                &= S(X|W)_{\sigma_0}-S(X|W)_{\sigma}                                 \\
                &\leq \delta' \log |X| + 2h\left(\frac12\delta'\right),  
\end{align*}
and the right \aw{hand side} of the above inequality vanishes for \aw{$\delta\rightarrow 0$}.
\end{proof}

\section{Towards the full rate region}\label{sec: full problem}
In this section, we consider the full rate region of the distributed compression of 
a \aw{classical-quantum} source.

\begin{theorem}
\label{unknown.theorem}
In the unassisted model, for distributed compression of a \aw{classical-quantum} source, 
the rate pairs satisfying the following inequalities are achievable:
\begin{equation}\begin{split}
  \label{eq:inner}
  R_X      &\geq S(X|B),                             \\
  R_B      &\geq\frac{1}{2}\left(S(B)+S(B|X)\right), \\
  R_X+2R_B &\geq S(B)+S(XB). 
\end{split}\end{equation}
\end{theorem}

\begin{proof}
From  the Devetak-Winter code, Eq.~(\ref{eq:DW}), and 
the code based on state merging, Theorem \ref{State_merging_rate},  two rate
points in the unassisted (and hence also in the unlimited entanglement-assisted)
rate region are:
\begin{align*}
(R_X,R_B) &= (S(X|B),S(B)), \\
  (R_X,R_B) &= \left(S(X),\frac12(S(B)+S(B|X))\right).
\end{align*}
Their upper-right convex closure is hence an inner bound to the rate region,
depicted schematically in Fig.~\ref{fig:inner}.
\end{proof}

\begin{figure}[!t]
\centering
  \includegraphics[width=0.8\textwidth]{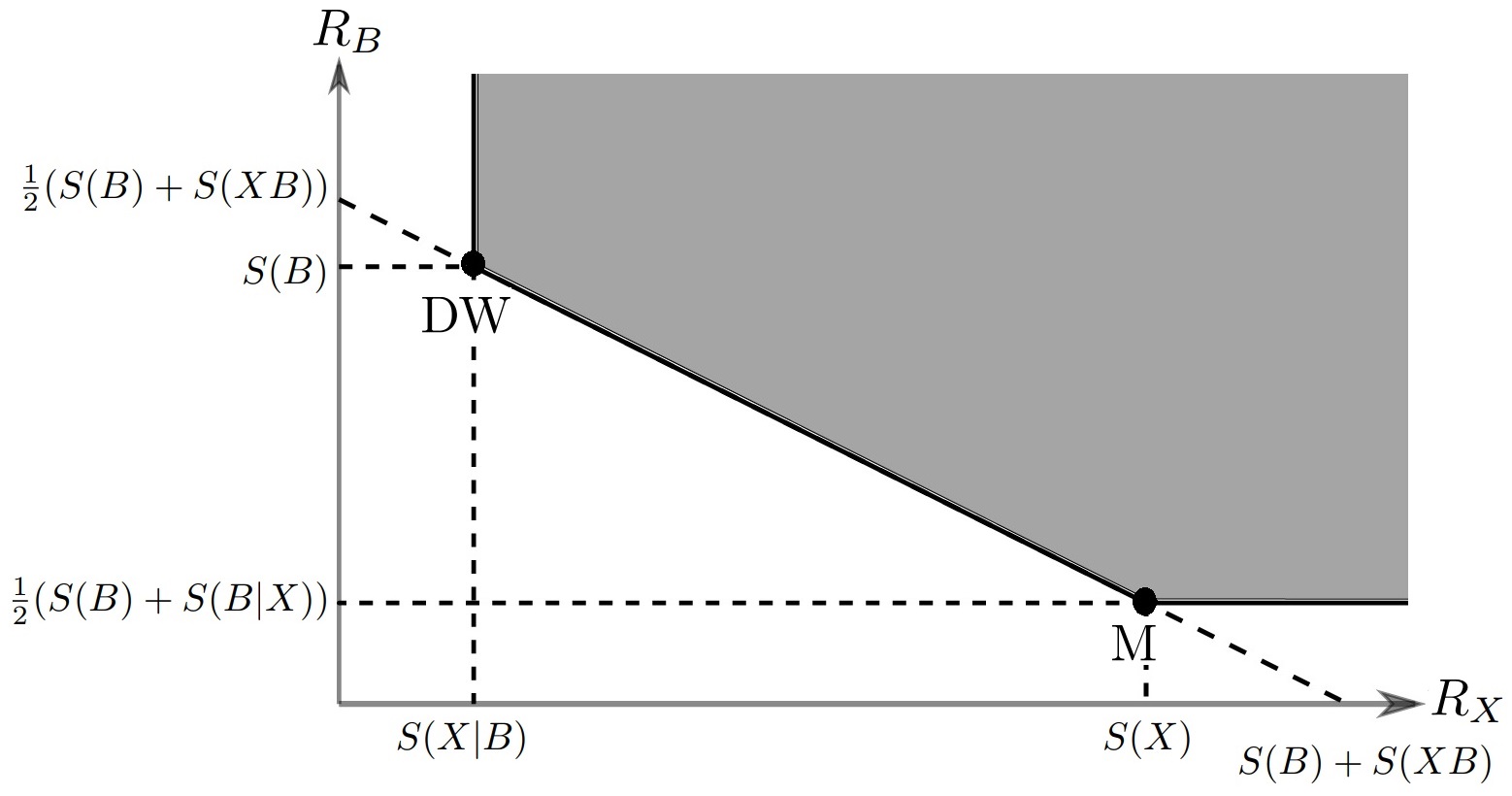}
  \caption{\aw{The region of all pairs $(R_X,R_B)$ satisfying the three conditions of 
           Eq.~(\ref{eq:inner}); it is the upper-right convex closure of the
           Devetak-Winter (DW) and the merging (M) point.
           All of these points are achievable in the unassisted model}.}
  \label{fig:inner}
\end{figure}

\aw{For generic sources we find \aw{that this is in fact the} rate region.
However, in general, we only present some outer bounds and inner bounds (achievable rates), 
which show the rate region to be much more complicated than the rate region of \aw{the}
classical Slepian-Wolf problem.}

\aw{\subsection{General converse bounds}
\label{sec:Converse Bounds in General}
For distributed compression of a \aw{classical-quantum} source 
in general, we start with a general converse bound.

\begin{theorem}
\label{theorem:full converse}
In the entanglement-assisted model, the asymptotic rate pairs for distributed compression of a 
\aw{classical-quantum} source are lower bounded as 
\begin{equation}\begin{split}
  \label{eq:general-converse}
  R_X        &\geq S(X|B),                                                \\
  R_B        &\geq \frac{1}{2}\left( S(B)+S(B|X)-\widetilde{I}_0 \right), \\ 
  R_X + 2R_B &\geq  S(B)+S(BX)-\widetilde{I}_0.\\
\end{split}\end{equation}
In the unassisted model, in addition to the above lower bounds, 
the asymptotic rate pairs are bounded as
\[ 
  R_X + R_B \geq S(XB). 
\]
\end{theorem}

\begin{proof}
The individual lower bounds have been established already: 
$R_X\geq S(X|B)$ is from \cite{Devetak2003,Winter1999}, in a slightly different source model. 
However, it also holds in our system model if Bob sends his information using 
unlimited communication such that \aw{Debbie} can decode it perfectly. Namely, notice 
that the fidelity (\ref{F_QCSW_unassisted1}) is more stringent than the decoding 
criterion of \cite{Devetak2003,Winter1999}, so any converse bound considering the 
decoding criterion of \cite{Devetak2003,Winter1999} is also a converse bound in our system model. 
The bound $R_B\geq \frac{1}{2}(S(B)+S(B|X)-\widetilde{I}_0)$ is from
Theorem \ref{theorem:generic optimal rate}. These two bounds hold in the unassisted,
as well as the entanglement-assisted model.

In the unassisted model, the rate sum lower bound $R_X + R_B \geq S(XB)$ has been
argued in \cite{Devetak2003,Winter1999}, too. As a matter of fact, for any distributed compression
scheme for the source, $\cE_X\otimes\cE_B$ jointly describes a Schumacher compression scheme 
with asymptotically high fidelity. Thus, its rate must be asymptotically lower bounded
by the joint entropy of the source, $S(XB)$ \cite{Schumacher1995,Jozsa1994_1,Barnum1996,Winter1999}.

This leaves the bound $R_X + 2R_B \geq S(B)+S(BX)-\widetilde{I}_0$ to be proved in
the entanglement-assisted model, which we tackle now.
The encoders of Alice and Bob are isometries $U_X:{X^n \to C_X W_X}$ and 
$U_B:{B^nB_0 \to C_B W_B B_0'}$, respectively. They send their respective compressed systems 
$C_X$ and $C_W$ to \aw{Debbie} and keep the environment parts $W_X$ and $W_B$ for themselves. 
Then, \aw{Debbie} applies the decoding isometry $V:{C_X C_B D_0 \to \hat{X}^n\hat{B}^n W_D D_0'}$,
where $\hat{X}^n\hat{B}^n D_0'$ are the output systems, and $W_D$ and $D_0'$ are the 
environment of \aw{Debbie}'s decoding isometry and her output entanglement, respectively.
We first bound the following sum rate:
\begin{align}
  \label{eq sum rate}
  &nR_X+nR_B+S(D_0) \nonumber \\
  &\geq S(C_X)+S(C_B)+S(D_0) \nonumber\\
    &\geq S(C_XC_BD_0)                        \nonumber\\
    &=    S(\hat{X}^n\hat{B}^n W_D D_0')      \nonumber\\
    &=    S(\hat{X}^n\hat{B}^n) + S(W_D D_0'|\hat{X}^n\hat{B}^n)     \nonumber\\
    &\geq S(\hat{X}^n\hat{B}^n) + S(W_D D_0'|\hat{X}^n\hat{B}^nX'^n) \nonumber\\
    &\geq\! \!S(\!X^n \!B^n\!) \!\!+\! S(\!W_D \!D_0'|\hat{X}^n\!\hat{B}^n\!X'^n\!) \!\! \nonumber \\
    & \quad \quad\quad\quad\quad\quad\quad\quad-\! n\!\sqrt{\!2\epsilon} \!\log(\!|X| |B|\!) \!\!-\! \!h(\!\sqrt{\!2\epsilon}\!) \nonumber\\
    &\geq S(X^n B^n) + S(W_D D_0'|X'^n) - 2n\delta(n,\epsilon)                                    \nonumber\\
    &\geq \!S(\!X^n B^n\!) \!\!+ \!\!S(\!W_XW_B B_0'|X'^n\!) \!\! \nonumber \\
    &\quad \quad\quad\quad\quad\quad\quad\quad- \!\!S(R^n\hat{B}^n\hat{X}^n|X'^n) \!\!-\!\! 2n\delta(\!n,\epsilon\!) \nonumber\\
    &\geq  \!S(X^n B^n) \!\!+ \!\!S(W_XW_B B_0'|X'^n) \!\!-\!\! 2n\delta(n,\epsilon)\!\!-\!\!n\delta'(n,\epsilon)                                 \nonumber\\
    &= \!S(\!X^n\! B^n\!)\! \!+ \!\!S(\!W_X|X'^n\!) \!\!+\! \!S(\!W_B B_0'|X'^n\!) \!\! \nonumber \\
    &\quad \quad\quad\quad\quad\quad\quad\quad- \!\!2n\delta(\!n,\epsilon\!) \!\!-\!\! n\delta'(\!n,\epsilon\!)                     \nonumber\\
    &\geq S(X^n B^n) + S(W_B B_0'|X'^n) - 2n\delta(n,\epsilon)-n\delta'(n,\epsilon),
\end{align}
where the third line is by subadditivity, the equality in the third line follows because 
the decoding isometry $V$ does not change the entropy. Then, in the fifth and sixth line 
we use the chain rule and strong subadditivity of entropy.
The inequality in the seventh line follows from the decodability of the systems $X^nB^n$:
the fidelity criterion (\ref{F_QCSW_assisted}) implies that the output state on systems 
$\hat{X}^n\hat{B}^n$ is $2\sqrt{2\epsilon}$-close to the original state $X^nB^n$ in trace norm;
then apply \aw{the} Fannes inequality (\aw{Lemma} \ref{Fannes-Audenaert inequality}). 
The eighth  line follows from the decoupling condition (Lemma \ref{decoupling condition}),
which implies that 
$I(W_D D_0':\hat{X}^n\hat{B}^n|{X'}^n) \leq n\delta(n,\epsilon) 
                                       =    4n\sqrt{6\epsilon} \log(|X| |B|) + 2 h(\sqrt{6\epsilon})$.
In the ninth line, we use that for any given $x^n$, the overall state of
$W_XW_B W_D B_0' D_0'R^n\hat{B}^n\hat{X}^n$ is pure, and invoking subadditivity.
In line tenth, we use the decoding fidelity (\ref{F_QCSW_assisted}) once
more, saying that the output state on systems 
$\hat{X}^n\hat{B}^nR^n{X'}^n$ is $2\sqrt{2\epsilon}$-close to the original 
state $X^nB^nR^n{X'}^n$ in trace norm;
then apply the Alicki-Fannes inequality (Lemma~\ref{AFW lemma}) in the following equation; notice that given $x^n$ the state on systems $X^nB^nR^n$ is pure, therefore  $S(X^nB^nR^n|{X'}^n)=0$, and we obtain:
\begin{align}\label{eq: delta'}
   & \abs{S(\hat{X}^n\hat{B}^nR^n|{X'}^n)-S(X^nB^nR^n|{X'}^n)}  \nonumber \\
    &\quad= S(\hat{X}^n\hat{B}^nR^n|{X'}^n) \nonumber\\
    &\quad\leq 2n \sqrt{2\epsilon} \log |X| |B| |R| + (1+\sqrt{2\epsilon})h(\frac{\sqrt{2\epsilon}}{1+\sqrt{2\epsilon}}) \nonumber\\
    &\quad\leq 4n \sqrt{2\epsilon} \log |X| |B|   + (1+\sqrt{2\epsilon})h(\frac{\sqrt{2\epsilon}}{1+\sqrt{2\epsilon}}) \nonumber\\
    &\quad:= \delta'(n,\epsilon),
\end{align}
where in the penultimate line, we can without loss of generality assume $|R| \leq |X| |B|$.
The equality in the twelfth line of Eq.~(\ref{eq sum rate}) follows because for 
a given $x^n$ the encoded states of Alice and Bob are independent.  

Moreover, we bound $R_B$ as follows:
\begin{align}
  \label{eq RB lower bound}
  nR_B &\geq S(C_B)                       \nonumber\\
       &\geq S(C_B|W_BB_0')               \nonumber\\
       &=    S(C_BW_BB_0') - S(W_BB_0')   \nonumber\\
       &=    S(B^nB_0) - S(W_BB_0')       \nonumber\\
       &=    S(B^n) + S(B_0) - S(W_BB_0'). 
\end{align}
Adding \aw{Eqs.}~(\ref{eq sum rate}) and (\ref{eq RB lower bound}), and after
cancellation of $S(B_0)=S(D_0)$, we get
\begin{align}
  \label{eq: lower bound R_X+2R_B}
  R_X& \! \!+2R_B \nonumber \\
  &\geq \!S(B)  \! \!+ \! \! S(X B) \! \! -  \!\! \frac{1}{n}I(X'^n  \!\!: \!W_B B_0')  \! \!-  \!2\delta( \!n,\epsilon \!)  \!\!-  \!\!\delta'( \!n,\epsilon \!)                              \nonumber\\
           &\geq  \!S(B)  \!+ \! S(X B)  \!- \! \frac{1}{n}I_{n\delta( \!n,\epsilon \!)}({\omega^{\otimes n}})   \!\!-  \!2\delta( \!n,\epsilon \!)   \!\!- \! \!\delta'( \!n,\epsilon \!) \nonumber\\
           &=   \!  S(B)  \!+  \!S(X B)  \!- \! I_{\delta(n,\epsilon)}({\omega})  \!- \! 2\delta(n,\epsilon)-\delta'(n,\epsilon),
\end{align}
where given that $I(R^n:B_0'W_B|X'^n) \leq \delta(n,\epsilon)$, which we have from the 
decoupling condition (Lemma \ref{decoupling condition}), the second equality follows directly 
from Definition \ref{I_delta}, just as in the proof of Theorem \ref{converse_QCSW}. 
The equality in the last line follows from Lemma \ref{lemma:I-delta}. 
In the limit of $n\rightarrow\infty$ and $\epsilon\rightarrow 0$, we have
$\delta(n,\epsilon) \rightarrow 0$ and $\delta'(n,\epsilon) \rightarrow 0$, and so $I_{\delta(n,\epsilon)}$ \aw{converges} 
to $\widetilde{I}_0$.
\end{proof}}

\subsection{General achievability bounds}
\label{sec:Achievability Bounds in General}
For general, non-generic sources, the achievability bounds of Theorem \ref{unknown.theorem}
and the outer bounds of Theorem \ref{theorem:full converse} do not match. Here we present several more
general achievability results that go somewhat towards filling in the unknown
area in between, without, however, resolving the question completely.

\begin{theorem}
\label{thm:achieve}
In the entanglement-assisted model, for distributed compression of a \aw{classical-quantum} source, 
\aw{any rate pairs satisfying} the following inequalities are achievable: \aw{with $\alpha=\frac{2I(X:B)}{I(X:B)+I_0}$},
\begin{equation}\begin{split}
  \label{eq:funny-region}
  R_X             &\geq S(X|B),                                    \\
  R_B             &\geq \frac{1}{2}\left( S(B)+S(B|X)-I_0 \right), \\ 
  R_X +\alpha R_B &\geq S(X|B)+ \alpha S(B).
\end{split}\end{equation}

\aw{More generally,} for any auxiliary random variable $Y$ such that \aw{$Y$--$X$--$B$}
is a Markov chain, \aw{all the following rate pairs (and hence also their upper-right convex closure)} 
are achievable:
\begin{align*}
  R_X &= I(X:Y) + S(X|BY)                    = S(X|B) + I(Y:B), \\
  R_B &= \frac{1}{2}(S(B) + S(B|Y) - I(Y:W))\\ 
  &= S(B)-\frac{1}{2}\left(I(Y:B)+I(Y:W)\right), 
\end{align*}
where $C$ and $W$ are the system and environment of an \aw{isometry $V:{B\rightarrow CW}$ 
with $I(W:R|Y)=0$}. 
\end{theorem}

\begin{proof}
\aw{The region described by Eq.~(\ref{eq:funny-region}) is precisely the upper-right convex
closure of the two corner points $(S(X|B),S(B))$ and $(S(X),\frac{1}{2}(S(B)+S(B|X)-I_0))$. Their
achievability} follows from \aw{Theorems} \ref{theorem: generic full rate region} and 
\ref{QSR_achievability}. 

We use the following two achievable points to show the 
\aw{second} statement: 
\begin{align*}
  \left(S(X|B),S(B)\right) \quad \text{and} \quad \left(S(X),\frac{1}{2}(S(B)+S(B|X)-I_0)\right).  
\end{align*}
Namely, Alice and \aw{Debbie} (the receiver) use the Reverse Shannon Theorem to simulate the channel
taking $X$ to $Y$ in i.i.d.~fashion, which costs $I(X:Y)$ bits of classical communication \cite{Bennett2002}. 
Now we are in a situation that we know, Bob has to encode $B^n$ with side information $Y^n$ at the decoder, 
which can be done \aw{at the rate $\frac{1}{2}(S(B)+S(B|Y)-I(Y:W))$, by the quantum state redistribution
protocol of Theorem \ref{QSR_achievability}}. 
Then Alice has to send some more information to allow the receiver to decode $X^n$ which is an instance of 
classical compression of $X$ with quantum side information $BY$ that is already at the decoder, 
hence costing another $S(X|BY)$ bits in communication, \aw{by the Devetak-Winter protocol \cite{Devetak2003,Winter1999}}. 
For $Y=X$, \aw{we recover the rate point $\left(S(X),\frac{1}{2}(S(B)+S(B|X)-I_0)\right)$, 
and for $Y=\emptyset$ we recover $\left(S(X|B),S(B)\right)$.}
\end{proof}

\medskip
In Fig.~\ref{fig:full}, we show the situation for a general source, depicting
the most important inner and outer bounds on the rate region in the entanglement-assisted
model.

\begin{figure}[!t]
\centering
  \includegraphics[width=0.8\textwidth]{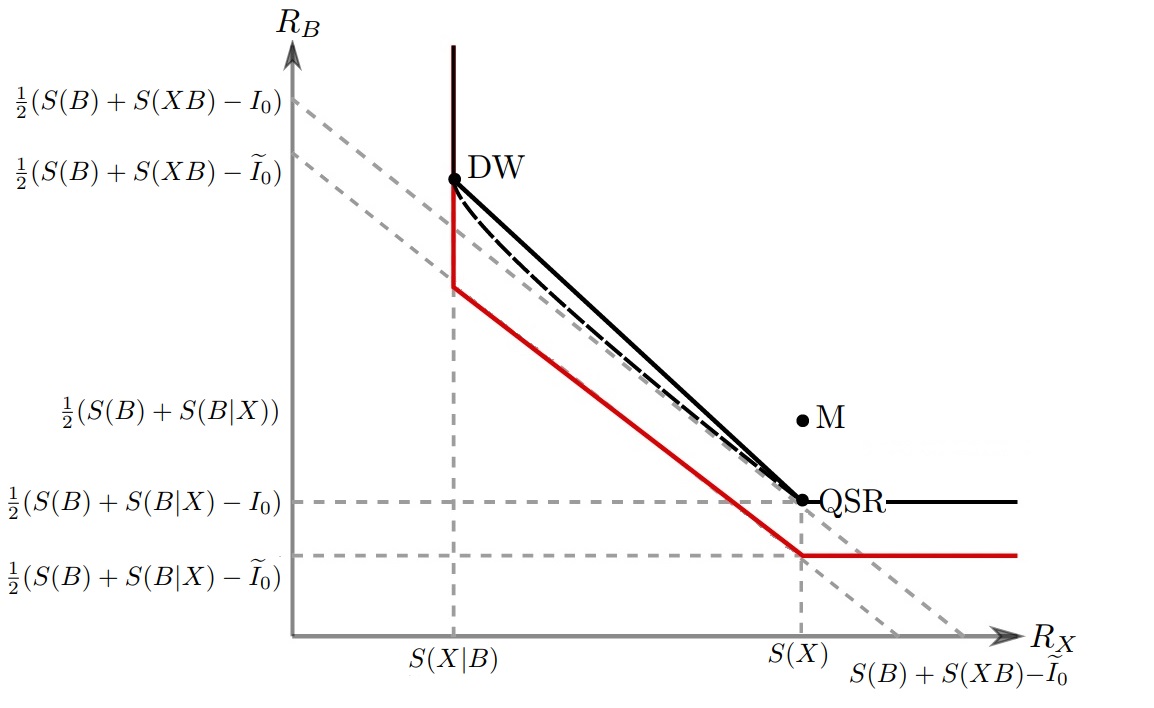}
  \caption{\aw{General outer (converse) bound, in red, and inner (achievable) bounds, in black, 
           on the entanglement-assisted rate region, assuming unlimited entanglement. 
           In general, our achievable points, the one from Devetak-Winter (DW), and the ones
           using merging (M) and quantum state redistribution (QSR) are no longer on the
           boundary of the outer bound. The achievable region is potentially slightly larger than
           the upper-right convex closure of the points DW and QSR, connected by a solid black
           straight line; indeed, the second part of Theorem \ref{thm:achieve} allows us to
           interpolate between DW and QSR along the black dashed curve}.}
  \label{fig:full}
\end{figure}


\subsection{Rate region for generic sources}
\label{sec:generic full}
In this subsection, we find the complete rate region for generic sources, generalizing
the insight of Theorem \ref{theorem:generic optimal rate} for the subproblem
of quantum compression with classical side information at the decoder.

\begin{theorem}
\label{theorem: generic full rate region}
\aw{In both unassisted and entanglement-assisted models, for a generic classical-quantum source, in particular one where there is
an $x$ such that $\psi_x^B$ has full support}, the optimal asymptotic rate region for distributed 
compression is \aw{the} set of rate pairs satisfying 
\begin{align*}
  R_X      &\geq S(X|B), \\
  R_B      &\geq\frac{1}{2}\left(S(B)+S(B|X)\right),\\
  R_X+2R_B &\geq S(B)+S(XB). 
\end{align*}
Moreover, \aw{there are protocols achieving these bounds requiring no prior} entanglement. 
\end{theorem}

\aw{\begin{proof}
We have argued the achievability already at the start of this section
(Theorem \ref{unknown.theorem}).
As for the converse, we have shown in Theorem \ref{theorem:generic optimal rate}
that for a generic source, $\widetilde{I}_0=0$, hence the claim follows from
the outer bounds of Theorem \ref{theorem:full converse}.
\end{proof}}

This means that for generic sources, which we recall are the complement of a set of measure
zero, the rate region has the shape of Fig.~\ref{fig:inner}.
\section{Discussion and open problems} 
\label{sec:discuss}
\aw{After seeing no progress for over 15 years in the problem of distributed
compression of quantum sources, we have decided to take a fresh look at the
classical-quantum sources considered in \cite{Devetak2003,Winter1999}. There,
the problem of compressing the classical source using the quantum part as
side information at the decoder was solved; here we analyzed the full
rate region, in particular we were interested in the other extreme of compressing 
the quantum source using the classical part as side information at the decoder. 
Like in the classical Slepian-Wolf coding, the former problem exhibits no rate loss, 
in that the quantum part of the source is compressed to the Schumacher rate, 
the local entropy, and the sum rate equals the joint entropy of the source.
Interestingly, this is not the case for the latter problem: clearly, if the
classical side information were available both at the encoder and the decoder,
the optimal compression rate would be the conditional entropy $S(B|X)$, which
would again imply no sum rate loss. However, since the classical side information
is supposed to be present only at the decoder, we have shown that in general 
the rate sum is strictly larger, in fact generically by $\frac12 I(X:B)$, and
with this additional rate there is always a coding scheme achieving asymptotically
high fidelity. This additional rate could be called ``the price of ignorance'', 
as it corresponds to the absence of the side information at the encoder.

To deal with general classical-quantum sources, we introduced information quantities 
$I_0$ and $\widetilde{I}_0$ (Definition \ref{I_delta}), to upper and lower bound the 
optimal quantum compression rate as
\begin{align*}
\frac12  \! \!\left( \! S(B) \!+ \!S(B|X) \!\!-\! \!\widetilde{I}_0\  \! \! \! \right )  \! \leq  \! R_B^*  \! \leq  \! \frac12  \! \left(S(B) \!+ \! S(B|X)\! \!- \!  {I}_0 \!\right)\!, 
\end{align*}
when unlimited entanglement is available.
For generic sources, $I_0 = \widetilde{I}_0 = 0$, but in general we do not
understand these quantities very well, and the first set of open problems that 
we would like to mention
is about them: is $I_0 = \widetilde{I}_0$ in general, or are there examples of
gaps? How can one calculate either one of these quantities, given that
a priori the auxiliary register $W$ is unbounded? In fact, can one
without loss of generality put a finite bound on the dimension of $W$,
for either optimization problem?

Further open problems concern the need for prior shared entanglement to achieve
the optimal quantum compression rate $R_B^*$. As a matter of fact, it would already
be interesting to know whether the rate $\frac12\left(S(B)+S(B|X)-{I}_0\right)$ 
requires in general pre-shared entanglement.

The full rate region inherits these features: while it is simple, and
in fact generated by the optimal codes for the two compression-with-side-information 
problems (quantum compression with classical side information, and classical
compression with quantum side information), in the generic case, in general the
picture is very complicated, and we have only been able to give several outer and
inner bounds on the rate region, whose determination remains an open problem.}

\medskip
We also would like to comment on the source model that we consider in this chapter, 
and its relation to the classical Slepian-Wolf coding. 
Our classical-quantum source is characterised by a classical source, 
the random variable $X$, and a quantum source $B$, which is described by a
density matrix $\rho_x^B$, but realized as quantum correlation with a purifying 
reference system $R$: $\rho_x^B = \tr_R \ketbra{\psi_x}^{BR}$. 
A source code in our sense reproduces the states $\ket{\psi_x}^{BR}$ with high
fidelity on average, which implies that, for any ensemble decomposition
$\rho_x^B = \sum_y p(y|x) \ketbra{\psi_{xy}}^B$, it reproduces the 
states $\ket{\psi_{xy}}^B$ with high fidelity on average (with respect to
the ensemble probabilities $p(x)p(y|x)$). If we only demand the latter, there
is no need for the purifying system $R$, and the source can be described compactly by 
the cccq-state
\begin{align}\label{eq: ensemble cqSW source}
&\sigma^{X'XYB}= \nonumber \\
&\> \!\!\!\! \!\!\sum_{x \in \mathcal{X},y \in \mathcal{Y}} \!\! \!\!\!\!p(x)p(y|x) \!\!\ketbra{x}^{X'} \!\!\!\!\otimes\!\ketbra{x}^X \!\!\!\otimes\! \ketbra{y}^Y \!\!\!\otimes  \!\ketbra{\psi_{xy}}^{\!B}\!\!\!,
\end{align}
where $X'$ and $Y$ are reference systems with which the correlation is preserved in a compression protocol. This now includes the well-known classical correlated source considered 
by Slepian and Wolf \cite{Slepian1973}, namely if the system $B$ is classical 
with orthonormal states $\ket{\psi_{xy}}=\ket{y}$. 
In the Schumacher's single compression problem \cite{Schumacher1995}, both 
source models, that is, the ensemble source and the purified source, lead to the same 
compression rate. However, when there is side information or more generally in the
distributed setting, different source models, albeit sharing the reduced states on $XB$, 
do not lead to the same compression rate \cite{Winter1999}. 
Our results provide a clear manifestation of this: recall that the minimum
compression rate of Bob in the Slepian-Wolf setting is $S(B|X)$, with the 
ensemble fidelity criterion. On the other hand, if the distributions $p(y|x)$ 
have pairwise overlapping support, or theorem regarding generic sources
applies, resulting in the strictly larger minimum rate $\frac12(S(B)+S(B|X))$
when the average entanglement fidelity criterion is used. The difference can be 
attributed to the harder task of maintaining the entanglement with the reference
system, rather than ``only'' classical correlation.

More broadly, a quantum source can be defined as a quantum system together with 
correlations with a reference system, in our case any state $\rho^{ABR}$. The 
compression task is to reproduce this state with high fidelity by coding and decoding 
of $A$ and $B$. While this problem is far from understood in the general case,
what we saw here is that the compression rate may depend on the concrete 
correlation with the reference system. In the present chapter, we have considered 
both a globally purifying quantum system and an ensemble of purifications, and in this
final discussion, implicitly looked at a classical system keeping track of an
ensemble of states subject to a probability distribution.

Finally, we mention that both models of quantum data compression with classical side information with partially purified source of  Eq.~(\ref{eq: extended source model}) and the ensemble model defined in Eq.~(\ref{eq: ensemble cqSW source}) are special cases of the model that we consider in the next chapter.
There we define an ensemble extension of the QSR source, namely the ensemble $\{ p(x), \ketbra{\psi_x}^{ACBR}\}$ with corresponding cqqqq-state $\sum_x p(x)\ketbra{x}^{X} \otimes \ketbra{\psi_x}^{ACBR}$ where Alice who has access to side information system $C$ wants to compress system $A$ and send it, via a noiseless quantum channel, to Bob who has access to side information system $B$. We let the encoder and decoder share free entanglement and consider two decodability critera: per-copy fidelity and block fidelity where in the former the fidelity is preserved for each copy of the source while in the latter the fidelity is preserved for the whole block of $n$ systems similar to the fidelity defined in Eq.~(\ref{F_QCSW_unassisted1}).
For the former criterion we find the optimal quantum communication  rate and for the latter criterion we find a converse bound and an achievable rate which match up to an asymptotic error and an unbounded auxiliary
system. 
Our new results  imply that in the compression of system $B$ with classical side information at the decoder $X$ in the source model of Eq.~(\ref{eq: extended source model}), the converse bound of Theorem~\ref{converse_QCSW}, i.e. the following rate is optimal in the entanglement-assisted model with per-copy fidelity: 
\begin{align*}
  R_B =\frac{1}{2}\left( S(B)+S(B|X)-\widetilde{I}_0 \right).
\end{align*}

\chapter{Quantum state redistribution for ensemble sources}
\label{chap:QSR ensemble}

In this chapter, we consider a generalization of the quantum state redistribution task, 
where pure multipartite states from an ensemble source are distributed among 
an encoder, a decoder and a reference system. The encoder, Alice, has access 
to two quantum systems: system $A$ which she compresses and sends to the decoder, 
Bob, and the side information system $C$ which she wants to keep at her site. 
Bob has access to quantum side information in a system $B$, wants to decode 
the compressed information in such a way to preserve the correlations with the 
reference system on average. 
%

As figures of merit, we consider both block error (which is the usual one 
in source coding) and per-copy error (which is more akin to rate-distortion theory), 
and find the optimal compression rate for the second criterion, 
and achievable and converse bounds for the first. The latter almost match in general, 
up to an asymptotic error and an unbounded auxiliary system; for so-called irreducible 
sources they are provably the same.
This chapter is based on the publications in \cite{ZK_QSR_ensemble_ISIT_2020,QSR_ensemble_full}.

\section{The source model}
Quantum state redistribution (QSR) is a source compression task where both 
encoder and decoder have access to side information systems \cite{Devetak2008_2,Yard2009,Oppenheim2008}. 
Namely, Alice, Bob and a reference system share asymptotically many copies of a 
pure state $\ket{\psi}^{ACBR}$, where Alice aims to compress the quantum system $A$ 
and send it to Bob via a noiseless quantum channel, while she has access to a side 
information quantum system $C$, and Bob has access to the side information quantum system 
$B$. Bob upon receiving the compressed information reconstructs system $A$, and the 
figure of merit of this task is to preserve the entanglement fidelity between the 
reconstructed systems and the purifying reference system $R$.

Quantum state redistribution generalizes Schumacher's compression, which 
is recovered as the extreme case that neither encoder nor decoder have any 
side information \cite{Schumacher1995}: the source is simply described by a pure state
$\ket{\psi}^{AR}$ shared between the encoder and a reference system.
However, besides this model, and originally, Schumacher considered a source 
generating an ensemble of pure states, i.e. $\cE = \{ p(x), \proj{\psi_x}^{A} \}$, 
and showed both source models lead to the same optimal compression rate 
(cf. Barnum \emph{et al.} \cite{Barnum1996}, as well as \cite{Winter1999}),
namely the von Neumann entropy of the reduced or average state of $A$, respectively. 

In the presence of side information systems though, an ensemble model and a purified 
source model can lead to different compression rates. An example of this is 
the classical-quantum Slepian-Wolf problem considered in \cite{ZK_cqSW_2018,ZK_cqSW_ISIT_2019}, 
where the compression rate can be strictly smaller than that of the corresponding 
purified source.

The general correlated ensemble source $\cE = \{ p(x), \proj{\psi_x}^{AB} \}$ was 
considered first in \cite{Winter1999} and then developed in \cite{Devetak2003}
and by Ahn \emph{et al.} \cite{Ahn2006}, 
with $A$ the system to be compressed and $B$ the side information system at the decoder. 
It is an ensemble version of the coherent state merging task introduced 
in \cite{family,Abeyesinghe2009}. 
In \cite{Devetak2003}, the source is $\ket{\psi_x}^{AB} = \ket{f(x)}^A\ket{\phi_x}^B$.
The optimal compression rate for an irreducible source of product states and a source 
generating Bell states is found in \cite{Ahn2006}, however, in general case the problem had 
been left open.

In the present chapter, we consider an even more general ensemble source where both 
encoder and decoder have access to side information systems, and which thus
constitutes an ensemble generalization of the pure QSR source. More precisely,
we consider a source which is given by an ensemble 
$\cE = \{ p(x), \proj{\psi_x}^{ACBR} \}$ of pure states 
$\psi_x=\proj{\psi_x}\in\cS(A \otimes C \otimes B \otimes R)$,
$\ket{\psi_x}\in A \otimes C \otimes B \otimes R$, with a Hilbert space 
$A \otimes C \otimes B \otimes R$, which in this chapter we assume to be of
finite dimension $|A|\cdot|C|\cdot|B|\cdot|R|<\infty$; 
$\cS(A \otimes C \otimes B \otimes R)$ denotes the set of states (density operators).
Furthermore, $x\in\cX$ ranges over a discrete alphabet, so we can describe 
the source equivalently by the classical-quantum (cq) state 
$\omega^{ACBRX} = \sum_x p(x) \proj{\psi_x}^{ACBR} \otimes \proj{x}^X$.
In this model,  $A$ and $C$ are Alice's information to be sent and side 
information system, respectively. System $B$ is the side information of Bob, 
and $R$ and $X$ are inaccessible reference systems used only to define the task.

The ensemble model of the previous chapter
as well as those models that have been considered in \cite{Winter1999,Ahn2006,ZK_Eassisted_ISIT_2019,ZK_cqSW_2018,ZK_cqSW_ISIT_2019} 
are all special cases of the model that we consider here. We find the optimal 
compression rate under the per-copy fidelity criterion, and achievable 
and converse rates under the block-fidelity criterion which almost match, up to an 
asymptotic error and an unbounded auxiliary system. In the generic case
of so-called \emph{irreducible} ensembles, they are provably the same.

\section{The compression task}
\label{sec:The Compression task QSR ensemble}
We consider the information theoretic setting of many copies of the 
source $\omega^{ACBRX}$, i.e.~$\omega^{A^nC^nB^nR^nX^n}=(\omega^{ACBRX})^{\otimes n}$:
\[
  \omega^{A^nC^nB^nR^nX^n}
    \!\!\! = \!\!\!\!
             \sum_{x^n \in \mathcal{X}^n}\!\!\!\! p(x^n)  \proj{\psi_{x^n}}^{A^nC^nB^nR^n}
                       \!\otimes\! \proj{x^n}^{X^n}\!\!\!\!\!,
\]
using the notation
\begin{alignat*}{3}
  x^n              &= x_1 x_2 \ldots x_n,                  
  &p(x^n)           &= p(x_1) p(x_2) \cdots p(x_n), \\
  \ket{x^n}        &= \ket{x_1} \ket{x_2} \cdots \ket{x_n}, \  
  &\ket{\psi_{x^n}} &= \ket{\psi_{x_1}} \ket{\psi_{x_2}} \cdots \ket{\psi_{x_n}}.
\end{alignat*}

We assume that the encoder, Alice, and the decoder, Bob, have initially a maximally 
entangled state $\Phi_K^{A_0B_0}$ on registers $A_0$ and $B_0$ (both of dimension $K$).
Alice, who has access to $A^n$ and the side information system $C^n$, performs the 
encoding compression operation $\cE:A^n C^n A_0 \longrightarrow M \hat{C}^n$ 
on $A^nC^n$ and her part $A_0$ of the entanglement, which is a quantum channel,
i.e.~a completely positive and trace preserving (CPTP) map. 
Notice that as functions, CPTP maps act on the operators (density matrices) over  
the respective input and output Hilbert spaces, but as there is no risk of confusion,
we will simply write the Hilbert spaces when denoting a CPTP map.
Alice's encoding operation produces the state $\sigma^{M  \hat{C}^n B^n B_0 R^n X^n}$ 
with $M$, $\hat{C}^n$ and $B_0$ as the compressed system of Alice, the reconstructed 
side information system of Alice and Bob's part of the entanglement, respectively.
The dimension of the compressed system is without loss of 
generality not larger than the dimension of the
original source, i.e. $|M| \leq  \abs{A}^n$. 
The system $M$ is then sent via a noiseless quantum channel to Bob, who performs
a decoding operation $\mathcal{D}:M B^n B_0 \longrightarrow \hat{A}^n \hat{B}^n$ 
on the compressed system $M$, his side information $B^n$ and his part of the entanglement 
$B_0$, to reconstruct the original systems, now denoted $\hat{A}^n$ and $\hat{B}^n$.
We call 
$\frac1n \log|M|$ the \emph{quantum rate} of the compression protocol.
We say an encoding-decoding scheme (or code, for short) has \emph{block fidelity} 
$1-\epsilon$, or \emph{block error} $\epsilon$, if 
\begin{align}
  \label{eq:block fidelity criterion}
  F &:= F(\omega^{A^nC^nB^nR^nX^n},\xi^{\hat{A}^n\hat{C}^n\hat{B}^nR^nX^n}) \nonumber \\
     &= \sum_{x^n} p(x^n) F\left( \psi_{x^n}^{A^nC^nB^nR^n} \!\!,
                  \xi_{x^n}^{\hat{A}^n \hat{C}^n \hat{B}^n R^n}  \right) \geq 1-\epsilon,
\end{align}
where 
\begin{align*}
\xi^{\hat{A}^n\hat{C}^n\hat{B}^nR^nX^n}
   &=\sum_{x^n \in \mathcal{X}^n}\!\!\!\! p(x^n)  \xi_{x^n}^{\hat{A}^n \hat{C}^n \hat{B}^n R^n} 
                       \!\otimes\! \proj{x^n}^{X^n}\\
   &=\left((\mathcal{D}\circ\cE)\otimes \id_{R^nX^n}\right) \omega^{A^nC^nB^n R^nX^n }.
\end{align*}
We say a code has \emph{per-copy fidelity} $1-\epsilon$, 
or \emph{per-copy error} $\epsilon$, if 
\begin{align}
  \label{eq:per-copy-error-RD}
  \overline{F} &:= \frac{1}{n}\sum_{i=1}^n F(\omega^{A_iC_iB_iR_iX^n},\xi^{\hat{A}_i\hat{C}_i\hat{B}_iR_iX^n}) \nonumber\\
     &= \sum_{x^n} p(x^n) \frac{1}{n}\sum_{i=1}^n 
                            F\left( \psi_{x_i}^{ACBR} \!\!, 
                                    \xi_{x^n}^{\hat{A}_i \hat{C}_i \hat{B}_i R_i} \right) 
     \geq 1-\epsilon. 
\end{align}
By the monotonicity of the fidelity under the partial trace
(over $X_{[n]\setminus i}$), this implies the easier to verify 
condition
\begin{align}
  \label{eq:per copy fidelity criterion}
  \widetilde{F} 
    :=   \frac{1}{n}\sum_{i=1}^n F\left(\omega^{ACBRX},\xi^{\hat{A}_i\hat{C}_i\hat{B}_iR_iX_i}\right) 
    \geq 1-\epsilon,
\end{align}
where 
$\xi^{\hat{A}_i\hat{C}_i\hat{B}_iR_iX_i}=\Tr_{[n]\setminus i}\,\xi^{\hat{A}^n\hat{C}^n\hat{B}^nR^nX^n}$, 
and `$\Tr_{[n]\setminus i}$' denotes the partial trace over all systems 
with indices in $[n]\setminus i$. 

Conversely, Eq. \eqref{eq:per copy fidelity criterion} can be shown to imply
the criterion \eqref{eq:per-copy-error-RD} with $(1-\epsilon)^2 \geq 1-2\epsilon$
on the right hand side. Indeed, note that 
\begin{align*}
&F\left(\omega^{ACBRX},\xi^{\hat{A}_i\hat{C}_i\hat{B}_iR_iX_i}\right) \\
& \quad    = \sum_{x_i} p(x_i) F\left(\psi_{x_i}^{ACBR},
                               \sum_{x_{[n]\setminus i}} p(x_{[n]\setminus i}) 
                                                         \xi_{x^n}^{\hat{A}_i \hat{C}_i \hat{B}_i R_i} \right).
\end{align*}
Thus, by the convexity of the square function and Jensen's inequality,
\[\begin{split}
  (1-\epsilon)^2 
   &\leq \left( \frac{1}{n}\sum_{i=1}^n F(\omega^{ACBRX},\xi^{\hat{A}_i\hat{C}_i\hat{B}_iR_iX_i}) \right)^2 \\
   &\leq \frac{1}{n}\sum_{i=1}^n \sum_{x^n} p(x^n) 
                       F\left(\psi_{x_i}^{ACBR},\xi_{x^n}^{\hat{A}_i \hat{C}_i \hat{B}_i R_i} \right),
\end{split}\]
and the last line is the left hand side of Eq. \eqref{eq:per-copy-error-RD}.

Correspondingly, we say $Q_b$ and $Q_c$ are an asymptotically achievable block-error rate 
and an asymptotically achievable per-copy-error rate, respectively, 
if for all $n$ there exist codes such that the block fidelity and per-copy fidelity converge 
to $1$, and the quantum rate converges to $Q_b$ and $Q_c$, respectively. Because 
of the above demonstrated relations
$\widetilde{F}^2 \leq \overline{F} \leq \widetilde{F}$ it
doesn't matter which of the two version of per-copy fidelity we take. 


According to Stinespring's theorem \cite{Stinespring1955}, the encoding and decoding 
CPTP maps $\cE$ and $\cD$ can be dilated respectively to the isometries 
$U_{\cE}: A^n C^n A_0\hookrightarrow M\hat{C}^n W_n$  and 
$U_{\cD}: M B^n B_0 \hookrightarrow \hat{A}^n \hat{B}^nV_n$,
with $W_n$ and $V_n$ as the environment systems of the encoder and decoder, respectively.

\vspace{-0.2cm}

\section{Main Results}
\label{sec: main results of qsr ensemble}
In Theorem~\ref{thm:main theorem} we obtain the main results of this chapter concerning 
optimal (minimum) block-error rate $Q_b^*$ and optimal per-copy-error rate $Q_c^*$. 
These rates are expressed in terms of the following single-letter function.


\begin{definition}\label{def:Q_epsilon}
  For a state 
  $\omega^{ACBRX}=\sum_x p(x) \proj{\psi_x}^{ACBR}\otimes \proj{x}^{X}$ and $\epsilon \geq 0$ define:
  \begin{align*}
    Q(\epsilon) :=&  
        \inf \frac{1}{2} I(Z:RXX'|B)_{\sigma} 
                  \text{ over CPTP maps } \\
                  &\cE_{\epsilon}: AC \rightarrow Z \hat{C}  \text{ and } \cD_{\epsilon}:ZB \rightarrow \hat{A}\hat{B} \text{ s.t.} \\
    & F( \omega^{ACBRX},\xi^{\hat{A} \hat{C} \hat{B} RX})  \geq 1- \epsilon, 
\end{align*}
where
\begin{align*}
  \sigma^{Z\hat{C}BRX}\!
     &:=\! (\cE_{\epsilon}\otimes \id_{BRX})\omega^{ACBRX}  
       \!=\! \sum_x p(x) \sigma_x^{Z\hat{C}BR} \!\otimes \!\proj{x}^{X}\!\!, \\
  \xi^{\hat{A}\hat{C}\hat{B}RX}
     \!&:=\! (\cD_{\epsilon} \otimes \id_{\hat{C}RX}) \sigma^{Z\hat{C}BRX}  
       \!=\! \sum_x p(x) \xi_x^{\hat{A}\hat{C}\hat{B}R}\!\otimes\! \proj{x}^{X}\!\!.
\end{align*}
Moreover, define $\widetilde{Q}(0):=\lim_{\epsilon \to 0+} Q(\epsilon)$.
\end{definition}

The function $Q(\epsilon)$ is defined for the specific source $\omega^{ACBRX}$; this dependency is dropped to simplify the notation.

\begin{theorem}
\label{thm:main theorem}
The minimum asymptotically achievable rate with per-copy error is
\begin{align*}
      Q_c^*=\widetilde{Q}(0).
\end{align*}
Instead, the minimum asymptotically achievable rate with block error is bounded 
from above and below as follows:
\begin{align*}
  \widetilde{Q}(0)
     \leq Q_b^* \leq Q(0).
\end{align*}
\end{theorem}
\vspace{-0.2cm}

\begin{proof}
We prove the achievability here and leave the converse proof to the next section. 


Let $U_0: AC\hookrightarrow Z \hat{C} W$ and $\widetilde{U}_0: ZB\hookrightarrow \hat{A} \hat{B} V$ be respectively the isometric extension of the CPTP maps $\cE_0$ and $\cD_0$ in 
Definition~\ref{def:Q_epsilon} with fidelity $1$ (i.e. $\epsilon = 0$). 
To achieve the block-error rate $Q_b = Q(0)$, 
Alice applies the isometry $U_0$, 
after which the purified state shared between the parties is
\begin{align*}
\ket{\sigma_0}^{Z\hat{C}WBRXX'}
     =\sum_x \sqrt{p(x)} \ket{\sigma_0(x)}^{Z\hat{C}WBR}\otimes \ket{x}^{X}\otimes \ket{x}^{X'}.
\end{align*}
%
Then the parties apply the QSR protocol to many copies of the above source where Alice 
sends system $M$ to Bob and systems $\hat{C}$ and $W$ are her side information. 
The rate achieved by the QSR protocol is
\begin{align*}
Q_b=\frac{1}{2}I(Z:RXX'|B)_{\sigma_0}.
\end{align*}

After executing the QSR protocol, Bob has $Z^n$, 
and the state shared between the parties is 
$\hat{\sigma}_0^{Z^n\hat{C}^nW^nB^nR^nX^n{X'}^n}$, 
which satisfies the following entanglement fidelity:
\begin{align}\label{eq: fidelity 3}
  F\left( (\sigma_0^{Z\hat{C}WBRXX'})^{\otimes n},
          \hat{\sigma}_0^{Z^n\hat{C}^nW^nB^nR^nX^n{X'}^n} \right) \to 1,
\end{align}
as $n\to\infty$.
Then, Bob applies to each system the CPTP map $\cD_0:ZB \longrightarrow \hat{A}\hat{B}$. 
Due to the monotonicity of the fidelity under CPTP maps, we obtain 
from Eq.~(\ref{eq: fidelity 3})
\begin{align}\label{eq: fidelity 4}
  \!  \!  \!  \! F\left(\!\! (\cD_0^{\otimes n}\!\otimes \!\id)(\sigma_0^{Z\hat{C}BRX})^{\otimes n}\!\!,
         (\cD_0^{\otimes n}\!\otimes\! \id)\hat{\sigma}_0^{Z^n\hat{C}^nB^n R^n X^n } \!\! \right) \!\!\to\!\! 1
\end{align}
as $n \to \infty$,
where the identity channel $\id$ acts on systems ${\hat{C}^nR^nX^n}$.
Notice that by the definition of $\cD_0$,
\begin{align*}
  (\omega^{ACBRX})^{\otimes n}
    =(\cD_0^{\otimes n}\otimes \id_{\hat{C}^nR^nX^n})(\sigma_0^{Z\hat{C}BRX})^{\otimes n}.
\end{align*}
Thus, the block fidelity criterion of Eq.~(\ref{eq:block fidelity criterion}) holds.


Now, let $U_{\epsilon}: AC\hookrightarrow Z \hat{C} W$ and $\widetilde{U}_{\epsilon}: ZB\hookrightarrow \hat{A} \hat{B} V$ be respectively the isometric extension of the CPTP maps $\cE_{\epsilon}$ and $\cD_{\epsilon}$ in 
Definition~\ref{def:Q_epsilon} with fidelity $1-\epsilon$.
To achieve the per-copy-error rate $Q_c^*$,
to each copy of the source Alice 
applies the isometry $U_{\epsilon}$. 
Then the purified state shared between the parties is
\begin{align*}
\ket{\sigma_{\epsilon}}^{Z\hat{C}WBRXX'}
     =\sum_x \sqrt{p(x)} \ket{\sigma_{\epsilon}(x)}^{Z\hat{C}WBR}\otimes \ket{x}^{X}\otimes \ket{x}^{X'}.
\end{align*}
The parties apply the QSR protocol to many copies of the above source where Alice 
sends system $Z$ to Bob and systems $\hat{C}$ and $W$ are her side information. 
The rate achieved by the QSR protocol is
\begin{align}
 Q_c &=\frac{1}{2}I(M:RXX'|B)_{\sigma_{\epsilon}}. \nonumber
\end{align}
After executing the QSR protocol, Bob has $Z^n$, and the state shared between 
the parties is $\hat{\sigma}_{\epsilon}^{Z^n\hat{C}^nW^nB^nR^nX^n{X'}^n}$, 
which satisfies the following 
entanglement fidelity:
\begin{align*}
  F\left( (\sigma_{\epsilon}^{Z\hat{C}WBRXX'})^{\otimes n},
          \hat{\sigma}_{\epsilon}^{Z^n\hat{C}^nW^nB^nR^nX^n{X'}^n} \right) \to 1
\end{align*}
as $n \to \infty$.
Due to monotonicity of the fidelity under partial trace, we obtain the per-copy fidelity, 
\begin{equation}\label{eq: fidelity 1}
  F(\sigma_{\epsilon}^{Z\hat{C}BRX},\hat{\sigma}_{\epsilon}^{Z_i\hat{C}_i B_i R_i X_i }) \to 1,
\end{equation}
for all $i \in [n]$ and $n \to \infty$.
Then, to each system $i$, Bob applies the CPTP map 
$\cD_{\epsilon}$.
We obtain 
\begin{align}\label{eq: fidelity 2}
 F \! \left(\! ( \! \cD_{\epsilon} \! \otimes \!  \id_{\hat{C}RX} \! )\sigma_{\epsilon}^{Z\hat{C}BRX}\!,\!
     ( \! \cD_{\epsilon}\otimes \id_{\hat{C}RX} \! )\hat{\sigma}_{\epsilon}^{Z_i\hat{C}_i B_i  \! R_i \!  X_i }  \! \!\right) \!\to \! 1
\end{align}
for all $i \in [n]$ and $n \to \infty$, 
which follows from Eq.~(\ref{eq: fidelity 1}) due to monotonicity of the fidelity 
under CPTP maps. 
On the other hand, the state 
$\xi_{\epsilon}^{\hat{A}\hat{C}\hat{B} RX}
  =(\cD_{\epsilon} \otimes \id_{\hat{C}RX})\sigma_{\epsilon}^{Z\hat{C}BRX}$ 
has high fidelity with the original source state, directly from the definition of $\cD_{\epsilon}$:
\begin{align*}
  F(\xi_{\epsilon}^{\hat{A}\hat{C}\hat{B} RX},\omega^{ACBRX}) \to 1.
\end{align*}
Therefore, from the above fidelity and Eq.~(\ref{eq: fidelity 2}) we obtain
\begin{align*}
  F\left( \omega^{ACBRX},
          (\cD_{\epsilon}\otimes \id_{\hat{C}RX})\hat{\sigma}_{\epsilon}^{M_i\hat{C}_i B_i R_i X_i } \right) 
     \to 1
\end{align*}
for all $i \in [n]$ and $n \to \infty$,
which satisfies the per-copy fidelity criterion in Eq.~(\ref{eq:per copy fidelity criterion}).
%
%
\end{proof}

Now, we define a new single-letter function 
%
which then we use to obtain simplified rates in Lemma~\ref{lemma: lower bound on Q_tilde(0)}   
and Corollary~\ref{cor:irreducible} which both are proved in \cite{QSR_ensemble_full}.
\begin{definition}
  \label{def:K_epsilon}
  For a state 
  $\omega^{ACBRX}=\sum_x p(x) \proj{\psi_x}^{ACBR}\otimes \proj{x}^{X}$ and $\epsilon \geq 0$ define:
  \begin{align*}
    K_\epsilon(\omega) &:= \sup I(W:X|\hat{C})_{\sigma} 
                                \text{ over isometries } \\
                       &\phantom{=====}
                        U: AC \rightarrow Z \hat{C} W \text{ and } 
                        \widetilde{U}:ZB \rightarrow \hat{A}\hat{B}V \text{ s.t.} \\
                       &\phantom{=====}
                        F( \omega^{ACBRX},\xi^{\hat{A} \hat{C} \hat{B} RX})  \geq 1- \epsilon, 
\end{align*}
where 
\begin{align*}
  \sigma^{Z\hat{C}WBRX}
     &:= (U\otimes \1_{BRX})\omega^{ACBRX} (U\otimes \1_{BRX})^{\dagger} \\
      & = \sum_x p(x) \proj{\sigma_x}^{Z\hat{C}WBR}\otimes \proj{x}^{X}, \\
  \xi^{\hat{A}\hat{C}\hat{B}WVRX}
     &:= (\widetilde{U}\otimes \1_{\hat{C}WRX}) \sigma^{Z\hat{C}WBRX}  
         (\widetilde{U}\otimes \1_{\hat{C}WRX})^{\dagger}\\
       &= \sum_x p(x) \proj{\xi_x}^{\hat{A}\hat{C}\hat{B}WVR}\otimes \proj{x}^{X}, \\
  \xi^{\hat{A}\hat{C}\hat{B}RX}
     &:= \Tr_{VW} \xi^{\hat{A}\hat{C}\hat{B}WVRX}.
\end{align*}
Moreover, define $\widetilde{K}_0:=\lim_{\epsilon \to 0+} K_{\epsilon}(\omega)$.
\end{definition}

\begin{remark}
Definition~\ref{def:K_epsilon} directly implies that $K_{0}(\omega) \leq \widetilde{K}_{0}(\omega)$ 
because $K_{\epsilon}(\omega)$ is a non-decreasing function of $\epsilon$. 
Furthermore, $K_{0}(\omega)$ can be strictly positive, for example, for a source 
with trivial system $C$ where $\psi_{x}^{A}\psi_{x'}^{A}=0$ holds for $x\neq x'$, 
we obtain $K_{0}(\omega)=S(X)$.
This follows because Alice can measure her system and obtain the value of $X$ and 
then copy this classical information to the register $W$. 
\end{remark}

\begin{lemma}\label{lemma: lower bound on Q_tilde(0)}
The rate $\widetilde{Q}(0)$ is lower bounded as:
\begin{align*}
    \widetilde{Q}(0) &\!\geq\! \frac{1}{2} \left(S(A|B)+S(A|C) \right) \!-\!\frac{1}{2}\widetilde{K}_{0} \\
   &\!=\!\frac{1}{2}I(A:RXX'|B)_{\omega}\!-\!\frac{1}{2}\widetilde{K}_{0},
\end{align*}
where the above conditional mutual information is precisely the communication rate 
of QSR for the purified source
\begin{align}\label{eq: purified source}
 \ket{\omega}^{ACBRXX'}
   =\sum_x \sqrt{p(x)} \ket{\psi_x}^{ACBR} \otimes \ket{x}^X \otimes \ket{x}^{X'}.
   \end{align}
Moreover, if system $C$ is trivial, then $\widetilde{Q}(0)$ is equal to this lower bound. 
\end{lemma}


%
%

\begin{definition}[{Barnum~\emph{et~al.}~\cite{Barnum2001_2}}]
  \label{def:reducibility  QSR ensemble}
  An ensemble  
  $\cE=\{p(x),\proj{\psi_x}^{ACBR} \}_{x\in \mathcal{X}}$  of pure states
  is called \emph{reducible} if its states fall into two or more orthogonal subspaces.
  Otherwise the ensemble $\cE$ is called \emph{irreducible}.
  We apply the same terminology to the source state $\omega^{ACBRX}$.
\end{definition}

\begin{corollary}
\label{cor:irreducible}
For an irreducible source $\omega^{ACBRX}$, $K_0=\widetilde{K}_0=0$. Hence, the 
optimal asymptotically achievable per-copy-error rate and block-error rate 
are equal and
\begin{align*}
  Q^*_c=Q^*_b=\frac{1}{2}\left(S(A|C)+S(A|B) \right).
\end{align*}
\end{corollary}


\section{Converse}
In this section, we first show some properties of the function $Q(\epsilon)$, which then we use to prove the converse for Theorem~\ref{thm:main theorem}.

\begin{lemma}
  \label{lemma:Q-convex}
For $0 \leq \epsilon \leq 1$, $Q(\epsilon)$ is a monotonically non-increasing, convex function of $\epsilon$. Consequently, for $0<\epsilon <1$ it is also continuous. 
\end{lemma}
\vspace{-0.2cm}
\begin{proof}
The monotonicity directly follows from the definition. 
For the convexity, we verify Jensen's inequality, that is we start with maps 
$\cE_1,\cD_1$ eligible for error $\epsilon_1$ with the output state $\xi_1^{\hat{A}\hat{C}\hat{B}RX}$, and $\cE_2,\cD_2$ eligible for error $\epsilon_2$ with the output state $\xi_2^{\hat{A}\hat{C}\hat{B}RX}$,
and $0\leq p \leq 1$. By embedding into larger Hilbert spaces if necessary, we
can w.l.o.g. assume that the maps act on the same systems for $i=1,2$.
We define the following two maps:
\vspace{-0.2cm}
\begin{align*}
  \cE(\rho) &:= p \cE_1(\rho) \otimes \proj{1}^{Z'} + (1-p) \cE_2(\rho) \otimes \proj{2}^{Z'}, \\
  \cD(\rho) &:= \cD_1(\bra{1}^{Z'} \rho \ket{1}^{Z'}) + \cD_2(\bra{2}^{Z'} \rho \ket{2}^{Z'}).
\end{align*}
They evidently realise the output state $\xi^{\hat{A}\hat{C}\hat{B}RX} = p\xi_1^{\hat{A}\hat{C}\hat{B}RX} + (1-p)\xi_2^{\hat{A}\hat{C}\hat{B}RX}$ with the following fidelity:
\begin{align}
    F&(\omega^{ACBRX} ,\xi^{\hat{A} \hat{C} \hat{B} RX} ) \nonumber\\
      &= F(\omega^{ACBRX}  ,p \xi_1^{\hat{A} \hat{C} \hat{B} RX}
        + (1-p) \xi_2^{\hat{A} \hat{C} \hat{B} RX}) \nonumber \\
      &\geq p F(\omega^{ACBRX} \!,\!\xi_1^{\hat{A} \hat{C} \hat{B} RX} \!)
            \!+\!(1\!-\! p)F(\omega^{ACBRX} \!,\xi_2^{\hat{A} \hat{C} \hat{B} RX} \!) \nonumber\\
     &\geq 1-\left( p\epsilon_1 +(1-p)\epsilon_2 \right), \nonumber
\end{align}
where the third line is due to simultaneous concavity of the fidelity in both arguments.
The last line follows by the definitions of the states $\xi_1$ and $\xi_2$.
Therefore, the maps $\cE$ and $\cD$ yield a fidelity of at least 
$1-\left( p\epsilon_1 +(1-p)\epsilon_2 \right) =: 1-\epsilon$. 
Thus, 
\begin{align*}
  Q(\epsilon) &\leq I(ZZ':RXX'|B)_\xi \\
        &=    p I(Z:RXX'|B)_{\xi_1} + (1-p) I(Z:R|B)_{\xi_2},
\end{align*}
and taking the infimum over maps $\cE_i,\cD_i$ shows convexity. 

The continuity statement follows from a mathematical folklore fact, stating that 
any real-valued function that is convex on an interval, is continuous on the 
interior of the interval.
\end{proof}


\begin{proof-of}[of Theorem \ref{thm:main theorem} (converse)]
We prove the converse for the per-copy fidelity criterion, therefore, the same converse bound holds for the block fidelity criterion as well. 
Consider a block length $n$ code per-copy fidelity $1-\epsilon$. The number 
of qubits, $\log|M|$, can be lower bounded as follows, with respect to 
the encoded state $(\cE\otimes \id_{B_0B^nR^nX^nX'^n})\omega^{A^nC^nB^nR^nX^nX'^n} \otimes \Phi_K^{A_0B_0}$ of the purified source:
\begin{align*}
  2\!\log|M| \!&\!\geq 2 S(M) \\
           &\geq I(M:R^nX^n X'^n|B^nB_0) \\
           &=  \!  I(\underbrace{MB_0}_{Z}:R^nX^nX'^n|B^n) \!\!-\!\! I(B_0:R^nX^nX'^n|B^n) \\
           &=    I(Z:R^nX^nX'^n|B^n)                           \\
           &=    \sum_{i=1}^n I(Z:R_iX_iX'_i|B^nR_{<i}X_{<i}X'_{<i}) \\
           &\quad \quad+ \sum_{i=1}^n I(R_{<i}X_{<i}X'_{<i}B_{[n]\setminus i}:R_iX_iX'_i|B_i) \\
           &=    \sum_{i=1}^n I(ZR_{<i}X_{<i}X'_{<i}B_{[n]\setminus i}:R_iX_iX'_i|B_i) \\
           &\geq \sum_{i=1}^n I(\underbrace{ZB_{[n]\setminus i}}_{Z_i}:R_iX_iX'_i|B_i), 
           \vspace{-0.3cm}
\end{align*} 
where in the first two inequalities we use standard entropy inequalities;
the equation in the third line is due to the chain rule, and the second 
conditional information is $0$ because $B_0$ is independent of $B^nR^n$;
the fourth line introduces a new register $Z$, noting that 
the encoding together with the entangled state defines a CPTP map 
$\cE_0:A^n \rightarrow Z\hat{C}^n$, via $\cE_0(\rho) = (\cE \otimes \id_{B_0})(\rho \otimes\Phi_K^{A_0B_0})$;
in the fifth we use the chain rule iteratively, and in the second term we introduce, 
each summand is $0$ because for all $i$, $R_{<i}B_{[n]\setminus i}$ is independent of $R_iB_i$;
in the sixth line we use again the chain rule for all $i$, and the
last line is due to data processing.
%
%

Now, for the $i$-th copy of the source $\omega^{A_iC_iB_iR_iX_i}$, we define maps 
$\cE_i:A_i C_i \rightarrow Z_i \hat{C}_i$ and $\cD_i:B_iZ_i \rightarrow \hat{A}_i \hat{B}_i$,
as follows:
\begin{itemize}
\item[$\cE_i$:] Alice tensors her system $A_i$ with a dummy state $\omega^{\otimes [n]\setminus i}$ and
      with $\Phi_K^{A_0B_0}$ (note that all systems are in her possession). 
      Then she applies $\cE:A^nC^nA_0 \rightarrow M \hat{C}^n$, and sends 
      $Z_i := M B_0 B_{[n]\setminus i}$ to Bob, while keeping $\hat{A}_i \hat{C}_i$.
      All other systems, i.e. $ \hat{A}_{[n]\setminus i} \hat{C}_{[n]\setminus i} R_{[n]\setminus i}X_{[n]\setminus i}$, are trashed.
\item[$\cD_i$:] Bob applies $\cD$ to $Z_iB_i = M B_0 B^n $ and keeps 
      $\hat{A}_i \hat{B}_i$, trashing the rest $\hat{A}_{[n]\setminus i}\hat{B}_{[n]\setminus i}$.
\end{itemize}
By definition, the output state 
\begin{align*}
\zeta^{\hat{A}_i\hat{C}_i\hat{B}_iR_iX_i} 
\!\!= \!(\cD_i \otimes \id_{\hat{A}_i\hat{C}_iR_iX_i}\!)\!\circ\!(\cE_i \otimes \id_{B_iR_iX_i}\!)\omega^{A_iC_iB_iR_iX_i}
\end{align*}
equals $\xi^{\hat{A}_i\hat{C}_i\hat{B}_iR_iX_i}$ which has fidelity $1-\epsilon_i$ with the source $\omega^{ACBRX}$, and the fidelity for all copies satisfy $  \frac{1}{n} \sum_i (1-\epsilon_i) \geq 1-\epsilon$. 
Thus, we obtain, with respect to the states $(\cE_i \otimes \id_{B_iR_iX_iX'_i})\omega^{A_iC_iB_iR_iX_iX'_i}$
\[\begin{split}
  \frac{1}{n}\log|M| &\geq \frac{1}{n} \sum_{i=1}^n \frac12 I(Z_i:R_iX_iX'_i|B_i)\\ 
                    & \geq \frac{1}{n} \sum_{i=1}^n Q(\epsilon_i)               
                     \geq Q\left( \frac{1}{n} \sum_{i=1}^n \epsilon_i\right)  
                     \geq Q(\epsilon), 
\end{split}\]
continuing from before, then by definition of $Q(\epsilon_i)$ since 
the pair $(\cE_i,\cD_i)$ results in fidelity $1-\epsilon_i$, in the next inequality 
by convexity and finally by monotonicity of $Q(\epsilon)$ (Lemma \ref{lemma:Q-convex}).
By the taking the limit of $\epsilon \to 0$ and $n \to \infty$, the claim follows.
\end{proof-of}

\vspace{-0.2cm}

\section{Discussion}
\label{sec: discussion qsr ensemble}
We considered a variant of the quantum state redistribution task, where pure 
multipartite states from an ensemble are distributed between an encoder, 
a decoder and a reference system. 
We distinguish two figures of merit for the information processing, per-copy 
fidelity and block fidelity, and define the corresponding quantum communication 
rates depending on the fidelity criterion, when unlimited entanglement is available.
For the per-copy fidelity criterion, we find that the optimal qubit rate of compression is equal to $\widetilde{Q}(0)$ from Definition~\ref{def:Q_epsilon}, 
which is bounded from below by the rate of the conventional QSR task minus the limit of the 
single-letter non-negative function $\widetilde{K}_0$ from Definition~\ref{def:K_epsilon}:
\begin{align*}
  \widetilde{Q}(\!0\!)\!\geq\! \frac{1}{2}\! \left( \!S(A|B)\!+\!S(A|C) \right) \!-\! \frac{1}{2}\! \widetilde{K}_0
        \!=\! \frac{1}{2} I(A:RXX'|B)_{\omega} \!-\! \!\frac{1}{2} \widetilde{K}_0,
\end{align*}
where the conditional mutual information is the rate of QSR
for the purified source in Eq.~(\ref{eq: purified source}). This lower bound is tight if system $C$ is trivial (state merging scenario).


For the block fidelity criterion, we have found converse and achievability bounds:
\vspace{-0.2cm}
\begin{align*}
 \widetilde{Q}(0)\leq Q_b \leq Q(0).
\end{align*}
The two bounds would match if we knew that the function $Q(\epsilon)$ were 
continuous at $\epsilon=0$. However, we don not know this; for one thing, one 
cannot use compactness to show continuity because the output system $W$ in 
Definition~\ref{def:K_epsilon} is as priori unbounded. 

For irreducible sources though, we show here $K_0=\widetilde{K}_0=0$, which implies 
that the purified source model  and the ensemble  model lead to the same compression 
rate. For reducible sources the information that the encoder can obtain about the 
classical variable of the ensemble, i.e. system $X$, is effectively used as 
side information to achieve a smaller compression rate. Thus we reproduce the 
result of \cite[Thm.~III.3]{Ahn2006}, which was proven only for irreducible 
product state ensembles. 

There are other sources for which we know $K_0=\widetilde{K}_0=0$ to hold. First, 
the ``generic'' sources  in \cite[Thm.~11]{ZK_cqSW_2018}, where it is shown that the 
function $\widetilde{I}_0=0$; this function is a special case of the function $\widetilde{K}_0$.
Indeed, the source there is described by an ensemble $\{p(x), \ket{\psi_x}^{AR}\ket{x}^B\}$,
which is always completely reducible, but generically the reduced states $\psi_x^A$ have
pairwise overlapping support, which is the condition under which 
vanishing $\widetilde{K}_0$ is shown. 
Secondly, the ensemble of four Bell states considered in \cite[Thm.~IV.1]{Ahn2006},
$\{ p(ij), \ket{\Phi_{ij}}^{AB} \}_{i,j=0,1}$,
where the side information system $C$ and the reference system $R$ are trivial; 
for this source, the mutual information between Alice's system and the classical system 
$X$ is zero, i.e. $I(A:X)=0$. Thus, due to data processing inequality, we have
$I(W:X) \leq I(A:X)=0$. Our main result reproduces the achievable rate 
$\frac12 H(p)$, and also the optimality, by very different, and somewhat more 
natural methods. 

There are other special cases of the source model of this chapter that have been previously
studied in the literature for which $K_0 > 0$ or at least $\widetilde{K}_0 > 0$. 
For instance in the source of \cite{Devetak2003}, where Alice's system is classical 
with $A=X$ and system $C$ is trivial, one can observe that $K_0 = S(X)$ holds. 
The rate we get is $Q^* = \frac12 S(X|B)$ under either error criterion, half of the
quantity reported in \cite{Devetak2003} because of the free entanglement in our
model, which allows for dense coding.
Furthermore, the visible variant of Schumacher compression in \cite{Barnum1996,Winter1999}, 
where Alice's side information system is classical with $C=X$, the function has the value 
$K_0 = S(X)$, and the optimal rate is $Q^*=\frac12 S(A)$, again half of the optimal
rate without entanglement, because we can use remote state preparation and dense coding. 
A third example is the ensemble $\{\frac13,\ket{\psi_i}^A\ket{\phi_i}^B\}_{i=1}^3$ from
\cite[Sec.~V.A]{Ahn2006}, which is reducible, but where the reduced ensembles 
on systems $A$ and $B$ are both irreducible; it is shown there that the optimal 
compression rate is strictly smaller than $(S(A)+S(A|B))/2$.


Finally, recall that in our definition of the compression task we have assumed 
that the encoder and decoder share free entanglement. This was motivated so as 
to make a smoother connection to QSR.
However, it is not known whether the pre-shared entanglement is always necessary to achieve 
the corresponding quantum rates. There are certainly cases where QSR does not require prior 
entanglement, such as when Alice's side information $C$ is trivial, which would carry
over to our setting whenever $K_0=\widetilde{K}_0=0$, for instance for an irreducible
ensemble.
More generally, in future work we plan to consider the trade-off between the quantum 
 and  entanglement rates.

\part{Quantum Thermodynamics}
\chapter{Resource theory of charges and entropy}
\label{chap:resource theory}

In this chapter, we consider asymptotically many non-interacting systems with multiple conserved quantities or charges. 
We generalize the seminal results of 
Sparaciari, Oppenheim and Fritz [\emph{Phys. Rev. A} 96:052112, 2017]
to the case of multiple, in general non-commuting charges. 
To this aim we formulate a resource theory of thermodynamics of asymptotically 
many non-interacting systems with multiple conserved quantities or charges. 
To any quantum state, we associate a vector with entries of the expected charge 
values and entropy of that state. We call the set of all these vectors the 
phase diagram of the system, and show that it characterizes the equivalence 
classes of states under asymptotic unitary transformations that approximately 
conserve the charges. This chapter is based on the results from \cite{thermo_ZBK_2020}.

\section{Resource theory of charges and entropy}
\label{sec:resource-theory}
%
%

Resource theory is a rigorous mathematical framework initially developed to characterize the 
role of entanglement in quantum information processing tasks. Later the framework was extended 
to characterize coherence, non-locality, asymmetry and many more, including quantum Shannon theory itself, see \cite{bcp14,Winter-Dong-2016,c&g16,m&s16,V&S17,msz16,sap16,srb17,g&w19,cpv18,sha19,Vic14,d&a18,Devetak2008_1}. 
The resource theory approach applies also to classical theories. In general, the resource 
theories have the following common features: (1) a well-defined set of resource-free states, 
and any states that do not belong to this set has a non-vanishing amount of resource; 
(2) a well-defined set of resource-free operations, also known as allowed operations, 
that cannot create or increase resource in a state. These allow one to quantify the 
resources present in the states or operations and characterize their roles in the transformations 
between the states or the operations. In particular, it enables one to define and rigorously 
bound or even determine various resource measures; determine which states can be 
transformed to the others using allowed operation; how the property of states may be 
changed, and how these changes are bounded under the allowed operations, etc.

\medskip

A system in our resource theory is a quantum system $Q$ with a finite-dimensional Hilbert space 
(denoted $Q$, too, without danger of confusion), together with a
Hamiltonian $H=A_1$ and other quantities (``charges'') $A_2, \ldots, A_c$, all of which are 
Hermitian operators that do not necessarily commute with each other. We consider composition of 
$n$ non-interacting systems, where the Hilbert space of the \textit{composite} system $Q^n$ is 
the tensor product $Q^{\otimes n} = Q_1 \otimes \cdots \otimes Q_n$ of the Hilbert spaces of 
the \textit{individual} systems, and the $j$-th charge of the composite system is the sum of 
charges of individual systems as follows,
\begin{equation}
  A^{(n)}_j = \sum_{i=1}^{n} \1^{\otimes (i-1)} \otimes A_j \otimes \1^{\otimes (n-i)}, 
                \quad j=1,2,\ldots,c.
\end{equation}
For ease of notation, we will write throughout
$A_j^{[Q_i]} = \1^{\otimes (i-1)} \otimes A_j \otimes \1^{\otimes (n-i)}$.

We wish to build a resource theory where the objects are states on a quantum system, 
which are transformed under thermodynamically meaningful operations.
To any quantum state $\rho$ is assigned the point 
$(\und{a},s) = (a_1,\ldots,a_c,s) 
= \bigl( \Tr \rho A_1, \ldots, \Tr \rho A_c, S(\rho) \bigr) \in \mathbb{R}^{c+1}$,
which is an element in the \emph{phase diagram} that has been originally introduced,
for $c=1$, as energy-entropy diagram in \cite{Sparaciari2016}; there it is shown,
for a system where energy is the only conserved quantity, that the diagram is a convex set.
In the case of commuting multiple conserved quantities, the charge-entropy diagram has been 
generalised and further investigated in \cite{brl19}. 
Note that the set of all these vectors, denoted $\mathcal{P}^{(1)}$, is not in 
general convex (unless the quantities commute pairwise). 
An example is a qubit system with charges $\sigma_x$, $\sigma_y$ and $\sigma_z$ where 
charge values uniquely determine the state as a linear function of the $\tr \rho\sigma_i$, 
hence the entropy, while the von Neumann entropy itself is well-known to be strictly concave.

Moreover, the set of these points for a composite system with charges $A_1^{(n)}, \ldots, A_c^{(n)}$, 
which we denote $\mathcal{P}^{(n)}$ contains, but is not necessarily equal to $n\mathcal{P}^{(1)}$
(which however is true for commuting charges). Namely, consider the point 
$g=\left(\frac{1}{2}\Tr (\rho_1+\rho_2) A_1, \ldots, \frac{1}{2}\Tr (\rho_1+\rho_2) A_c,
\frac{1}{2}S(\rho_1)+\frac{1}{2}S(\rho_2)\right)$, which does not necessarily 
belong to $\mathcal{P}^{(1)}$ but belongs to its convex hull; 
however, $2g \in \mathcal{P}^{(2)}$ due to the state $\rho_1 \otimes \rho_2$.
Therefore, we consider the convex hull of the set $\mathcal{P}^{(1)}$ and call it the 
\emph{phase diagram} of the system, denoted
\begin{equation}
  \overline{\mathcal{P}} 
     \equiv \overline{\mathcal{P}}^{(1)} 
     := \left\{ \left(\sum_i p_i \Tr \rho_i A_1, \ldots, \sum_i p_i\Tr \rho_i A_c, \sum_i p_i S(\rho_i)\right) : 
                0 \leq p_i \leq 1,\,\sum_i p_i = 1 \right\}.
\end{equation} 
The interpretation is that the objects of our resource theory are ensembles of 
states $\{p_i,\rho_i\}$, rather than single states. 

We define the \emph{zero-entropy diagram} and \emph{max-entropy diagram}, 
respectively, as the sets
\begin{align*}
  \overline{\mathcal{P}}_0^{(1)} 
     &= \{(\und{a},0): \Tr \rho A_j = a_j \text{ for a state } \rho \}, \\
  \overline{\mathcal{P}}_{\max}^{(1)} 
     &= \left\{\bigl(\und{a},S(\tau(\und{a}))\bigr): \Tr \rho A_j = a_j \text{ for a state } \rho \right\}, 
\end{align*}
where $\tau(\und{a})$ is the unique state maximising the entropy among all states 
with charge values $\Tr \rho A_j = a_j$ for all $j$, which is called generalized 
thermal state, or generalized Gibbs state, or also generalized grand canonical state \cite{Liu2007}. 
Note that, as a linear image of the compact convex set of states, the zero-entropy diagram is 
compact and convex.
We similarly define the set $\mathcal{P}^{(n)}$, the phase diagram $\overline{\mathcal{P}}^{(n)}$, 
zero-entropy diagram $\overline{\mathcal{P}}_0^{(n)}$ and max-entropy diagram 
$\overline{\mathcal{P}}_{\max}^{(n)}$ 
for the composition of $n$ systems with charges $A_1^{(n)}, \ldots, A_c^{(n)}$. 

\begin{figure}[!t]
  \includegraphics[width=1\textwidth]{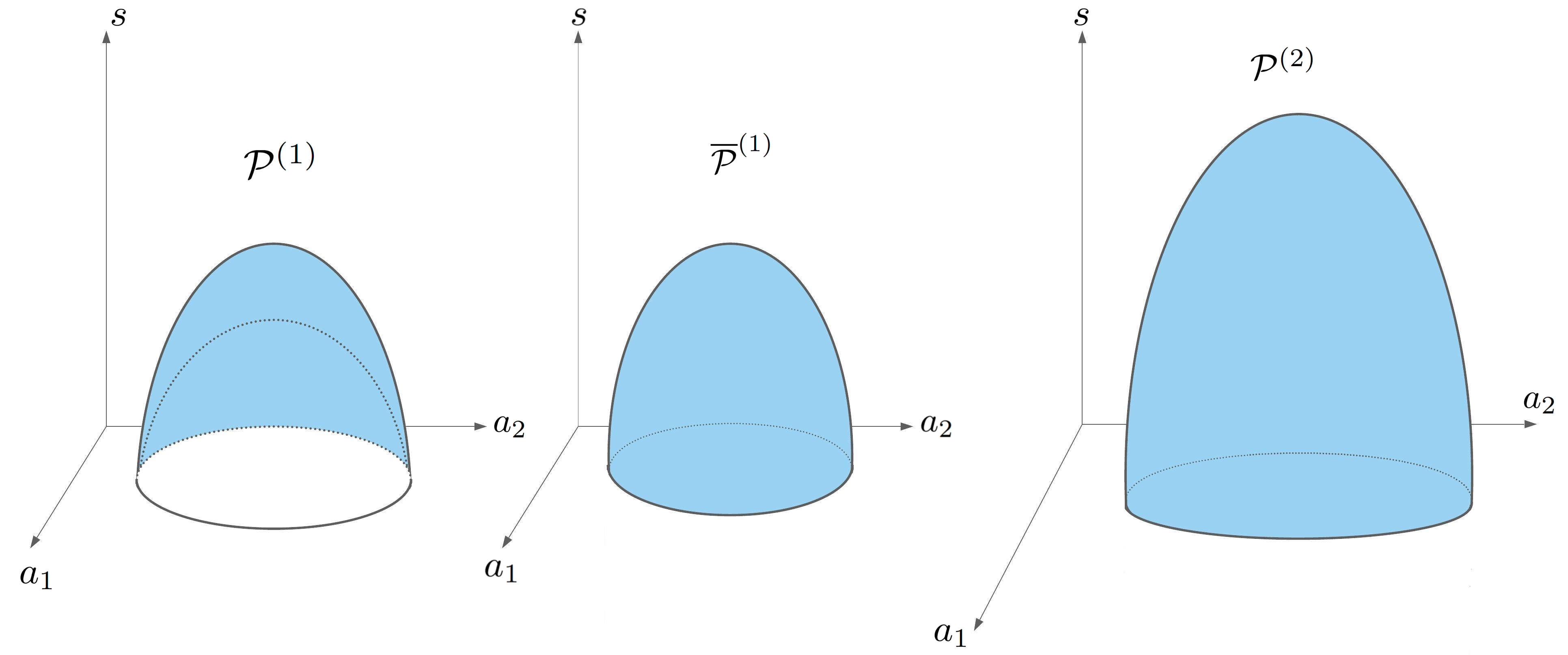}
  \caption{Schematic of the phase diagrams $\mathcal{P}^{(1)}$, $\mathcal{P}^{(2)}$
           and $\overline{\mathcal{P}}$. As seen, $\mathcal{P}^{(1)}$ is not convex, and there is a hole inside the diagram.}
  \label{fig:phase-diagrams}
\end{figure}

\begin{lemma}
\label{lemma:phase diagram properties}
For an individual and composite systems with charges $A_j$ and $A^{(n)}_j$, respectively, 
we have: 
\begin{enumerate}
  \item $\overline{\mathcal{P}}^{(n)}$, for $n \geq 1$, is a compact and convex 
        subset of $\mathbb{R}^{c+1}$.

  \item $\overline{\mathcal{P}}^{(n)}$, for $n \geq 1$, is the convex hull of the union 
        $\overline{\mathcal{P}}_{0}^{(n)} \cup \overline{\mathcal{P}}_{\max}^{(n)}$, 
        of the zero-entropy diagram and the max-entropy diagram.
     
    
  \item $\overline{\mathcal{P}}^{(n)} = n \overline{\mathcal{P}}^{(1)}$ for all $n \geq 1$. 
    
  \item $\mathcal{P}^{(n)}$ is convex for all $n \geq 2$, and indeed 
        $\mathcal{P}^{(n)} = \overline{\mathcal{P}}^{(n)} = n \overline{\mathcal{P}}^{(1)}$. 
    
    
    
  \item Every point of $\mathcal{P}^{(n)}$ is realised by a suitable tensor product state 
        $\rho_1 \otimes \cdots \otimes \rho_n$, for all $n \geq d$.

  \item All points $\bigl(\und{a},S(\tau(\und{a}))\bigr) \in \overline{\mathcal{P}}_{\max}$
        are extreme points of $\overline{\mathcal{P}}$.
\end{enumerate}
\end{lemma}

\begin{proof}
1. The phase diagram is convex by definition. Further, $\Tr \rho A_j^{(n)}$ and $S(\rho)$ 
are continuous functions defined on the set of quantum states which is a compact set; 
hence, the set $\mathcal{P}^{(n)}$ is also a compact set. The cxonvex hull of a 
finite-dimensional compact set is compact, so the phase diagram is a compact set.  

2. Any point in the phase diagram according to the definition is a convex combination of the form 
\[(a_1,\ldots,a_c,s)
= \left(\sum_i p_i \Tr(\rho_i A_1), \ldots, \sum_i p_i\Tr(\rho_i A_c),\sum_i p_i S(\rho_i)\right).\]
The point $(a_1,\ldots,a_c,0)$ belongs to $\overline{\mathcal{P}}_{0}^{(1)}$ 
because the state $\rho=\sum_i p_i \rho_i$ has charge values $a_1, \ldots, a_c$. 
Moreover, the state with charge values $a_1, \ldots, a_c$ of maximum entropy is 
the generalized thermal state $\tau(\und{a})$, so we have 
\begin{align*}
    S(\tau(\und{a})) \geq S(\rho) \geq \sum_i p_i S(\rho_i),
\end{align*}
where the second inequality is due to concavity of the entropy. Therefore, any point
$(\und{a},s)$ can be written as the convex combination of the points $(\und{a},0)$ and 
$(\und{a},S(\tau(\und{a}))$.

3. Due to item 2, it is enough to show that 
$\overline{\mathcal{P}}_{0}^{(n)} = n\overline{\mathcal{P}}_{0}^{(1)}$, and 
$\overline{\mathcal{P}}_{\max}^{(n)}= n\overline{\mathcal{P}}_{\max}^{(1)}$. 
The former follows from the definition. The latter is due to the fact that the 
thermal state for a composite system is the tensor power of the thermal state of 
the individual system. 

4. Let $\tau(\und{a}) = \sum_i p_i \ketbra{i}{i}$ be the diagonalization of the 
generalized thermal state.
For $n \geq 2$, define $\ket{v} = \sum_i \sqrt{p_i} \ket{i}^{\ox n}$. Obviously, the charge 
values of the states $\tau(\und{a})^{\ox n}$ and $\ketbra{v}{v}$ are the same, since they
have the same reduced states on the individual systems;
thus, there is a pure state for any point in the zero-entropy diagram of the composite system.
Now, consider the state $\lambda \ketbra{v}{v}+(1-\lambda)\tau(\und{a})^{\ox n}$, which has the 
same charge values as $\tau(\und{a})^{\ox n}$ and $\ketbra{v}{v}$. 
The entropy $S\bigl(\lambda \ketbra{v}{v}+(1-\lambda)\tau(\und{a})^{\ox n}\bigr)$ is a continuous function 
of $\lambda$; hence, for any value $s$ between $0$ and $S(\tau(\und{a})^{\ox n})$, there is a state 
with the given values and entropy $s$. 

5. For $n\geq d$, it is elementary to see that any state $\rho$ can be decomposed 
into a uniform convex combination of $n$ pure states, i.e. 
$\rho=\frac{1}{n} \sum_{i=1}^n \ketbra{\psi_i}{\psi_i}$. 
Observe that the state $\psi^{n} = \ketbra{\psi_1} \ox \cdots \ox \ketbra{\psi_n}$ has the same 
charge values as the state $\rho^{\otimes n}$, but as it is pure it has entropy $0$.
Further, consider the thermal state $\tau$ with the same charge values as $\rho$, but 
the maximum entropy consistent with them. Now let
$\rho_i := \lambda \ketbra{\psi_i}{\psi_i}+(1-\lambda)\tau$, and 
observe that $\rho_\lambda^{n} = \rho_1 \ox \cdots \rho_n$ has the same charge values as 
$\psi^{n}$, $\rho^{n}$ and $\tau^{\ox n}$.
Since the entropy $S(\rho_\lambda^{n})$ is a continuous function of $\lambda$, 
thus interpolating smoothly between $0$ and $n S(\tau)$, there is a 
tensor product state with the same given charge values and prescribed 
entropy $s$ in the said interval. 
%

6. This follows from the strict concavity of the von Neumann entropy $S(\rho)$ as a
function of the state, which imparts the strict concavity on $\und{a} \mapsto S(\tau(\und{a}))$
\end{proof}

\medskip
The penultimate point of Lemma \ref{lemma:phase diagram properties} motivates us to define 
a resource theory where the objects are sequences of states on composite systems
of $n\rightarrow\infty$ parts. 
Inspired by \cite{Sparaciari2016}, the allowed operations in this resource theory 
are those that respect basic principles of physics, namely entropy and charge 
conservation. We point out right here, that ``physics'' in the present context 
does not necessarily refer to the fundamental physical laws of nature, but to 
any rule that the system under consideration obeys. 
It is well-known that quantum operations that preserve entropy for all states are 
unitaries. The class of unitaries that conserve charges of a system are precisely 
those that commute with all charges of that system. 
However, it turns out that these constraints are too strong if imposed literally, 
when many charges are to be conserved, as it could easily happen that only
trivial unitaries are allowed.
Our way out is to consider the thermodynamic limit and at the same time relax the 
allowed operations to approximately entropy and charge conserving ones.    
As for the former, we couple the composite system to an ancillary system with corresponding 
Hilbert space $\mathcal{K}$ of dimension $2^{o(n)}$ where restricting the dimension of 
the ancilla ensures that the \emph{average} entropy of an individual system, that is, 
entropy of the composite system per $n$ does not change in the limit of large $n$. 
Moreover, as for charge conservation, we consider unitaries that preserve the average 
charges of an individual system, and we allow unitaries that are \emph{almost} commuting 
with the total charges of the composite system and the ancilla. The precise definition 
goes as follows:  

\begin{definition}
\label{almost-commuting unitaries}
A unitary operation $U$ acting on a composite system coupled to an ancillary system with 
Hilbert spaces $\mathcal{H}^{\otimes n}$ and $\mathcal{K}$ of dimension $2^{o(n)}$, respectively,
is called an \emph{almost-commuting unitary} with the total charges of a composite system and an 
ancillary system if the operator norm of the normalised commutator for all total charges 
vanishes asymptotically for large $n$:
\begin{align*}
 & \lim_{n \to \infty} \frac{1}{n} \norm{ [U,A_j^{(n)}+A_j']}_{\infty} =\\
 & \qquad    \lim_{n \to \infty} \frac{1}{n}\norm{ U (A_j^{(n)}+A_j')-(A_j^{(n)}+A_j') U }_{\infty} 
     = 0 
  \qquad j=1,\ldots,c.
\end{align*}
where $A_j^{(n)}$ and $A_j'$ are respectively the charges of the composite system 
and the ancilla, such that $\norm{A_j'}_{\infty} \leq o(n)$.  
\end{definition}

We stress that the definition of almost-commuting unitaries automatically implies that 
the ancillary system has a relatively small dimension and charges with small operator 
norm compared to a composite system.    
%
%
The first step in the development of our resource theory is a precise characterisation 
of which transformations between sequences of product state are possible using almost  
commuting unitaries.
To do so, we define \textit{asymptotically equivalent} states as follows:

\begin{definition}
\label{Asymptotic equivalence definition}
Two sequences of product states $\rho^n=\rho_1 \otimes \cdots \otimes \rho_n$ and 
$\sigma^n=\sigma_1 \otimes \cdots \otimes \sigma_n$ of a composite system with 
charges $A_j^{(n)}$ for $j=1,\ldots,c$, are called \emph{asymptotically equivalent} if  
\begin{align*}
  \lim_{n \to \infty} \frac{1}{n} \abs{S(\rho^n) - S(\sigma^n)}                       &= 0, \\ 
  \lim_{n \to \infty} \frac{1}{n} \abs{\Tr \rho^n A_j^{(n)} - \Tr \sigma^n A_j^{(n)}} &= 0 \text{ for } j=1,\ldots,c.
\end{align*}
In other words, two sequences of product states are considered equivalent if their
associated points in the normalised phase diagrams $\frac1n \mathcal{P}^{(n)}$ differ
by a sequence converging to $0$.
\end{definition}


The \emph{asymptotic equivalence theorem} of \cite{Sparaciari2016} characterizes 
feasible state transformations via \textit{exactly} commuting unitaries where energy 
is the only conserved quantity of a system, showing that it is precisely 
given by asymptotic equivalence. 
We prove an extension of this theorem for systems with multiple conserved quantities,
by allowing almost-commuting unitaries. 

\begin{theorem}[Asymptotic (approximate) Equivalence Theorem]
\label{Asymptotic equivalence theorem} 
Let $\rho^n=\rho_1 \otimes \cdots \otimes \rho_n$ and 
$\sigma^n=\sigma_1 \otimes \cdots \otimes \sigma_n$ 
be two sequences of product states of a composite system with 
charges $A_j^{(n)}$ for $j=1,\ldots,c$.
These two states are asymptotically equivalent if and only if 
there exist ancillary quantum systems with corresponding Hilbert space $\mathcal{K}$ 
of dimension $2^{o(n)}$ and an almost-commuting unitary $U$ acting on 
$\mathcal{H}^{\otimes n} \otimes \mathcal{K}$ such that 
\begin{align*}
  \lim_{n \to \infty} \norm{U \rho^n \otimes \omega' U^{\dagger} - \sigma^n \otimes \omega}_1 &= 0, \\
\end{align*}
where $\omega$ and $\omega'$ are states of the ancillary system, and charges of the ancillary 
system are trivial, $A_j' = 0$.
\end{theorem}

\medskip
The \emph{proof} of this theorem is given in Section~\ref{proof-AET}, as it relies on a
number of technical lemmas, among them a novel construction of approximately 
microcanonical subspaces (Section~\ref{section: Approximate microcanonical (a.m.c.) subspace}). 

\medskip
By grouping the $Q$-systems into blocks of $k$, we do not of course change the 
physics of our system, except that now in the asymptotic limit we only consider
$n = k\nu$ copies of $Q$, but the state $\rho^n$ is asymptotically equivalent to
$\rho^{n+O(1)}$ via almost-commuting unitaries according to Definition \ref{almost-commuting unitaries} 
and Theorem \ref{Asymptotic equivalence theorem}. 
But now we consider $Q^k$ with its charge observables $A_j^{(k)}$ as elementary 
systems, which have many more states than the $k$-fold product states we began with. 
Yet, Lemma \ref{lemma:phase diagram properties} shows that the phase diagram for 
the $k$-copy system is simply the rescaled single-copy phase diagram,
$\overline{\mathcal{P}}^{(k)} = k \overline{\mathcal{P}}^{(1)}$, and indeed 
for $k\geq d$, $\mathcal{P}^{(k)} = k \overline{\mathcal{P}}^{(1)}$. This means 
that we can extend the equivalence relation of asymptotic equivalence and the
concomitant Asymptotic Equivalence Theorem (AET) \ref{Asymptotic equivalence theorem}
to any sequences of states that factor into product states of blocks $Q^k$, for 
any integer $k$, which freedom we shall exploit in our treatment of thermodynamics.

\section{Approximate microcanonical (a.m.c.) subspace}
\label{section: Approximate microcanonical (a.m.c.) subspace}

In this section, we recall the definition of 
approximate microcanonical (a.m.c.) and give a new proof that it exists 
for certain explicitly given parameters.
For charges $A_j$ and average values $v_j$,  a.m.c. is basically a \textit{common} subspace for the spectral projectors of $A_j^{(n)}$ with corresponding values close to $n v_j$; that is, a subspace onto which a state projects with high probability if and only if it projects onto the spectral projectors of the charges with high probability. We show in Theorem~\ref{thm:symmetric-micro-exists} that for a large enough $n$ such a subspace exists.
An interesting property of an a.m.c. subspace is that any unitary acting on this subspace 
is an almost-commuting unitary with charges $A_j^{(n)}$.

\begin{definition}
\label{defi:microcanonical}
An \emph{approximate microcanonical (a.m.c.) subspace}, or more precisely
a \emph{$(\epsilon,\eta,\eta',\delta,\delta')$-approximate microcanonical subspace},
$\cM$ of $\cH^{\ox n}$, with projector $P$, 
for charges $A_j$ and values $v_j = \langle A_j \rangle$
is one that consists, in a certain precise sense, of exactly the
states with ``very sharp'' values of all the $A_j^{(n)}$. 
Mathematically, the following has to hold:
\begin{enumerate}
  \item Every state $\omega$ with support contained in $\cM$ satisfies
        $\tr \omega\Pi^\eta_j \geq 1-\delta$ for
        all $j$.
  \item Conversely, every state $\omega$ on $\cH^{\ox n}$ such that
        $\tr \omega\Pi^{\eta'}_j \geq 1-\delta'$ for
        all $j$, satisfies $\tr\omega P \geq 1-\epsilon$.
\end{enumerate}
Here, $\Pi^\eta_j := \bigl\{ nv_j-n\eta\Sigma(A_j) \leq A_j^{(n)} \leq nv_j+n\eta\Sigma(A_j) \bigr\}$
is the spectral projector of $A_j^{(n)}$ of values close to $n v_j$,
and $\Sigma(A) = \lambda_{\max}(A)-\lambda_{\min}(A)$ is the spectral diameter
of the Hermitian $A$, i.e.~the diameter of the smallest disc
covering the spectrum of $A$.
\end{definition}

\medskip

\begin{remark}
It is shown in Theorem~3 of~\cite{Halpern2016} that for every $\epsilon > c\delta' > 0$,
$\delta > 0$ and $\eta > \eta' > 0$, and for all sufficiently
large $n$, there exists a nontrivial 
$(\epsilon,\eta,\eta',\delta,\delta')$-a.m.c.~subspace.
However, there are two (related) reasons why one might be not completely
satisfied with the argument in~\cite{Halpern2016}: First, the proof uses
a difficult result of Ogata~\cite{Ogata} to reduce the non-commuting
case to the seemingly easier of commuting observables; while this
is conceptually nice, it makes it harder to perceive the nature
of the constructed subspace. Secondly, despite the fact that the
defining properties of an a.m.c.~subspace are manifestly permutation symmetric
(w.r.t.~permutations of the $n$ subsystems), the resulting construction
does not have this property.

Here we address both these concerns. Indeed, we shall show by 
essentially elementary means how to obtain an a.m.c.~subspace
that is by its definition permutation symmetric.
\end{remark}

\begin{theorem}
  \label{thm:symmetric-micro-exists}
  Under the previous assumptions, for every $\epsilon > 2(n+1)^{3d^2}\delta' > 0$,
  $\eta >\eta' > 0$ and $\delta>0$, for all sufficiently large $n$
  there exists an approximate microcanonical subspace projector.
  In addition, the subspace can be chosen to be stable under permutations
  of the $n$ systems: $U^\pi\cM = \cM$, or equivalently $U^\pi P (U^\pi)^\dagger = P$,
  for any permutation $\pi\in S_n$ and its unitary action $U^\pi$.

  More precisely, given $\eta > \eta' > 0$ and $\epsilon > 0$, there exists a $\alpha > 0$ such
  that there is a non-trivial $(\epsilon,\eta,\eta',\delta,\delta')$-a.m.c. subspace
  with
  \begin{align*}
    \delta  &= (c+3)(5n)^{5d^2} e^{-\alpha n} \text{ and } \\ 
    \delta' &= \frac{\epsilon}{2(n+1)^{3d^2}} - (c+3)(5n)^{2d^2} e^{-\alpha n}.
  \end{align*}
  Furthermore, we may choose $\alpha = \frac{(\eta-\eta')^2}{8c^2(d+1)^2}$.
\end{theorem}
\begin{proof}
For $s>0$, partition the state space $\cS(\cH)$ on $\cH$ into 
\begin{align*}
  \cC_s(\underline{v}) &= \bigl\{ \sigma : \forall j\ |\tr\sigma A_j - v_j| \leq s\Sigma(A_j) \bigr\}, \\
  \cF_s(\underline{v}) &= \bigl\{ \sigma : \exists j\ |\tr\sigma A_j - v_j| >    s\Sigma(A_j) \bigr\}
                        = \cS(\cH) \setminus \cC_s(\underline{v}),
\end{align*}
which are the sets of states with $A_j$-expectation values ``close'' 
to and ``far'' from $\underline{v}$. 
Note that if $\rho \in \cC_s(\underline{v})$ and 
$\sigma \in \cF_t(\underline{v})$, $0 < s < t$, then $\|\rho-\sigma\|_1 \geq t-s$.

Choosing the precise values of $s>\eta'$ and $t<\eta$ later, 
we pick a universal distinguisher $(P,P^\perp)$
between $\cC_s(\underline{v})^{\ox n}$ and $\cF_t(\underline{v})^{\ox n}$,
according to Lemma~\ref{lemma:universal-test} below:
\begin{align}
  \label{eq:s-close}
  \forall \rho\in\cC_s(\underline{v})\   \tr\rho^{\ox n} P^\perp &\leq (c+2)(5n)^{2d^2} e^{-\zeta n}, \\
  \label{eq:t-far}
  \forall \sigma\in\cF_t(\underline{v})\ \tr\sigma^{\ox n} P     &\leq (c+2)(5n)^{2d^2} e^{-\zeta n},
\end{align}
with $\zeta = \frac{(t-s)^2}{2c^2(2d^2+1)}$.
Our a.m.c.~subspace will be $\cM := \operatorname{supp} P$;
by Lemma~\ref{lemma:universal-test}, $P$ and likewise $\cM$ are permutation symmetric.

It remains to check the properties of the definition. First, let $\omega$
be supported on $\cM$. Since we are interested in $\tr\omega \Pi_j^\eta$, we may
without loss of generality assume that $\omega$ is permutation symmetric.
Thus, by the ``constrained de Finetti reduction'' (aka ``Postselection 
Lemma'')~\cite[Lemma~18]{Duan2016},
\begin{align}
  \label{eq:dF}
  \omega \leq (n+1)^{3d^2} \int {\rm d}\sigma\,\sigma^{\ox n} F(\omega,\sigma^{\ox n})^2,
\end{align}
with a certain universal probability measure ${\rm d}\sigma$ on $\cS(\cH)$, 
and the fidelity $F(\rho,\sigma) = \|\sqrt{\rho}\sqrt{\sigma}\|_1$ between
states. We need the monotonicity of the fidelity under cptp maps, which we
apply to the test $(P,P^\perp)$:
\[
  F(\omega,\sigma^{\ox n})^2 \leq F\bigl( (\tr\sigma^{\ox n}P,1-\tr\sigma^{\ox n}P),(1,0) \bigr)^2
                             \leq \tr\sigma^{\ox n}P,
\]
which holds because $\tr\omega P = 1$.
Thus,
\begin{align}
  \label{eq:dF-P}
  \tr\omega(\Pi_j^\eta)^\perp \leq (n+1)^{3d^2} \int {\rm d}\sigma\,\bigl(\tr\sigma^{\ox n}(\Pi_j^\eta)^\perp\bigr) 
                                                                                               (\tr\sigma^{\ox n}P).
\end{align}

Now we split the integral on the right hand side of Eq.~(\ref{eq:dF-P}) into two parts,
$\sigma\in\cC_t(\underline{v})$ and $\sigma\not\in\cF_t(\underline{v})$:
If $\sigma\in\cF_t(\underline{v})$, then by Eq.~(\ref{eq:t-far}) 
we have
\[
  \tr\sigma^{\ox n}P \leq (c+2)(5n)^{2d^2} e^{-\zeta n}.
\]
On the other hand, if $\sigma\in\cC_t(\underline{v})$, then because of $t < \eta$ we have
\[
  \tr\sigma^{\ox n}(\Pi_j^\eta)^\perp \leq 2 e^{-2(\eta-t)^2 n},
\]
which follows from Hoeffding's inequality~\cite{DemboZeitouni}: 
Indeed, let $Z_\ell$ be the i.i.d.~random variables obtained by the
measurement of $A_j$ on the state $\sigma$. They take values in
the interval $[\lambda_{\min}(A_j),\lambda_{\max}(A_j)]$, their expectation 
values satisfy $\EE Z_j = \tr\sigma A_j \in [v_j \pm t\Sigma(A_j)]$, while
\[\begin{split}
  \tr\sigma^{\ox n}(\Pi_j^\eta)^\perp &= \Pr\left\{\frac1n\sum_\ell Z_\ell \not\in[v_j \pm \eta\Sigma(A_j)\right\} \\
                     &\leq \Pr\left\{\frac1n\sum_\ell Z_\ell \not\in[\tr\sigma A_j \pm (\eta-t)\Sigma(A_j)\right\},
\end{split}\]
so Hoeffding's inequality applies.
All taken together, we have
\[\begin{split}
  \tr\omega(\Pi_j^\eta)^\perp 
      &\leq (n+1)^{3d^2} \left( (c+2)(5n)^{2d^2} e^{-\zeta n} +  2 e^{-2(\eta-t)^2 n} \right) \\
      &\leq (c+3)(5n)^{5d^2} e^{-2(\eta-t)^2 n},
\end{split}\]
because we can choose $t$ such that
\begin{equation}
  \label{eq:t}
  \eta-t = \frac{t-s}{2c\sqrt{2d^2+1}} \geq \frac{t-s}{4cd}.
\end{equation}

Secondly, let $\omega$ be such that $\tr\omega \Pi_j^\eta \geq 1-\delta'$;
as we are interested in $\tr\omega P$, we may again assume
without loss of generality that $\omega$ is permutation symmetric,
and invoke the constrained de Finetti reduction~\cite[Lemma~18]{Duan2016},
Eq.~(\ref{eq:dF}). 
From that we get, much as before,
\[
  \tr\omega P^\perp \leq (n+1)^{3d^2} \int {\rm d}\sigma\, (\tr\sigma^{\ox n}P^\perp) F(\omega,\sigma^{\ox n})^2,
\]
and we split the integral on the right hand side into two parts,
depending on $\sigma\in\cF_s(\underline{v})$ or
$\sigma\in\cC_s(\underline{v})$: In the latter case, 
$\tr\sigma^{\ox n}P^\perp \leq (c+2)(5n)^{2d^2} e^{-\zeta n}$,
by Eq.~(\ref{eq:s-close}). In the former case, there exists a $j$ such
that $\tr\sigma A_j = w_j \not\in [v_j \pm s\Sigma(A_j)]$, and so
\[\begin{split}
  F(\omega,\sigma^{\ox n})^2 
      &\leq F\bigl( (1-\delta',\delta'), (\tr\sigma^{\ox n}\Pi_j^{\eta'},1-\tr\sigma^{\ox N}\Pi_j^{\eta'}) \bigr) \\
      &\leq \left( \sqrt{\delta'} + \sqrt{\tr\sigma^{\ox n}\Pi_j^{\eta'}} \right)^2                          \\
      &\leq 2 \delta' + 2 \tr\sigma^{\ox n}\Pi_j^{\eta'}                             \\
      &\leq 2 \delta' + 4 e^{-2(s-\eta')^2 n},
\end{split}\]
the last line again by Hoeffding's inequality; indeed, with the previous notation,
\[\begin{split}
  \tr\sigma^{\ox n}\Pi_j^{\eta'} &= \Pr\left\{ \frac1n \sum_\ell Z_\ell \in[v_j \pm \eta'\Sigma(A_j) \right\} \\
                      &\leq \Pr\left\{ \frac1n \sum_\ell Z_\ell \not\in[w_j \pm (s-\eta')\Sigma(A_j) \right\}.
\end{split}\]
All taken together, we get
\[\begin{split}
  \tr\omega P^\perp
      &\leq (n+1)^{3d^2} \left( (c+2)(5n)^{2d^2} e^{-\zeta n} +  4 e^{-2(s-\eta')^2 n} + 2 \delta' \right) \\
      &\leq (n+1)^{3d^2} (c+3)(5n)^{2d^2} e^{-2(s-\eta')^2 n} + 2(n+1)^{3d^2}\delta',
\end{split}\]
because we can choose $s$ such that
\begin{equation}
  \label{eq:s}
  s-\eta' = \frac{t-s}{2c\sqrt{2d^2+1}} \geq \frac{t-s}{4cd}.
\end{equation}

From eqs.~(\ref{eq:t}) and (\ref{eq:s}) we get by summation
\[
  \eta-\eta' = t-s + \frac{t-s}{c\sqrt{2d^2+1}} \leq (t-s)\left( 1+\frac{1}{cd} \right),
\]
from which we obtain
\[
  s-\eta' = \eta-t \geq \frac{\eta-\eta'}{4c(d+1)},
\]
concluding the proof.
\end{proof}

\medskip

\begin{lemma}
  \label{lemma:universal-test}
  For all $0 < s < t$ there exists $\zeta > 0$, such that
  for all $n$ there exists a permutation symmetric 
  projector $P$ on $\cH^{\ox n}$ with the properties
  \begin{align}
    \forall \rho\in\cC_s(\underline{v})\   \tr\rho^{\ox n} P^\perp &\leq (c+2)(5n)^{2d^2} e^{-\zeta n}, \\
    \forall \sigma\in\cF_t(\underline{v})\ \tr\sigma^{\ox n} P     &\leq (c+2)(5n)^{2d^2} e^{-\zeta n}.
  \end{align}
  The constant $\zeta$ may be chosen as
  $\zeta = \frac{(t-s)^2}{2c^2(2d^2+1)}$.
\end{lemma}
\begin{proof}
We start by showing that there is a POVM $(M,\1-M)$ with 
\begin{align}
  \forall \rho\in\cC_s(\underline{v})\            \tr\rho^{\ox n} (\1-M) &\leq c e^{-\frac{(t-s)^2}{2c^2}n}, \\
  \forall \sigma\in\cF_t(\underline{v})\ \ \qquad \tr\sigma^{\ox n} M    &\leq       e^{-\frac{(t-s)^2}{2c^2}n}.
\end{align}
Namely, for each $\ell=0,\ldots,n$ choose $j_\ell \in \{1,\ldots,c\}$ uniformly
at random and measure $A_{j_\ell}$ on the $\ell$-th system. Denote the outcome
by the random variable $Z_\ell^{j_{\ell}}$ and let $Z_\ell^j = 0$ for $j\neq j_\ell$.
Thus, for all $j$, the random variables $Z_\ell^j$ are i.i.d.~with
mean $\EE Z_\ell^j = \frac{1}{c}\tr\rho A_j$, if the measured state is $\rho^{\ox n}$.

Outcome $M$ corresponds to the event
\[
  \forall j\ \frac1n \sum_\ell Z_\ell^j \in \frac{1}{c}\left[v_j \pm \frac{s+t}{2}\Sigma(A_j)\right];
\]
outcome $\1-M$ corresponds to the complementary event
\[
  \exists j\ \frac1n \sum_\ell Z_\ell^j \not\in \frac{1}{c}\left[v_j \pm \frac{s+t}{2}\Sigma(A_j)\right].
\]
We can use Hoeffding's inequality to bound the traces in question.\\
For $\rho\in \cC_s(\underline{v})$, we have $|\EE Z_\ell^j - v_j | \leq \frac{s}{c}\Sigma(A_j)$
for all $j$, and so:
\[\begin{split}
  \tr\rho^{\ox n}(\1-M) &=    \Pr\left\{ \exists j\ \frac1n \sum_\ell Z_\ell^j 
                                          \not\in \frac{1}{c}\left[v_j \pm \frac{s+t}{2}\Sigma(A_j)\right] \right\} \\
                        &\leq \sum_{j=1}^c \Pr\left\{ \frac1n \sum_\ell Z_\ell^j 
                                          \not\in \frac{1}{c}\left[v_j \pm \frac{s+t}{2}\Sigma(A_j)\right] \right\} \\
                        &\leq \sum_{j=1}^c \Pr\left\{ \frac1n \sum_\ell Z_\ell^j 
                                          \not\in \frac{1}{c}\left[v_j \pm \frac{s+t}{2}\Sigma(A_j)\right] \right\} \\
                        &\leq \sum_{j=1}^c \Pr\left\{ \left| \frac1n \sum_\ell Z_\ell^j - \EE Z_1^j \right|
                                                                     > \frac{t-s}{2c} \Sigma(A_j) \right\} \\
                        &\leq c e^{-\frac{(t-s)^2}{2c^2}n}.
\end{split}\]
For $\sigma\in\cF_t(\underline{v})$, there exists a $j$ such that
$|\EE Z_\ell^j - v_j | > \frac{t}{c}\Sigma(A_j)$. Thus,
\[\begin{split}
  \tr\sigma^{\ox n} M     &\leq \Pr\left\{ \frac1n \sum_\ell Z_\ell^j 
                                            \in \frac{1}{c}\left[v_j \pm \frac{s+t}{2}\Sigma(A_j)\right] \right\} \\
                          &\leq \Pr\left\{ \left| \frac1n \sum_\ell Z_\ell^j - \EE Z_1^j \right|
                                                                   > \frac{t-s}{2c} \Sigma(A_j) \right\} \\
                          &\leq e^{-\frac{(t-s)^2}{2c^2}n}.
\end{split}\]

This POVM is, by construction, permutation symmetric, but $M$ is not a 
projector. To fix this, choose $\lambda$-nets $\cN_C^\lambda$ in $\cC_s(\underline{v})$
and $\cN_F^\lambda$ in $\cF_t(\underline{v})$, with
$\lambda = e^{-\zeta n}$, with $\zeta = \frac{(t-s)^2}{2c^2(2d^2+1)}$.
This means that every state $\rho\in\cC_s(\underline{v})$
is no farther than $\lambda$ in trace distance from a $\rho'\in\cN_C^\lambda$,
and likewise for $\cF_t(\underline{v})$.
By~\cite[Lemma~III.6]{Hayden2006} (or rather, a minor variation of its proof), 
we can find such nets with 
$|\cN_C^\lambda|,\ |\cN_F^\lambda| \leq \left( \frac{5n}{\lambda} \right)^{2d^2}$
elements.
Form the two states
\begin{align*}
  \Gamma &:= \frac{1}{|\cN_C^\lambda|} \sum_{\rho\in\cN_C^\lambda} \rho^{\ox n}, \\
  \Phi   &:= \frac{1}{|\cN_F^\lambda|} \sum_{\sigma\in\cN_F^\lambda} \sigma^{\ox n},
\end{align*}
and let
\[
  P := \{ \Gamma-\Phi \geq 0 \}
\]
be the Helstrom projector which optimally distinguishes $\Gamma$ from $\Phi$.
But we know already a POVM that distinguishes the two states, hence
$(P,P^\perp=\1-P)$ cannot be worse:
\[
  \tr\Gamma P^\perp + \tr\Phi P \leq \tr\Gamma (\1-M) + \tr\Phi M
                                \leq (c+1) e^{-\frac{(t-s)^2}{2c^2}n},
\]
thus for all $\rho\in\cN_C^\lambda$ and $\sigma\in\cN_F^\lambda$,
\[
  \tr\rho^{\ox n}P^\perp,\ \tr\sigma^{\ox n}P 
      \leq (c+1)\left( \frac{5n}{\lambda} \right)^{2d^2} e^{-\frac{(t-s)^2}{2c^2}n}.
\]
So, by the $\lambda$-net property, we find 
for all $\rho\in\cC_s(\underline{v})$ and $\sigma\in\cF_t(\underline{v})$,
\[
  \tr\rho^{\ox n}P^\perp,\ \tr\sigma^{\ox n}P 
      \leq \lambda + (c+1)\left( \frac{5n}{\lambda} \right)^{2d^2} e^{-\frac{(t-s)^2}{2c^2}n}
      \leq (c+2)(5n)^{2d^2} e^{-\zeta n},
\]
by our choice of $\lambda$.
\end{proof}

\begin{corollary}\label{corollary: a.m.c. projection}
For charges $A_j$, values $v_j = \langle A_j \rangle$ and $n>0$, Theorem~\ref{thm:symmetric-micro-exists} implies that there is an a.m.c. subspace $\mathcal{M}$ of $\mathcal{H}^{\otimes n}$ for any $ \eta' > 0$ with the following parameters:
\begin{align*}
  \eta     &=2 \eta',\\
  \delta'  &=\frac{c+3}{2}(5n)^{2d^2} e^{-\frac{n \eta'^2}{8c^2(d+1)^2}}, \\ 
  \delta   &=(c+3)(5n)^{2d^2} e^{-\frac{n \eta'^2}{8c^2(d+1)^2}},\\ 
  \epsilon &=2(c+3)(n+1)^{3d^2}(5n)^{2d^2} e^{-\frac{n \eta'^2}{8c^2(d+1)^2}}. 
\end{align*}

Moreover, let $\rho^n=\rho_1 \otimes \cdots \otimes \rho_n$ be a state with $\frac{1}{n}\abs{\Tr(\rho^n A_j^{(n)})- v_j}\leq \frac{1}{2} \eta' \Sigma(A_j)$. Then, $\rho^n$ projects onto a.m.c. subspace with probability $\epsilon$: 
\[\Tr(\rho^n P) \geq 1- \epsilon.\]
\end{corollary}
\begin{proof}
For simplicity of notation we drop the subscript $j$ from $A_j$, $v_j$ and $\Pi^{\eta'}_j$, so let $\sum_{l=1}^d E_l \ketbra{l}{l}$ be the spectral decomposition of $A$. Define independent random variables $X_i$ for $i=1,\ldots,n$ taking values in the set $\set{E_1,\ldots,E_d}$ with probabilities $p_i(E_l)=\Tr(\rho_i \ketbra{l}{l} )$.
Furthermore, define random variable $\overline{X}=\frac{X_1+\ldots+X_n}{n}$ which has the following expectation value 
\begin{align*}
    \mathbb{E}(\overline{X})=\frac{1}{n} \Tr(\rho^n A^{(n)}).
\end{align*}
Therefore, we obtain
\begin{align*}
    1-&\Tr(\rho^n \Pi^{\eta'})\\
    &=\sum_{\substack{l_1,\ldots,l_n: \\ \abs{E_{l_1}+\ldots+E_{l_n}-n v} \geq n \eta' \Sigma(A) }} \bra{l_1} \rho_1 \ket{l_1} \ldots \bra{l_n} \rho_n \ket{l_n} \\
    &=\Pr \left( \abs{\overline{X}- v} \geq \eta' \Sigma(A)  \right) \\
    &=\Pr \left( \overline{X}- \mathbb{E}(\overline{X})  \geq \eta'\Sigma(A)+v - \mathbb{E}(\overline{X})  \>\bigcup \> \overline{X}- \mathbb{E}(\overline{X})  \leq -\eta' \Sigma(A)+v - \mathbb{E}(\overline{X})  \right) \\
    &\leq \exp \left(-\frac{2n (\eta'\Sigma(A)+v - \mathbb{E}(\overline{X}))^2 }{(\Sigma(A))^2} \right)+\exp \left(-\frac{2n (\eta'\Sigma(A)-v + \mathbb{E}(\overline{X}))^2 }{(\Sigma(A))^2} \right)\\
    & \leq 2 \exp (-\frac{n \eta'^2}{2})\\
    & \leq \delta',
\end{align*}
where the second line follows because random the variables $X_1,\ldots,X_n$ are independent and as a result $\Pr\{\overline{X}=\frac{E_{l_1}+\ldots+E_{l_n}}{n}\}=\bra{l_1} \rho_1 \ket{l_1} \cdots \bra{l_n} \rho_n \ket{l_n}$. The fourth line is due to Hoeffding's inequality (Lemma~\ref{Hoeffding's inequality}). The fifth line is due to assumption $\abs{\mathbb{E}(\overline{X})- v} \leq \frac{1}{2} \eta' \Sigma(A)$.

Thus, by the definition of a.m.c. subspace $\Tr(\rho^n P) \geq 1- \epsilon$.
\end{proof}

\section{Proof of the AET Theorem~\ref{Asymptotic equivalence theorem}}
\label{proof-AET}
Here, we first prove the following lemma where we will use points 3 and 4 to prove the main theorem. Corollary~\ref{corollary: a.m.c. projection} implies that assuming $\frac{1}{n}\Tr \rho^n A^{(n)}_j \approx \frac{1}{n}\Tr \sigma^n A^{(n)}_j \approx v_j$ the states $\rho^n$ and
$\sigma^n$ project onto the a.m.c. subspace with high probability. 
Hence, in Lemma~\ref{lemma: timmed state}, we show that one can find states $\widetilde{\rho}$ and $\widetilde{\sigma}$ with support inside the a.m.c. subspace which are very close to the original states in trace norm, that is, $\widetilde{\rho} \approx \rho^n$ and $\widetilde{\sigma} \approx \sigma^n$, and there are unitaries $V_1$ and $V_2$ that factorizes these states to the tensor product of maximally mixed states $\tau$ and $\tau'$ and some other state of very small dimension:
\begin{align*}
    V_1\widetilde{\rho}V_1^{\dagger}=\tau \otimes \omega \quad \text{and} \quad 
    V_2\widetilde{\sigma}V_2^{\dagger}=\tau' \otimes \omega'.
\end{align*}
Further, assuming that the states $\rho^n$ and $\sigma^n$ have very close entropy rates, i.e. 
$\frac{1}{n}S(\rho^n) \approx \frac{1}{n}S(\sigma^n)$, one can find states $\tau$ and $\tau'$ with the same dimension that is $\tau=\tau'$. Thus, we observe that two states $\widetilde{\rho}\otimes \omega'$ and $\widetilde{\sigma}\otimes \omega$ have exactly the same spectrum, so there is unitary acting on the a.m.c. subspace and the ancillary system taking one state to another. Based on the properties of the a.m.c. subspace, we show that this unitary is an almost-commuting unitary with the charges $A_j^{(n)}$.

\begin{lemma}\label{lemma: timmed state}
Let subspace $\mathcal{M}$ of $\mathcal{H}^{\otimes n}$ with projector $P$ be a high probability subspace for state $\rho^n=\rho_1 \otimes \cdots \otimes \rho_n$, i.e. $\Tr(\rho^n P) \geq 1- \epsilon$.
Then, for sufficiently large $n$ there is a subspace $\widetilde{\mathcal{M}} \subseteq \mathcal{M}$ with projector $\widetilde{P}$ and state $\widetilde{\rho}$ with support inside $\widetilde{\mathcal{M}}$  such that the following holds:  
\begin{enumerate}
    \item $\Tr(\Pi^n_{\alpha,\rho^n}\rho^n \Pi^n_{\alpha,\rho^n}  \widetilde{P})  \geq 1- 2\sqrt{\epsilon}-\frac{1}{O(\alpha)}$.
    \item $2^{-\sum_{i=1}^n S(\rho_i)-2\alpha \sqrt{n}}\widetilde{P} \leq \widetilde{P} \Pi^n_{\alpha,\rho^n}\rho^n \Pi^n_{\alpha,\rho^n} \widetilde{P} \leq 2^{-\sum_{i=1}^n S(\rho_i)+ \alpha \sqrt{n}}\widetilde{P}$.
    \item There is a unitary $U$ such that $U\widetilde{\rho}U^{\dagger}=\tau \otimes \omega $
     where $\tau$ is a maximally mixed state of dimension $2^{\sum_{i=1}^n S(\rho_i) - O(\alpha \sqrt{n})}$, and $\omega$ is a  state of dimension $2^{O(\alpha \sqrt{n })}$.
    \item $\norm{\widetilde{\rho}-\rho^n}_1 \leq 2\sqrt{\epsilon}+\frac{1}{O(\alpha)}+2\sqrt{2\sqrt{\epsilon}+\frac{1}{O(\alpha)}} $.
\end{enumerate}
\end{lemma}

\begin{proof} 
1. Let $E\geq 0$ and $F\geq 0$ be two positive operators such that $E+F=P \Pi^n_{\alpha ,\rho^n } P$ where all eigenvalues of $F$ are smaller than $2^{-\alpha \sqrt{n}}$, and define $\widetilde{P}$ to be the projection onto the support of $E$. 
In other words, $\widetilde{P}$ is the projection onto the support of $P \Pi^n_{\alpha ,\rho^n } P$ with corresponding eigenvalues greater $2^{-\alpha\sqrt{n}}$. Then, we  obtain
\begin{align*}
    \Tr&(\Pi^n_{\alpha,\rho^n}\rho^n \Pi^n_{\alpha,\rho^n}  \widetilde{P})\\
    & \geq \Tr(\Pi^n_{\alpha,\rho^n}\rho^n \Pi^n_{\alpha,\rho^n}  E)\\
    &=\Tr(\Pi^n_{\alpha,\rho^n}\rho^n \Pi^n_{\alpha,\rho^n}  P\Pi^n_{\alpha,\rho^n}P)-\Tr(\Pi^n_{\alpha,\rho^n}\rho^n \Pi^n_{\alpha,\rho^n}F)\\
    &\geq \Tr(\rho^n  P\Pi^n_{\alpha,\rho^n}P)-\norm{\Pi^n_{\alpha,\rho^n}\rho^n \Pi^n_{\alpha,\rho^n}-\rho^n}_1-\Tr(\Pi^n_{\alpha,\rho^n}\rho^n \Pi^n_{\alpha,\rho^n}F)\\ 
    &\geq \Tr(\rho^n \Pi^n_{\alpha,\rho^n})-\norm{P\rho^n P-\rho^n}_1-\norm{\Pi^n_{\alpha,\rho^n}\rho^n \Pi^n_{\alpha,\rho^n}-\rho^n}_1-\Tr(\Pi^n_{\alpha,\rho^n}\rho^n \Pi^n_{\alpha,\rho^n}F)\\ 
    &\geq \Tr(\rho^n \Pi^n_{\alpha,\rho^n})-\norm{P\rho^n P-\rho^n}_1-\norm{\Pi^n_{\alpha,\rho^n}\rho^n \Pi^n_{\alpha,\rho^n}-\rho^n}_1-2^{-\alpha \sqrt{n}}\\ 
    & \geq 1-\frac{\beta}{\alpha^2}-2\sqrt{\epsilon}-2\frac{\sqrt{\beta}}{\alpha}-2^{-\alpha \sqrt{n}},
\end{align*}
where the first line follows from the fact that $\widetilde{P} \geq E$. The third, forth and fifth lines are due to H\"{o}lder inequality. The last line follows from Lemma \ref{lemma:typicality properties } and gentle operator lemma \ref{Gentle Operator Lemma}.

\medskip
2. By the fact that in the typical subspace the eigenvalues of $\rho^n$ are bounded (Lemma \ref{lemma:typicality properties }), we obtain 
\begin{align*}
      \widetilde{P} \Pi^n_{\alpha,\rho^n}\rho^n \Pi^n_{\alpha,\rho^n} \widetilde{P} &\leq 2^{-\sum_{i=1}^n S(\rho_i)+\alpha\sqrt{n}}\widetilde{P} \Pi^n_{\alpha ,\rho^n }  \widetilde{P} \\
      &\leq 2^{-\sum_{i=1}^n S(\rho_i)+\alpha\sqrt{n}}\widetilde{P}.
\end{align*}
For the lower bound notice that 
\begin{align*}
      \widetilde{P} \Pi^n_{\alpha,\rho^n}\rho^n \Pi^n_{\alpha,\rho^n} \widetilde{P}& \geq 2^{-\sum_{i=1}^n S(\rho_i)-\alpha\sqrt{n}}\widetilde{P} \Pi^n_{\alpha ,\rho^n }  \widetilde{P}\\ 
      &=2^{-\sum_{i=1}^n S(\rho_i)-\alpha\sqrt{n}}\widetilde{P}P \Pi^n_{\alpha ,\rho^n } P \widetilde{P} \\ 
      &\geq 2^{-\sum_{i=1}^n S(\rho_i)-2\alpha\sqrt{n}} \widetilde{P},
\end{align*}
where the equality holds because $\widetilde{P} \subseteq \mathcal{M}$, therefore $\widetilde{P} P=\widetilde{P}$. The last inequality follows because $\widetilde{P}$ is the projection onto support of $P \Pi^n_{\alpha ,\rho^n } P$ with eigenvalues greater $2^{-\alpha\sqrt{n}}$.


\medskip
3. Consider the unnormalized state $\widetilde{P} \Pi^n_{\alpha,\rho^n}\rho^n \Pi^n_{\alpha,\rho^n} \widetilde{P}$ with support inside $\widetilde{\mathcal{M}}$.
From from point 2, we know that all the eigenvalues of this state belongs to the interval $[2^{-\sum_{i=1}^n S(\rho_i)-2\alpha \sqrt{n}},2^{-\sum_{i=1}^n S(\rho_i)+\alpha \sqrt{n}}]$ which we denote it by $[p_{\min},p_{\max}]$. We divide this interval to $b=2^{\floor{5\alpha \sqrt{n}}}$ many intervals (bins) with equal length of $\Delta p=\frac{p_{\max}-p_{\min}}{b}$. Now, we \textit{trim} the eigenvalues of this unnormalized state in three steps as follows.\\
\begin{enumerate}[(a)]
    \item  Each eigenvalue belongs to a bin which is an interval $[p_k,p_{k+1})$ for some $0 \leq k \leq b -1$ with $p_k=p_{\min}+\Delta p \times k$. For example, eigenvalue  $\lambda_l$ is equal to $p_k+q_l$ for some $k$ such that $0\leq q_l< \Delta p$. We throw away $q_l$ part of each eigenvalue $\lambda_l$. The sum of these parts over all eigenvalues is very small
    \begin{align*}
        \sum_{l=1}^{|\widetilde{\mathcal{M}}|} q_l \leq \Delta p |\widetilde{\mathcal{M}}| \leq 2^{-2\alpha \sqrt{n}+1},
    \end{align*}
    where the dimension of the subspace $\widetilde{\mathcal{M}}$ is bounded as $|\widetilde{\mathcal{M}}|\leq 2^{\sum_{i=1}^n S(\rho_i)+2\alpha \sqrt{n}}$ which follows from point 2 of the lemma.

    \item  We throw away the bins which contain less than $2^{\sum_{i=1}^n S(\rho_i)-10\alpha \sqrt{n}}$ many eigenvalues. The sum of all the eigenvalues that are thrown away is bounded by
    \begin{align*}
       2^{\sum_{i=1}^n S(\rho_i)-10\alpha \sqrt{n}} \times 2^{5\alpha \sqrt{n}} \times 2^{-\sum_{i=1}^n S(\rho_i)+\alpha  \sqrt{n}} \leq 2^{-4\alpha \sqrt{n}},
    \end{align*}
    in the left member, the first number is the number of eigenvalues in the bin; the second  is the number of bins, and the third is the maximum eigenvalue. 
    
    \item  If a bin, e.g. $k$th bin,  is not thrown away in the previous step, it contains $M_k$ many eigenvalues with the same value with 
    \begin{align}\label{eq:bounds of bin size}
         2^{\sum_{i=1}^n S(\rho_i)-10\alpha  \sqrt{n}} \leq  M_k \leq 2^{\sum_{i=1}^n S(\rho_i)+2\alpha  \sqrt{n}}.
    \end{align}
    Let
    \begin{align}\label{eq:L}
     L= 2^{\floor{\sum_{i=1}^n S(\rho_i) -10 \alpha\sqrt{n}}}  
    \end{align}
    and for the $k$th bin, let $m_{k}$ be an integer number such that 
    \begin{align}\label{M_k_L}
          m_{k} L\leq M_k  \leq (m_{k}+1) L.     
    \end{align}
    Then, $m_{k}$ is bounded as follows
    \begin{align}\label{m_k}
        m_{k} \leq 2^{12\alpha\sqrt{n}}. 
    \end{align}
    From the $k$th bin, we keep $m_{k} L$ number of eigenvalues and throw away the rest where there are $M_k-m_{k} L \leq L$ many of them; the sum of the eigenvalues that are thrown away in this step is bounded by
    \begin{align}
        \sum_{k=0}^{b-1}  p_{k}(M_k-m_k L) \leq L\sum_{k=0}^{b-1} p_k \leq  2^{-4\alpha\sqrt{n}}. \nonumber 
    \end{align}
\end{enumerate}

Therefore, for sufficiently large $n$ the sum of the eigenvalues thrown away in 
the last three steps is bounded by
\begin{align}\label{eq:thrown away sum}
    2^{-2\alpha\sqrt{n}+1}+2^{-4\alpha\sqrt{n}}+2^{-4\alpha\sqrt{n}}\leq 2^{-\alpha\sqrt{n}}
\end{align}
The kept eigenvalues of all bins form an $L$-fold degenerate unnormalized state of dimension $\sum_{k=0}^{b-1} m_{k} L$ because each eigenvalue has at least degeneracy of the order of $L$. Thus, up to  unitary $U^{\dagger}$, it can be factorized into the tensor product of a maximally mixed state $\tau$ and  unnormalized state $\omega'$ of dimensions $L$ and $\sum_{k=0}^{b-1} m_{k} $, respectively. 
From (\ref{m_k}), the dimension of $\omega'$ is bounded by
    \begin{align}
         \sum_{k=0}^{b-1} m_{k} \leq  2^{12\alpha\sqrt{n}} \times 2^{5 \alpha\sqrt{n}}=2^{17 \alpha\sqrt{n}}.    \nonumber 
    \end{align}
Then, let $\omega =\frac{\omega'}{\Tr(\omega')}$ and define 
\begin{align*}
   \widetilde{\rho}=U \tau \otimes \omega U^{\dagger}.
\end{align*}    
   
\medskip
4. From points 3 and 1, we obtain
\begin{align}\label{eq: Tr of omega'}
   \Tr(\omega') &= \Tr(\tau \otimes \omega') \\
   &\geq \Tr(\widetilde{P} \Pi^n_{\alpha,\rho^n}\rho^n \Pi^n_{\alpha,\rho^n} \widetilde{P})-2^{-\alpha\sqrt{n}}\\
   &\geq 1-2\sqrt{\epsilon}-2\frac{\sqrt{\beta}}{\alpha}-\frac{\beta}{\alpha^2}-2^{-\alpha \sqrt{n}+1}.
\end{align}
Thereby, we get the following
\begin{align*}
    \norm{\widetilde{\rho}-\rho^n}_1& \leq \norm{\widetilde{\rho}-U \tau \otimes \omega' U^{\dagger}}_1+\norm{U \tau \otimes \omega' U^{\dagger}-\widetilde{P} \Pi^n_{\alpha,\rho^n}\rho^n \Pi^n_{\alpha,\rho^n} \widetilde{P}}_1 \\
    &\quad \quad \quad+ \norm{\widetilde{P} \Pi^n_{\alpha,\rho^n}\rho^n \Pi^n_{\alpha,\rho^n} \widetilde{P} -\rho^n}_1 \\
    &\leq 1-\Tr(\omega')+\norm{U \tau \otimes \omega' U^{\dagger}-\widetilde{P} \Pi^n_{\alpha,\rho^n}\rho^n \Pi^n_{\alpha,\rho^n} \widetilde{P}}_1 \\
    &\quad \quad \quad+ \norm{\widetilde{P} \Pi^n_{\alpha,\rho^n}\rho^n \Pi^n_{\alpha,\rho^n} \widetilde{P} -\rho^n}_1\\
    &\leq 1-\Tr(\omega')+2^{-\alpha \sqrt{n}}+\norm{\widetilde{P} \Pi^n_{\alpha,\rho^n}\rho^n \Pi^n_{\alpha,\rho^n} \widetilde{P} -\rho^n}_1\\
    &\leq 1-\Tr(\omega')+2^{-\alpha \sqrt{n}} +2\sqrt{2\sqrt{\epsilon}+2\frac{\sqrt{\beta}}{\alpha}+\frac{\beta}{\alpha^2}+2^{-\alpha \sqrt{n}}} \\
    &= 2\sqrt{\epsilon}+2\frac{\sqrt{\beta}}{\alpha}+\frac{\beta}{\alpha^2}+2^{-\alpha \sqrt{n}+1}+2\sqrt{2\sqrt{\epsilon}+2\frac{\sqrt{\beta}}{\alpha}+\frac{\beta}{\alpha^2}+2^{-\alpha \sqrt{n}}},
\end{align*}
where the first line is due to triangle inequality. The second, third and fourth lines are 
due to Eqs. (\ref{eq: Tr of omega'}) and (\ref{eq:thrown away sum}), 
and Lemma \ref{Gentle Operator Lemma}, respectively.
\end{proof}

\begin{proof-of}[{of Theorem \ref{Asymptotic equivalence theorem}}]
We first prove the \textit{if} part. If there is an almost-commuting unitary $U$ and an ancillary system with the desired properties stated in the theorem, then we obtain
\begin{align}
    \frac{1}{n}\abs{S(\rho^n)-S(\sigma^n)}
    &\leq \frac{1}{n}\abs{S(\rho^n \otimes \omega')-S(\sigma^n\otimes \omega)}
    +\frac{1}{n}\abs{S(\omega')-S(\omega)} \nonumber \\
    &\leq \frac{1}{n}\abs{S(\rho^n\otimes \omega')-S(\sigma^n\otimes \omega)}
    +\frac{2}{n}\log 2^{o(n)} \nonumber \\
    &= \frac{1}{n}\abs{S(U(\rho^n\otimes \omega' )U^{\dagger})-S(\sigma^n\otimes \omega)}
    +o(1) \nonumber \\
    &\leq \frac{1}{n}  o(1) \log (d^{n} \times 2^{ o(n)})
    +\frac{1}{n}h\left(o(1)\right)+o(1)\nonumber \\
    &= o(1) \nonumber,
\end{align}
where the first line follows from additivity of the von Neumann entropy and triangle inequality. The second line is due to the fact that von Neumann entropy of a state is upper bounded by the logarithm of the dimension.
The penultimate line follows from continuity of von Neumann entropy \cite{Fannes1973,Audenaert2007}  where $h(x) = -x \log x - (1 - x) \log(1 - x)$ is the binary entropy function.   
Moreover, we obtain
\begin{align}\label{approximate average charge conservation}
    \frac{1}{n}&\abs{\Tr(\rho^n A_j^{(n)})-\Tr(\sigma^n A_j^{(n)})} \nonumber\\
    &=\frac{1}{n}\abs{\Tr\left( \rho^n \otimes \omega'  (A_j^{(n)}+A_j')\right)- \Tr\left( \sigma^n\otimes \omega  (A_j^{(n)}+A_j')\right) } \nonumber \\
    &\leq\frac{1}{n}\abs{\Tr\left( \rho^n \otimes \omega'  (A_j^{(n)}+A_j')\right)- \Tr\left( U\rho^n\otimes \omega'U^{\dagger}  (A_j^{(n)}+A_j')\right) }\nonumber\\
    &\quad \quad +\frac{1}{n}\abs{\Tr\left( U\rho^n\otimes \omega'U^{\dagger}  (A_j^{(n)}+A_j')\right)- \Tr\left( \sigma^n\otimes \omega  (A_j^{(n)}+A_j')\right) }\nonumber\\ 
    &=\frac{1}{n}\abs{\Tr\left( \rho^n \otimes \omega'  \left(A_j^{(n)}+A_j' -U^{\dagger}(A_j^{(n)}+A_j')U\right)\right) }\\
   &\quad \quad \quad  +\frac{1}{n}\abs{\Tr\left( \left(U\rho^n\otimes  \omega'U^{\dagger} - \sigma^n\otimes \omega \right)  (A_j^{(n)}+A_j')\right)} \nonumber\\ 
    &\leq  \frac{1}{n} \Tr(\rho^n \otimes \omega') \norm{U(A_j^{(n)}+A_j')U^{\dagger} - (A_j^{(n)}+A_j')}_{\infty}\\
    &\quad \quad \quad + \frac{1}{n} \norm{U\rho^n\otimes  \omega'U^{\dagger} - \sigma^n\otimes \omega}_1 \norm{A_j^{(n)}+A_j'}_{\infty} \nonumber\\
    & = o(1),  
\end{align}
the second line follows because $A_j'=0$ for all $j$. The third and fifth lines are due to  
triangle inequality and H\"{o}lder's inequality, respectively. 

\medskip
Now, we prove the \textit{only if} part. Assume for the sates $\rho^n$ and $\sigma^n$ the following holds:
\begin{align*}
   & \frac{1}{n} \abs{S(\rho^n)-S(\sigma^n)} \leq \gamma_n \\ 
   & \frac{1}{n} \abs{ \Tr(A_j^{(n)}\rho^n)-\Tr(A_j^{(n)}\sigma^n)}\leq \gamma'_n  \quad \quad j=1,\ldots,c,
\end{align*}
for vanishing $\gamma_n $ and $\gamma'_n $.
According to Theorem \ref{thm:symmetric-micro-exists}, for charges $A_j$, values $ v_j=\frac{1}{n} \Tr(\rho^n A_j^{(n)})$, $\eta'>0$ and any $n>0$, there is an a.m.c. subspace $\mathcal{M}$ of $\mathcal{H}^{\otimes n}$ with projector $P$ and the following parameters:
\begin{align*}
  &\eta=2 \eta',\\
  &\delta'=\frac{c+3}{2}(5n)^{2d^2} e^{-\frac{n \eta'^2}{8c^2(d
  +1)^2}}, \\ 
  &\delta=(c+3)(5n)^{2d^2} e^{-\frac{n \eta'^2}{8c^2(d
  +1)^2}},\\ 
  &\epsilon=2(c+3)(n+1)^{3d^2}(5n)^{2d^2} e^{-\frac{n \eta'^2}{8c^2(d
  +1)^2}}. 
\end{align*}
Choose $\eta'$ as the following such that $\delta$, $\delta'$ and $\epsilon$ vanish for large $n$:
\begin{align*}
\eta'=\left\{
                \begin{array}{ll}
                  \frac{\sqrt{8}c(d+1) }{n^{\frac{1}{4}} \Sigma(A)_{\min}} \quad &\text{if} \quad \gamma'_n \leq \frac{1}{n^{\frac{1}{4}}}\\
                  \frac{\sqrt{8}c(d+1)\gamma'_n}{ \Sigma(A)_{\min}} \quad &\text{if} \quad \gamma'_n > \frac{1}{n^{\frac{1}{4}}}
                \end{array}
              \right.
\end{align*}
where $\Sigma(A)_{\min}$ is the minimum spectral diameter among all spectral diameters of 
charges $\Sigma(A_j)$. Since $\frac{1}{n} \Tr(\rho^n A_j^{(n)})= v_j$ and $\abs{\frac{1}{n}\Tr(\sigma^n A_j^{(n)})- v_j} \leq \frac{1}{2}\eta' \Sigma (A_j)$, Corollary~\ref{corollary: a.m.c. projection} implies that states 
$\rho^n$ and $\sigma^n$ project onto this a.m.c. subspace with probability $\epsilon$:
\begin{align*}
    &\Tr(\rho^n P)\geq 1-\epsilon,\\
    &\Tr(\sigma^n P)\geq 1-\epsilon.
\end{align*}
Moreover, consider the typical projectors $\Pi^n_{\alpha ,\rho^n }$ and $\Pi^n_{\alpha ,\sigma^n }$ of states $\rho^n$ and $\sigma^n$, respectively, with $\alpha=n^{\frac{1}{3}}$. Then point 3 and 4 of Lemma~\ref{lemma: timmed state} implies that there are states $\widetilde{\rho}$ and $\widetilde{\sigma}$ with support inside the a.m.c. subspace $\mathcal{M}$ and unitaries $V_1$ and $V_2$ such that
\begin{align}\label{eq:4 formulas}
 &\norm{\widetilde{\rho}-\rho^n}_1 \leq o(1) , \nonumber\\
 &\norm{\widetilde{\sigma}-\sigma^n}_1 \leq o(1), \nonumber\\
 &V_1\widetilde{\rho}V_1^{\dagger}=\tau \otimes \omega, \nonumber\\
 &V_2\widetilde{\sigma}V_2^{\dagger}=\tau' \otimes \omega',
\end{align}
where $\tau$ and $\tau'$ are maximally mixed states; since $\abs{S(\rho^n)-S(\sigma^n)} \leq n \gamma_n$, one may choose the dimension of them in Eq.~(\ref{eq:L}) to be exactly the same as $L=2^{\floor{\sum_{i=1}^n S(\rho_i)-10 z}}$ with $z=\max \{\alpha \sqrt{n}, n \gamma_n\}$, hence, we obtain $\tau=\tau'$. Then, $\omega$ and $\omega'$ are states with support inside Hilbert space $\mathcal{K}$ of dimension $2^{o(z)}=2^{o(n)}$.
%
%
Then, it is immediate to see that the states $\widetilde{\rho} \otimes \omega'$ and $\widetilde{\sigma} \otimes \omega$ on Hilbert space $\mathcal{M}_t=\mathcal{M} \otimes \mathcal{K}$ have exactly the same spectrum; thus, there is a unitary $\widetilde{U}$ on  subspace $\mathcal{M}_t$ such that 
\begin{align}\label{eq:exact spectrum states}
    \widetilde{U} \widetilde{\rho} \otimes \omega' \widetilde{U}^{\dagger} =\widetilde{\sigma} \otimes \omega.
\end{align}
We extend the unitary $\widetilde{U}$ to $U=\widetilde{U} \oplus \1_{\mathcal{M}_t^{\perp}}$
acting on $\mathcal{H}^{\otimes n} \otimes \mathcal{K}$ and obtain
\begin{align*}
&\norm{U \rho^n  \otimes \omega' U^{\dagger} - \sigma^n \otimes \omega}_1 \\
&\quad \quad  \leq
\norm{U \rho^n \otimes \omega' U^{\dagger} - U \widetilde{\rho} \otimes \omega' U^{\dagger}}_1 +\norm{   \sigma^n \otimes \omega-\widetilde{\sigma}\otimes \omega}_1 
  +\norm{ U \widetilde{\rho} \otimes \omega' U^{\dagger}-\widetilde{\sigma}\otimes \omega}_1 \\
 & \quad  \quad =\norm{U \rho^n \otimes \omega' U^{\dagger} - U \widetilde{\rho} \otimes \omega' U^{\dagger}}_1 +\norm{   \sigma^n \otimes \omega-\widetilde{\sigma}\otimes \omega}_1 \\
 & \quad  \quad  \leq o(1),
\end{align*}
where the second and last lines are due to Eqs. (\ref{eq:exact spectrum states}) and (\ref{eq:4 formulas}), respectively.

As mentioned before, $\mathcal{M}_t=\mathcal{M} \otimes \mathcal{K}$ is a subspace of $\mathcal{H}^{\otimes n} \otimes \mathcal{K}$ with projector $P_t=P \otimes \1_{\mathcal{K}}$ where $P$ is the corresponding projector of a.m.c. subspace. We define total charges $A_j^t=A_j^{(n)}+A_j'$ and let $A_j'=0$ for all $j$ and show that every unitary of the form $U=U_{\mathcal{M}_t} \oplus \1_{\mathcal{M}_t^{\perp}}$  asymptotically commutes with all total charges: 
\begin{align*}
\norm{U A_j^t U^{\dagger}-A_j^t}_{\infty} &= \norm{(P_t+P_t^{\perp})(U A_j^t U^{\dagger}-A_j^t)(P_t+P_t^{\perp})}_{\infty} \\
& \leq \norm{P_t(U A_j^t U^{\dagger}-A_j^t)P_t}_{\infty}+\norm{P_t^{\perp}(U A_j^t U^{\dagger}-A_j^t)P_t}_{\infty} \\
& \quad +\norm{P_t(U A_j^t U^{\dagger}-A_j^t)P_t^{\perp}}_{\infty}+\norm{P_t^{\perp}(U A_j^t U^{\dagger}-A_j^t)P_t^{\perp}}_{\infty}\\
& = \norm{P_t(U A_j^t U^{\dagger}-A_j^t)P_t}_{\infty}+2\norm{P_t^{\perp}(U A_j^t U^{\dagger}-A_j^t)P_t}_{\infty} \\
& \leq 3  \norm{(U A_j^t U^{\dagger}-A_j^t)P_t}_{\infty}\\
& = 3  \norm{(U A_j^t U^{\dagger} -n v_j \1+ nv_j \1-A_j^t)P_t}_{\infty}\\
& \leq 3  \norm{(U A_j^t U^{\dagger} -n v_j \1)P_t}_{\infty}+3  \norm{(A_j^t- nv_j \1)P_t}_{\infty}\\
& = 6   \norm{(A_j^t- nv_j \1)P_t}_{\infty}\\
& = 6 \max_{\ket{v} \in \mathcal{M}_t} \norm{(A_j^t-n v_j \1)\ket{v}}_2\\
& = 6 \max_{\ket{v} \in \mathcal{M}_t} \norm{ (A_j^t-n v_j \1)(\Pi_j^{\eta} \otimes \1_{\mathcal{K}} +\1-\Pi_j^{\eta} \otimes \1_{\mathcal{K}}) \ket{v} 
 }_2  \\
& \leq 6 \max_{\ket{v} \in \mathcal{M}_t} \norm{(A_j^t-n v_j \1) \Pi_j^{\eta} \otimes \1 \ket{v}}_2 \\
& \quad \quad \quad+6 \max_{\ket{v} \in \mathcal{M}_t} \norm{(A_j^t-n v_j \1) (\1-\Pi_j^{\eta} \otimes \1) \ket{v}}_2  \\
& \leq  6n \Sigma (A_j) \eta +6 \max_{\ket{v} \in \mathcal{M}_t} \norm{(A_j^t-n v_j I) (\1-\Pi_j^{\eta} \otimes \1) \ket{v}}_2,
\end{align*}
where the first line is due to the fact that $P_t+P_t^{\perp}=\1_{\mathcal{H}^{\ox n}} \ox \1_{\mathcal{K}}$. The forth line follows because $U A_j^t U^{\dagger}-A_j^t$ is a Hermitian operator with zero eigenvalues in the subspace $P_t^{\perp}$. 
The fifth line is due to Lemma~\ref{lemma:norm inequality}. The twelfth line is due to the definition of the a.m.c. subspace. Now, bound the second term in the above:
\begin{align*}
&6 \max_{\ket{v} \in \mathcal{M}_t} \norm{(A_j^t-n v_j I) (\1-\Pi_j^{\eta} \otimes \1) \ket{v}}_2  \\
& \quad \quad \quad \leq  6 \max_{\ket{v} \in \mathcal{M}_t} \norm{A_j^t-n v_j \1 }_{\infty} \norm{ (\1-\Pi_j^{\eta} \otimes \1) \ket{v}}_2 \\
& \quad \quad \quad = 6 \norm{A_j^t-n v_j \1}_{\infty}  \max_{\ket{v} \in \mathcal{M}_t} \sqrt{ \Tr ((\1-\Pi_j^{\eta} \otimes \1)\ketbra{v}{v})}\\
& \quad \quad \quad = 6 n\norm{A_j- v_j \1}_{\infty} \max_{v \in \mathcal{M}} \sqrt{ \Tr ((\1-\Pi_j^{\eta} )v)}\\
&\quad \quad \quad \leq 6 n\norm{A_j- v_j \1}_{\infty} \sqrt{ \delta},
\end{align*}
the first line is due to Lemma~\ref{lemma:norm inequality}. 
The last line is by definition of the a.m.c. subspace. 
Thus, for vanishing $\delta$ and $\eta$ we obtain:
\begin{align*}
    \frac{1}{n}\norm{U A_j^t U^{\dagger}-A_j^t}_{\infty} \leq o(1),
\end{align*}
concluding the proof.
%
\end{proof-of}


\section{Discussion} 
\label{sec:discussion}
We have considered an asymptotic resource theory with states of tensor product structure 
as the objects and allowed operations which are thermodynamically meaningful, 
namely operations which preserve the entropy and and charges of a system asymptotically. 
The allowed operations classify the objects into asymptotically equivalent objects
that are interconvertible under allowed operations. The basic result on which 
our theory is built is that the objects are interconvertible via allowed operations 
if and only if they have the same average entropy and average charge values in the 
asymptotic limit. 

The existence of the allowed operations between the objects of the same class is based on two pillars:
First, for objects with the same average entropy there are states with sublinear dimension which 
can be coupled to the objects to make their spectrum asymptotically identical.
Second, objects with the same average charge values project onto a common subspace of the 
charges of the system which has the property that any unitary acting on this subspace is an almost-commuting unitary with the corresponding charges. Therefore, the spectrum of the objects of the same class can be modified using small ancillary systems and then they are interconvertible via unitaries that asymptotically preserve the charges of the system.
The notion of a common subspace for different charges, which are Hermitian operators, 
is introduced in \cite{Halpern2016} as approximate microcanonical (a.m.c.) subspace.
In this chapter, for given charges and parameters, we show the existence of an a.m.c. which is by construction a permutation-symmetry subspace, which is not guaranteed by the construction in \cite{Halpern2016}.

\chapter{Asymptotic thermodynamics of multiple conserved quantities}
\label{chap:thermo}

%

As a thermodynamic theory, or even as a resource theory in general, transformations
by almost-commuting unitaries, which we developed in the previous chapter, do not appear to be the most fruitful: they are reversible
and induce an equivalence relation among the sequences of product states. 
In particular, every point $(\und{a},s)$ of the phase diagram $\overline{\mathcal{P}}^{(1)}$
defines an equivalence class, namely of all state sequences with charges and 
entropy converging to $\und{a}$ and $s$, respectively. 

To make the theory more interesting, and more resembling of ordinary thermodyanmics, 
including irreversibility as expressed in its first and second laws,
we now specialise to a setting considered in many previous papers in the resource theory 
of thermodynamics, both with with single or multiple conserved quantities. 
Specifically, we consider an asymptotic analogue of the setting proposed 
in \cite{Guryanova2016} concerning the interaction of thermal baths with a 
quantum system and batteries, where it was shown that the second law constrains
the combination of extractable charge quantities. 
In \cite{Guryanova2016}, explicit protocols for state transformations to saturate 
the second law are presented, that store each of several commuting charges in its 
corresponding battery. However, for the case of non-commuting charges, one battery, 
or a so-called reference frame, stores all different types of charges \cite{Halpern2016,Popescu2018}.
Only recently it was shown that reference frames for non-commuting charges
can be constructed, at least under certain conditions, which store the different 
charge types in physically separated subsystems \cite{Popescu2019}.
Moreover, the size of the bath required to perform the transformations is not 
addressed in these works, as only the limit of asymptotically large bath was
considered. 
We will address these questions in a similar setting but in the asymptotic regime, 
where Theorem~\ref{Asymptotic equivalence theorem} provides the necessary and sufficient
condition for physically possible state transformations.
In this new setting, the \emph{asymptotic} second law constrains the combination of 
extractable charges; we provide explicit protocols for realising transformations 
satisfying the second law, where each battery can store its corresponding type of
work in the general case of non-commuting charges. Furthermore, we determine the 
minimum number of thermal baths of a given type that is required to perform a transformation.

\section{System model, batteries and the first law}
\label{subsec:model}
We consider a system being in contact with a bath and suitable batteries,
with a total Hilbert space 
$Q=S\otimes B\otimes W_1\otimes \cdots \otimes W_c$,
consisting of many non-interacting subsystems; namely, the work system, the thermal bath and 
$c$ battery systems with Hilbert spaces ${S}$, ${B}$ and ${W}_j$ for 
$j=1,\ldots,c$, respectively.  
We call the $j$-th battery system the $j$-type battery as it is designed to absorb
$j$-type work.
The work system and the thermal bath have respectively the charges $A_{S_j}$ and $A_{B_j}$ 
for all $j$, but $j$-type battery has only one nontrivial charge $A_{W_j}$, and all 
its other charges are zero because it is meant to store only the $j$-th charge. 
The total charge is the sum of the charges of the sub-systems $A_j=A_{S_j}+A_{B_j}+A_{W_j}$ 
for all $j$. Furthermore, for a charge $A$, let $\Sigma(A)=\lambda_{\max}(A)-\lambda_{\min}(A)$ 
denote the spectral diameter, where $\lambda_{\max}(A)$ and $\lambda_{\min}(A)$ are 
the largest and smallest eigenvalues of the charge $A$, respectively. 
We assume that the total spectral 
diameter of the work system and the thermal bath is bounded by the spectral diameter of the 
battery, that is $\Sigma(A_{S_j})+\Sigma(A_{B_j}) \leq \Sigma(A_{W_j})$ for all $j$; this 
assumption ensures that the batteries can absorb or release charges for transformations.   

As we discussed in the previous chapter, the generalized thermal state $\tau(\und{a})$
is the state that maximizes the entropy subject to the
constraint that the charges $A_j$ have the values $a_j$. 
This state is equal to $\frac{1}{Z}e^{-\sum_{j=1}^c \beta_j A_{j}}$ for real
numbers $\beta_j$ called inverse temperatures and chemical potentials; 
each of them is a smooth function of charge values $a_1,\ldots,a_c$, and 
$Z=\Tr e^{-\sum_{j=1}^c \beta_j A_{j}}$ is the generalized partition function. 
Therefore, the generalized thermal state can be equivalently denoted $\tau(\und{\beta})$ 
as a function of the inverse temperatures, associated uniquely with the charge
values $\und{a}$. 
We assume that the thermal bath is initially in a generalized thermal state
$\tau_b(\und{\beta})$, for globally fixed $\und{\beta}$. 
This is because in \cite{Halpern2016} it was argued that these are precisely the
\emph{completely passive} states, from which no energy can be extracted into 
a battery storing energy, while not changing any of the other conserved quantity, 
by means of almost-commuting unitaries and even when unlimited copies of the state are available. 
We assume that the work system with state $\rho_s$ and the thermal bath are initially 
uncorrelated, and furthermore that the battery systems can acquire only pure states.  

Therefore, the initial state of an \emph{individual} global system $Q$ 
is assumed to be of the following form,
\begin{equation}
  \label{eq:initial composite global}
  \rho_{SBW_1\ldots W_c} 
     = \rho_S \ox \tau(\und{\beta})_B \ox \proj{w_1}_{W_1} \ox \cdots \ox \proj{w_c}_{W_c},
\end{equation}
and the final states we consider are of the form 
\begin{equation}
  \label{eq:final composite global}
  \sigma_{SBW_1\ldots W_c} 
     = \sigma_{SB} \ox \proj{w_1'}_{W_1} \ox \cdots \ox \proj{w_c'}_{W_c},
\end{equation}
where $\rho_S$ and $\sigma_{SB}$ are states of the system and system-plus-bath, respectively, 
and $w_j$ and $w_j'$ label pure states of the $j$-type battery before and after the transformation. 
The notation is meant to convey the expectation value of the $j$-type work, i.e. $w_j^{(\prime)}$
is a real number and $\Tr \proj{w_j^{(\prime)}}A_{W_j} = w_j^{(\prime)}$.

The established resource theory of thermodynamics treats the batteries and the bath as 
`enablers' of transformations of the system $S$, and we will show first and second laws 
that express the essential constraints that any such transformation has to obey. 
We start with the batteries. With the notations $\und{W} = W_1\ldots W_c$, 
$\ket{\und{w}} = \ket{w_1}\cdots\ket{w_c}$, and $\ket{\und{w}'} = \ket{w_1'}\cdots\ket{w_c'}$,
let us look at a sequence $\rho^n = \rho_{S^n} = \rho_{S_1} \ox\cdots\ox \rho_{S_n}$ of initial
system states, and a sequence $\proj{\und{w}}^n = \proj{\und{w}_1}_{\und{W}_1} \ox\cdots\ox \proj{\und{w}_n}_{\und{W}_n}$
of initial battery states, recalling that the baths are initially all in the same thermal
state, $\tau_{B^n} = \tau(\und{\beta})^{\ox n}$; furthermore a sequence of target states 
$\sigma^n = \sigma_{S^nB^n} = \sigma_{S_1B_1} \ox\cdots\ox \sigma_{S_nB_n}$ of the system and bath, and a 
sequence $\proj{\und{w}'}^n = \proj{\und{w}_1'}_{\und{W}_1} \ox\cdots\ox \proj{\und{w}_n'}_{\und{W}_n}$
of target states of the batteries. 

\begin{definition}
  \label{definition:regular}
  A sequence of states $\rho^n$ on any system $Q^n$ is called \emph{regular} if 
  its charge and entropy rates converge, i.e. if
  \begin{align*}
    a_j &= \lim_{n\rightarrow\infty} \frac1n \Tr \rho^n A_j^{(n)},\ j=1,\ldots,c, \text{ and} \\
    s   &= \lim_{n\rightarrow\infty} \frac1n S(\rho^n) 
  \end{align*}
  exist. To indicate the dependence on the state sequence, 
  we write $a_j(\{\rho^n\})$ and $s(\{\rho^n\})$.
\end{definition}

\medskip
According to the AET and the other results of the previous chapter, every point $(\und{a},s)$
in the phase diagram $\overline{\mathcal{P}}^{(1)}$ labels an equivalence class of 
regular sequences of product states under transformations by almost-commuting unitaries. 

In the rest of the chapter we will essentially focus on regular sequences, so that
we can simply identify them, up to asymptotic equivalence, with a point in the phase 
diagram. However, it should be noted that at the expense of clumsier expressions, 
most of our expositions can be extended to arbitrary sequences of product states or 
block-product states. 

\medskip
Now, for regular sequences $\rho_{S^n}$ of initial states of the system and 
final states of the system plus bath, $\sigma_{S^nB^n}$, as well as regular
sequences of initial and final battery states, $\proj{\und{w}}^n$ and $\proj{\und{w}'}^n$, 
respectively, define the asymptotic rate of $j$-th charge change of the $j$-type battery as 
\begin{equation}
  \label{eq: W_j definition}
  \Delta A_{W_j} := a_j(\{\proj{w_j'}^n\})-a_j(\{\proj{w_j}^n\})
                  = \lim_{n\rightarrow\infty} \frac1n  \Tr (\proj{w_j'}^n-\proj{w_j}^n)A_{W_j}^{(n)}.
\end{equation}
Where there is no danger of confusion, we denote this number also as $W_j$, 
the $j$-type work extracted (if $W_j < 0$, this means that the work $-W_j$ is done 
on system $S$ and bath $B$). 

Similarly, we define the asymptotic rate of $j$-th charge change of the work system
and the bath as 
\begin{align*}
  \Delta A_{S_j} &:= a_j(\{\sigma_{S^n}\})-a_j(\{\rho_{S^n}\})
                   = \lim_{n\rightarrow\infty} \frac1n \Tr (\sigma_{S^n}-\rho_{S^n})A_{S_j}^{(n)}, \\ 
  \Delta A_{B_j} &:= a_j(\{\sigma_{B^n}\})-a_j(\{\tau(\und{\beta})_{B^n}\})
                   = \lim_{n\rightarrow\infty} \frac1n \Tr (\sigma_{B^n}-\tau(\und{\beta})_B^{\ox n})A_{B_j}^{(n)}, 
\end{align*}
where we denote $\sigma_{S^n} = \tr_{B^n} \sigma_{S^nB^n}$ and likewise 
$\sigma_{B^n} = \tr_{S^n} \sigma_{S^nB^n}$.

\begin{theorem}[First Law]
\label{thm:first-law}
Under the above notations, if the regular sequences 
$\rho_{S^nB^n\und{W}^n} = \rho_{S^n} \ox \tau(\und{\beta})_B^{\ox n} \ox \proj{\und{w}}^n$ 
and $\sigma_{S^nB^n\und{W}^n} = \sigma_{S^nB^n} \ox \proj{\und{w}'}^n$
are equivalent under almost-commuting unitaries, then 
\begin{align*}
  s(\{\sigma_{S^nB^n}\}) &= s(\{\rho_{S^n}\}) + S(\tau(\und{\beta})) \text{ and} \\ 
  W_j                    &= -\Delta A_{S_j}-\Delta A_{B_j} \text{ for all } j=1,\ldots,c. 
\end{align*}

Conversely, given regular sequences $\rho_{S^n}$ and $\sigma_{S^nB^n}$ of product
states such that 
\[
  s(\{\sigma_{S^nB^n}\}) = s(\{\rho_{S^n}\}) + S(\tau(\und{\beta})), 
\]
and assuming that the spectral radius of the battery observables $W_{A_j}$ is large enough 
(see the discussion at the start of this chapter), then there exist regular sequences of
product states of the $j$-type battery, $\proj{w_j}^n$ and $\proj{w_j'}^n$, for all
$j=1,\ldots,c$, such that 
\begin{align}
  \label{eq:initial}
  \rho_{S^nB^n\und{W}^n}   &= \rho_{S^n} \ox \tau(\und{\beta})_B^{\ox n} \ox \proj{\und{w}}^n \text{ and} \\
  \label{eq:final}
  \sigma_{S^nB^n\und{W}^n} &= \sigma_{S^nB^n} \ox \proj{\und{w}'}^n 
\end{align} 
can be transformed into each other by almost-commuting unitaries.
\end{theorem}

\begin{proof}
The first part is by definition, since the almost-commuting unitaries 
asymptotically preserve the entropy rate and the work rate of all 
charges. 

In the other direction, all we have to do is find states $\proj{w_j}$ 
and $\proj{w_j'}$ of the $j$-type battery $W_j$, such that
$W_j = \Delta A_{W_j} = -\Delta A_{S_j}-\Delta A_{B_j}$, for all $j=1,\ldots,c$. 
This is clearly possible if the spectral radius of $W_{A_j}$ is large enough. 
With this, the states in Eqs. (\ref{eq:initial}) and (\ref{eq:final}) have the 
same asymptotic entropy and charge rates. 
Hence, the claim follows from the AET, Theorem~\ref{Asymptotic equivalence theorem}.
\end{proof}

\begin{remark}\normalfont
The second part of Theorem~\ref{thm:first-law} says that for regular product state 
sequences, as long as the initial and final states of the work system and the thermal bath 
have asymptotically the same entropy, they can be transformed one into the another 
because there are always batteries that can absorb or release the necessary charge difference. 
Furthermore, we can even fix the initial (or final) state of the batteries and 
design the matching final (initial) battery state, assuming that the charge 
expectation value of the initial (final) state is far enough from the edge of the
spectrum of $A_{W_j}$.
\end{remark}

For any such states, we say that there is a \emph{work transformation} 
taking one to the other, denoted 
$\rho_{S^n} \ox \tau(\und{\beta})_B^{\ox n} \rightarrow \sigma_{S^nB^n}$.
This transformation is always feasible, implicitly assuming the 
presence of suitable batteries for all $j$-type works to balance to books
explicitly. 

\begin{remark}\normalfont
As a consequence of the previous remark, we now change our point of view of what a 
transformation is. Of our complicated $S$-$B$-$\und{W}$ compound, we only 
focus on $SB$ and its state, and treat the batteries as implicit. Since we insist
that batteries need to remain in a pure state, which thus factors off and 
does not contribute to the entropy, and due to the above first law Theorem \ref{thm:first-law}, 
we can indeed understand everything that is going on by looking at how 
$\rho_{S^nB^n}$ transforms into $\sigma_{S^nB^n}$. 
\end{remark}

Note that in this context, it is in a certain sense enough that the initial states 
$\rho_{S^n}$ form a regular sequence of product states and that the target states 
$\sigma_{S^nB^n}$ form a regular sequence. This is because the first part of 
the first law, Theorem \ref{thm:first-law}, only requires regularity, and 
since the target state defines a unique point $(\und{a}',s')$ in the phase 
diagram, we can find a sequence of product states $\widetilde{\sigma}_{S^nB^n}$ in 
its equivalence class, and use the second part of Theorem \ref{thm:first-law}
to realise the work transformation 
$\rho_{S^n} \ox \tau(\und{\beta})_B^{\ox n} \rightarrow \widetilde{\sigma}_{S^nB^n}$.

\section{The second law}
\label{subsec:secondlaw}
If the first law in our framework arises from focusing on the system-plus-bath 
compound $SB$, while making the batteries implicit, the second law comes about 
from trying to understand the action on the work system $S$ alone, through the
concomitant back-action on the bath $B$. 
Following \cite{Guryanova2016,Halpern2016}, the second law constrains the different 
combinations of commuting conserved quantities that can be extracted from the work 
system. We show here that in the asymptotic regime, the second law similarly bounds 
the extractable work rate via the rate of free entropy of the system. 

The \emph{free entropy} for a system with state $\rho$, charges $A_j$ and inverse temperatures 
$\beta_j$ is defined in \cite{Guryanova2016} as
\begin{align}
  \label{free entropy}
  \widetilde{F}(\rho) = \sum_{j=1}^c \beta_j \Tr \rho A_j - S(\rho).
\end{align}
It is shown in \cite{Guryanova2016} that the generalized thermal state 
$\tau(\und{\beta})$ is the state that minimizes the free entropy for 
fixed $\beta_j$.

For any work transformation 
$\rho_{S^n} \ox \tau(\und{\beta})_B^{\ox n} \rightarrow \sigma_{S^nB^n}$
between regular sequences of states, 
we define the asymptotic rate of free entropy change for the work system and the 
thermal bath respectively as follows:
\begin{equation}\begin{split}
  \label{eq:free entropy rates}
  \Delta\widetilde{F}_S 
    &:= \lim_{n \to \infty} \frac{1}{n}\left(\widetilde{F}(\sigma_{S^n})
                                             -\widetilde{F}(\rho_{S^n}) \right), \\
  \Delta\widetilde{F}_B
    &:= \lim_{n \to \infty} \frac{1}{n}\left(\widetilde{F}(\sigma_{B^n})
                                             -n \widetilde{F}(\tau_B)\right),
\end{split}\end{equation}
where the free entropy is with respect to the charges of the work system and the thermal 
bath with fixed inverse temperatures $\beta_j$.

\begin{figure}[ht]
\begin{center}
  \includegraphics[width=10cm,height=8cm]{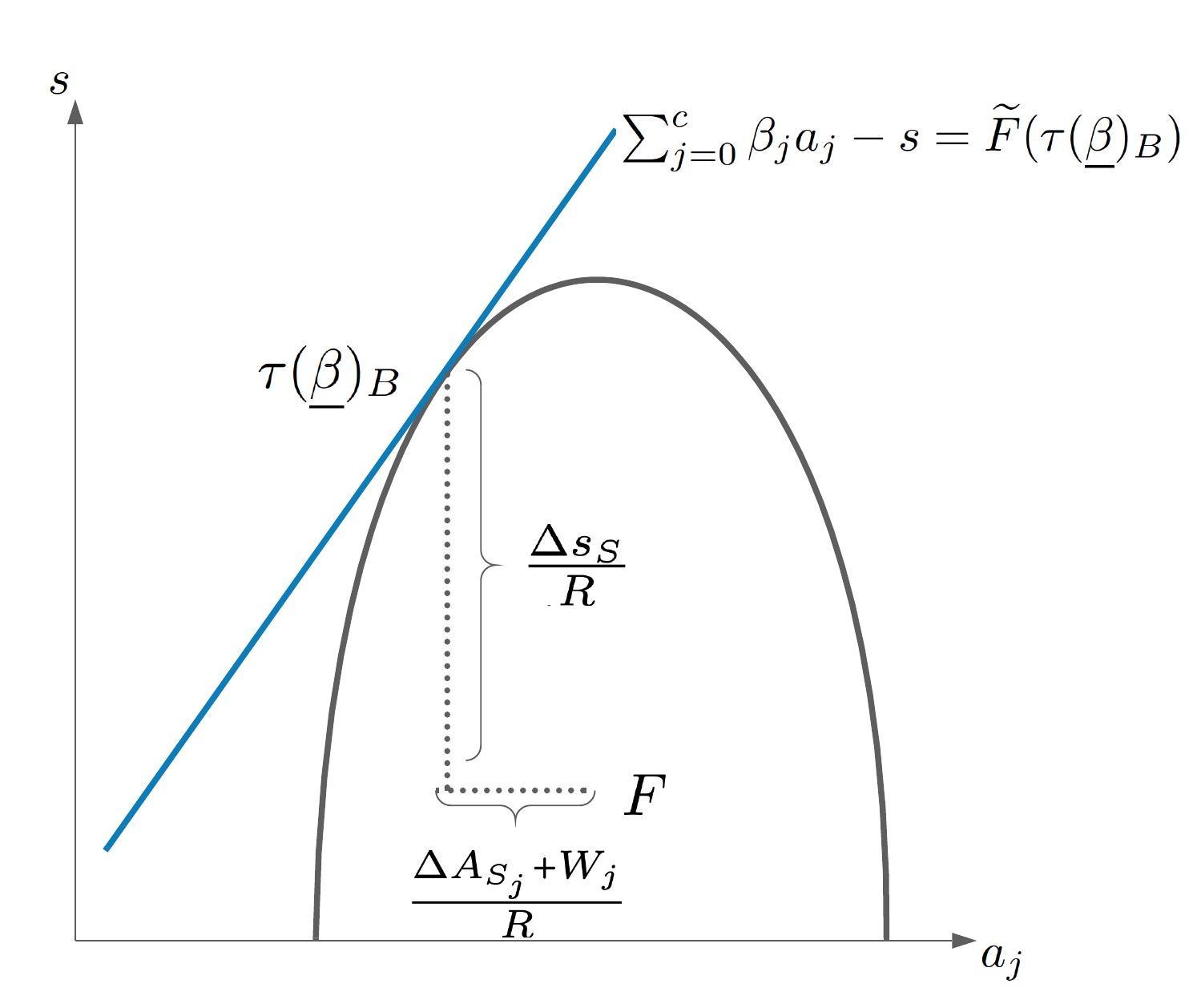}
 \end{center}
  \caption{State change of the bath for a given work transformation under extraction of 
           $j$-type work $W_j$, viewed in the phase diagram of the bath $\overline{\cP}_B$. 
           The blue line represents the tangent hyperplane at the corresponding point 
           of the generalized thermal state $\tau(\und{\beta})_B$, $R$ is the number of copies 
           of the elementary baths in the proof of Theorem \ref{asymptotic second law}, 
           and $F$ is the point corresponding to the final state of the bath.}
  \label{fig:second-law}
\end{figure}

\begin{theorem}[Second Law]
\label{asymptotic second law}
For any work transformation 
$\rho_{S^n} \ox \tau(\und{\beta})_B^{\ox n} \rightarrow \sigma_{S^nB^n}$
between regular sequences of states, the $j$-type works $W_j$ that are extracted 
(and they are necessarily $W_j = -\Delta A_{S_j}-\Delta A_{B_j}$ according 
to the first law) are constrained by the rate of free entropy change of the system:
\[
   \sum_{j=1}^c \beta_j W_j  \leq -\Delta\widetilde{F}_S.
\]

Conversely, for arbitrary regular sequences of product states, 
$\rho_{S^n}$ and $\sigma_{S^n}$, and any real numbers $W_j$ with  
$\sum_{j=1}^c \beta_j W_j < -\Delta\widetilde{F}_S$, 
there exists a bath system $B$ and a regular sequence of product states 
$\sigma_{S^nB^n}$ with $\Tr_{B^n}\sigma_{S^nB^n} = \sigma_{S^n}$, such that 
there is a work transformation 
$\rho_{S^n} \ox \tau(\und{\beta})_B^{\ox n} \rightarrow \sigma_{S^nB^n}$
with accompanying extraction of $j$-type work at rate $W_j$. 
This is illustrated in Fig.~\ref{fig:second-law}.
\end{theorem}

\begin{proof}
We start with the first statement of the theorem. Consider the global system transformation 
$\rho_{S^n} \ox \tau(\und{\beta})_B^{\ox n} \rightarrow \sigma_{S^nB^n}$ by 
almost-commuting unitaries. 
We use the definition of work (\ref{eq: W_j definition}) and free 
entropy (\ref{free entropy}), as well as the first law, Theorem \ref{thm:first-law}, 
to get
\begin{equation}
\label{work expansion formula}
\begin{split}
    \sum_j \beta_j W_j &= -\sum_j \beta_j(\Delta A_{S_j}+\Delta A_{B_j})\\
                       &= -\Delta\widetilde{F}_S-\Delta\widetilde{F}_B-\Delta s_S -\Delta s_B .
\end{split}
\end{equation}
The second line is due to the definition in Eq. (\ref{eq:free entropy rates}). 
Now observe that 
\begin{align}\label{eq: positive Delta_SB}
  \Delta s_S +\Delta s_B 
     &=     \lim_{n \to \infty} \frac1n \bigl(S(\sigma_{S^n})-S(\rho_{S^n})\bigr) 
                                + \frac1n \bigl(S(\sigma_{B^n})-nS(\tau(\und{\beta})_B)\bigr) \nonumber\\
     &\geq  \lim_{n \to \infty} \frac1n \bigl(S(\sigma_{{SB}^n})-S(\rho_{S^n})-S(\tau(\und{\beta})_B^{\ox n})\bigr)
            = 0, 
\end{align}
where the inequality is due to sub-additivity of von Neumann entropy, and the 
final equation due to asymptotic entropy conservation. 
Further, the generalized thermal state $\tau(\und{\beta})_B$ has 
the minimum free entropy \cite{Guryanova2016}, hence $\Delta\widetilde{F}_B \geq 0$.

For the second statement of the theorem, the achievability part of the second law, 
we aim to show that there is a work transformation 
$\rho_{S^n} \ox \tau(\und{\beta})_B^{\ox n} \rightarrow \sigma_{S^n} \ox \xi_{B^n}$, 
with a suitable regular sequences of product states, 
and works $W_1,\ldots,W_c$ are extracted. 
This will be guaranteed, by the first law, Theorem \ref{thm:first-law}, 
and the AET, Theorem \ref{Asymptotic equivalence theorem}, if 
\begin{equation}\begin{split}
  s(\{\xi_{B^n}\})   &= S(\tau(\und{\beta})_B)  - \Delta s_S, \\
  a_j(\{\xi_{B^n}\}) &= \Tr \tau(\und{\beta})_B A_{B_j} - \Delta A_{S_j} - W_j \quad \text{for all } j=1,\ldots,c.
  \label{eq:state assumptions}
\end{split}\end{equation}
The left hand side here defines a point $(\und{a},s)$ in the charges-entropy 
space of the bath, and our task is to show that it lies in the phase diagram,
for which purpose we have to define the bath characteristics suitably. 
On the right hand side, 
$\bigl(\Tr\tau(\und{\beta})_B A_{B_1}, \ldots, \Tr\tau(\und{\beta})_B A_{B_c}, S(\tau(\und{\beta})_B)\bigr)$
is the point corresponding to the initial state of the bath, which due to 
its thermal nature is situated on the upper boundary of the region. 
At that point, the region has a unique tangent hyperplane, which has the equation 
$\sum_j \beta_j a_j-s = \widetilde{F}(\tau(\und{\beta})_B)$, and the phase diagram 
is contained in the half space $\sum_j \beta_j a_j-s \geq \widetilde{F}(\tau(\und{\beta})_B)$, 
corresponding to the fact that their free entropy is larger than that of the
thermal state. In fact, due to the strict concavity of the entropy, and hence 
of the upper boundary of the phase diagram, the phase diagram, with the exception 
of the thermal point $\bigl(\Tr \tau(\und{\beta})_B \underline{A_{B}}, S(\tau(\und{\beta})_B)\bigr)$ 
is contained in the open half space $\sum_j \beta_j a_j-s > \widetilde{F}(\tau(\und{\beta})_B)$. 

One of many ways to construct a suitable bath $B$ is as several ($R\gg 1$) non-interacting 
copies of an ``elementary bath'' $b$: $B=b^R$ and charges $A_{B_j}=A^{(R)}_{b_j}$, so that 
the GGS of $B$ is $\tau(\und{\beta})_B = \tau(\und{\beta})_b^{\otimes R}$. 
We claim that for large enough $R$, the left hand side of Eq. (\ref{eq:state assumptions}) 
defines a point in the phase diagram of $B$. Indeed, we can express the conditions 
in terms of $b$, assuming that we aim for a regular sequence of product states
$\xi_{b^{nR}}$:
\begin{equation}\begin{split}
  s(\{\xi_{b^{nR}}\})   &= S(\tau(\und{\beta})_b) - \frac1R \Delta s_S, \\
  a_j(\{\xi_{b^{nR}}\}) &= \Tr \tau(\und{\beta})_b A_{b_j} - \frac1R (\Delta A_{S_j} + W_j) 
                                                             \quad \text{for all } j=1,\ldots,c.
  \label{eq:R-bath-assumptions}
\end{split}\end{equation}
For all sufficiently large $R$, these points $(\und{a},s)$ are arbitrarily close to
where the bath starts off, at 
$(\und{a}_{\und{\beta}},s_{\und{\beta}}) 
 = \bigl(\Tr\tau(\und{\beta})_b A_{b_1}, \ldots, \Tr\tau(\und{\beta})_b A_{b_c}, S(\tau(\und{\beta})_b)\bigr)$,
while they always remains in the open half plane 
$\sum_j \beta_j a_j-s > \widetilde{F}(\tau(\und{\beta})_b)$. Indeed, 
they all lie on a straight line pointing from $(\und{a}_{\und{\beta}},s_{\und{\beta}})$
into the interior of that half plane. Hence, for sufficiently large $R$,
$(\und{a},s) \in \overline{\cP}$, the phase diagram of $b$, and by 
point 5 of Lemma~\ref{lemma:phase diagram properties} there does indeed exist
a regular sequence of product states corresponding to it.
\end{proof}

\section{Finiteness of the bath: tighter constraints and negative entropy}
\label{subsec:finitebath}
In the previous two sections we have elucidated the traditional statements of 
the first and second law of thermodynamics, as emerging in our resource theory. 
In particular, the second law is tight, if sufficiently large baths are allowed 
to be used. 

Here, we specifically look at the the second statement (achievability) 
of the second law in the presence of an explicitly given, finite bath $B$. It will 
turn out that typically, equality in the second law cannot be attained, only 
up to a certain loss due to the finiteness of the bath. We also discover 
a purely quantum effect whereby the system and the bath remain entangled after 
effecting a certain state transformation, allowing quantum engines to perform 
tasks impossible classically (i.e. with separable correlations). 
The question we want to address is the following refinement of the one answered 
in the previous section: 

\begin{quote}
Given regular sequences $\rho_{S^n}$ and $\sigma_{S^n}$ of product states, 
and numbers $W_j$, are there extensions $\sigma_{S^nB^n}$ of $\sigma_{S^n}$ 
forming a regular sequence of product states, such that the work 
transformation $\rho_{S^n} \ox \tau(\und{\beta})_B^{\ox n} \rightarrow \sigma_{S^nB^n}$
is feasible, with accompanying extraction of $j$-type work at rate $W_j$?
\end{quote}

To answer it, we need the following \emph{extended phase diagram}.
For a give state $\sigma_S$ of the system $S$, and a bath $B$, define the 
the following set:
\begin{equation}
  \cP^{(1)}_{|\sigma_S} := \left\{ \bigl(\Tr\xi_B A_1^{(B)},\ldots,\Tr\xi_B A_c^{(B)},S(B|S)_\xi\bigr) : 
                                   \xi_{SB} \text{ state with } \Tr_B\xi_{SB}=\sigma_S \right\},
\end{equation}
furthermore its $n$-copy version
\begin{align}
 &\cP^{(n)}_{|\sigma_{S^n}} 
    := \nonumber
    \\ & \left\{\! \bigl(\Tr\xi_{B^n} A_1^{(B^n)}\!,\!\ldots\!,\!\Tr\xi_{B^n} A_c^{(B^n)}\!\!,S(B^n|S^n)_\xi\bigr)\! : 
               \xi_{S^nB^n} \text{ state with } \Tr_{B^n}\xi_{S^nB^n}\!=\sigma_S^{\otimes n} \!\right\}\!.
\end{align}
Finally, define the \emph{conditional entropy phase diagram} as 
\begin{align}
&  \overline{\cP}_{|s_0} := \overline{\cP}^{(1)}_{|s_0}
    := \nonumber\\
&\quad    \left\{ \bigl(\und{a},s\bigr) : a_j = \Tr\xi_B A_j^{(B)},\,
                                       -\min\{s_0,S(\tau(\und{a}))\} \leq s \leq S(\tau(\und{a})) 
                                       \text{ for a state } \xi_B \right\},
\end{align}
and likewise its $n$-copy version $\overline{\cP}^{(n)}_{|ns_0}$, 
for a number $s$ (intended to be an entropy or entropy rate). 
These concepts are illustrated in Fig.~\ref{fig:extended-phase-diagram}.
The relation between the sets, and the name of the latter, are explained in the 
following lemma.

\begin{figure}[ht]
\begin{center}
  \includegraphics[width=10cm,height=8cm]{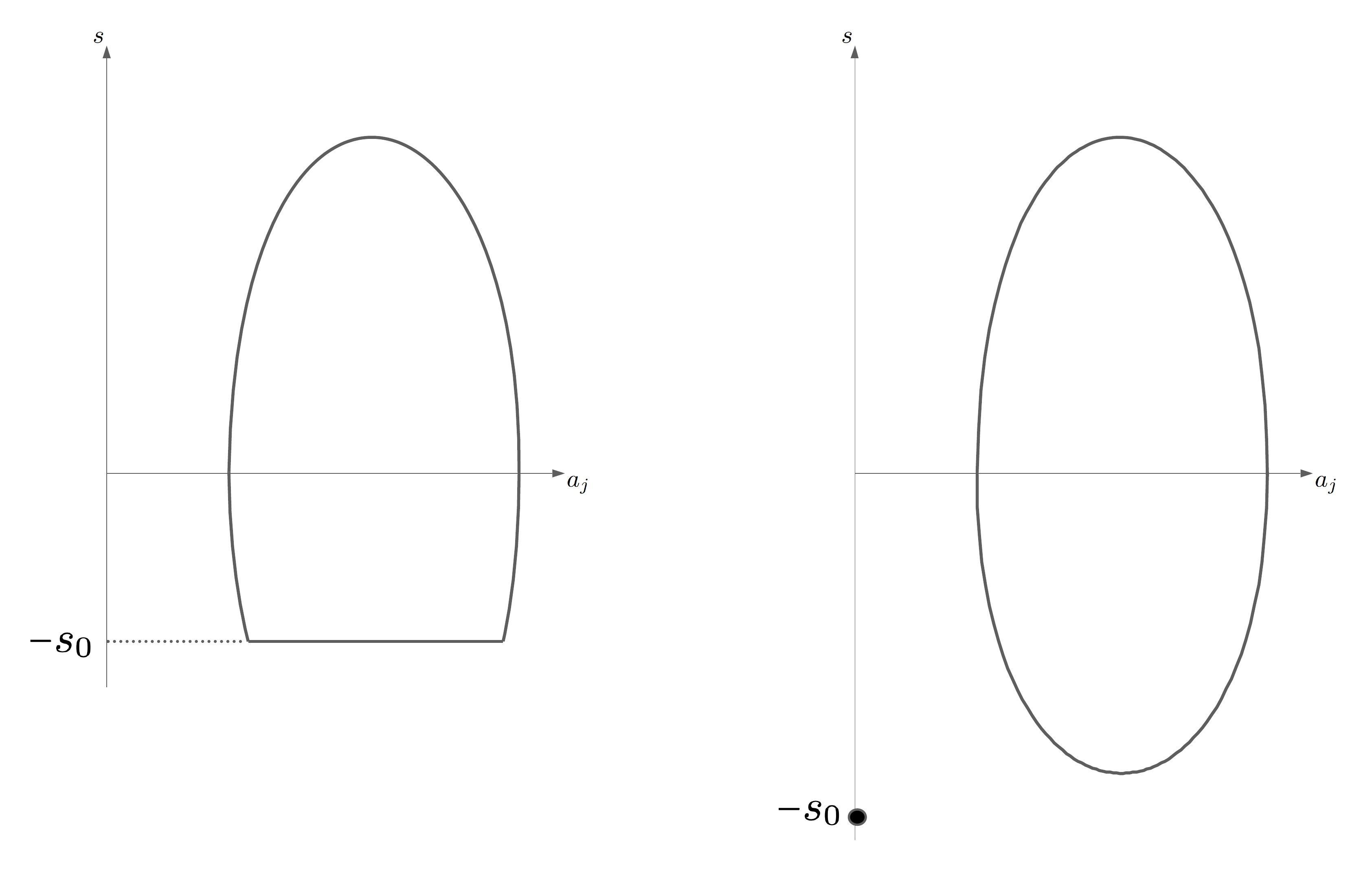}
  \end{center}
  \caption{Schematic of the extended phase diagram $\overline{\cP}_{|s_0}$. 
           Depending on the value of $s_0$, whether it is smaller or larger than 
           $\log|B|$, the diagram acquires either the left hand or the right hand 
           one of the above shapes.}
  \label{fig:extended-phase-diagram}
\end{figure}

\begin{lemma}
\label{lemma:extended-phase-diagram}
With the previous notation, we have:
\begin{enumerate}
  \item For all $k$, $\cP^{(k)}_{|\sigma_{S^k}} \subset \overline{\cP}^{(k)}_{|S(\sigma_{S^k})}$, 
        and the latter is a closed convex set.
  \item For all $k$, $\overline{\cP}^{(k)}_{|ks_0} = k \overline{\cP}^{(1)}_{|s_0}$.
  \item For a regular sequence $\{\sigma_{S^k}\}$ of product states with entropy rate
        $s_0=s(\{\sigma_{S^k}\})$, every point in $\overline{\cP}_{|s}$ is arbitrarily well 
        approximated by points in $\frac1k \cP^{(k)}_{|\sigma_{S^k}}$ for all sufficiently large $k$. 
        I.e., $\displaystyle{\overline{\cP}_{|s_0} = \lim_{k\to\infty} \frac1k \cP^{(k)}_{|S(\sigma_{S^k})}}$.        
\end{enumerate}
\end{lemma}

\begin{proof}
1. We only have to convince ourselves that for a state
$\xi_{S^kB^k}$ with $\Tr_{B^k}\xi_{S^kB^k}=\sigma_{S^k}$, 
\[
  - \min\{S(\sigma_{S^k}),kS(\tau(\und{a}))\} \leq S(B^k|S^k)_\xi \leq k S(\tau(\und{a})), 
\]
where $\und{a}=(a_1,\ldots,a_c)$ with $a_i = \frac1k \Tr\xi_{B^k} A_i^{(B^k)}$.
The upper bound follows from subadditivity, since 
$S(B^k|S^k)_\xi \leq S(B^k)_\xi \leq k S(\tau(\und{a}))$.
The lower bound consists of two inequalities: first, by purifying $\xi$ to 
a state $\ket{\phi} \in S^kB^kR$ and strong subadditivity, 
$S(B^k|S^k)_\xi \geq S(B^k|S^kR)_\phi = - S(B^k)_\xi \geq - k S(\tau(\und{a}))$.
Secondly, $S(B^k|S^k)_\xi \geq -S(S^k)_\xi = -S(\sigma_{S^k})$.

2. Follows easily from the definition. 

3. It is enough to show that the points of the minimum entropy diagram 
\[
  \overline{\cP}_{\min|s} 
    := \left\{\bigl(\und{a},-\min\{s_0,S(\tau(\und{a}))\}\bigr) 
              : \Tr \xi_B A_j^{(B)} = a_j \text{ for a state } \xi_B \right\}
\]
can be approximated as claimed by an admissible $k$-copy state $\xi_{S^kB^k}$. This 
is because the maximum entropy diagram $\overline{\cP}_{\max}^{(k)}$ is realized 
by states $\vartheta_{S^kB^k} := \sigma_{S^k}\ox\tau(\und{a})_B^{\ox k}$, and by 
interpolating the states, i.e. $\lambda\xi + (1-\lambda)\vartheta$ for $0\leq\lambda\leq 1$,
we can realize the same charge values $\und{a}$ with entropies in the 
whole interval $[S(B^k|S^k)_\xi;k S(\tau(\und{a}))]$.

The approximation of $\overline{\cP}_{\min|s}$ can be proved invoking 
results from quantum Shannon theory, specifically quantum state merging, 
the form of which we need here is stated below as a Lemma. 
For this, consider a tuple $\und{a} \in \overline{\cP}_0$ and a purification 
$\ket{\Psi} \in S^kB^kR^k$ of the state $\vartheta_{S^kB^k} = \sigma_{S^k}\ox\tau(\und{a})_B^{\ox k}$, 
which can be chosen in such a way as to be a product state itself: 
$\ket{\Psi} = \ket{\Psi_1}_{S_1B_1R_1}\ox\cdots\ox\ket{\Psi_k}_{S_kB_kR_k}$. 
Now we distinguish two cases, depending on which of the entropies 
$S(\sigma_{S^k})$ and $kS\bigl(\tau(\und{a})_B\bigr)$ is the smaller.

\begin{enumerate}[{(i)}]
  \item $S(\sigma_{S^k}) \geq S\bigl(\tau(\und{a})_B\bigr)$: We shall 
    construct $\xi_{S^kB^k}$ in such a way that $\xi_{S^k} = \sigma_{S^k}$ and 
    $\xi_{B^k} \approx \tau\bigl(\und{a}\bigr)_B^{\otimes k}$. To this end, choose 
    a pure state $\phi_{CR'}$ with entanglement entropy 
    $S(\phi_C) = \frac1k S(\sigma_{S^k})-S\bigl(\tau(\und{a})_B\bigr) + \frac12 \epsilon$,
    and consider the state 
    $\widetilde{\Psi}^{S^kB^kC^kR^k{R'}^k} = \Psi_{S^kB^kR^k}\ox\phi_{CR'}^{\ox k}$.
    Now we apply state merging (Lemma \ref{lemma:marge}) twice to 
    this state (which is a tensor product of $k$ systems), with a random 
    rank-one projector $P$ on the combined system $R^k{R'}^k$: 
    first, by splitting the remaining parties $S^k : B^kC^k$, 
    and second by splitting them $B^k : S^kC^k$. 
    By construction, in both bipartitions it is the solitary system 
    ($S^k$ and $B^k$, resp.) that has the smaller entropy by at least $\frac12\epsilon k$, 
    showing that the post-measurement state $\widetilde{\xi}(P)_{S^kB^kC^k}$ with 
    high probability approximates the marginals of $\vartheta_{S^kB^k}$ on $S^k$ 
    and on $B^k$ simultaneously.
    Choose a typical subspace projector $\Pi$ of $\phi_C^{\ox k}$ with 
    $\log \rank \Pi \leq S(\sigma_{S^k})-k S\bigl(\tau(\und{a})_B\bigr) + \epsilon k$,
    and let
    \[
      \ket{\xi(P)}_{S^kB^kC^k} := \frac1c (\1_{S^kB^k}\Pi_{C^k})\ket{\widetilde{\xi}(P)},
    \]
    with a normalization constant $c$. 
    Merging and properties of the typical subspace imply that for sufficiently large $k$, 
    \begin{align}
      \label{eq:xiP-S}
      \frac12 \left\| \xi(P)_{S^k} - \sigma_{S^k} \right\|_1            &\leq \epsilon, \\
      \label{eq:xiP-B}
      \frac12 \left\| \xi(P)_{B^k} - \tau(\und{a})_B^{\ox k} \right\|_1 &\leq \epsilon.
    \end{align}
    Now, we invoke Uhlmann's theorem applied to purifications of $\sigma_{S^k}$ and  
    of $\xi(P)_{S^kB^k}$, together with the well-known relations between fidelity 
    and trace norm applied to Eq.~(\ref{eq:xiP-S}), 
    to obtain a state $\xi_{S^kB^k}$ with $\xi_{S^k} = \sigma_{S^k}$ and 
    $\frac12 \left\| \xi(P)_{S^kB^k} - \xi_{S^kB^k} \right\|_1 \leq \sqrt{\epsilon(2-\epsilon)}$,
    thus by Eq.~(\ref{eq:xiP-B})
    \[
      \frac12 \left\| \xi_{B^k} - \tau(\und{a})_B^{\ox k} \right\|_1 \leq \epsilon + \sqrt{\epsilon(2-\epsilon)}.
    \]
    From the latter bound it follows that
    \[
      \left| \frac1k \tr \xi_{B^k}A_j^{(B^k)} - a_j \right| 
               \leq \|A_{B_j}\| \left(\epsilon + \sqrt{\epsilon(2-\epsilon)}\right).
    \]
    It remains to bound the conditional entropy: 
    \[\begin{split}
      \frac1k& S(B^k|S^k)_\xi \\
      &=    \frac1k S\bigl(\xi_{S^kB^k}\bigr) - \frac1k S(\xi_{S^k})           \\
                         &\leq \frac1k S\bigl(\xi(P)_{S^kB^k}\bigr) \!- \!\frac1k S(\sigma_{S^k}) 
                                  \!+\! \left(\! \epsilon \!+\! \sqrt{\!\epsilon(2\!-\!\epsilon)\!} \right)\log(|S||B|) 
                                  \!+\! h\left(\! \epsilon \!+\! \sqrt{\!\epsilon(2\!-\!\epsilon)\!} \right)           \\
                         &\leq \frac1k \log \rank\Pi \!-\! \frac1k S(\sigma_{S^k}) 
                                  \!+\! \left(\! \epsilon \!+ \!\sqrt{\!\epsilon(2\!-\!\epsilon)\!} \right)\log(|S||B|) 
                                  \!+\! h\left( \!\epsilon \!+\! \sqrt{\!\epsilon(2-\epsilon)\!} \!\right)           \\
                         &\leq \frac1k \left( S(\sigma_{S^k}) - kS\bigl(\tau(\und{a})\bigr) \right) 
                                  - \frac1k S(\sigma_{S^k}) 
                                  + \left( 2\epsilon + \sqrt{\epsilon(2-\epsilon)} \right)\log(|S||B|)\\ 
                               &\quad \quad \quad \quad \quad \quad \quad   + h\left( \epsilon + \sqrt{\epsilon(2-\epsilon)} \right)            \\
                         &=    -S\bigl(\tau(\und{a})\bigr) 
                                  + \left( 2\epsilon + \sqrt{\epsilon(2-\epsilon)} \right)\log(|S||B|) 
                                  + h\left( \epsilon + \sqrt{\epsilon(2-\epsilon)} \right) ,
    \end{split}\]
    where in the second line we have used the Fannes inequality on the continuity 
    of the entropy \cite{Fannes1973,Audenaert2007}, with the binary entropy
    $h(x)=-x\log x-(1-x)\log(1-x)$;
    in the third line that $\xi(P)_{S^kB^k}$ has rank at most $\rank \Pi$;
    and in the fourth line the upper bound on the latter rank by construction. 

 \item $S(\sigma_{S^k}) < S\bigl(\tau(\und{a})_B \bigr)$: We shall 
    construct $\xi_{S^kB^k}$ such that $\xi_{S^k} = \sigma_{S^k}$ and 
    $\tr\xi_{B^k}A_j^{(B^k)} \approx \tr\tau\bigl(\und{a}\bigr)_B A_{B_j}$ for 
    all $j=1,\ldots,c$. Here, choose a pure state $\phi_{CR'}$ with entanglement entropy 
    $S(\phi_C) = \epsilon$, and define
    $\widetilde{\Psi}^{S^kB^kC^kR^k{R'}^k} = \Psi_{S^kB^kR^k}\ox\phi_{CR'}^{\ox k}$.
    Now we apply state merging (Lemma \ref{lemma:marge}) to
    this state (which is a tensor product of $k$ systems), with a random 
    rank-one projector $P$ on the combined system $R^k{R'}^k$,
    by splitting the remaining parties $S^k : B^kC^k$, which ensures that
    $S^k$ has the smaller entropy by at least $\epsilon k$, 
    showing that the post-measurement state $\widetilde{\xi}(P)_{S^kB^kC^k}$ with 
    high probability approximates the marginal of $\vartheta_{S^kB^k}$ on $S^k$.
    Proceed as before with a typical subspace projector $\Pi$ of $\phi_C^{\ox k}$ 
    such that  
    $\log \rank \Pi \leq S(\sigma_{S^k})-k S\bigl(\tau(\und{a})_B\bigr) + \epsilon k$,
    and let
    \(
      \ket{\xi(P)}_{S^kB^kC^k} := \frac1c (\1_{S^kB^k}\Pi_{C^k})\ket{\widetilde{\xi}(P)},
    \)
    with a normalization constant $c$. 
    Merging and properties of the typical subspace thus imply that for sufficiently large $k$, 
    \begin{equation}
      \label{eq:xiPt-S}
      \frac12 \left\| \xi(P)_{S^k} - \sigma_{S^k} \right\|_1 \leq \epsilon. 
    \end{equation}
    Next we need to look at the charge values of $\xi(P)_{B^k}$. 
    Note that the expectation $\EE_P \xi(P)_{B^k}$ is approximately 
    equal to $\EE_P \widetilde{\xi}(P)_{B^k} = \tau(\und{a})_B^{\ox k}$.
    It follows from \cite[{Lemma III.5}]{Hayden2006}, that if $k$ is sufficiently 
    large, then with high probability 
    \begin{equation}
      \label{eq:xiPt-B-A}
      \left| \tr \bigl(\xi(P)_{B^k} - \tau(\und{a})_B^{\ox k}\bigr) A_j^{(B^k)} \right| \leq \|A_{B_j}\| \epsilon 
                      \quad \text{for all } j=1,\ldots,c. 
    \end{equation}
    So we just focus on a good instance of $P$, where both 
    Eqs.~(\ref{eq:xiPt-S}) and (\ref{eq:xiPt-B-A}) hold. Now we proceed as
    in the first case to find a state $\xi_{S^kB^k}$ with $\xi_{S^k} = \sigma_{S^k}$ and 
    $\frac12 \left\| \xi(P)_{S^kB^k} - \xi_{S^kB^k} \right\|_1 \leq \sqrt{\epsilon(2-\epsilon)}$, 
    using Uhlmann's theorem. Thus, as before we find
    \[
      \left| \frac1k \tr \xi_{B^k}A_j^{(B^k)} - a_j \right| 
                 \leq \|A_{B_j}\| \left(\epsilon + \sqrt{\epsilon(2-\epsilon)}\right).
    \]
    Regarding the conditional entropy, we have quite similarly as before,
    \[\begin{split}
      \frac1k &S(B^k|S^k)_\xi \\
      &=    \frac1k S\bigl(\xi_{S^kB^k}\bigr) - \frac1k S(\xi_{S^k})           \\
                             &\leq \frac1k S\bigl(\xi(P)_{S^kB^k}\bigr) \!- \!\frac1k S(\sigma_{S^k}) 
                                  \!+\! \left( \!\epsilon \!+ \!\sqrt{\!\epsilon(2\!-\!\epsilon)\!} \right)\log(|S||B|) 
                                  \!+\! h\left(\! \epsilon \!+\! \sqrt{\!\epsilon(2\!-\!\epsilon)\!} \!\right)           \\
                             &\leq \frac1k \log 2^{\epsilon k} - \frac1k S(\sigma_{S^k}) 
                                  + \left( \epsilon + \sqrt{\epsilon(2-\epsilon)} \right)\log(|S||B|) 
                                  + h\left( \epsilon + \sqrt{\epsilon(2-\epsilon)} \right)           \\
                             &\leq -\frac1k S(\sigma_{S^k}) 
                                  + \left( 2\epsilon + \sqrt{\epsilon(2-\epsilon)} \right)\log(|S||B|) 
                                  + h\left( \epsilon + \sqrt{\epsilon(2-\epsilon)} \right).
    \end{split}\]
\end{enumerate}

Since in both cases we knew the conditional entropy to be always
$\geq - \frac1k \min\left\{ S(\sigma_{S^k}),k S\bigl(\tau(\und{a})\bigr) \right\}$,
this concludes the proof.
\end{proof}

\medskip
\begin{lemma}[Quantum state merging \cite{Horodecki2007,SM_nature}]
\label{lemma:marge}
Given a pure product state 
$\Psi_{A^nB^nC^n}=(\Psi_1)_{A_1B_1C_1}\ox\cdots\ox(\Psi_n)_{A_nB_nC_n}$,
such that $S(\Psi_{A^n})-S(\Psi_{B^n}) \geq \epsilon n$, consider a Haar 
random rank-one projector $P$ on $C^n$. Then, for sufficiently large $n$ 
it holds except with arbitrarily small probability that the post-measurement state 
\[
  \psi(P)_{A^nB^n} = \frac{1}{\tr\Psi_{C^n}P}\tr_{C^n}\Psi(\1_{A^nB^n}\ox P)
\]
satisfies $\frac12\| \psi(P)-\Psi_{A^nB^n}\|_1 \leq \epsilon$.
\hfill$\blacksquare$
\end{lemma}

\medskip
\begin{remark}\normalfont
While we have seen that the upper boundary of the extended phase diagram 
$\overline{\cP}^{(k)}_{|S(\sigma_{S^k})}$ is exactly realized by points 
in $\cP^{(k)}_{|\sigma_{S^k}}$, namely those corresponding to the tensor product 
states $\sigma_{S^k} \ox \tau(\und{a})_B^{\ox k}$, 
it seems unlikely that we can achieve the analogous thing for the lower boundary: 
this would entail finding, for every (sufficiently large) $k$ a tensor product
state, or a block tensor product state, $\xi_{S^kB^k}$ with prescribed charge vector 
$\und{a}$ on $B^k$, and $S(B^k|S^k)_\xi = -\min\{kS\bigl(\tau(\und{a})\bigr),S(\sigma_{S^k})\}$. 

Now, for concreteness, consider the case that 
$kS\bigl(\tau(\und{a})\bigr) \leq S(\sigma_{S^k})$, so that the conditional entropy 
aimed for is $S(B^k|S^k)_\xi = -k S\bigl(\tau(\und{a})_B\bigr)$, which is the value 
of a purification of $\tau(\und{a})_B^{\ox k}$. In particular, it would mean that 
$S(\xi_{B^k}) = k S\bigl(\tau(\und{a})_B\bigr)$, and so -- recalling the
charge values and the maximum entropy principle -- it would follow that 
$\xi_{B^k} = \tau(\und{a})_B^{\ox k}$. 
However, from the equality conditions in strong subadditivity \cite{Hayden2004}, 
this in turn would imply that $\xi_{S^kB^k}$ is a probabilistic mixture of 
purifications of $\tau(\und{a})_B^{\ox k}$ whose restrictions to $S^k$ 
are pairwise orthogonal. This would clearly put constraints on the spectrum 
of $\sigma_{S^k}$ that are not generally met. 

In the other case that $kS\bigl(\tau(\und{a})\bigr) > S(\sigma_{S^k})$, the
conditional entropy should be $S(B^k|S^k)_\xi = - S(\sigma_{S^k})$, and 
since $\xi_{S^k} = \sigma_{S^k}$, this would necessitate a pure state 
$\xi_{S^kB^k}$. Looking at the proof of Lemma \ref{lemma:extended-phase-diagram}, 
however, we see that it leaves quite a bit of manoeuvring space, so it
may or may not be possible to satisfy all charge constraints
$\tr \xi_{B^k}A_j^{(B^k)} = a_j$ ($j=1,\ldots,c$).
\end{remark}

\medskip
Coming back to our question, if a work transformation 
$\rho_{S^n} \ox \tau(\und{\beta})_B^{\ox n} \rightarrow \sigma_{S^nB^n}$ is
feasible for regular sequences on the left hand side, 
by the first law this implies that 
\begin{align*}
  s(\{\sigma_{S^nB^n}\}) &= s(\{\rho_{S^n}\}) + S(\tau(\und{\beta})) \text{ and} \\ 
  W_j                    &= -\Delta A_{S_j}-\Delta A_{B_j}  \\
                         &= a_j(\{\rho_{S^n}\}) - a_j(\{\sigma_{S^n}\})
                            + a_j(\{\tau(\und{\beta})_{B^n}\}) - a_j(\{\sigma_{B^n}\}). 
\end{align*}
When $\sigma_{S^n}$ and the $W_j$ are given, this constrains the possible 
states $\sigma_{S^nB^n}$ as follows: for each $n$, 
\begin{align*}
  \frac1n S(B^n|S^n)_\sigma              &\approx S(\tau(\und{\beta})) - \Delta s_S, \\
  \frac1n \Tr \sigma_{B^n} A_{B_j}^{(n)} &\approx \Tr \tau(\und{\beta})_B A_{B_j} 
                                                  - \Delta A_{S_j} - W_j,\quad \text{for all } j=1,\ldots,c.
\end{align*}
Since by Lemma \ref{lemma:extended-phase-diagram} the left hand sides converge to the 
components of a point in $\overline{\cP}_{|s(\{\sigma_{S^n}\})}$, meaning that a necessary 
condition for the feasibility of the work transformation in question is that 
\begin{equation}\begin{split}
  \label{eq:conditional-point}
  (\und{a},t) \in \overline{\cP}_{|s(\{\sigma_{S^n}\})}, \text{ with } 
  a_j         &:=  \Tr \tau(\und{\beta})_B A_{B_j} - \Delta A_{S_j} - W_j, \\
  t           &:=  S(\tau(\und{\beta})) - \Delta s_S.
\end{split}\end{equation}
Again by Lemma \ref{lemma:extended-phase-diagram}, this is equivalent to 
all $a_j$ to be contained in the set of joint quantum expectations of the 
observables $A_{B_j}$, and 
\[
  -\min\left\{ s(\{\sigma_{S^n}\}),S\bigl(\tau(\und{a})\bigr) \right\} \leq t \leq S\bigl(\tau(\und{a})\bigr).
\]
The following theorem shows that this is also sufficient, when we allow blockings 
of the asymptotically many systems. 

\begin{figure}[ht]
  \begin{center}
    \includegraphics[width=10cm,height=8cm]{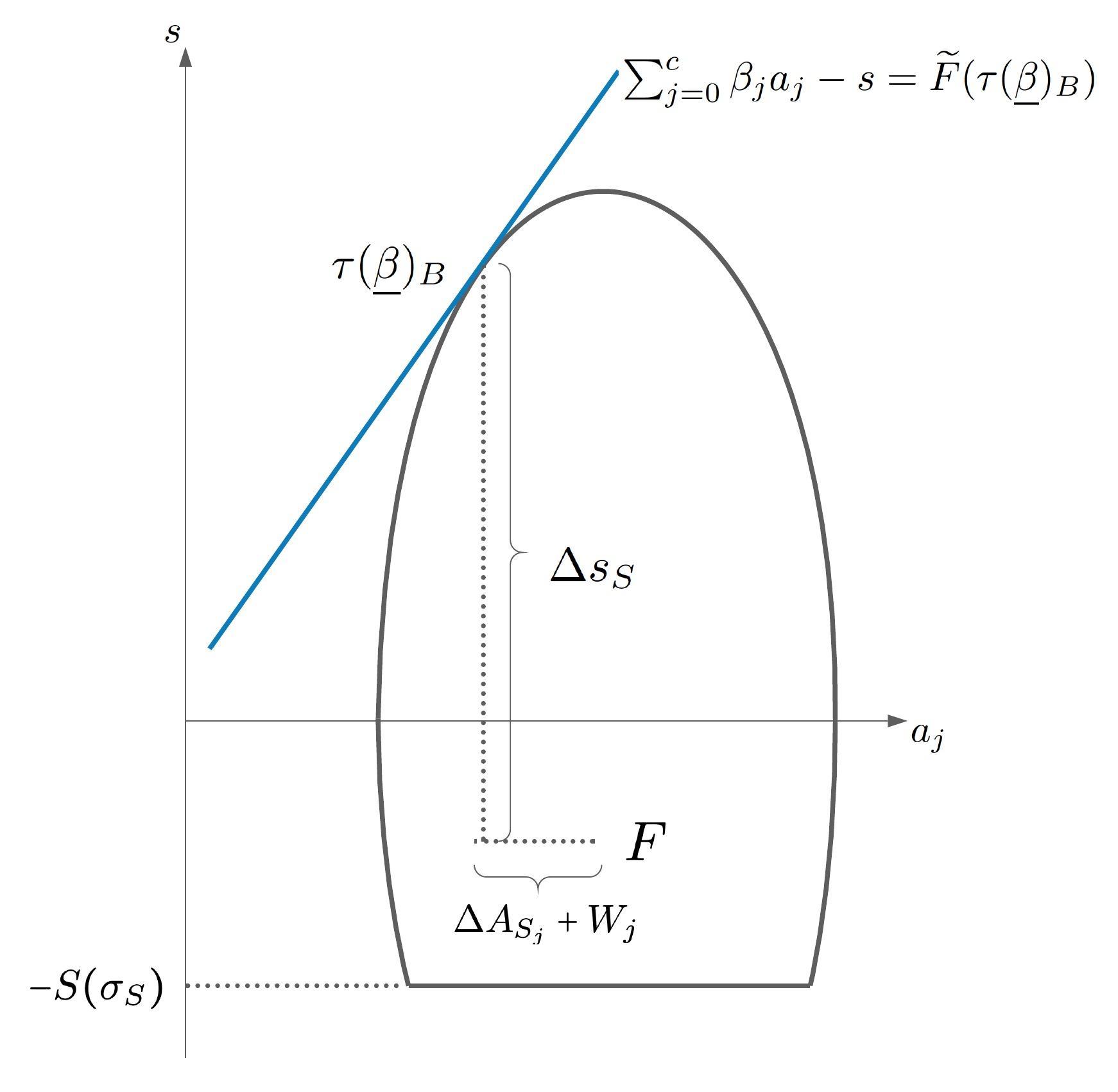}
  \end{center}
  \caption{State change of the bath for a given work transformation under the extraction of
           $j$-type work $W_j$, viewed in the extended phase diagram of the bath, which 
           initially is in the thermal state $\tau(\und{\beta})_B$, the blue line at the 
           corresponding point in the diagram representing the tangent hyperplane of the 
           diagram. The final states $\{\sigma_{S^nB^n}\}$ give rise to the point $F$ in 
           the extended diagram, whose charge values are those of $\{\sigma_{B^n}\}$, 
           while the entropy is $\frac1n S(B^n|S^n)_\sigma$.}
  \label{fig:second-law-finite}
\end{figure} 

\begin{theorem}[Second Law with fixed bath]
\label{thm:second-law-finite-bath}
For arbitrary regular sequences $\rho_{S^n}$ and $\sigma_{S^n}$ of product states, 
a given bath $B$, and any real numbers $W_j$, 
if there exists a regular sequence of block product states 
$\sigma_{S^nB^n}$ with $\Tr_{B^n}\sigma_{S^nB^n} = \sigma_{S^n}$, such that 
there is a work transformation 
$\rho_{S^n} \ox \tau(\und{\beta})_B^{\ox n} \rightarrow \sigma_{S^nB^n}$
with accompanying extraction of $j$-type work at rate $W_j$,
then Eq. (\ref{eq:conditional-point}) defines a point
$(\und{a},t) \in \overline{\cP}_{|s(\{\sigma_{S^n}\})}$.

Conversely, assuming additionally that $\sigma_{S^n} = \sigma_S^{\otimes n}$ is an i.i.d.~state,
if Eq. (\ref{eq:conditional-point}) defines a point
$(\und{a},t) \in \overline{\cP}_{|S(\sigma_S)}^0$ in the interior of the 
extended phase diagram, then for every $\epsilon>0$
there is a work transformation 
$\rho_{S^n} \ox \tau(\und{\beta})_B^{\ox n} \rightarrow \sigma_{S^nB^n}$
with block product states $\sigma_{S^nB^n}$ such that 
$\Tr_{B^n}\sigma_{S^nB^n} = \sigma_{S^n}$, and with accompanying 
extraction of $j$-type work at rate $W_j\pm\epsilon$.
This is illustrated in Fig.~\ref{fig:second-law-finite}.
\end{theorem}

\begin{proof}
We have already argued the necessity of the condition. It remains to show its 
sufficiency. 
Using Lemma \ref{lemma:extended-phase-diagram}, this is not hard: Namely, 
by its point 3, for sufficiently large $k$, $(\und{a},t) \in \overline{\cP}_{|s}$ 
is $\epsilon$-approximated by $\frac1k \cP^{(k)}_{|\sigma_S^{\ox k}}$, i.e. there 
exists a $\sigma_{S^kB^k}$ with $\tr_{B^k} \sigma_{S^kB^k} = \sigma_S^{\ox k}$ 
with $\frac1k S(B^k|S^k)_\sigma \leq t-\epsilon$ and $\frac1k \tr \sigma_{B^k} A_j^{(B^k)} \approx a_j$
for all $j=1,\ldots,c$. By mixing $\sigma$ with a small fraction of 
$\bigl(\tau(\und{a})_B\ox\sigma_S\bigr)^{\ox k}$, we can in fact assume that 
$\frac1k S(B^k|S^k)_\sigma = t$ while preserving $\frac1k \tr \sigma_{B^k} A_j^{(B^k)} \approx a_j$.
Now our target block product states will be 
$\sigma_{S^nB^n} := \bigl(\sigma_{S^kB^k}\bigr)^{\ox \frac{n}{k}}$ for $n$ a multiple of $k$.
By construction, this sequence has the same entropy rate as 
the initial regular sequence of product states $\rho_{S^n} \ox \tau(\und{\beta})_B^{\ox n}$, 
so by the first law, Theorem \ref{thm:first-law}, and the 
AET, Theorem \ref{Asymptotic equivalence theorem}, there is indeed a 
corresponding work transformation with $j$-type work extracted 
equal to $W_j\pm\epsilon$.
\end{proof}

\medskip
\begin{remark}\normalfont
One might object that tensor power target states are not general enough 
in Theorem \ref{thm:second-law-finite-bath}, as we 
had observed in the previous chapter that such states do not
generate the full phase diagram $\overline{\mathcal{P}}$ of the system $S$. 
However, by considering blocks of $\ell$ systems $S^\ell$, we can 
apply the theorem to block tensor power target states 
$\sigma_{S^n} = \bigl(\sigma_1\ox\cdots\ox\sigma_\ell\bigr)^{\ox \frac{n}{\ell}}$, 
and these latter are in fact a rich enough class to exhaust the entire 
phase diagram $\overline{P}$, when $\ell \geq \dim S$ 
(point 5 of Lemma \ref{lemma:phase diagram properties}).

More generally, we can allow as target \emph{uniformly regular} sequences of 
product states $\sigma_{S^n}$, by which we mean the following strengthening 
of the condition in Definition \ref{definition:regular}. 
Denoting $B_{N+1}^{N+n} := B_{N+1}\ldots B_{N+n}$, we require that for all 
$\epsilon > 0$ and uniformly for all $N$, it holds that for sufficiently large $n$, 
\[
  \left| a_j - \frac1n \Tr \sigma_{B_{N+1}^{N+n}} A_j^{(n)} \right| 
                                                       \leq \epsilon \text{ for all } j=1,\ldots,c, \text{ and } 
  \left| s - \frac1n S(\sigma_{B_{N+1}^{N+n}}) \right| \leq \epsilon. 
\]
\end{remark}

\section{Tradeoff between thermal bath rate and work extraction}  
\label{subsec:bath-rate}
Here we consider a different take on the question of the work deficit due
to finiteness of the bath. Namely, we still consider a given fixed finite 
bath system $B$, but now as which state transformations and associated 
generalized works are possible when for each copy of the subsystem $S$,
$R\geq 0$ copies of $B$ are present. It is clear what that means when 
$R$ is an integer, but below we shall give a meaning to this rate as a real number. 
We start off with the observation that ``large enough bath'' in 
Theorem \ref{asymptotic second law} can be taken to mean $B^R$, for the given 
elementary bath $B$ and sufficiently large integer $R$.

\begin{theorem}
\label{thm:large-R-2nd-law}
For arbitrary regular sequences of product states, 
$\rho_{S^n}$ and $\sigma_{S^n}$, and any real numbers $W_j$ with  
$\sum_{j=1}^c \beta_j W_j < -\Delta\widetilde{F}_S$, 
there exists an integer $R\geq 0$ and a regular sequence of product states 
$\sigma_{S^nB^{nR}}$ with $\Tr_{B^{nR}}\sigma_{S^nB^{nR}} = \sigma_{S^n}$, such that 
there is a work transformation 
$\rho_{S^n} \ox \tau(\und{\beta})_B^{\ox nR} \rightarrow \sigma_{S^nB^{nR}}$
with accompanying extraction of $j$-type work at rate $W_j$.
\end{theorem}

\begin{proof}
This was already shown in the achievability part of Theorem \ref{asymptotic second law}.
\end{proof}

\medskip
To give meaning to a rational rate $R = \frac{\ell}{k}$, group the systems of $S^n$,
for $n=\nu k$, into blocks of $k$, which we denote $\widetilde{S}=S^k$, 
and consider $\rho_{S^n} \equiv \rho_{\widetilde{S}^\nu}$ as a $\nu$-party state,
and likewise $\sigma_{S^n} \equiv \sigma_{\widetilde{S}^\nu}$. For each 
$\widetilde{S}=S^k$ we assume $\ell$ copies of the thermal bath, 
$\tau(\und{\beta})_B^{\otimes \ell} = \tau(\und{\beta})_{\widetilde{B}}$,
with $\widetilde{B} = B^\ell$. If $\{\rho_{S^n}\}$ and $\{\sigma_{S^n}\}$ are
regular sequences of product states, then evidently so are 
$\{\rho_{\widetilde{S}^\nu}\}$ and $\{\sigma_{\widetilde{S}^\nu}\}$.

Now, for the given sequences $\{\rho_{S^n}\}$ and $\{\sigma_{S^n}\}$ of initial 
and final states, respectively, as well as works $W_1,\ldots,W_c$ satisfying 
$\sum_j\beta_j W_j = -\Delta\widetilde{F}_S-\delta$, $\delta \geq 0$, 
we can ask what is the infimum over all rates $R = \frac{\ell}{k}$
such that there is a work transformation 
\[
  \rho_{S^n} \ox \tau(\und{\beta})_{B^{nR}} 
  \equiv \rho_{\widetilde{S}^\nu} \ox \tau(\und{\beta})_{\widetilde{B}}^{\ox \nu\ell} 
             \rightarrow 
             \sigma_{\widetilde{S}^\nu \widetilde{B}^{\nu\ell}}
             \equiv \sigma_{S^nB^{nR}},
\]
where as before the final state is intended to satisfy 
$\Tr_{\widetilde{B}^{\nu\ell}} \sigma_{\widetilde{S}^\nu \widetilde{B}^{\nu\ell}} = \sigma_{\widetilde{S}^\nu}$. 

%
%
%
%
%

We observe that if $S(\rho_{S^n})=S(\sigma_{S^n})$ and $\sum_j\beta_j W_j = -\Delta\widetilde{F}_S$,
then the work transformation is possible without 
using any thermal bath, which follows from Eq.~(\ref{work expansion formula}). 
That is, the thermal bath is not necessary for extracting work if the entropy of the work 
system does not change. 
Conversely, the role of the thermal bath is precisely to facilitate changes 
of entropy in the work system.

To answer the above question after the minimum bath rate $R^*$, 
we first show the following lemma.

\begin{lemma}
Consider regular sequences of product states, 
$\rho_{S^n}$ and $\sigma_{S^n}$, and real numbers $W_j$, and assume 
that for large enough rate $R$ there is a work transformation 
$\rho_{S^n} \otimes \tau(\und{\beta})_B^{\ox nR} \rightarrow \sigma_{S^nB^{nR}}$,
with $\sigma_{S^n}$ as the reduced final state on the work system, 
and works $W_1,\ldots,W_c$ are extracted. Then there is another work transformation 
$\rho_{S^n} \otimes \tau(\und{\beta})_B^{\ox nR} \rightarrow \sigma_{S^n}\ox\xi_{B^{nR}}$,
in which the final state of the work system and the thermal bath are uncorrelated,
$\xi_{B^{nR}}$ is a regular sequence of product states, 
and the same works $W_1,\ldots,W_c$ are extracted.
\end{lemma}

\begin{proof}
Assuming that $\rho_{S^n} \otimes \tau(\und{\beta})_B^{\ox nR} \rightarrow \sigma_{S^nB^{nR}}$ is a work transformation, the second law implies that $\sum_j\beta_j W_j = -\Delta\widetilde{F}_s-\delta$ for some $\delta \geq 0$, and we obtain
\begin{equation}\label{eq: coordinates_P_1}
\begin{split}
  s(\{\sigma_{B^{nR}}\})   &= S(\tau(\und{\beta})_B) - \frac1R \Delta s_S+\frac{\delta'}{R}, \\
  a_j(\{\sigma_{B^{nR}}\}) &= \Tr \tau(\und{\beta})_B A_{B_j} - \frac1R (\Delta A_{S_j} + W_j) 
                                                             \quad \text{for all } j=1,\ldots,c.
\end{split}\end{equation}
for $0 \leq \delta' \leq \delta$ where the first equality is due to the fact that $\Delta\widetilde{F}_B+\Delta s_S +\Delta s_B =\delta$ as seen in Eq.~(\ref{work expansion formula}) and positivity of the entropy rate change from Eq.~(\ref{eq: positive Delta_SB}). The second equality follows from the first law, Theorem \ref{thm:first-law}, 
and the AET, Theorem \ref{Asymptotic equivalence theorem}. 
If $R$ is large enough, due to the convexity of the phase diagram of the thermal bath $\overline{\mathcal{P}}_B^{(1)}$,  the following coordinates belong  to the phase diagram as well
\begin{equation}\label{eq: coordinates_P_2}
\begin{split}
  s(\{\xi_{B^{nR}}\})   &= S(\tau(\und{\beta})_B) - \frac1R \Delta s_S, \\
  a_j(\{\xi_{B^{nR}}\}) &= \Tr \tau(\und{\beta})_B A_{B_j} - \frac1R (\Delta A_{S_j} + W_j) 
                                                             \quad \text{for all } j=1,\ldots,c.
\end{split}\end{equation}
Therefore, due to points 3 and 5 of Lemma~\ref{lemma:phase diagram properties}, 
there is a tensor product state $\xi_{B^{nR}}$ with coordinate of 
Eq.~(\ref{eq: coordinates_P_2}) on $\overline{\cP}_B^{(1)}$. Hence the first law, 
Theorem~\ref{thm:first-law}, implies that the desired transformation exists, and
works $W_1,\ldots,W_c$ are extracted.
\end{proof}

\medskip
\begin{theorem}
\label{thm:optimal-rate}
For regular sequences of product states, $\rho_{S^n}$ and $\sigma_{S^n}$, 
and real numbers $W_j$ satisfying $\sum_j\beta_j W_j = -\Delta\widetilde{F}_s-\delta$, 
let $R^*$ be the infimum of rates such that there is a work transformation 
$\rho_{S^n} \otimes \tau(\und{\beta})_B^{\ox nR} \rightarrow \sigma_{S^n}\ox\xi_{B^{nR}}$
under which works $W_1,\ldots,W_c$ are extracted, and 
$\xi_{B^{nR}}$ is a regular sequence of product states. 

Then, this minimum $R^*$ is achieved for a state $\xi_{B^{nR}}$ on the boundary of the 
phase diagram $\overline{\mathcal{P}}_B$ of the thermal bath. Indeed, it is point 
where the line given by Eq.~(\ref{eq:R-bath-assumptions}) intersects the boundary of 
the phase diagram; see Fig.~\ref{fig:optimal-rate}.
Equivalently, it is the smallest $R$ such that the point in 
Eq.~(\ref{eq:R-bath-assumptions}) is contained in $\overline{\mathcal{P}}_B$. 

For $\delta \ll 1$, the minimum rate can be written as
\begin{equation}
  \label{eq:rate-vs-heatcapacity}
  R \approx -\frac{1}{2\delta} \sum_{ij} \frac{\partial\beta_j}{\partial a_i}(\Delta A_{S_i}+W_i)(\Delta A_{S_j}+W_j),
\end{equation}
where $\Delta A_{S_j} = a(\{\sigma_{S^n}\}) - a(\{\rho_{S^n}\})$.
\end{theorem}

\begin{figure}[ht]
  \begin{center}
    \includegraphics[width=10cm,height=8cm]{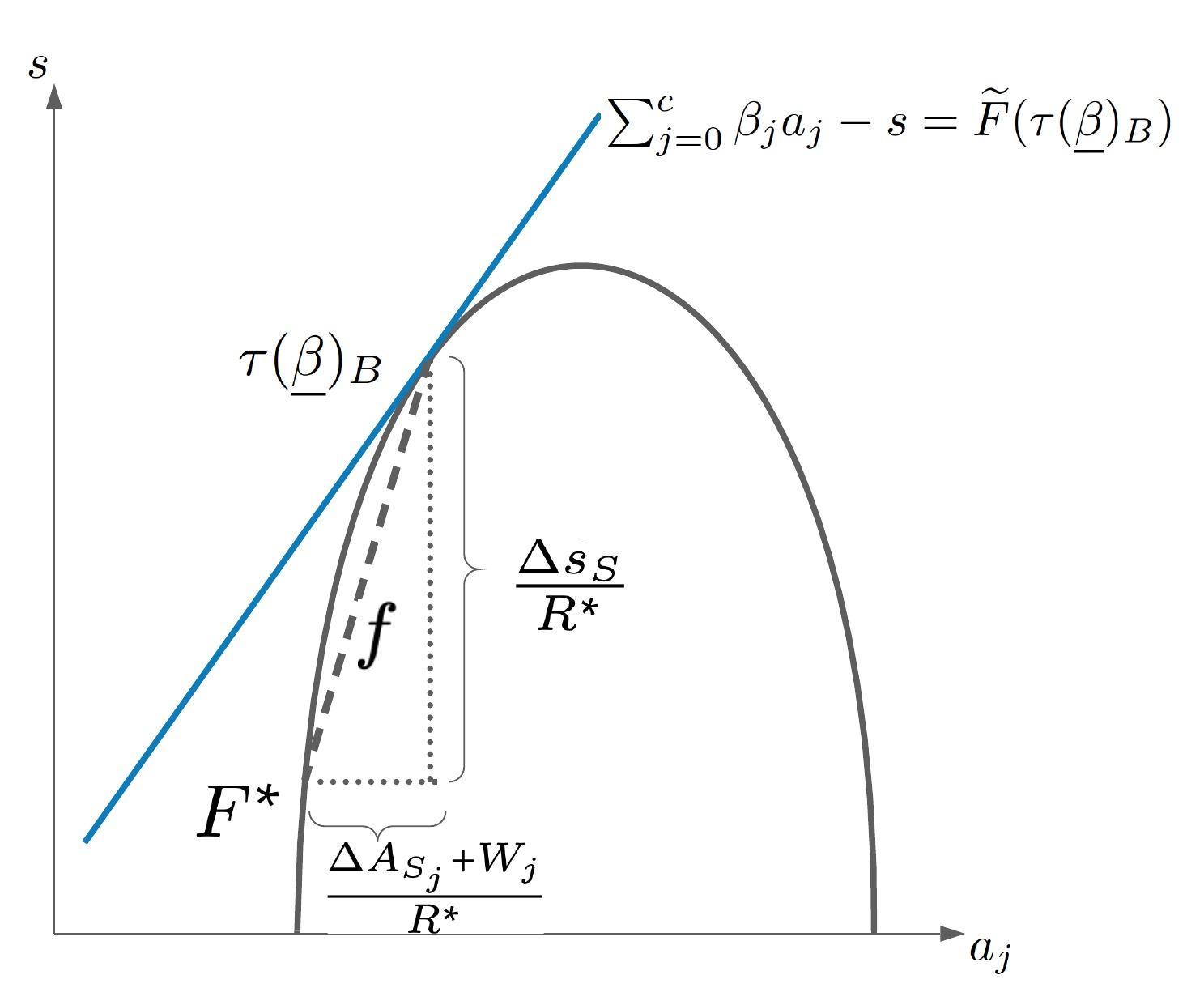}
  \end{center}
  \caption{Graphical illustration of $R^*$, the minimum bath rate for a 
           work transformation $\{\rho_{S^n}\} \rightarrow \{\sigma_{S^n}\}$ 
           satisfying the second law, according to Theorem~\ref{thm:optimal-rate}. 
           The initial state is the generalized thermal state $\tau(\und{\beta})$, its 
           corresponding point marked on the upper boundary of the phase diagram. The final 
           bath states correspond to points on the line denoted $f$, and they 
           are feasible if and only they fall into the phase diagram. 
           Consequently, $F^*$ is the point corresponding to the minimum rate.}
  \label{fig:optimal-rate}
\end{figure}

\begin{proof}
The final state of the thermal bath $\xi_{{B}^{ nR}}$ is a tensor product state, so  
the first law, Theorem~\ref{thm:first-law}, 
and the AET, Theorem~\ref{Asymptotic equivalence theorem} imply that
\begin{equation}\label{eq:coordinates} 
\begin{split}
  s(\{\xi_{B^{nR}}\})   &= S(\tau(\und{\beta})_B) - \frac1R \Delta s_S, \\
  a_j(\{\xi_{B^{nR}}\}) &= \Tr \tau(\und{\beta})_B A_{B_j} - \frac1R (\Delta A_{S_j} + W_j) 
                                                             \quad \text{for all } j=1,\ldots,c,
\end{split}\end{equation}
where $\Delta s_{S} = s(\{\sigma_{S^n}\}) - s(\{\rho_{S^n}\})$. 
Due to point 3 of Lemma~\ref{lemma:phase diagram properties}, 
the above coordinates belong to $\overline{\mathcal{P}}_B^{(1)}$. 
%
For $R=R^*$ assume that the above coordinates belong to the point $(\und{a},s)$ on the boundary 
of the phase diagram $\overline{\mathcal{P}}^{(1)}_B$. Then, for $R>R^*$ the point 
of Eq.~(\ref{eq:coordinates}) is a convex combination of the points 
$(\und{a},s)$ and the corresponding point of the state $\tau(\und{\beta})_B$, 
so it belongs to the phase diagram due to its convexity. Therefore, all points with $R>R^*$ 
are inside the diagram.

To approximate the minimum $R$ for small $\delta$, define the function 
$S(\und{a}):=S(\tau(\und{a})_B)$ for $\und{a}=(a_1,\ldots,a_c)$. 
Its Taylor expansion around the point corresponding to the
initial thermal state $\tau(\und{\beta})_B \equiv S\left(\tau(\und{a}^0)_B\right)$ of the bath
gives the approximation
\begin{equation}
  \label{eq:S-taylor}
  S(\und{a}) \approx S(\und{a}^0) + \sum_j \beta_j (a_j-a_j^0)
                                  + \frac12 \sum_{ij} \frac{\partial\beta_j}{\partial a_i}(a_j-a_j^0) (a_i-a_i^0),
\end{equation}
where we have used the well-know relation $\frac{\partial S}{\partial a_i}=\beta_i$. 
From Eq.~(\ref{eq:coordinates}), we obtain
\begin{align*}
  S(\und{a})-S(\und{a}^0) &= -\frac{\Delta s_S}{R},\\
  a_j-a_j^0               &=  \frac{1}{R} (-\Delta A_{S_j}-W_j),
\end{align*}
and by substituting these values in the Taylor approximation (\ref{eq:S-taylor}), 
using the definition of the free entropy and of the deficit $\delta$, 
we arrive at the claimed Eq. (\ref{eq:rate-vs-heatcapacity}).
\end{proof}

\medskip
\begin{remark}\normalfont
For a single charge, $c=1$, which we traditionally interpret as the internal 
energy $E$ of a system, Eq.~(\ref{eq:rate-vs-heatcapacity}) takes on the very 
simple form
\[
  R \approx -\frac{1}{2\delta} \frac{\partial\beta}{\partial E}(\Delta E_S+W)^2.
\]
Here we can use the usual thermodynamic definitions to rewrite 
$\frac{\partial\beta}{\partial E} = \frac{\partial\frac{1}{T}}{\partial E} = -\frac{1}{T^2}\frac{1}{C}$,
with the heat capacity $C = \frac{\partial E}{\partial T}$, all derivatives taken with
respect to corresponding Gibbs equilibrium states. Thus, 
\begin{equation}
  R \approx \frac{1}{T^2}\frac{1}{C}\cdot\frac{1}{2\delta}(\Delta E_S+W)^2,
\end{equation}
resulting in a clear operational interpretation of the heat capacity in terms 
of the rate of the bath to approach the second law tightly. 

For larger numbers of charges, the matrix 
$\bigl[ \frac{\partial\beta_j}{\partial a_i} \bigr]_{ij} 
  = \bigl[ \frac{\partial^2 S}{\partial a_i\partial a_j} \bigr]_{ij}$ is actually 
the Hessian of the entropy $S\bigl(\tau(\und{a})_B\bigr)$ with respect to the charges, 
and the r.h.s. side of Eq.~(\ref{eq:rate-vs-heatcapacity}) is $\frac{1}{2\delta}$
times the corresponding quadratic form evaluated on the vector
$(\Delta A_{S_1}+W_1,\ldots,\Delta A_{S_c}+W_c)$. Note that by the strict 
concavity of the generalized Gibbs entropy, this is a negative definite 
symmetric matrix, thus explaining the minus sign in Eq. (\ref{eq:rate-vs-heatcapacity}).
In the same vein as the single-parameter discussion before, the Hessian matrix can 
be read as being composed of generalized heat capacities, which likewise receive 
their operational interpretation in terms of the required rate of the bath. 
\end{remark}


\section{Discussion} 
The traditional framework of thermodynamics assumes a system containing an asymptotically large number of particles interacts with an even larger bath. So that all the thermodynamic quantities of interest, e.g., energy, entropy, etc., can be expressed in terms of average or mean values. Also, the notion of temperature there remains meaningful as any exchange of energy hardly drives the bath away from equilibrium as it is considerably large. The quantum thermodynamics attempts to go beyond this assumption. For instance, the system that interacts with a large bath may have a fewer number of quantum particles. In this case, the average quantities are not sufficient to characterize the system as there may be large quantum fluctuations that cannot be ignored. To address this issue, the resource theory of quantum thermodynamics is developed and it shows that the classical laws are not sufficient to characterize the thermodynamic transformations. One rather needs many second laws associated with many one-shot free energies (based on Renyi $\alpha$-relative entropies) \cite{Horodecki2011,Brandao2015}. However, this formalism is still not enough to study the situation where a quantum system interacts with a bath and they are of comparable size. Clearly, the very notion of temperature is questionable as the bath may get driven out of equilibrium after an interaction with the system. To address this, a resource theory is developed based on information conservation \cite{Sparaciari2016, brl19} and it is only applicable to the regime where asymptotically large number system-bath composites are considered. This in turn also allows one to consider the system and bath on the same footing.

Here we have developed a resource theoretic formalism applicable to a more general scenario where a system with multiple conserved quantities (i.e., charges) interacts with a bath, and the system and bath may be of comparable size. These charges may not commute with each other, as allowed by quantum mechanics. The non-commutative nature implies that any (unitary) evolution cannot strictly conserve all these changes simultaneously. We overcome this problem by considering the notion of approximate micro-canonical ensembles, initially developed in \cite{Halpern2016}. This is an essential requirement and forms the basis of the (approximate) first law for thermodynamics with non-commuting charges. With this, we have developed a resource theory for work and heat for thermodynamics with non-commuting charges. We introduce the charge-entropy diagram that conceptually captures all the essential aspects of thermodynamics and an equivalence theorem to show the thermodynamic equivalence between quantum states sharing the same point on the charge-entropy diagram. Then we have derived the second law with the help of the diagram to characterize the state transformations and to quantify the thermodynamics resources such as works corresponding to different charges.  We have also considered the situation where the bath is finite and quantified the rate of state transformations. Interestingly the rate of transformation has been shown to have a direct link with the generalized heat-capacity of the bath. All these then extended to the cases where the systems have (quantum) correlation with the bath. There the charge-entropy diagram has been expressed in terms of conditional-entropy of the bath which may get negative in presence of entanglement and, using that, the second law has been derived.

\appendix
\chapter{Miscellaneous definitions and facts}
\label{Miscellaneous_Facts}

In this Appendix, we list a number of useful
definitions and facts that we often refer to in 
various chapters.


For an operator $X$, the \emph{trace norm}, the 
\emph{Hilbert-Schmidt norm} and 
the \emph{operator norm} are defined respectively 
in terms of $|X| = \sqrt{X^\dagger X}$:
\begin{align*}
 \|X\|_1        &= \Tr |X|, \\
 \|X\|_2        &= \sqrt{\Tr |X|^2}, \\
 \|X\|_{\infty} &= \lambda_{\max}(|X|),
\end{align*}
where $\lambda_{\max}(X)$ is the largest eigenvalue of $X$.

\begin{lemma}[Cf.~\cite{bhatia97}]
For any operator $X$,
\begin{align}
\|X\|_1 \leq \sqrt{d} \|X\|_2 \leq  d \|X\|_{\infty},  
\end{align}
where $d$ equals the rank of $X$.
\qed
\end{lemma}

\begin{lemma}[Cf.~\cite{bhatia97}]
\label{norm1_trace} 
For any self-adjoint operator $X$,
\begin{align*}
\quad \quad\quad\quad\quad\quad\quad \quad \quad\quad \quad
  \norm{X}_1=\max_{-\1 \leq Q \leq \1}\Tr(QX). 
\quad \quad\quad\quad\quad\quad\quad \quad \quad\quad \quad
\blacksquare
\end{align*}
\end{lemma}

\begin{lemma}[Cf.~\cite{bhatia97}]
\label{T_norm1_inequality}
For \aw{any self-adjoint operator $X$ and any operator $T$,}
\begin{align*}
\quad \quad\quad\quad\quad\quad\quad \quad \quad\quad \quad
  \norm{TXT^{\dagger}}_1 \leq \norm{T}_{\infty}^2 \norm{X}_1.
\quad \quad\quad\quad\quad\quad\quad \quad \quad\quad \quad
\blacksquare
\end{align*}
\end{lemma}

\begin{lemma}[Cf.~Bhatia~\cite{bhatia97}]
\label{lemma:norm inequality}
For operators $A$, $B$ and $C$ and for any norm $p \in [1,\infty]$ the following holds 
\begin{align*}
    \norm{ABC}_p \leq \norm{A}_{\infty} \norm{B}_p \norm{C}_{\infty}.
\end{align*}
\end{lemma}

\begin{lemma}[Hoeffding's inequality, Cf.~\cite{DemboZeitouni}]
\label{Hoeffding's inequality}
Let $X_1, X_2,\ldots,X_n$ be independent random variables with $a_i \leq X_i \leq b_i$. Define the empirical mean of these variables as $\overline{X}=\frac{X_1+\ldots+X_n}{n}$, then for any $t>0$
\begin{align*}
  \Pr\left\{\overline{X}-\mathbb{E}(\overline{X}) \geq  t\right\} 
                  &\leq \exp\left(-\frac{2n^2t^2}{\sum_{i=1}^n(b_i-a_i)^2}\right),\\
  \Pr\left\{\overline{X}-\mathbb{E}(\overline{X}) \leq -t\right\} 
                  &\leq \exp\left(-\frac{2n^2t^2}{\sum_{i=1}^n(b_i-a_i)^2}\right).
\end{align*}
\end{lemma}

The fidelity of two states is defined as
\begin{align*}
F(\rho,\sigma)= \Tr \sqrt{\sigma^{\frac{1}{2}} \rho \sigma^{\frac{1}{2}}}.
\end{align*}
When one of the arguments is pure, then 
\begin{align*}
  F(\rho,\ketbra{\psi}{\psi})
     =\sqrt{\Tr (\rho \ketbra{\psi}{\psi})}
     =\sqrt{\bra{\psi}\rho\ket{\psi}}.
\end{align*}

\begin{lemma}
\label{lemma:FvdG}
The fidelity is related to the trace norm as follows \cite{Fuchs1999}:
\begin{align*}
  1- F(\rho,\sigma) \leq \frac{1}{2}\|\rho-\sigma\|_1 \leq \sqrt{1-F(\rho,\sigma)^2} =: P(\rho,\sigma),
\end{align*}
where $P(\rho,\sigma)$ is the so-called purified distance,
or Bhattacharya distance, between quantum states.
\qed
\end{lemma}

\begin{lemma}[{Pinsker's inequality, cf.~\cite{Schumacher2002}}]
\label{lemma:Pinsker}
The trace norm and relative entropy are related by 
\begin{align*}
  \|\rho-\sigma\|_1 \leq \sqrt{2 \ln 2 S(\rho\|\sigma)}. 
\end{align*} \qed
\end{lemma}

\begin{lemma}[{Uhlmann~\cite{UHLMANN1976}}]
\label{lemma:Uhlmann1}
Let $\rho^A$ and $\sigma^A$ be two quantum states with fidelity $F(\rho^A,\sigma^A)$. Let $\rho^{AB}$ and $\sigma^{AC}$ be purifications of these two states, then there exists an isometry $V:{B\to C} $ such that  
\begin{align*}
 \phantom{========:}
 F\left( (\1_A \otimes V^{B\to C})\rho^{AB}(\1_A \otimes V^{B\to C})^{\dagger},\sigma^{AC} \right)  
= F(\rho^A,\sigma^A). 
\phantom{========:} \blacksquare
\end{align*}
\end{lemma}

A consequence of this, due to \cite[Lemma~2.2]{Devetak2008_1}, is as follows.
\begin{lemma}
\label{lemma:Uhlmann2}
Let $\rho^A$ and $\sigma^A$ be two quantum states with trace distance 
$\frac12 \|\rho^A-\sigma^A\|_1 \leq \epsilon$, and
let $\rho^{AB}$ and $\sigma^{AC}$ be purifications of these two states.
Then there exists an isometry $V:{B\to C}$ such that  
\begin{align*}
  \phantom{=========:}
  \left\| (\1_A \otimes V^{B\to C})\rho^{AB} (\1_A \otimes V^{B\to C})^{\dagger}
                                   \! \! -\! \sigma^{AC} \right\|_1 \!  
                                   \leq\! \sqrt{\epsilon(2-\epsilon)} \,.
  \phantom{=========} \blacksquare
\end{align*}
\end{lemma}

\begin{lemma}[{Fannes~\cite{Fannes1973}; Audenaert~\cite{Audenaert2007}}]
\label{Fannes-Audenaert inequality}
Let $\rho$ and $\sigma$ be two states on Hilbert space 
$A$ with trace distance 
$\frac12\|\rho-\sigma\|_1 \leq \epsilon$, then
\begin{align*}
  |S(\rho)-S(\sigma)| \leq \epsilon\log |A| + h(\epsilon),
\end{align*}
where $h(\epsilon)=-\epsilon \log \epsilon -(1-\epsilon)\log (1-\epsilon)$ \aw{is the binary entropy}.
\end{lemma}

\aw{There is also an extension} of the Fannes inequality for the conditional entropy; 
this lemma is very useful \aw{especially} when the dimension of the 
system \aw{conditioned on} is unbounded.

\begin{lemma}[{Alicki-Fannes~\cite{Alicki2004}; Winter~\cite{Winter2016}}]
\label{AFW lemma}
Let $\rho$ and $\sigma$ be two states on a bipartite Hilbert space 
$A\otimes B$ with trace distance $\frac12\|\rho-\sigma\|_1 \leq \epsilon$, then
\begin{align*}
  |S(A|B)_{\rho}\!-\!S(A|B)_{\sigma}| \leq 2\epsilon \log |A| \!+\! (1+\epsilon)h\left(\frac{\epsilon}{1+\epsilon}\right)\!.
\end{align*}\qed
\end{lemma}

\begin{lemma}
\label{full_support_lemma}
Let $\rho$ be a state with full support on the Hilbert space \aw{$A$}, i.e.~it 
\aw{has positive minimum} eigenvalue $\lambda_{\min}$, and
let $\ket{\psi}^{AR}$ be a purification of $\rho$ on the Hilbert space \aw{$A \otimes R$}. 
Then any purification of another state $\sigma$ on \aw{$A$} is of the form 
\begin{align*}
  (\1_A \otimes T) \ket{\psi}^{AR},
\end{align*}
where $T$ is an operator acting on system $R$ with $\| T \|_{\infty} \le \frac{1}{\sqrt{\lambda_{\min}}}$.
\begin{proof}
Let $\rho=\sum_i \lambda_i \ketbra{e_i}{e_i}$ and $\sigma=\sum_j \mu_j \ketbra{f_j}{f_j}$ be spectral decompositions  of the states. The purification of $\rho$  is  $\ket{\psi}^{AR}=\sum_i \sqrt{\lambda_i} \ket{e_i} \ket{i}$. Define $\ket{\phi}^{AR}=\sum_j \sqrt{\mu_j} \ket{f_j} \ket{j}$.
Any purification of the state $\sigma$ is of the form  
$(\1_A \otimes V) \ket{\phi}^{AR}$ where $V$ is an isometry acting on system $R$. 
Write the eigenbasis $\set{\ket{f_j}}$ as linear combination of eigenbasis $\set{\ket{e_j}}$, that is, $\ket{f_j}=\sum_i\alpha_{ij}\ket{e_i}$. Then, we have $\ket{\phi}^{AR}=\sum_{i,j} \sqrt{\mu_j} \alpha_{ij} \ket{e_i} \ket{j}$. Define the operator  $P=\sum_{jk} p_{jk} \ketbra{j}{k}$ where $p_{jk}=\alpha_{kj} \sqrt{\frac{\mu_j}{\lambda_k}}$. It is immediate to see that
\begin{align*}
    \ket{\phi}^{AR}=(\1_A \otimes P) \ket{\psi}^{AR}. 
\end{align*}
Thus, we have $(\1_A \otimes V) \ket{\phi}^{AR} = (\1_A \otimes VP) \ket{\psi}^{AR}$. 
Defining $T=VP$, we then have 
\begin{align*}
\lambda_{\max} (T^{\dagger}T)&=\lambda_{\max} (P^{\dagger}P)\\
   &\leq \Tr (P^{\dagger}P) \\
 &=\sum_{j,k}|p_{jk}|^2 \\ 
 &= \sum_{j,k}\frac{|\alpha_{kj}|^2\mu_j}{\lambda_k}  \\
 &\leq \frac{1}{\lambda_{\min}},
\end{align*}
where the last inequality follows \aw{from the orthonormality of} the basis $\set{\ket{f_j}}$. 
\end{proof}
\end{lemma}

\begin{lemma}[Gentle Operator Lemma \cite{winter1999_2,Ogawa2007,wilde_2013}]
\label{Gentle Operator Lemma}
If a quantum state $\rho$ with diagonalization $\rho=\sum_j p_j \pi_j$ projects onto operator $\Lambda$ with probability  $1- \epsilon$, which is bounded as $0 \leq \Lambda \leq I$, i.e. $\Tr(\rho \Lambda) \geq 1- \epsilon$ then
\begin{align*}
    \sum_j p_j\norm{\pi_j-\sqrt{\Lambda}\pi_j \sqrt{\Lambda}}_1 \leq 2\sqrt{\epsilon}.
\end{align*}
\end{lemma}

\begin{definition}
Let $\rho_1,\ldots,\rho_n$ be quantum states on a $d$-dimensional Hilbert space $\mathcal{H}$ with diagonalizations $\rho_i=\sum_j p_{ij} \pi_{ij}$ and one-dimensional projectors $\pi_{ij}$. For $\alpha >0$ and $\rho^n=\rho_1 \otimes \cdots \otimes \rho_n$ define the set of entropy typical sequences as
\begin{align} 
    \mathcal{T}_{\alpha,\rho^n }^n=\left\{j^n=j_1 j_2 \ldots j_n : \abs{\sum_{i=1}^n -\log p_{ i j_i}-S(\rho_i) } \leq \alpha \sqrt{n}\right\}. \nonumber
\end{align}
Define the entropy typical projector of $\rho^n$ with constant $\alpha$ as
\begin{align*}
    \Pi^n_{\alpha ,\rho^n }=\sum_{j^n \in \mathcal{T}_{\alpha, \rho^n}^n} \pi_{1j_1} \otimes \cdots \otimes \pi_{nj_n}.
\end{align*}
\end{definition}

\begin{lemma}\label{lemma:typicality properties }
(Cf. \cite{csiszar_korner_2011}) There is a constant $0<\beta \leq \max \set{(\log 3)^2,(\log d)^2}$ such that the entropy typical projector has the following properties for any $\alpha >0$, $n>0$ and arbitrary state $\rho^n=\rho_1 \otimes \cdots \otimes \rho_n$:
\begin{align*}
  \Tr \left(\rho^n \Pi^n_{\alpha ,\rho^n }\right) &\geq 1-\frac{\beta}{\alpha^2}, \\
  2^{-\sum_{i=1}^n S(\rho_i)-\alpha \sqrt{n}}  \Pi^n_{\alpha,\rho^n} 
                                &\leq \Pi^n_{\alpha,\rho^n}\rho^n \Pi^n_{\alpha,\rho^n} 
                                 \leq 2^{-\sum_{i=1}^n S(\rho_i)+\alpha \sqrt{n}}\Pi^n_{\alpha,\rho^n},\quad \text{and} \\
  \left(1-\frac{\beta}{\alpha^2}\right) 2^{ \sum_{i=1}^n S(\rho_i)-\alpha \sqrt{n}} 
                                &\leq  \Tr \left(\Pi^n_{\alpha ,\rho^n }\right) 
                                 \leq  2^{\sum_{i=1}^n S(\rho_i)+\alpha \sqrt{n}}.
\end{align*}
\end{lemma}

\bibliography{References}

\printindex

\end{document}